\documentclass[12pt]{article}
\usepackage[a4paper, margin=0.75in]{geometry}
\usepackage[utf8]{inputenc}         
\usepackage[english]{babel}
\usepackage{lmodern}
\usepackage{graphicx}

\usepackage[colorinlistoftodos]{todonotes}

\usepackage[colorlinks=true]{hyperref}

\usepackage[dvipsnames]{xcolor}

\usepackage[]{color-edits}
\addauthor{UF}{blue}

\addauthor{VG}{red}
\newcommand{\vgc}[1]{{\VGcomment{#1}}}
\newcommand{\vge}[1]{{\VGedit{#1}}}

\usepackage{amsthm}
\usepackage{amsmath}
\usepackage{amssymb}
\usepackage{amsfonts}
\usepackage{mathrsfs}
\usepackage{mathtools}
\usepackage{verbatim}
\usepackage{footnote}
\usepackage{lineno}
\usepackage{caption}
\usepackage{tikzsymbols}
\usepackage[capitalize, nameinlink]{cleveref}
\usepackage{tikz}
\usetikzlibrary{quantikz2}

\usepackage{icomma}
\usepackage{enumitem}
\usepackage{array}
\usepackage{multirow}
\usepackage{setspace}
\usepackage{euscript}
\usepackage{indentfirst}
\usepackage{epigraph}
\usepackage{fancybox, fancyhdr}
\usepackage{titlesec}
\usepackage{dsfont}
\usepackage{csquotes}

\usepackage{algorithm}              
\usepackage{algorithmicx}           
\usepackage[noend]{algpseudocode}   
\usepackage{listings}

\floatname{algorithm}{Algorithm}

\newtheorem{definition}{Definition}[section]

\newtheorem{theorem}{Theorem}[section]

\newtheorem{proposition}{Proposition}[section]

\newtheorem{corollary}{Corollary}[section]
\newtheorem{lemma}{Lemma}[section]
\newtheorem{observation}{Observation}[section]

\newtheorem{claim}{Claim}[section]

\newtheorem{remark}{Remark}[section]

\crefname{conjecture}{\textbf{Conjecture}}{}
\crefname{lemma}{\textbf{Lemma}}{}
\crefname{claim}{\textbf{Claim}}{}
\crefname{theorem}{\textbf{Theorem}}{}
\crefname{corollary}{\textbf{Corollary}}{}
\crefname{observation}{\textbf{Observation}}{}
\crefname{proposition}{\textbf{Proposition}}{}
\crefname{assumption}{\textbf{A.}}{}

\renewcommand{\S}{\mathbb{S}}

\usepackage{xparse}

\DeclarePairedDelimiterX\brackets[1]{(}{)}{
	
	#1
}
\DeclarePairedDelimiterX\squarebrackets[1]{[}{]}{
	
	#1
}
\DeclarePairedDelimiterX\figbrackets[1]{\{}{\}}{
	
	#1
}

\DeclarePairedDelimiterX\ntrs[1]{\langle}{\rangle}{
	
	#1
}
\DeclarePairedDelimiterX\kets[1]{\lvert}{\rangle}{
	
	#1
}
\DeclarePairedDelimiterX\bras[1]{\langle}{\rvert}{
	
	#1
}
\DeclarePairedDelimiterX\brakets[2]{\langle}{\rangle}{
	
	#1
    \delimsize\vert
    
    #2
}

\DeclarePairedDelimiterX\dprs[2]{\langle}{\rangle}{
	
	#1
    ,
    
    #2
}

\DeclarePairedDelimiterX\ketbras[2]{\lvert}{\rvert}{
	
	#1
    \delimsize\rangle\delimsize\langle
    
    #2
}
\DeclarePairedDelimiterX\qvals[3]{\langle}{\rangle}{
	
	#1
    \delimsize\vert
    #2
    \delimsize\vert
    
    #3
}
\DeclarePairedDelimiterX\qexps[2]{\langle}{\rangle}{
	
	#1
    \delimsize\vert
    #2
    \delimsize\vert
    
    #1
}

\let\P\undefined
\let\o\undefined
\let\O\undefined

\NewDocumentCommand{\P}{o g}{
	\IfNoValueTF{#1}{
		\IfNoValueTF{#2}{
			\mathop{\mathds{P}}
		}{
			\mathop{\mathds{P}}\squarebrackets*{#2}
		}
	}{
		\IfNoValueTF{#2}{
			\mathop{\mathds{P}}_{#1}
		}{
			\mathop{\mathds{P}}_{#1}\squarebrackets*{#2}
		}
	}
}

\NewDocumentCommand{\E}{o g}{
	\IfNoValueTF{#1}{
		\IfNoValueTF{#2}{
			\mathop{\mathds{E}}
		}{
			\mathop{\mathds{E}}\squarebrackets*{#2}
		}
	}{
		\IfNoValueTF{#2}{
			\mathop{\mathds{E}}_{#1}
		}{
			\mathop{\mathds{E}}_{#1}\squarebrackets*{#2}
		}
	}
}

\NewDocumentCommand{\D}{o g}{
	\IfNoValueTF{#1}{
		\IfNoValueTF{#2}{
			\mathop{\mathds{D}}
		}{
			\mathop{\mathds{D}}\squarebrackets*{#2}
		}
	}{
		\IfNoValueTF{#2}{
			\mathop{\mathds{D}}_{#1}
		}{
			\mathop{\mathds{D}}_{#1}\squarebrackets*{#2}
		}
	}
}

\NewDocumentCommand{\cov}{o g g}{
	\IfNoValueTF{#1}{
		\IfNoValueTF{#2}{
			\mathop{\mathrm{cov}}
		}{
			\IfNoValueTF{#3}{
				\mathop{\mathrm{cov}}
			}{
				\mathop{\mathrm{cov}}\brackets*{#2, #3}
			}
		}
	}{
		\IfNoValueTF{#2}{
			\mathop{\mathrm{cov}}_{#1}
		}{
			\IfNoValueTF{#3}{
				\mathop{\mathrm{cov}}_{#1}
			}{
				\mathop{\mathrm{cov}}_{#1}\brackets*{#2, #3}
			}
		}
	}
}

\NewDocumentCommand{\I}{g}{
	\IfNoValueTF{#1}{
		\mathop{\mathds{I}}
	}{
		\mathop{\mathds{I}}\figbrackets*{#1}
	}
}

\NewDocumentCommand{\V}{o g}{
	\IfNoValueTF{#1}{
		\IfNoValueTF{#2}{
			\mathop{\mathds{V}}
		}{
			\mathop{\mathds{V}}\squarebrackets*{#2}
		}
	}{
		\IfNoValueTF{#2}{
			\mathop{\mathds{V}}_{#1}
		}{
			\mathop{\mathds{V}}_{#1}\squarebrackets*{#2}
		}
	}
}

\NewDocumentCommand{\Tr}{o g}{
	\IfNoValueTF{#1}{
		\IfNoValueTF{#2}{
			\mathop{\mathrm{Tr}}
		}{
			\mathop{\mathrm{Tr}}\squarebrackets*{#2}
		}
	}{
		\IfNoValueTF{#2}{
			\mathop{\mathrm{Tr}}_{#1}
		}{
			\mathop{\mathrm{Tr}}_{#1}\squarebrackets*{#2}
		}
	}
}

\NewDocumentCommand{\o}{g}{\operatorname{\mathrm{o}}\IfNoValueTF{#1}{}{\brackets*{#1}}}
\NewDocumentCommand{\O}{g}{\operatorname{\mathrm{O}}\IfNoValueTF{#1}{}{\brackets*{#1}}}
\NewDocumentCommand{\Th}{g}{\operatorname{\Theta}\IfNoValueTF{#1}{}{\brackets*{#1}}}
\NewDocumentCommand{\om}{g}{\operatorname{\omega}\IfNoValueTF{#1}{}{\brackets*{#1}}}
\NewDocumentCommand{\Om}{g}{\operatorname{\Omega}\IfNoValueTF{#1}{}{\brackets*{#1}}}

\NewDocumentCommand{\ntr}{g}{\IfNoValueTF{#1}{}{\ntrs*{#1}}}
\let\ket\undefined
\NewDocumentCommand{\ket}{g}{\IfNoValueTF{#1}{}{\kets*{#1}}}
\let\bra\undefined
\NewDocumentCommand{\bra}{g}{\IfNoValueTF{#1}{}{\bras*{#1}}}
\NewDocumentCommand{\dpr}{g g}{\IfNoValueTF{#1}{}{\IfNoValueTF{#2}{}{\dprs*{#1}{#2}}}}
\let\braket\undefined
\NewDocumentCommand{\braket}{g g}{\IfNoValueTF{#1}{}{\IfNoValueTF{#2}{}{\brakets*{#1}{#2}}}}
\NewDocumentCommand{\ketbra}{g g}{\IfNoValueTF{#1}{}{\IfNoValueTF{#2}{}{\ketbras*{#1}{#2}}}}
\NewDocumentCommand{\qval}{g g g}{\IfNoValueTF{#1}{}{\IfNoValueTF{#2}{}{\IfNoValueTF{#3}{}{\qvals*{#1}{#2}{#3}}}}}
\NewDocumentCommand{\qexp}{g g}{\IfNoValueTF{#1}{}{\IfNoValueTF{#2}{}{\qexps*{#1}{#2}}}}

\newcommand{\polylog}{\mathop{\mathrm{polylog}}}

\newcommand{\minimize}{\mathop{\mathrm{minimize}}}
\newcommand{\subto}{\mathop{\mathrm{subject\ to}}}

\newcommand{\const}{\mathrm{const}}

\newcommand{\mms}{\mathrm{MMS}}

\newcommand{\aps}{\mathrm{APS}}

\newcommand{\eps}{\varepsilon}

\newcommand{\Ff}{\mathcal{F}}

\newcommand{\Aa}{\mathcal{A}}
\newcommand{\Bb}{\mathcal{B}}
\newcommand{\Ss}{\mathcal{S}}

\newcommand{\Pp}{\mathcal{P}}

\newcommand{\Mm}{\mathcal{M}}

\newcommand{\Ee}{\mathcal{E}}

\newcommand{\Nn}{\mathcal{N}}

\newcommand{\OPT}{\mathrm{OPT}}

\newcommand{\oW}{\overline{W}}

\newcommand{\hp}{\hat{p}}

\newcommand{\oU}{\overline{U}}
\newcommand{\Hss}{\hat{\mathcal{S}}}
\newcommand{\hSs}{\hat{\mathcal{S}}}
\newcommand{\hv}{\hat{v}}
\newcommand{\hlam}{\hat{\lambda}}
\newcommand{\hS}{\hat{S}}
\newcommand{\hbeta}{\hat{\beta}}

\newcommand{\hgam}{\hat{\gamma}}

\newcommand{\wgamma}{\tilde{\gamma}}

\newcommand{\wv}{\tilde{v}}

\newcommand{\items}{\mathcal{M}} 
\newcommand{\agents}{\mathcal{N}}

\date{}
\title{On MMS, APS and XOS}

\author{Uriel Feige\thanks{Weizmann Institute of Science, Israel. {\tt uriel.feige@weizmann.ac.il}} \and Vadim Grinberg\thanks{Weizmann Institute of Science, Israel. {\tt vadim.grinberg@weizmann.ac.il}}}

\begin{document}

\maketitle
\begin{abstract}
    We consider allocations of a set of $m$ indivisible goods to $n$ agents of equal entitlements that have valuations from the class XOS. A previous sequence of works showed allocations that obtain an $\alpha$-approximation for the maximin share (MMS), for values of $\alpha$ that gradually approach $\frac{1}{4}$ from below (the currently known ratio is $\frac{4}{17}$). In this work we attempt to obtain ratios better than $\frac{1}{4}$, and manage to do so for sufficiently large $n$. Our methodology is to first investigate the gap between the anyprice share (APS) and the MMS when all agents have the same XOS valuations, for which we design an allocation algorithm and prove that each agent receives at least $\alpha > \frac{11}{40}$ times the APS. Then, we derive inspiration from this algorithm, and modify it so that it applies also when agents have different XOS valuations. Using this modified version, we show that for some sufficiently large $n_0$, there is an $\alpha$-MMS allocation (in fact, an $\alpha$-APS allocation) for every $n \geq n_0$.
\end{abstract}

\tableofcontents

\newpage
\section{Introduction}
\label{sec:intro}

We consider fair allocations of $m$ indivisible goods to $n$ agents.
Let $\Mm$ denote the set of $m$ items to allocate, and let the set of agents be $\Nn := [n]$.
Every agent $i \in \Nn$ has a \textit{valuation function} $v_i$ defined on the subsets of $\Mm$.
We assume that $v_i$ satisfies $v_i(\varnothing) = 0$, and for every $S_1 \subseteq S_2$ it holds $v_i(S_1) \leq v_i(S_2)$, i.e $v_i$ is monotone and non-decreasing. 
Every agent $i$ may have different entitlement $b_i > 0$ to the goods $\Mm$, and all entitlements $b_i$ for $i \in [n]$ satisfy $\sum_{i = 1}^nb_i = 1$.
In this work, we shall always have $b_i = 1/n$ for all agents $i \in \Nn$.

An \textit{allocation} of items $\Mm$ to agents $\Nn$ is a partition $(P_1, \ldots, P_n)$ of $\Mm$ into $n$ disjoint sets, where agent $i$ ``receives'' the set $P_i$.
We would like this allocation to be \textit{fair}, for a certain definition of fairness.
In this work, we will focus on a share-based notion of fairness called the \textit{maximin share}, or MMS for short.
\begin{definition}
    Let $\Pp_n$ be the set of all partitions $(P_1, \ldots, P_n)$ of items $\Mm$ into $n$ disjoint sets.
    For an agent $i$ with valuation $v_i$ and entitlement $b_i = \frac{1}{n}$, her maximin share is
    \[
        \mms(\Mm, v_i, 1/n) = \max_{(P_1, \ldots, P_n)\in \Pp_n}\min_{j \in [n]}v_i(P_j).
    \]
\end{definition}
This notion of fairness was first considered in \cite{Budish11}.
A partition $(P_1, \ldots, P_n)$ of $\Mm$ is called an \textit{MMS-allocation} 
if for every agent $i \in \Nn$ it holds $v_i(P_i)\geq \mms(\Mm, v_i, 1/n)$.
That is, each agent $i$ values the bundle $P_i$ given to her at least as high as her maximin share.
While in general one may study MMS allocations for general monotone valuations $v_i$, it is common practice to classify $v_i$ according to the \textit{complement free hierarchy}, first introduced in \cite{LLN}.
\begin{definition}
    Let $v$ be an arbitrary monotone non-decreasing valuation, satisfying $v(\varnothing) = 0$.
    \begin{itemize}
        \item $v$ is \textbf{additive} if $v(S) = \sum_{e \in S}v(e)$ for every set $S \subseteq \Mm$.

        \item $v$ is \textbf{submodular} if $v(S_1\cup \{e\}) - v(S_1) \geq v(S_2\cup \{e\}) - v(S_2)$ for every item $e \in \Mm$ and sets $S_1\subseteq S_2 \subseteq \Mm$.

        \item $v$ is \textbf{XOS} if there exist $t$ additive valuations $v^1,\ldots, v^t$, such that $v(S) = \max_{j \in [t]}v^j(S)$ for every set $S\subseteq \Mm$.
        This class is equivalent to fractionally subadditive valuations (\cite{Feige09}).

        \item $v$ is \textbf{subadditive} if $v(S_1) + v(S_2) \geq v(S_1\cup S_2)$ for all sets $S_1, S_2 \subseteq \Mm$.
    \end{itemize}
\end{definition}
As was shown in \cite{LLN}, each of these valuation classes is strictly contained in the one that follows it.
In this work, we will consider only valuations $v_i$ that are XOS.

In general, MMS-allocations need not exist, even {if} the valuations are additive  \cite{KPW18-1}.
Because of this, it is reasonable to consider allocations in which every agent receives a fraction of her MMS.
In particular, for value $\alpha \in (0, 1]$, we say that allocation $(P_1, \ldots, P_n)$ is an $\alpha$-MMS allocation, if for every agent $i \in \Nn$ it holds $v_i(P_i) \geq \alpha\cdot \mms(\Mm, v_i, 1/n)$.
{In this work, we} focus on showing the existence of $\alpha$-MMS allocations for XOS agents, for certain values of $\alpha$.
{Below, we give} a brief overview of known results on MMS for XOS agents.
As for works on other classes of valuations, we refer the reader to \cref{sec:related}.

The allocations achieving $\alpha$-MMS for XOS valuations were first considered in \cite{GhodsiHSSY22}, who showed the existence of $\frac{1}{5}$-MMS allocations.
In the same work, the authors proved that for XOS agents it is impossible to guarantee allocations that achieve better than $\frac{1}{2}$-MMS.
Subsequently, the authors of \cite{SS24} improved the approximation constant to $\frac{1}{4.6}$, which was later increased to $\frac{3}{13}$ in \cite{AkramiMSS23}.
The ratio was then improved to $\frac{4}{17}$-MMS in \cite{FG25}, which is the currently known lower bound.

Notably, the existence proofs for the ratios $\frac{1}{5}, \frac{3}{13}$ and $\frac{4}{17}$ of respectively \cite{GhodsiHSSY22}, \cite{AkramiMSS23} and \cite{FG25} are  built on the same approach, first introduced in \cite{GhodsiHSSY22} and extended in subsequent works.
On a high level, the approach is as follows: first give away carefully chosen sets of at most $k$ items of total value at least $\alpha$ to some agents, and remove them from the game, and then proceed to allocate the remaining items among the remaining agents.
Taking $k = 1$ gives $\frac{1}{5}$-MMS of \cite{GhodsiHSSY22}, while taking $k =3, 4$ gives $\frac{3}{13}$-MMS of \cite{AkramiMSS23} and $\frac{4}{17}$-MMS of \cite{FG25} respectively.
A similar procedure for $k =2$ would give us a ratio of $\frac{2}{9}$, suggesting that it may be possible to obtain $\frac{k}{4k + 1}$-MMS allocations for larger values of $k$, using a similar approach.
However, even if there are $\frac{k}{4k + 1}$-MMS allocations for larger $k$, these techniques seem to have a natural barrier of $\frac{1}{4}$.
Therefore, if one wishes to obtain an approximation ratio of $\alpha$-MMS for $\alpha > \frac{1}{4}$, one should seek for an alternative approach to achieving such an allocation.

\subsection{Our results}

We present a completely new approach to finding $\alpha$-MMS allocations for XOS agents, which allows us to improve over the previously known bounds as long as the number $n$ of agents is sufficiently large.
Specifically, we prove the following theorem.
\begin{theorem}\label{mainresultmms}
    There exists a constant $\alpha > \frac{11}{40}$ and number $n_\alpha$, such that for every $n\geq n_\alpha$ the following holds.
    Consider a fair allocation instance with items $\Mm$ and agents $\Nn$, $|\Nn| = n$, where for every $i \in [n]$ valuation $v_i$ is XOS.
    Then, there exists an allocation $(P_1, \ldots, P_n)$ of $\Mm$ to $\Nn$ such that every agent $i \in \Nn$ receives at least $\alpha$-fraction of her MMS, {i.e $v_i(P_i) \geq \alpha\cdot \mms(\Mm, v_i, 1/n)$.}
\end{theorem}
{As $n_\alpha$ grows}, the value of the constant $\alpha$ {converges to} a solution to the equation $2(12\alpha- 3)\ln(3\alpha) = (1 - 3\alpha)(3\ln3 - 4)$, which gives $\alpha \approx 0.2767738$.
{We alert the reader that to achieve $\alpha \geq 1/4$, our proof uses $n_\alpha > 100000$.} 


Interestingly enough, we prove an approximation ratio of $\alpha > \frac{11}{40}$ not just for MMS-allocations, but also for \textit{APS} allocations. {The APS has two equivalent definitions~\cite{BEF21APS}, one based on prices, and the other based on fractional partitions. We present the latter definition, as this is the definition used by our proofs.}

\begin{definition}
\label{def:APS}
    {Let $\Bb \subseteq 2^{\Mm}$ be a collection of subsets of $\Mm$.}
    For $\rho \in (0, 1)$, we call a collection $\{(B, \lambda_B)\}_{B \in \Bb}$ a \textbf{fractional $\rho$-partition}, if 
    \begin{itemize}
        \item for all $B \in \Bb$, $\lambda_B \geq 0$, and $\sum_{B \in \Bb}\lambda_B = 1$;
        \item for every item $e\in \Mm$ it holds $\sum_{B \in \Bb: B\ni e}\lambda_B \leq \rho$.
    \end{itemize}
    Let $\Ff\Pp_\rho$ be all fractional $\rho$-partitions of $\Mm$.
    For agent $i$ with valuation $v_i$ and entitlement $b_i$, her anyprice share (APS) is
    \[
        \aps(\Mm, v_i, b_i) = \max_{\{(B, \lambda_B)\}_{B \in \Bb} \in \Ff\Pp_{b_i}}\min_{B \in \Bb}v_i(B).
    \]
\end{definition}
For agent $i \in [n]$, a fractional $b_i$-partition $\{(B, \lambda_B)\}_{B \in \Bb}$ is called an \textit{APS-partition} for agent $i$, if 
$\min_{B \in \Bb}v_i(B) = \aps(\Mm, v_i, b_i)$.
Similarly, a partition $(P_1, \ldots, P_n)$ is called an \textit{APS-allocation}, if for every agent $i \in [n]$ we have $v_i(P_i) \geq \aps(\Mm, v_i, b_i)$.

In our case, for every agent $i \in [n]$ we have $b_i = 1/n$.
Then, one can think of fractional $\frac{1}{n}$-partitions as follows: instead of a partition $(P_1, \ldots, P_n)$ where each item is contained in exactly one of the sets $P_j$, we consider a \textit{distribution} $\{\lambda_B\}_{B \in \Bb}$ over sets $B \in \Bb$, such that every item $e \in \Mm$ is contained in at most $1/n$-fraction of all bundles {($\frac{1}{n}$-fraction in terms of the weights $\lambda_B$).}
The caveat with APS is that sets $B \in \Bb$ are \textbf{not} necessarily disjoint, and can have arbitrary intersections.
It is easy to see that in the event of $\{B\}_{B \in \Bb}$ all being disjoint, the fractional $\frac{1}{n}$-partition $\{(B, \lambda_B)\}_{B \in \Bb}$ corresponds to some partition $(P_1',\ldots, P_n')$ of $\Mm$.
Therefore, any MMS-partition of agent $i$ with $b_i = 1/n$ is also a {fractional $1/n$-partition}, hence $\aps(\Mm, v_i, 1/n) \geq \mms(\Mm, v_i, 1/n)$ for every $i$.

APS-allocations were first introduced in \cite{BEF21APS}, and in general consider arbitrary entitlements $b_i$ of agents $i \in \Nn$.
When it comes to XOS valuations, \cite{FG25} extend their result from MMS to APS allocations for agents with equal entitlements, thus giving $\frac{4}{17}$-APS allocation, which is currently the best known approximation.
The reader can further explore other results on APS-allocations in \cref{sec:related}.
Our approach for \cref{mainresultmms} extends from MMS to APS, thus giving $\alpha$-APS allocations for XOS agents with equal entitlements and $\alpha > \frac{11}{40}$, provided that the number of agents $n$ is sufficiently large.

\begin{theorem}\label{mainresultaps}
    Let $\alpha, n_\alpha$ be as in \cref{mainresultmms}.
    For every $n \geq n_\alpha$, the following holds.
    Consider a fair allocation instance with a set $\Mm$ of items, and a set $\Nn$ of $n$ agents with equal entitlements, where for every $i \in [n]$, her valuation $v_i$ is XOS.
    Then, there exists an allocation $(P_1, \ldots, P_n)$ of $\Mm$ to $\Nn$ such that every agent $i \in \Nn$ receives at least $\alpha$-fraction of her APS, {i.e $v_i(P_i) \geq \alpha\cdot \aps(\Mm, v_i, 1/n)$.}
\end{theorem}

{In addition, we show that for every $n$ and any XOS valuation $v$, the multiplicative gap between the values of the APS and the MMS is less than $40/11$.}

\begin{theorem}\label{apsmmsgap}
    Let $\alpha$ be as in \cref{mainresultmms} and \cref{mainresultaps}.
    For any $n$, $\Mm$ and any XOS valuation $v$ it holds that $\mms(\Mm, v, 1/n)\geq \alpha\cdot  \aps(\Mm, v, 1/n)$.
\end{theorem}
Our approach for \cref{mainresultmms} and \cref{mainresultaps} originated from the work on the proof for \cref{apsmmsgap}.
We provide an overview of the proofs of our results in Section~\ref{sec:overview}.
{For convenience, in non-rigorous parts of our exposition, we may describe the approach either in terms of MMS or in terms of APS. However, all the rigorous technical proofs all apply to APS.} 

The allocations for \cref{mainresultmms}, \cref{mainresultaps} and \cref{apsmmsgap} can be obtained in polynomial time, as long as the MMS/APS 
partitions of all agents are provided in advance.
{Computing the values of the MMS and the APS (and corresponding partitions) for XOS valuations is APX-hard. This hardness result holds even for the more restricted class of submodular valuations (see~\cite{BUF23}).} {Nevertheless,  the value of the APS and an associated fractional partition can be computed in polynomial time for any monotone  (including XOS) valuation $v$, if the algorithm has access to the valuations not just via value queries, but also via what we refer to as {\em anyprice queries}. The APS has two equivalent definitions. One is Definition~\ref{def:APS}. The other is price-based, and postulates that when items have non-negative prices (scaled so that they sum up to 1), then an agent with entitlement $b$ ($b = \frac{1}{n}$ in our case) is entitled to any set of items of her choice, of price at most $b$. In this context, the following anyprice queries are natural. The query specifies nonnegative prices to the items, summing up to 1. The reply is the set $S$ of items that the agent desires, namely, the one maximizing 
$v(S)$, subject to the sum of prices of items in $S$ not exceeding $b$. If we allow both value queries and anyprice queries, the value of the APS can be computed using polynomially many queries (using the ellipsoid algorithm), and so can an associated fractional partition (using techniques involving linear programming duality). Consequently, given access to value queries and anyprice queries, our randomized allocation algorithm can be implemented so that its expected running time is polynomial. Further details on this aspect are omitted.}


\section{Overview of our proof}
\label{sec:overview}

Throughout this overview all valuations are assumed to be XOS, and $\rho$ serves as a generic name representing some unspecified constant, strictly larger than $\frac{1}{4}$. 

Our goal is to show the existence of $\rho$-MMS allocations. We break this goal into two main steps. One is the design of an allocation algorithm for which we have reasons to hope that it produces $\rho$-MMS allocations. (At this stage, we do not require this algorithm to be efficient, as we are mostly interested in the existence of $\rho$-MMS allocations, and treat the complexity of computing them as a secondary issue.) The other is to analyse the designed algorithm and prove that indeed it produces $\rho$-MMS allocations.

Our methodology for achieving the above two steps is to first change the problem as follows. Recall that the currently best approximation ratio with respect to the MMS, namely $\frac{4}{17}$, holds also with respect to the APS. Hence it seems reasonable to expect that if there are $\rho$-MMS allocations, then there also are $\rho$-APS allocations. Now we greatly {simplify} the goal of designing $\rho$-APS allocations, by considering only allocation instances in which all agents have the same XOS valuation function. We refer to this setting as that of {\em identical valuations}. Observe that identical valuations makes no sense if we wish to design $\rho$-MMS allocations, because the MMS partition with respect to the valuation function gives an MMS allocation, showing that $\rho = 1$. However, for APS, identical valuations do make sense, because the APS fractional partition need not be a legal allocation (bundles might overlap). Technically, the identical valuation setting is equivalent to proving that for XOS valuations, the multiplicative gap between the APS and the MMS is $\frac{1}{\rho} < 4$. (Prior to our work, it was only known that this gap is at most $\frac{17}{4} > 4$.) 

Handling the identical valuations setting already requires the development of new techniques. The techniques developed in this process help overcome some of the difficulties also in the general case in which valuations need not be identical, at least at the level of providing inspiration on how to address the general case. This allows us to achieve the task of handling the general case using two ``leaps" of intermediate difficulty (one is that of developing techniques that provably handle identical valuations, the other is developing only the additional techniques that are needed in order to extend the results to the general case), rather than one very big leap.

\subsection{$\rho$-APS allocations with identical valuations}
\label{sec:identical}

Here we consider the setting in which there is an XOS valuation $v$ with $APS(\items, v, {1/n}) = 1$, and we need to prove the existence of a partition $A_1, \ldots, A_n$ of $\items$, with $v(A_j) \ge \rho$ for every $j$. (This is a reformulation of the problem of producing $\rho$-APS allocations when agents have identical XOS valuations.) We refer to bundles of value at least $\rho$ as {\em acceptable}. As is standard in related settings, we may assume that no single item has value above $\rho$, because otherwise we can give this item to one of the agents without decreasing the APS value for the remaining agents.

We propose to use a greedy algorithm for the above problem. First, based on $v$, we set up a {\em potential set-function} $\beta$, with the following properties. 

\begin{enumerate}
    \item $\beta$ is monotone, namely, $\beta(S) \le \beta(S')$ if $S \subset S'$.
    \item $\beta(\items) \ge 1$.
    \item If $\beta(S) \ge \rho$, {then $S$ is an acceptable bundle.}
    \item Uniform and monotone lower bounds on transitions. {There is some (non-increasing) function $f : [\rho, \infty] \rightarrow (0, 1]$ such that} for every acceptable set $S$, 
    there is an acceptable  subset $T \subseteq S$ {satisfying $\beta(S \setminus T) \ge (1 - f(\beta(S)))\beta(S)$. That is,} the value $\beta(S \setminus T)$ can be lower bounded only as a function of $\beta(S)$ (and beyond that, independently of the contents of $S$ itself). 
\end{enumerate}

Our greedy algorithm has $n$ rounds. In each round $r$, a set $\items_r$ of yet unallocated items is still available (with $\items_1 = \items$). In round $r$ we greedily select an acceptable $A_r \subset \items_r$ that maximizes $\beta(\items_r \setminus A_r)$. Then we update $\items_{r+1} = \items_r \setminus A_r$, and proceed to the next round.

For the greedy algorithm to work, we need to select the function $\beta$ such that we can trace how it decreases throughout the greedy algorithm, and prove that even at the last round of the algorithm, $\beta(\items_n) \ge \rho$. Using such a $\beta$, the algorithm produces $n$ disjoint bundles, each of value at least $\rho$.

{We intend to apply $\beta$ to sets $S$ of the form $\items_r$, for some round $r$. Thus, to gain some intuition of how the definition of $\beta$ and its properties are used, the reader may think of each set $S$ referred to below as being one the sets $\items_r$.}

Our $\beta$ is defined as follows. Let $\{\lambda_B, B\}_{B \in \Bb}$ denote the APS partition for $v$. That is, $\lambda_B > 0$ for every $B \in \Bb$, $\sum_{B \in \Bb} \lambda_B = 1$, $v(B) \ge 1$ for every $B \in \Bb$ (as we assume that the APS value is~1), and $\sum_{B  :e \in B} \lambda_B \le \frac{1}{n}$ for every $e \in \items$. We define 

$$\beta(S) = \sum_B \lambda_B \cdot v(B \cap S).$$

\noindent The exact definition and the way we use it is given in Section~\ref{sec:greedalgo}.
In a sense, $\beta(S)$ quantifies the amount of value that the items of $S$ contribute to the APS fractional partition. Technically, it is the expected value of a bundle $B \cap S$ sampled with probability $\lambda_B$. One can easily verify that properties~1, 2 and~3 above indeed hold. We now discuss property~4.

Inspired by arguments used in \cite{GhodsiHSSY22} and \cite{FG25}, for every $S$, we define $D_S$ to be the distribution over acceptable bundles $T \subseteq S$ that minimizes $p_S$, where $p_S$ is defined {as follows}.
Let $p(e)$ for item $e$ denote $
\P[T \sim D_S]{e \in T}$. 
That is, $p(e)$ is the probability that item $e$ is included in a set selected at random from $D_S$. 
We choose $D_S$ to be the distribution minimizing (over all distributions) the maximum (over all items) $p(e)$. 
We refer to the resulting maximum value of $p(e)$ as $p_S$. 

For a given value $t \ge \rho$, we define $p_t$ to be the maximum $p_S$ over all those $S$ that satisfy $\beta(S) \ge t$. 
Observe that $p_t$ above is monotone non-increasing (the larger $t$ is, the fewer $S$ there are that we can maximize over). That is, as $\beta$ gets smaller, the respective $p_t$ grows.

For a given $S$ with $\beta(S) = t$, if we select at random an acceptable set $T$ from the distribution $D_S$, then every item $e \in S$ has probability {at} most $p_S \le p_t$ of belonging to $T$. 
{We make the following observation, which we prove rigorously in \cref{probval}.
\begin{observation}\label{processobs}
    Consider the potential function $\beta$, as defined earlier, w.r.t XOS valuation $v$.
    Then, for any $S\subseteq \Mm$ with $\beta(S) = t$, and a random acceptable set $T \sim D_S$, the expected value of $\beta(S\setminus T)$ is at least $(1 - p_t)\cdot t$.
    As a corollary, there exists a choice of an acceptable $T\subseteq S$ for which the potential is no less than $\beta(S\setminus T) \geq (1 - p_t)\cdot \beta(S)$.
\end{observation}
This serves as property~4 for our choice of $\beta$.}

To make use of our $\beta$, we need also to have an upper bound on $p_t$ as a function of $t$. To get a sense of what kind of upper bound we can expect, suppose that for a desired value of $\rho$, it happens that in every round of the algorithm, every bundle in the support of the APS fractional partition maintains a value of at least $\rho$. Then, picking such a bundle at random, we have that $p_t \le \frac{1}{n}$. 
In this case, after $n-1$ steps the value of $\beta$ will be $(1 - \frac{1}{n})^{n-1} > \frac{1}{e}$, implying that we can take $\rho = \frac{1}{e} \ge 0.36$.

However, the value of $p_t$ depends both on $\rho$ and on $t$. 
We emphasize here that the dependence on $\rho$ is in two respects. One is that $\rho$ serves as a lower bound for the values that bundles in the distribution must have. The other is that it serves as an upper bound on the value of every single item, a fact that we use when arguing that there must be distributions certifying that the value of $p_t$ is ``small".
As for the dependence on $t$, in the initial steps of the algorithm, when $t$ is large, the value of $p_t$ is smaller than $1/n$. For example, for $t=1$ we have that $p_t \le \frac{1}{2n}$, because every APS bundle can be partitioned into two disjoint bundles, each of value at least $\rho$ (this holds for $\rho \le \frac{1}{3}$, when no single item has value above $\rho$), and we can pick one such sub-bundle at random, rather than the whole APS bundle. On the other hand, as $t$ becomes smaller, $p_t$ increases. For example, the value of some APS bundles might drop below $\rho$, then they can no longer be chosen, and this increases the probability of being selected for items in other APS bundles. 
For all $t$ bounded away from $\rho$, this increase is by constant factors (compared to $\frac{1}{2n}$) which increase as we approach $\rho$. Towards $\rho$ the increase accelerates, and for $t = \rho$ we have that $p_t = 1$.

Hence, the question of whether we get a value of $\rho$ that is larger than $\frac{1}{4}$ depends on the exact shape of the function $p_t$, and specifically, on the rate in which $p_t$ increases as $t$ decreases. Fortunately, we are able to give strong enough upper bounds on $p_t$. This can be done using a relatively straightforward analysis. The main principles behind this analysis are as follows.  For every value $t' \ge \rho$, using the fact that no single item has value above $\rho$, we can determine distributions over acceptable bundles associated with a single hypothetical bundle $B^* = B \cap S$ of value $t'$. {Here $B$ is some APS bundle, and one can think of $S$ as the remaining items $\Mm_r$ for some round $r$.}
For example, if $\frac{3}{2}\rho \le t' < 2\rho$ then it can be shown that there must be three acceptable sub-bundles of $B^*$ such that each item is included in at most two of them, giving rise to a distribution over items of $B^*$ in which each item is selected with probability at most $\frac{2}{3}$. Thus, we afford to multiply $\lambda_B$ by $\frac{3}{2}$ without violating the constraint that each item is sampled with probability at most $\rho$. Doing so separately for every bundle in the APS fractional partition, the sum of new values of $\lambda_B$ need not be one. Normalizing it to~1 normalizes $\frac{1}{n}$ to become $p_S$. 
For more details, see Lemma~\ref{partition} and Lemma~\ref{distribution}.

Finally, to upper bound $p_t$ we maximize $p_S$ over all hypothetical profiles of values of $v(B \cap S)$ (where $B$ ranges over all bundles of the APS partition $\Bb$) that are consistent with $t$ (have average value $t$, where average is computed with respect to the weights $\lambda_B$). {It is not difficult to show that this maximum is attained when each value $v(B \cap S)$ is just slightly below a half integral multiple of $\rho$. Hence, for each $t$, providing an upper bound on $\rho_t$ can be modeled as the optimal solution for a linear program (LP) with finitely many variables. The value of these optimal solutions is characterized exactly, by exhibiting matching primal and dual solutions to the corresponding LPs.}

{Proposition~\ref{pro:valbound} states our upper bounds on the probability $p$. (For a proof see Section~\ref{sec:pBound}, where this proposition is rephrased as Corollary~\ref{valbound}.)

\begin{proposition}\label{pro:valbound}
    Suppose that no item has value larger than $\alpha$, and the value of the potential function is $\beta$, for some $\beta \ge\alpha$. Then there is a distribution over acceptable bundles in which no item has probability larger than $p$, for $p$ as listed below.
    \begin{enumerate}
        \item if $\alpha \leq 3/11$, then
        \[p \le \frac{1 - \alpha}{2(\beta - \alpha)n}\]
        \item if $\alpha \geq 3/11$ and $\beta \geq 3\alpha$, then
        \[p \le \frac{2(1 - 3\alpha)}{(\beta - 12\alpha + 3)n}\]
        \item if $\alpha \geq 3/11$ and $\beta < 3\alpha$, then
        \[p \le \frac{4\alpha}{3(\beta - \alpha)n}\]
    \end{enumerate}
    (When $\alpha = 3/11$, the three bounds coincide.)
\end{proposition}
}

\begin{remark}
\label{rem:smallItems}
    {The parameter $\alpha$ serves two roles in the above proposition, being both the minimum value of an acceptable bundle, and an upper bound on the value of individual items. If $\alpha$ serves only the first role, then as the upper bound on the value of individual items tends to~0, the upper bound on $p$ converges from above to $\frac{\alpha(1 - \alpha)}{(\beta - \alpha)n}$, for all $0 < \alpha < \beta \le 1$.}  
\end{remark}

For a plot of $p_t$ when {no item has value larger than} $\rho = 0.25$ and $n = 100$, see {Figure~\ref{fig:pt}}.
\begin{figure}[h]
\includegraphics[width=17cm]{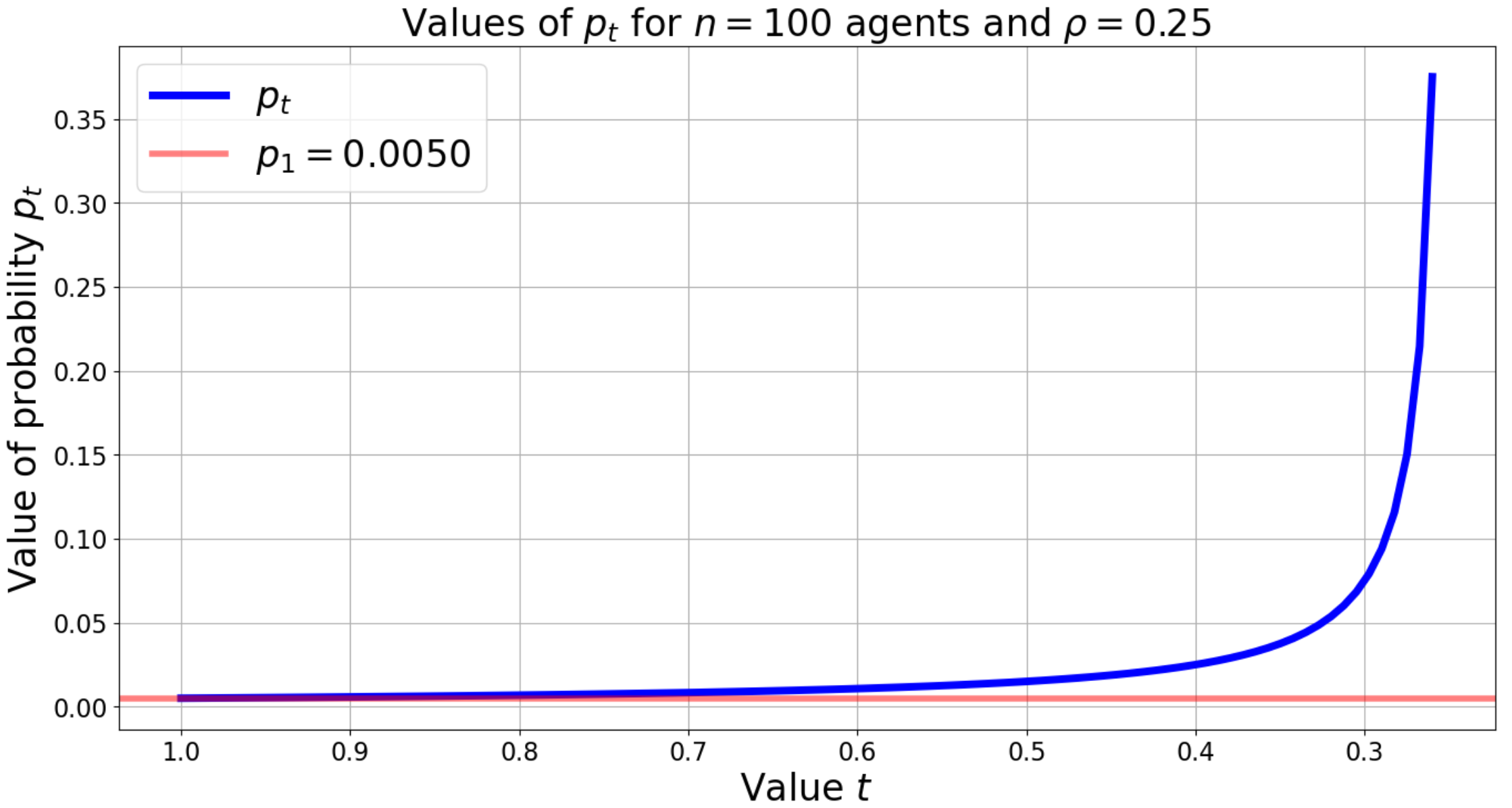}
\label{fig:pt}
\caption{The function $p_t$ depends on the value of $\rho$, and also has a weak dependency on the value on $n$. Hence, the above plot applies for specific values of $\rho$ and $n$, but its general form is the same for all $\rho$ and large $n$.}
\label{fig:pt}
\end{figure}
Observe that $p_1 = \frac{1}{2n}$, so the process starts of with very small $p_t$, but as $\beta$ approaches $\rho$, $p_t$ approaches~1, greatly accelerating towards the end ($p_t$ is a convex function of $t$).

For a given value of $n$, we say that a proposed value of $\rho$ is {\em eligible} if tracing the value of $\beta$ for $n-1$ steps (where in each step the previous value $t$ is updated to the new value $(1 - p_t)\cdot t$), the value of $\beta$ remains at least $\rho$. The largest eligible $\rho$ is the guaranteed MMS approximation for the greedy allocation algorithm.
{The fact that $\rho$ being eligible guarantees the existence of $\rho$-MMS allocations follows from Observation~\ref{processobs}.
So, to estimate the performance of the algorithm it suffices to trace the value of $\beta$ throughout this update process.}

It remains to lower bound the value of the largest eligible $\rho$. For constant values of $n$, this can easily be done (up to rounding errors) using a computer program that simply computes the outcome of $n$ steps of the update process. For very large values of $n$, this can be done with great accuracy by approximating the discrete updating process by a continuous process (the approximation error decreases as $n$ grows), and obtaining an analytic expression for the continuous process. 

{\begin{proposition}
\label{pro:bigrhoinf}
As $n$ tends to infinity, the approximation $\rho$ approaches a value that is at least as high as the solution to the equality $2(12\rho - 3)\ln(3\rho) = (1 - 3\rho)(3\ln 3 - 4)$. In particular, $\rho > 0.276 > \frac{11}{40}$.    
\end{proposition}

For the proof of Proposition~\ref{pro:bigrhoinf}, see Section~\ref{sec:littlealpha} (Theorem~\ref{rhoinfinity}) and Section~\ref{sec:bigAlpha} (Theorem~\ref{bigrhoinf}).}

\begin{remark}
    {If the upper bound on the value of individual items converges to~0 (as in Remark~\ref{rem:smallItems}), the optimal $\rho$ converges to a solution to $1 -\rho + \rho\ln\rho = \rho(1 - \rho)$, giving $\rho \approx 0.394$. Under the same condition, the approach initiated by~\cite{GhodsiHSSY22} (discussed in Section~\ref{sec:intro}) converges (from below) to $\frac{1}{4}$-MMS.}
\end{remark}

For intermediate values of $n$, we can get a nearly tight lower bound using a {\em doubling} procedure, as follows. Let $\rho_\epsilon(n)$ be the highest value of $\rho$ such that when there are $n$ agents, after $n$ steps the process has value $\rho + \epsilon$. (As $\epsilon$ tends to~0, $\rho_\epsilon(n)$ tends to the optimal $\rho(n)$.) 
We show that for every $\epsilon$, there is some $n_{\epsilon}$ such that for every $n > n_{\epsilon}$, $\rho_{\epsilon}(n) \ge \rho_{\epsilon}(2n)$.  This implies that beyond some $n_{\epsilon}$, the analytic bounds that can be derived for $\rho_{\epsilon}(n)$ when $n$ tends to infinity holds for all $n > n_{\epsilon}$.
The precise description of this process is given in Lemma~\ref{doublingineq} and Corollary~\ref{doubling}.

Using the principles above, we determine that some $\rho$ larger than $\frac{11}{40}$ is eligible for all $n$.

\begin{theorem}
\label{thm:APSMMSgap}
    For every XOS valuation and every value of $n$, the ratio between the APS and the MMS is less than $\frac{40}{11}$.
\end{theorem}

\subsection{Different valuations, but large $n$}
\label{sec:different}

We now consider settings with $n$ agents in which each agent $i$ has a potentially different XOS valuation $v_i$. For simplicity of the presentation, we also return to considering $\rho$-MMS allocations rather than $\rho$-APS allocations, though we remark that all our results extent to $\rho$-APS allocations. For each $i \in \agents$, the MMS partition for agent $i$ will be denoted by $B_1^i, \ldots, B_n^i$, and we assume without loss of generality that $v_i({B_j^i}) = 1$ for all $i$ and $j$ (all MMS bundles have value~1). 

To adapt the greedy algorithm to a setting in which the valuations of agents need not be identical, we need to define a new potential function that guides the greedy algorithm. Inspired by the identical valuations case, we wish to base this new potential function on $\beta$ that worked well for the identical valuation case. However, the definition of $\beta$ depends on the valuation function $v$. In the general case, there are $n$ different function $v_i$, and hence $n$ different $\beta_i$, where $\beta_i(S) = \sum_j \frac{1}{n} \cdot v_i(B_j^i \cap S)$. We could not find a way of aggregating the entries of the $n$-dimensional vector $(\beta_1, \ldots, \beta_n)$ into one scalar $\beta$ that will serve as a useful potential function for the greedy algorithm.

As an example, consider the following approach, that we refer to as the {\em greedy-average algorithm}. In round $r$, let $\agents_r$ denote the set of $n - r + 1$ active agents that did not yet receive a bundle. Select the agent $i \in \agents_r$ with currently highest $\beta_i$ value.  Then, greedily select a bundle $A_r \subset \items_r$ that is acceptable for $i$ and maximizes $\sum_{j \in \agents_r \setminus \{i\}} \beta_j(\items_r \setminus A_r)$. This version of the greedy algorithm aims to keep the average $\beta$ value high. However, it does not prevent the situation in which for a small number of agents, the value of their respective $\beta$ function will drop below $\rho$, and once this happens, the respective $\items_r$ might not contain any bundle that is acceptable for them.

Due to examples such as those above, we abandon the greedy approach, and replace it by a randomized approach. The basic structure of the randomized approach is similar to that of the greedy-average algorithm, except for a key difference: agent $i$ does not choose $A_r$ greedily, but rather by random sampling of an acceptable bundle from her associated distribution $D_{\items_r}$ at the time (a distribution that also depends on $v_i$, though we omit $i$ from the notation $D_{\items_r}$). The hope is that for every possible input allocation instance, the random process will offer positive probability of succeeding, implying that a $\rho$-APS allocation exists.

Unfortunately, we do not know how to analyse the randomized algorithm when the number of agents $n$ is small. However, we propose a concrete plan (which we will successfully implement) for how to analyse it when $n$ is sufficiently large. Our plan is based on the known fact that a sum of a collection of many random variables that are independent and bounded is concentrated around its expectation. When $n$ is large, the changes to the potential $\beta_i$ of agent $i$ along the run of the randomized algorithm are basically a sum of many small random changes, one change per round of the algorithm. Hence we may expect at every given round that the potential functions of all active agents will have roughly the same value. If these values were exactly the same, this would mimic the analysis for the case of identical valuations, for which we could prove that $\rho > \frac{11}{40}$. The hope is that being roughly the same rather than exactly the same does not make a big difference to the analysis. 

Guided with the above intuition, our plan has three steps. But before presenting these steps, we present two simplifying assumptions that we make in the first step, and remove in the other two steps.

\begin{enumerate}
    \item {\em No large intersections.} For every two agents $i$ and $j$, every MMS bundle $B_i^k$ for agent $i$ has only a small effect on $\beta_j$. That is, for some $c$ that will depend on $n$ (as $1/\polylog  n$), $\beta_j(\items \setminus B_i^k) \ge 1 - c$. Observe that $\beta_i(\items) = 1$, and {that} the small effect condition also implies that  $\beta_j(S \setminus B_i^k) \ge \beta_j(S) - c$ for every $S \subset \items$. The no large intersection assumption is useful for arguing that as the algorithm progresses, $\beta_j$ decreases in small steps (a property needed for concentration results), rather than in large jumps.

    \item {\em No large items}. For every agent $i$ and item $e$ we assume that $v_i(e) \le \rho$. This assumption is needed in order to establish sufficiently small values for $p_t$ as a function of $t$ (as explained in Section~\ref{sec:identical}). 
\end{enumerate}

Armed with the above assumptions we present our three steps. Later, we shall present a more detailed overview of how we implement them. {(In the full proof provided in the appendix, step (2) below is performed before step (1), so as to avoid duplications in proofs. In our overview we present step (1) first, because if it does not work by itself, there is no point in performing step (2).)}

\begin{enumerate}
    \item {\em Handling an ``easy" special case.} We start by analysing the randomized algorithm (for large $n$) while assuming the above two assumptions. This is done by establishing that the path followed by the $\beta$ values in the greedy process with identical valuations is an ``attractor" for the path followed {by $\beta_i$} in the randomized process. {In other words,} small fluctuations from the former path \textit{do not} start off a ripple effect that causes the latter {path} to diverge greatly from the former path. The need to prove the attractor property places an upper bound on how large $c$ can be in the {\em no large intersections} property. 

\item {\em Removing the assumption that there are no large intersections.} Without this assumption, we lose the concentration results, and the analysis completely breaks down. To overcome this difficulty, we modify the randomized algorithm, so that it gives special treatment to sets that violate the large intersection property. Note that all sets in an MMS partition of an agent might violate this property, each with respect to a different other agent. {Consequently, the treatment is  ``special" not in the sense of being applied only rarely, but rather in the sense of differing from the treatment given to sets that do not violate the large intersection property.} 
Our modifications work only if $c$ in the {\em no large intersections} property is sufficiently large, as otherwise their effect on $\rho$ becomes larger than what we can tolerate. Luckily, there are ``middle ground" values for $c$ that allow us to both handle the ``easy" special cases, and to remove the assumption of no large intersections.

\item {\em Removing the assumption that there are no large items.} In most algorithms for finding approximate MMS allocations {(including \cite{GhodsiHSSY22})}, this is the easiest step. If $v_i(e) \ge \rho$ for some agent $i$ and item $e$, simply give item $e$ to agent $i$, remove $i$ (as item $e$ suffices from her MMS approximation), and the MMS of the remaining agents in the remaining instance does not decrease. Hence the allocation problem with $n$ agents reduces to a seemingly easier one with $n-1$ agents. However, in our case, we cannot afford to reduce the number of agents, as our analysis requires large $n$. This seems problematic, because when there are many large items, each of value 1 to all agents, it is without loss of generality that each large item is allocated to a single agent, and the number of agents is reduced dramatically. So is the situation hopeless? Not really, because we may choose the identity of the agents that receive the large items. This way, even though we remain with an instance with few agents, their valuations are not arbitrary, but rather chosen carefully among the original $n$ valuations. We show that such a choice exists that allows us to give each one of them at least $\rho$-MMS.
    
\end{enumerate}

Now we present more details on each of the steps.

\subsubsection{Handling an ``easy" special case}

As explained above, we prove for each {agent $i$ that the path of $\beta_i$ values} that it follows along rounds of the randomized algorithm stays close to the ``ideal" path of $\beta$ values that is followed by the greedy allocation algorithm for the case of identical valuations. 
The proofs that the paths stay close is highly technical. It ``obviously" needs to use martingale concentration results, but the details of how they are used are quite complicated, due to positive correlations between deviations for different agents. For example, if the agent who selects an acceptable bundle in a certain round is below the ideal path in that round, then her associated $p$ value at that round is larger than that of the ideal path, amplifying the deviation of other agents from the ideal path.  

Due to these positive correlations, closeness of paths is shown to hold only if $c$ in the {\em no large intersections} assumption is sufficiently small, and only for the main part of the path --- once the value of $\beta$ gets very close to {$\rho$}, 
the proof breaks down. Nevertheless, this suffices in order to obtain a value of $\rho > \frac{11}{40}$. 

{
\begin{proposition}
\label{pro:difftheorem}
    Let $\rho$ be as in Proposition~\ref{pro:bigrhoinf}, and consider any $\alpha < \rho$. Then there exists $n_\alpha$ such that for all $n \geq n_\alpha$, if the {\em no large intersections} property holds (with $c$ sufficiently small), then the randomized algorithm (with high probability) produces an $\alpha$-MMS (and $\alpha$-APS) allocation. 
\end{proposition}

{We provide an overview of how Proposition~\ref{pro:difftheorem} is proved.  

\begin{enumerate}
\item Fix $\alpha < \rho$, for $\rho$ as in Proposition~\ref{pro:bigrhoinf}, and fix a sufficiently large $n$.
    \item Because $\alpha < \rho$, if we were in the case of identical valuations, then for our choice of $\alpha$, we would have an ``ideal" curve $\beta_{\alpha}$ with $\beta_{\alpha}(1) = 1$ and $\beta_{\alpha}(n) > \rho >  \alpha$. In the curve $\beta_{\alpha}$, for each round $r$, $\beta_{\alpha}(r+1) = (1 - p_r)\beta_{\alpha}(r)$, where $p_r$ (the maximum probability that an item is selected in round $r$) is computed as in Section~\ref{sec:identical}, based on the value $\beta_{\alpha}(r)$, and on the assumption that no item has value larger than $\alpha$. 
    \item For an appropriate choice of $\varepsilon < \beta_{\alpha}(n) - \alpha$, we consider a new curve $\beta_{\alpha,\varepsilon}$. It too starts at~1, but differs from $\beta_{\alpha}$ in the way that $p_r$ is computed in each round $r$. Rather than using the value $\beta_{\alpha,\varepsilon}(r)$ for this computation, we use $\beta_{\alpha, \varepsilon}(r) - \varepsilon$. Due to this, $\beta_{\alpha,\varepsilon}$ decreases at a faster rate than $\beta_{\alpha}$.
    However, for our choice of $\varepsilon$, we show that $\beta_{\alpha,\varepsilon}(n) > \alpha + \varepsilon$ holds. 
    \item Now we consider $n^2$ possible bad events, where bad event $\mathcal{E}_{i,r}$ (for $i,r \in [n]$) is the event that at round $r$, the {value $\beta_i$} for agent $i$ was smaller than $\beta_{\alpha, \varepsilon}(r) - \varepsilon$.  Conditioned on the assumption that those bad events associated with $r' < r$ did not happen (and hence each $p_{r'}$ is not larger than that used in order to compute the curve $\beta_{\alpha, \varepsilon}$), the {expectation of $\beta_i$} for agent $i$ in round $r$ is at least $\beta_{\alpha, \varepsilon}(r) - \eps$. Under this conditioning, a martingale concentration result (that uses the {\em no large intersection} property) shows that $\mathcal{E}_{i,r}$ happens with probability smaller than $\frac{1}{n^2}$. Taking a union bound, with positive probability no bad event happens. The absence of bad events implies (among other things) that every $\beta_i$ remains above $\alpha$ in every round $r \in [n]$, as desired.
\end{enumerate}}

In Section~\ref{sec:algorithm} we state Theorem~\ref{difftheorem}, which is stronger than Proposition~\ref{pro:difftheorem}, because it does not assume the {\em no large intersection} property. The proof of Proposition~\ref{pro:difftheorem} is implicit as part of the proof of Theorem~\ref{difftheorem}. }
Section~\ref{sec:martingale} is dedicated to the proof of martingale concentration.

\subsubsection{Allowing for large intersections}
\label{sec:largeIntersect}

We define the {\em damage} that a bundle $B$ acceptable for an agent $i$ can do to the other agents as the sum of decrease in $\beta$ values of all other agents caused by the removal of all items of $B$. 
Using an appropriate choice of parameters, we classify bundles into three classes, those doing little damage {(at most $c$)}, those doing medium damage, and those doing {large} damage {(greater than $D$)}.
For our choice of parameters {$c$ and $D$}, those bundles that do little damage do not form a problem, as they satisfy the {\em no large intersection} assumption. Those bundles that do {large} damage also do not form a problem -- it can be shown that they are rare, and simply not allowing the agent to choose them does not significantly affect the run of the randomized algorithm. 

The main difficulty is those bundles that do medium damage. There might be too many of them, and so we cannot simply forbid choosing them. But, if they are chosen, some other agents suffer a large drop in their $\beta$ value, and we no longer have the martingale concentration results that we need for our analysis.

Our fix for this is to depart from the greedy paradigm which states that once an agent $i$ gets an acceptable bundle $B$, this is an irreversible decision. Instead, we allow agents who suffered large damage by the allocation of $B$ to ``steal" back items from $B$, if the need arises. To make sure that this stealing keeps the remaining part of $B$ acceptable for agent $i$, we show how to carefully select a subset $S\subset B$ of items that can be stolen, and for each item $e \in S$, a subset of agents that are allowed to steal $e$. This selection is done in such a way that effectively makes $B$  into a bundle that makes only little damage, while ensuring that {after stealing, what remains of $B$} is still acceptable for $i$ (for a new approximation ratio $\rho'$ that is smaller from $\rho$ by a term that tends to~0 as $n$ grows).

We now explain the stealing principle in more detail. We do so on a simplified example, so as to avoid clutter that obscures the main idea. Also, we set the parameters associated with stealing not exactly to the same values as they have in our proofs.
For the complete description of the algorithm with the stealing procedure, see Section~\ref{sec:algorithm}, and the analysis of the stealing is given in Section~\ref{sec:stealinganalysis}.

Consider the first step of the algorithm. Agent $i$ has $n$ disjoint bundles in her MMS partition $(B^i_1, \ldots, B^i_n)$, and suppose that each $B^i_h$ is broken into two disjoint acceptable bundles that we rename as $B^i_{h}$ and $B^i_{h+n}$, each of value slightly larger than the required minimum of $\rho$-MMS. For example, for every {$k \in [2n]$} we may have $v_i(B^i_k) \ge \rho + n^{-0.4}$. We choose $k\in \{1, \ldots 2n\}$ uniformly at random and give $B^i_k$ to agent $i$. For every item $e$, the event that $e\in B^i_{k}$  has probability $p=\frac{1}{2n}$. 

Recall that any other agent $j$ has an MMS partition $(B^j_1, \ldots, B^j_n)$, where for every $\ell$, $v_j(B^j_{\ell}) = 1$. Associating a weight of $\frac{1}{n}$ with every bundle in the MMS partition, the associated value of $\beta_j(\items)$ is $\sum_{\ell} \frac{1}{n} \cdot v_j(B^j_{\ell}) = 1$. 

The {\em damage} that a bundle $B^i_{k}$ does to agent $j$ is defined as $\frac{1}{n} \cdot \sum_{\ell} v_j(B^j_{\ell} \cap B^i_{k})$. The {\em no large intersection} property requires that the damage be at most $c$. For concreteness here, let us fix $c = n^{-0.1}$. 
We classify the bundles $B^i_k$ according to the total damage that they do to all other agents $j \not= i$. 

One class contains those bundles that do a total damage of at least $n^{0.4}$ (large damage). As the total damage of all bundles is $n-1$ (there are $n-1$ other agents, and with maximum damage of~1), there are at most $n^{0.6}$ bundles in this class. We do not allow $i$ to choose any bundle from this class. As $n^{0.6}$ out of the $2n$ bundles in $i$'s partition are discarded, the value of $p$ now might slightly increase, by a factor of at most $\frac{2n}{2n - n^{0.6}} < (1 + n^{-0.4})$. This slight increase has only negligible effect of the analysis and the final value of $\rho$.

Another class contains those bundles $B^i_k$ that do not do a damage larger than $c$ to any other agent. These bundles do not violate the {\em no large intersection} property, and hence they do not need any special treatment.

The last class contains the remaining bundles. Consider one bundle $B^i_k$ from this class. It inflicts a damage larger than $c$ to some of the other agents. As the total damage of $B^i_k$ is at most $n^{0.4}$ and $c = n^{-0.1}$, the number of agents suffering damage larger than $c$ is at most $n^{0.5}$. Let $j$ be one such agent, and refer to the set of bundles in her MMS partition as $\Bb_j$. 

Order the bundles of $\Bb_j$ according to the $v_i$ value of their intersection with $B_k^i$, from largest to smallest. Let $\Bb_j^{i,k}(c)$ denote the bundles in the prefix of this order, where the prefix is maximal subject to the condition that the total damage inflicted by $B_k^i$ to all bundles in $\Bb_j^{i,k}(c)$ is at most $c$. 
As $c = n^{-0.1}$ and each bundle in $\Bb_j$ contributes at most $\frac{1}{n}$ to $\beta_j(\items)$, 
it means that $\Bb_j^{i,k}(c)$ contains 
at least $n^{0.9}$ different bundles. 
As $v_i(B_k^i) \le 1$,  there must be a bundle $B^j_{\ell^*} \in \Bb_j^{i,k}(c)$ for which $v_i(B^j_{\ell^*} \cap B^i_k) \le n^{-0.9}$. Consequently, it holds for every bundle $B^j_{\ell} \in \Bb_j \setminus \Bb_j^{i,k}(c)$ that $v_i(B^j_{\ell} \cap B^i_k) \le n^{-0.9}$.

{We can now explain the stealing operation. We limit the damage done by $B^i_k$ to agent $j$ to be the intersection of $B^i_k$ only with the bundles of $\Bb_j^{i,k}(c)$. As such, it is upper bounded by $n^{-0.1} = c$ as desired. All remaining items of $B^i_k$ are not considered to be excluded for agent $j$. That is, if in a future round, the acceptable bundle $A_j$ selected by agent $j$ is a subset of a bundle $B^j_{\ell}\not\in \Bb_j^{i,k}(c)$, then the items in $A_j \cap B^i_k$ are ``stolen" from agent $i$ and given to agent $j$. Note that by this stealing, the value of $B^i_k$ decreases by at most $n^{-0.9}$. Finally, recall that there are at most $n^{0.5}$ potential agents $j$ that might steal items from $i$. Hence, the combination of all stealing operations can reduce the value of $B^i_k$ by at most $n^{-0.4}$. As $B^i_k$ started with a value of at least $\rho + n^{-0.4}$, what remains of $B^i_k$ after the stealing is still acceptable for agent $i$.}

\subsubsection{Allowing for large items}

In this section we allow the instance to have large items (items $e$ of value at least $\rho$ to some of the agents), and show how to adapt our algorithm to such instances. At a high level, we classify such instances into four classes, and handle each class separately. The classes are defined using a key lemma that we formulate and prove about matchings in bipartite graphs.

The bipartite graph $G(U,W,E)$ that we have in mind is one in which the set $U$ of vertices on one side {represents} the $n$ agents (so $U = \agents$), the set $W$ of agents on the other side {represents} the $m$ items
(so $W = \items$), and there is an edge between agent $i \in U$ and item $j \in W$ if $v_i(j) \ge \rho$, namely, item $j$ is large for agent $i$.

\begin{lemma}\label{bipartite0}
    Let $0 < \eps \leq 1/240$ and $n$ be arbitrary.
    For every bipartite graph $G = G(U, W, E)$ with $|U| = n$ and $|W| \geq n$, one of the following must hold.
    \begin{enumerate}
        \item $G$ has a matching in which all of $U$ is matched.

        \item $G$ has a maximal matching with at most $(1 - \eps)n$ edges.

        \item There are partitions $(U_1, U_2)$ of $U$ and $(W_1, W_2)$ of $W$ with the following properties:
        \begin{enumerate}
            \item $|U_1| = |W_1|$, and $(U_1, W_1)$ form a perfect matching which is maximal in $G$;
            
            \item for every vertex $u \in U_2$ we have $|N(u)| \leq n - 3|U_2|$.
        \end{enumerate}

        \item There are partitions $(U_1, U_2, U_3)$ of $U$ and $(W_1, W_2, W_3)$ of $W$ with the following properties:
        \begin{enumerate}
            \item $|U_1|= |W_1|$, and $(U_1, W_1)$ form a perfect matching;

            \item there are no edges between $U_2$ and $W_3$ and between $U_3$ and $W_2 \cup W_3$;

            \item 
            $(1 - 60\eps)n \leq |W_2| < |U_2|$;

            \item for every subset $U_2'\subseteq U_2$ of size $|W_2|$ there is a perfect matching between $U_2'$ and $W_2$.
        \end{enumerate}
    \end{enumerate}
\end{lemma}

The proof of {Lemma~\ref{bipartite0} appears in \cref{sec:bigitems} as Lemma~\ref{bipartite}.} 

Now we explain how we get a $\rho$-MMS allocation in each of the first three cases of Lemma~\ref{bipartite0}, in the order in which they appear in the lemma.

\begin{enumerate}
    \item We allocate to each agent the large item that is matched to it.
    \item We consider a maximal matching of size at most $(1 - \eps)n$ and give each matched agent the large item to which it is matched. We remain with an instance with $n' \ge \epsilon n$ agents, no large items, and each agent $i$ has an MMS partition (of $W_2$) into at least $n'$ bundles, each of value at least~1. As $\epsilon$ is a fixed constant, if $n$ is sufficiently large, then so is $\epsilon n$. So, this case reduces to the case with no large items, handled in previous sections.
    \item We match the agents of $U_1$ with the large items in $W_1$. As to agents in $U_2$, condition 3(b) can be shown to imply that each of them has a partition of the remaining items (namely, $W_2$) into at least $2|U_2|$ bundles, each of value at least $1 - \rho$. As each of these partitions has at least twice as many bundles as the total number of agents, known techniques of~\cite{GhodsiHSSY22}  imply the existence of $\frac{1 - \rho}{2}$-MMS allocations. As our $\rho$ satisfies $\rho < \frac{1}{3}$, this approximation ratio is even better than $\rho$.
\end{enumerate}

The fourth case of Lemma~\ref{bipartite0} requires the most complicated allocation algorithm (which is given in Section~\ref{sec:bigcase4}). For this case we use the notation $n' = |U_2| + |U_3|$  and $k = n' - |W_2|$. {Lemma~\ref{bipartite0} implies that $n' \ge (1 - O(\eps))n$ and $k \le O(\eps n)$.} 
Our allocation algorithm has the following steps.

\begin{enumerate}
    \item For an arbitrary perfect matching between $U_1$ and $W_1$, give each agent in $U_1$ the large item to which it is matched.  
    \item Observe that $|U_3| \le k \le O(\eps n)$, and none of the remaining items (those of $W_2 \cup W_3$) is large for any agent in $U_3$. Disregarding the $|W_2|$ small items of $W_2$, it can be shown that each agent in $U_3$ has at least roughly $\frac{n}{2}$ disjoint bundles contained in $W_3$, each of value at least $1 - \rho$. Run the randomized allocation algorithm implicit in Section~\ref{sec:largeIntersect} on the agents of $U_3$ and items of $W_3$. Previously, we proved that the algorithm works if the number of agents is sufficiently large. Though we do not know whether $|U_3|$ is sufficiently large, we show that the algorithm does work for $U_3$, because it is run with favorable parameters. 
    Specifically, the associated values of $p$ (the probabilities of choosing items) are $O(\frac{1}{n})$ (as the number of disjoint bundles to choose from is $\Omega(n)$), whereas we need to allocate bundles to only $|U_3| \le O(\eps n)$ agents. 
    \item It remains to allocate bundles to agents of $U_2$. Prior to allocating bundles to agents of $U_3$, each agent in $U_2$ had at least $k$ disjoint bundles contained in $W_3$, each of value at least $1$. As the $|U_3|$ bundles allocated to agents of $U_3$ each contained each item with probability $O(\frac{1}{n})$, it follows that after allocating bundles to $U_3$, each agent of $U_2$ maintains in expectation a $(1 - \epsilon)$ fraction of the value of her associated $k$ disjoint bundles. Run on $U_2$ and the items remaining in $W_3$ the greedy-average algorithm mentioned in Section~\ref{sec:different}. Though in general, this algorithm need not allocate acceptable bundles to all agents, here we stop it after a very small number of steps, namely $k - |U_3|$, whereas the total number of agents in $U_2$ is very much larger. It is not difficult to prove that the algorithm does manage to allocate acceptable bundles to such a small number of agents.
    
    \item At this point, a set $U'_2 \subset U_2$ containing exactly $n'-k = |W_2|$ agents remain. By item 4(d) in Lemma~\ref{bipartite0}, there is a perfect matching between $U'_2$ and $W_2$, so each of these agents gets a large item. 
\end{enumerate}

This completes our overview for how we address allocation instances that do have large items (if $n$ is sufficiently large).

\section{Related work}
\label{sec:related}
The topic of fair division has been extensively studied in various different settings.
While in our work we focus only on allocations of indivisible goods, there are multiple results that also consider divisible goods, i.e when parts of the same item can be given to different agents.
In terms of notions of fairness, we consider only MMS and APS, which are so-called \textit{share based} notions.
There are many other share based notions, such as proportional share, MES and WMMS.
Among non-share based notions of fairness, there are comparison based ones, such as envy-freeness, e.g EF1, EFX, which studies allocations where each agent receives a bundle at least as good (or almost as good) as the others.
For a general overview on fairness notions for indivisible goods, a reader can check surveys \cite{AzLiMoWu22} and \cite{AmanatidisSurvey23}, while for divisible goods there is plenty of much earlier works, see for example a book \cite{Brams_Taylor_1996}.
When it comes to classes of valuation functions, our work considers only XOS valuations.
A significant work has been done to explore other valuation classes, including additive, submodular and subadditive valuations.
We will go into more detail on corresponding results below.

We will now move to discussing results more directly related to our work, specifically on the indivisible goods case and share-based fairness.
We start with describing what is known for additive valuations and MMS.
Following a sequence of works, the \cite{akrami2023breaking} showed that $(\frac{3}{4} + \frac{3}{3836})$-MMS allocations exist when all valuations are additive.
The approximation ratio has been recently improved in \cite{heidari2025}, where the authors managed to obtain a constant $\frac{10}{13}$, and \cite{huang2025} achieved a further improvement of $\frac{7}{9}$-MMS, which is currently best known for additive valuations.
As for to nonexistence results, \cite{KPW18-1, KPW18-2} were first to prove that MMS ($1$-MMS) allocations need not exist, even when the valuations are additive, for every $n \ge 3$.
Later, \cite{FST21} designed instances for for three agents with additive valuations in which no allocation gives more than $\frac{39}{40}$-MMS to all agents. For general $n$, \cite{FST21} showed an upper bound of $1 - \frac{1}{n^4}$. 

APS-allocations were introduced in \cite{BEF21APS}, where the authors showed that $\frac{3}{5}$-APS allocations exist if valuations are additive.

When it comes to allocations for submodular valuations, one of the first major steps was the work of \cite{BarKrish20} that proved existence of $0.21$-MMS.
The constant was later improved to $\frac{1}{3}$ by \cite{GhodsiHSSY22}, and further to $\frac{10}{27}$ in \cite{BUF23}.
On the upper bound side, \cite{GhodsiHSSY22} showed in the same work that $\frac{3}{4}$-MMS allocations do not exist for submodular valuations.
Later, \cite{KulkKulkMeh23} designed examples with three agents in which no allocation gives every agent more than $\frac{2}{3}$-MMS. 
The APS approximations for submodular valuations were first considered in \cite{BUF23}, where the authors achieve an approximation of $\frac{1}{3}$-APS. 

{
For XOS valuations, following a sequence of works~\cite{GhodsiHSSY22}, \cite{SS24}, \cite{AkramiMSS23}, \cite{FG25}, it is known that there always are $\frac{4}{17}$-MMS (and also $\frac{4}{17}$-APS) allocations. As mentioned in Section~\ref{sec:intro},  
a ratio of $\frac{1}{4}$ seems to be a natural barrier for the techniques used in these works. Our approach, based on different techniques, manages to pass the $\frac{1}{4}$-barrier and obtain an $\alpha$-MMS ($\alpha$-APS) allocation for a constant $\alpha > \frac{11}{40}$, provided that the number $n$ of agents is sufficiently large. }

Finally, there has been work on subadditive valuation, a class that contains all of the previously mentioned ones.
It was shown in \cite{GhodsiHSSY22} that it is impossible to achieve an approximation ratio greater than $\frac{1}{2}$-MMS.
Unlike for the other classes, in the subadditive case there are no known existential results for $c$-MMS allocations for any constant $c > 0$.
However, there are  works that obtain approximation ratios in terms of the number of agents $n$ or number of items $m$.
First works focused on approximations in terms of $m$, and the earlier mentioned work of \cite{GhodsiHSSY22} obtained $\Omega(\frac{1}{\log m})$-MMS allocations.
The approximation ratio was later improved to $\frac{(1 - o(1))\log\log m}{\log m}$ in the work of \cite{FG25}.
Moving to approximations in terms of $n$, \cite{SS24} achieved a major improvement of $\Omega(\frac{1}{\log n \log\log n})$-MMS, and the subsequent work of \cite{FH25} obtained a ratio of $\frac{1}{14\log n}$.
Recently, the authors of \cite{SS25} further pushed the approximation ratio to $\Omega(\frac{1}{(\log\log n)^2})$-MMS.
The work of \cite{Fei25} developed techniques that allow moving from \textit{multi-allocations} to allocations, and applied them to the result of \cite{SS25}, thus improving the ratio to the currently best-known $\frac{1}{8\log\log n}$-MMS.
When it comes to APS-allocations for subadditive valuations, the only approximation ratio we are aware of is $\frac{(1 - o(1))\log\log m}{\log m}$-APS from \cite{FG25}.

There has been additional work on MMS allocations for the instances when the number of agents $n$ is small.
For example, for the case of $n = 3$ agents, it was shown in \cite{AMNS17} that $\frac{7}{8}$-MMS allocations exist when all valuations are additive, which was further improved to $\frac{8}{9}$-MMS in \cite{GOURVES2019} and to $\frac{11}{12}$-MMS in \cite{FeNo22}.
For $n = 4$ additive agents, \cite{GhodsiHSSY22} proved an approximation ratio of $\frac{4}{5}$.
If one considers the same number $n = 4$ but allows the agents to be \textit{subadditive}, \cite{CCMS25} proved the existence of $\frac{1}{2}$-MMS allocations.
Our results, on the contrary, consider instances when number of agents $n$ is large.

In addition to the allocations of goods, there are also many works that consider allocations of \textit{chores}, i.e items of negative value that agents do not wish to get, but nevertheless have to be allocated.
The authors of \cite{ARSW17} were first to extend the definition of MMS to chores and showed an approximation ratio of $2$ for the case of additive valuations.
The constant was further improved to $\frac{4}{3}$ in \cite{BarKrish20}, and later to $\frac{11}{9}$ by \cite{HL21}, and finally to $\frac{13}{11}$ by \cite{HSH23}.

\section{Technical details}

An overview of our proofs appears in Section~\ref{sec:overview}. As the rigorous proofs are technical and long, they appear in the appendix. 
{In Appendix~\ref{identical}, we provide a complete description of the algorithm for identical valuation case, mentioned in Section~\ref{sec:identical}, together with analysis of the algorithm and proofs of correctness for all values of $n$.
In Appendix~\ref{sec:diffvalsintro}, we present the algorithm for the case of different valuations and large $n$, mentioned in Section~\ref{sec:different}, under the assumptions of \textit{no large intersections} and \textit{no large items}, and also show how to remove the assumption on large intersections.
Next, in Appendix~\ref{sec:martingale} we give a martingale-based proof that for large values of $n$, under the aforementioned assumptions, with high probability the individual $\beta$-paths of the agents do not deviate much from the greedy $\beta$-path of the identical valuations case.
The idea of the proof was briefly covered in Proposition~\ref{pro:difftheorem}.
After that, in Appendix~\ref{sec:bigitems}, we show how to remove the assumption {of} no large items, proving \cref{bipartite0} and presenting the allocation algorithms for each of its cases.
Appendices~\ref{sec:smallalpha} and~\ref{sec:bigalpha} contain technical proofs for the lower bounds for the $\beta$-processes.
Appendix~\ref{sec:doubling} contains the analysis of the doubling procedure.
} 

\subsection*{Acknowledgements}

This research was supported in part by the Israel Science Foundation (grant No. 1122/22).

\bibliographystyle{alpha}
\bibliography{main}

\appendix

\section{Gap between APS and MMS}
\label{identical}

In this section, we prove \cref{mainresultaps} for the case where all agents $i \in [n]$ have the same valuation function $v_i = v$.
This result immediately implies \cref{apsmmsgap}.
Consider some XOS valuation $v$, and let $\{(S, \lambda_S)\}_{S \in \Ss}$ be the APS partition of $\Mm$ with entitlement $1/n$.
{The following lemma is well-known, but for completeness we provide a proof.}
\begin{lemma}\label{giveitem}
    For any item $e \in \Mm$ and any valuation $v$, $\aps(\Mm\setminus\{e\}, v, \frac{1}{n - 1}) \geq \aps(\Mm, v, \frac{1}{n})$.
\end{lemma}
\begin{proof}
    Consider the original APS partition $\{(S, \lambda_S)\}_{S\in \Ss}$.
    It satisfies $\sum_{S \in \Ss}\lambda_S = 1$ and for every $e \in \Mm$, $\sum_{S \ni e}\lambda_S = 1/n$.
    Construct a new APS partition $\{(S, \mu_{S})\}_{S \in \Ss}$ as follows.
    For every $S \in \Ss$ such that $e \in S$, set $\mu_S = 0$, and for every $S \in \Ss$ such that $e \notin S$, set $\mu_S = \frac{n}{n - 1}\lambda_S$.
   Observe that 
   \[
        \sum_{e\notin S}\lambda_S = \sum_{S}\lambda_S - \sum_{e\in S}\lambda_S = 1 - \frac{1}{n} = \frac{n - 1}{n} \implies \sum_{S}\mu_S = \frac{n}{n - 1}\sum_{e\notin S}\lambda_S = 1,
   \]
   and at the same time for any $e'\neq e$:
   \[
        \sum_{S\ni e'}\mu_S  = \sum_{\substack{S\ni e',\\j\notin S}}\mu_S  = \frac{n}{n - 1}\sum_{\substack{S\ni e',\\e\notin S}}\lambda_S \leq \frac{n}{n - 1}\sum_{S\ni e'} \lambda_S =\frac{1}{n - 1}. 
   \]
   So, $\{(S, \mu_{S})\}_{S \in \Ss}$ is a valid APS partition.
   Meanwhile, since the {collection} of sets $S$ with $\mu_S > 0$ {is a subset of the collection of sets with $\lambda_S >0$}, 
   the minimum value of such a set {does not decrease}. 
   Hence the APS-value for this partition is at least the APS value of the original $\{(S, \lambda_{S})\}_{S \in \Ss}$.
\end{proof}

The following lemma will provide us with convenient structure of the APS partition $\{(S, \lambda_S)\}_{S\in \Ss}$.
\begin{lemma}\label{wlog}
    Let $\{(S, \lambda_S)\}_{S\in \Ss}$ be a fractional $\frac{1}{n}$-partition corresponding to $\aps(\Mm, v, \frac{1}{n})$, we refer to it as $\aps$-partition.
    Without loss of generality, we can assume that 
    \begin{enumerate}
        \item there are at most $m + 1$ bundles in the support of $\{(S, \lambda_S)\}_{S\in \Ss}$;

        \item for every $S \in \Ss$ with $\lambda_S > 0$, $v(S) = 1$, and $v(S) = v_S(S)$ for some linear function $v_S$ defined only on the items from $S$;

        \item every ``minimal'' bundle (one in which removing any item decreases its value) is a sub-bundle of some $S \in \Ss$;    
    \end{enumerate}
\end{lemma}
\begin{proof}
    For 1), consider the linear program representation of the APS-definition.
    For a value $z \geq 0$, let $\Ss_z$ denote all $S \in 2^{\Mm}$ such that $v(S) \geq z$.
    We seek for the maximal value of $z$ such that the following LP has a feasible solution: the variables are $\lambda_S \geq 0$ for $S \in \Ss_z$, and the constraints are $\sum_{S \in \Ss}\lambda_S = 1$ and $\sum_{S \ni j}\lambda_S = 1/n$ for every $j \in \Mm$.
    {The LP has $m+1$ constraints. Consequently, it has a} basic feasible solution 
    {with} at most $m + 1$ non-zero variables $\lambda_S$. {This gives an APS partition supported on at most $m + 1$ bundles.}

    {To see 2), we first scale valuation $v$ so that {$\aps(\Mm, v, 1/n) = 1$. Then,} for every $S \in \Ss$ with $\lambda_S > 0$, $v(S) \geq 1$.
    Next, let $S \in \Ss$, and let $v_k$ be a linear function from the XOS-definition of $v$ such that $v(S) = v_k(S)$ (i.e, $v_k$ is the maximizing linear funciton for $v$).
    Introduce a new linear function $\wv_{S}$ defined as follows: for any $e \in S$, $\wv_{S}(e) = v_k(e)/v_k(S)$, and for any $e \notin S$, $v_{S}(e) = 0$.
    We repeat this process for every $S \in \Ss$ with $\lambda_S > 0$, and let $\wv$ be an XOS valuation function defined as maximum over all $\wv_{S}$ obtained as above.
    By construction, $\aps(\Mm, \wv, 1/n) = 1 = \aps(\Mm, v, 1/n)$. In addition, for every $S \in \Ss$ with $\lambda_S > 0$, $\wv(S) = 1$, and $\wv(S) = v_{S}(S)$ for some linear $v_S$ which is nonzero only on $S$.
    {Since $v(S) \geq \wv(S)$ for every $S$,} 
    proving the theorem for $\wv$ implies the same result for the original $v$.
    
    For 3), assume that the valuation $v$ satisfies the conditions of 2).
    Then, for every \textit{subset of items} $M \subseteq \Mm$, $v(M) = \max_{S \in \Ss}v_S(M)$ where $v_S$ is a linear function that is equal to $0$ for any $e \notin S$.
    Consider some $M \subseteq \Mm$ that is not contained in any $S \in \Ss$, and let $S^* \in \Ss$ be such that linear function $v_{S^*}(M)$ maximizes $v(S) = \max_{S \in \Ss}v_S(M)$.
    Then, there exists $e \in M\setminus S^*$, hence $v_{S^*}(e) = 0$.
    As a result, $v(M \setminus \{e\}) \geq v_{S^*}(M\setminus \{e\}) = v_{S^*}(M) = v(M)$, implying that $M$ is not minimal.
    So, every minimal bundle must be a sub-bundle of some $S \in \Ss$.    
    }
\end{proof}

In addition, we will need the following assumption on the values of items in $\Mm$.
\begin{claim}\label{smallitems}
    Let $\{(S, \lambda_S)\}_{S\in \Ss}$ be a fractional $\frac{1}{n}$-partition corresponding to $\aps(\Mm, v, \frac{1}{n})$ for XOS valuation $v$.
    Without loss of generality, we can assume that for every $e \in \Mm$, $v(e) < \alpha$.
\end{claim}
\begin{proof}
    Assume that the APS-value for $\{(S, \lambda_{S})\}_{S \in \Ss}$ is $1$.
    If there exists an item $e \in \Mm$ with $v(e) \geq \alpha$, give this item to some agent.
    If we give away all items $e$ of value at least $v(e) \geq \alpha$ to agents, those agents receive $\alpha$-fraction of their APS, while by \cref{giveitem} the APS-value for the remaining agents and items {does not decrease}.
    As a result, we can assume that all items of value at least $\alpha$ have been given away to some agents, and it remains to find an allocation for the remaining items and agents.
\end{proof}

\subsection{Allocation algorithm}
\label{sec:greedalgo}

We produce an allocation in $n$ steps. At step $i$ we pick a bundle for agent $i$, removing the allocated items from $\Mm$.
For $i \in [n]$, let $\Mm^i$ denote the items that remain at the beginning of step $i$ (so $\Mm^1 = \Mm$), and for a bundle $S \in \Ss$, let $S^i := S \cap \Mm^i$.
Define $\beta^i := \sum_{S}\lambda_S\cdot v(S^i)$, {interpreted as} the total remaining value of the APS partition at the beginning of step $i$.
By \cref{wlog}, $\beta^1 = 1$.
{We call a bundle $B$ \textit{acceptable} if $v(B) \geq \alpha$. 
At step $i$, we pick some minimal acceptable bundle $B \subseteq \Mm^i$, such that the choice of $B$ maximizes $\beta^{i + 1}$.
Note that by \cref{wlog}, $B$ being minimal implies that there exists $(S, \lambda_S) \in \Ss$ with  $\lambda_S > 0$ such that $B\subseteq S$ (i.e, it is a subset of an APS bundle).
We give the chosen $B$ to agent $i$ and remove the items of $B$ from $\Mm$, i.e $\Mm^{i + 1} := \Mm^i\setminus B$.}

\begin{algorithm}[h]
\caption{Greedy allocation algorithm}\label{basealgo}
	\begin{algorithmic}[1]
		\Require{$(n, m, \Mm, v)$, and an APS partition $\{(S, \lambda_S)\}_{S \in \Ss}$ satisfying \cref{wlog} and \cref{smallitems}.}
		\Ensure{Disjoint sets $M_1, \ldots, M_n$ such that $v(M_i) \geq \alpha$ for all $i \in [n]$.}
		\State Initialize $\Mm^1 = \Mm$
		\For{$i = 1,\ldots, n$}
		      \State Let $B\subseteq \Mm^i$ be a minimal acceptable bundle maximizing $\beta^{i + 1} = \sum_S\lambda_S \cdot v(S \cap (\Mm^i \setminus B))$
            \State 
            Give agent $i$ bundle $M_i:= B$.
            \State Update $\Mm^{i + 1} := \Mm^i \setminus B$.
		\EndFor
		\Return $M_1, \ldots, M_n$.
	\end{algorithmic}
\end{algorithm}
For the purpose of analysis, we will consider another, ``efficient'' version of this algorithm, where instead of picking {the ``best'' minimal} 
acceptable bundle at every step, the algorithm chooses $B$ from a specific list of $O(m)$ acceptable bundles.

Suppose that we want to pick a bundle $B$ that only contains items from some subset $\Mm'\subseteq \Mm$.
We partition all bundles $S \in \Ss$ into sub-bundles based on the value $v(S\cap \Mm')$, as follows.
\begin{lemma}\label{partition}
    Let $\Mm' \subseteq \Mm$ be fixed.
    For $t \geq 1$, let $\Ss(\Mm', t\alpha)$ contain all bundles $S \in \Ss$ such that $(t - 1)\alpha \leq v(S\cap \Mm') < t\alpha$.
    For every $t\geq 2$ and every $S \in \Ss(\Mm', t\alpha)$, the following holds.
    \begin{enumerate}
        \item if $t$ is even, there exist $t/2$ sub-bundles $B_1, \ldots, B_{t/2} \subseteq S$ such that for any $k \in [t/2]$, $v(B_k \cap \Mm') \geq \alpha$, and for every $j \in S$, item $j$ is contained in \textbf{at most one} 
        sub-bundle $B_k$.

        \item if $t$ is odd, there exist $t$ sub-bundles $B_1, \ldots, B_t \subseteq S$ such that for any $k \in [t]$, $v(B_k\cap \Mm')\geq \alpha$, and for every $j \in S$, item $j$ is contained in \textbf{at most two} of sub-bundles $B_k$.
    \end{enumerate}
\end{lemma}
\begin{proof}
    For a given set $S$, perform the following procedure: if two items $e, e'$ have combined value $v_S(\{e, e'\} \cap \Mm') < \alpha$, we ``unify'' them, replacing $e, e'$ in $S$ with a combined item $\{e, e'\}$ of value $v_S(\{e\}\cap \Mm') + v_S(\{e'\}\cap \Mm')$ and treating both $e, e'$ like a single item.
    This is a ``virtual'' procedure that is performed in order to simplify the analysis, so that if some property / action is applied to the unified item, it is applied simultaneously to all of the items in the union.
    We repeat the process until in given set $S$ any two (unified) items $e, e'$ have combined $v_S(\cdot \cap M)$-value at least $\alpha$ (so we can no-longer unify them without making combined value at least $\alpha$).

    Let $S \in \Ss(\Mm', t\alpha)$, then set $S$ must contain at least $t$ (unified) items that have not been picked yet.
    Let $e_1, \ldots, e_t \in S$ be arbitrary different items, note that any two of them together have value at least $\alpha$.
 
    {If $t$ is even} then 
    create $t/2$ bundles $B_1, \ldots, B_{t/2}$ of value at least $\alpha$ as follows: {$B_k = \{e_{2k - 1}, e_{2k}\}$ for $1 \leq k \le t/2$}.
    Every $e \in S$ is contained in {at most} 
    one of $B_1, \ldots, B_{t/2}$, and all $B_k$ have value at least $\alpha$, as each contains at least two items.

    {If $e$ is odd then}
    create $t$ bundles $B_1, \ldots, B_t\subseteq S$ of value at least $\alpha$ as follows: $B_k = \{e_{k}, e_{k + 1}\}$ for $1 \leq k < t$, and {$B_t = \{e_t, e_1\}$}. 
    Every $j \in S$ is contained in at most two of $B_1, \ldots, B_t$, and all $B_k$ for $k \in [t]$ have value at least $\alpha$, as each contains at least two items.
\end{proof}

We apply \cref{partition} with set $\Mm' = \Mm^i$, and obtain classes $\Ss^i(t\alpha) := \Ss(\Mm^i, t\alpha)$.
Since we consider $\alpha \leq 1/3$, and by \cref{wlog} there are no sets of value greater than $1$, we will only have bundles of classes $\Ss^i(\alpha)$, $\Ss^i(2\alpha)$, $\Ss^i(3\alpha)$ and $\Ss^i(4\alpha)$, where for every $S \in \Ss^i(4\alpha)$ it holds $3\alpha \leq v(S\cap \Mm^i) \leq 1$.
We denote $\Ss^i(1) := \Ss^i(4\alpha)$ for clarity.
Using \cref{partition}, we are ready to give the description of the efficient algorithm (\cref{algo}).


\begin{algorithm}[h]
\caption{Efficient allocation algorithm}\label{algo}
	\begin{algorithmic}[1]
		\Require{$(n, m, \Mm, v)$, and an APS partition $\{(S, \lambda_S)\}_{S \in \Ss}$ satisfying \cref{wlog} and \cref{smallitems}.}
		\Ensure{Disjoint sets $B^1, \ldots, B^n$ such that $v(B^i) \geq \alpha$ for all $i \in [n]$.}
		\State Initialize $\Mm^1 = \Mm$
		\For{$i = 1,\ldots, n$}
            \State Apply \cref{partition} with $\Mm' = \Mm^i$, obtaining classes $\Ss^i(\alpha), \Ss^i(2\alpha), \Ss^i(3\alpha), \Ss^i(1)$ 
            \State Let $\Bb^i$ be a collection of sets, initially $\Bb^i = \varnothing$
            \For{$S \in \Ss$ with $\lambda_S > 0$}
                \If{$S \in \Ss^i(2\alpha)$} add a copy of $S \cap \Mm^i$ to $\Bb^i$ \EndIf
                \If{$S \in \Ss^i(3\alpha)$} let $B_1, B_2, B_3$ be sub-bundles of {$S \cap \Mm^i$} from \cref{partition}
                    \State for every $k = 1, 2, 3$, add a copy of $B_k$ to $\Bb^i$ \EndIf
                \If{$S \in \Ss^i(1)$} let $B_1, B_2$ be sub-bundles of {$S\cap \Mm^i$} from \cref{partition}
                    \State for every $k = 1, 2$, add a copy of $B_k$ to $\Bb^i$ \EndIf
            \EndFor
            \State Let $B \in \Bb^i$ be an acceptable bundle maximizing $\beta^{i + 1} = \sum_S\lambda_S \cdot v(S \cap (\Mm^i \setminus B))$
            \State 
            Give agent $i$ bundle $B^i := B$.
            \State Update $\Mm^{i + 1} := \Mm^i \setminus B$.
		\EndFor
		\Return $B^1, \ldots, B^n$.
	\end{algorithmic}
\end{algorithm}

The algorithm surely succeeds if $\beta^n \geq \alpha$, as by definition there would exist at least one $S \in \Ss$ with $\lambda_S > 0$ and $v(S^n) \geq \alpha$.
Let $\alpha^*$ denote the largest value such that for every $\alpha \leq \alpha^*$ it holds that $\beta^n \geq \alpha$.
We present a lower bound on the value of $\alpha^*$.

One of the key observations that we use in the analysis of \cref{algo} is the following.
Suppose that at every step $i$, instead of picking a maximizing bundle $B$, we would select $B$ randomly, according to some distribution over acceptable bundles.
If the distribution with which we pick $B$ is such that every item $e \in \Mm$ is chosen with some small probability, then we can guarantee that in expectation, the total APS value should not decrease by a lot.
\begin{lemma}\label{probval}
    Let $i \in [n]$, and let $\Bb^i$ be as in \cref{algo}.
    Let $\mu^i$ be an arbitrary probability distribution over bundles in $\Bb^i$.
    For $e \in \Mm$, let $\mu^i(e)$ be the probability that we pick item $e$ if we sample a random $B \sim \mu^i$.
    Suppose that for every $e \in \Mm$, the value $\mu^i(e)$ is upper bounded by some $p^i$.
    Then, in expectation over $B \sim \mu^i$
    \[
        \E[\mu^i]{\beta^{i + 1}\mid \beta^i} = \E[B\sim\mu^i]{\sum_S\lambda_S \cdot v(S \cap (\Mm^i \setminus B))} \geq (1 - p^i)\cdot \beta^i.
    \]
\end{lemma}
\begin{proof}
    {By \cref{wlog}, for a given set $S$ it holds $v(S) = v_S(S)$ for some additive function $v_S$ that is $0$ for all items $e \notin S$.
    Similarly, for every $i \in [n]$ it holds $v(S^i) = v_S(S^i)$.
    For convenience, define a linear function $v_S^i$ which is equal to $v_S$ on all items from $S^i$ and $0$ elsewhere.}
    
    For $e \in \Mm$, the expected loss $\beta^i = \sum_Sv(S^i)\lambda_S =  \sum_Sv^i_S(S)\lambda_S$ suffers by giving item $e$ is
    \[
        \left(\sum_{S\ni e}v^i_S(e)\lambda_S\right)\cdot \mu^i(e).
    \]
    Therefore, the total expected loss over $B \sim \mu^i$ that $\beta^i$ suffers after step $i$ is
    \[
        \sum_e\left(\sum_{S\ni e}v^i_S(e)\lambda_S\right)\cdot \mu^i(e) = \sum_e\sum_{S\ni e}v^i_S(e)\mu^i(e)\lambda_S = \sum_S\left(\sum_{e \in S}v^i_S(e)\mu^i(e)\right)\lambda_S.
    \]
    Since for every $e \in \Mm$, $\mu^i(e)\leq p^i$,
    the expected total loss of $\beta^i$ after step $i$ is at most 
    \[
        \sum_S\left(\sum_{e \in S}v^i_S(e)\mu^i(e)\right)\lambda_S \leq p^i\cdot \sum_S\sum_{e \in S}v^i_S(e)\lambda_S = p^i\cdot \sum_Sv^i_S(S)\lambda_S = p^i\cdot \beta^i,
    \]
    and therefore $\E{\sum_S\lambda_S \cdot v(S \cap (\Mm^i \setminus B))} \geq \E{\sum_S\lambda_S \cdot v^i_S(S \cap (\Mm^i \setminus B))} \geq (1 - p^i)\cdot \beta^i$.
\end{proof}
Denote $\Ss^i_\alpha, \Ss^i_{2\alpha}, \Ss^i_{3\alpha}, \Ss^i_1 = \Ss^i(\alpha), \Ss^i({2\alpha}), \Ss^i({3\alpha}), \Ss^i(1)$ for convenience.
In the following corollary, we provide a concrete distribution $\mu^i$ and probability bound $p^i$ for every step $i$ of the algorithm.
\begin{corollary}\label{choice}
    Consider a partition $\Ss^i_\alpha, \Ss^i_{2\alpha}, \Ss^i_{3\alpha}, \Ss^i_1$ at step $i \in [n]$.
    Let $\Lambda^i_\alpha := \sum_{S \in \Ss^i_\alpha}\lambda_S$, and $\Lambda^i_{2\alpha}$, $\Lambda^i_{3\alpha}$, $\Lambda^i_1$ are defined accordingly as sums of weights of sets in corresponding classes.
    Let 
    \[
        \Lambda^i := 0\cdot \Lambda^i_\alpha + 1\cdot \Lambda^i_{2\alpha} + \frac{3}{2}\cdot \Lambda^i_{3\alpha} + 2\cdot \Lambda^i_{1}.
    \]
    It holds for every $i \in [n]$ that $\beta^{i + 1} \geq (1 - \frac{1}{n\Lambda^i})\cdot \beta^i$.
\end{corollary}
\begin{proof}
    It is sufficient to show that at every step $i \in [n]$ there exists a choice of an acceptable bundle $B$ such that $\sum_S\lambda_S \cdot v(S \cap (\Mm^i \setminus B))\geq (1 - \frac{1}{n\Lambda^i})\cdot \beta^i$.
    To do so, we are going to construct a probability distribution $\mu^i$ over the bundles $\Bb^i$ in \cref{algo}, such that for every $j \in \Mm$ the probability to pick item $j$ if we sample a random $B \sim \mu^i$ is at most $\frac{1}{n\Lambda^i}$.
    Then, by \cref{probval} in expectation over $B \sim \mu^i$ it holds that $\sum_S\lambda_S \cdot v(S \cap (\Mm^i \setminus B)) \geq (1 - \frac{1}{n\Lambda^i})\cdot \beta^i$, implying that there exists at least one choice of $B \in \Bb^i$ that guarantees $\beta^{i + 1} \geq (1 - \frac{1}{n\Lambda^i})\cdot \beta^i$.

    Note that $\Bb^i$ is a multiset, as by construction we may insert several copies of the same $B$ into $\Bb^i$ in \cref{algo}.
    Consequently, the distribution $\mu^i$ over $\Bb^i$ will have separate values of $\mu^i_B$ for different copies of $B \in \Bb^i$. 
    To define the values of $\mu^i$, we consider all $S \in \Ss$ with $\lambda_S > 0$.
    \begin{itemize}
        \item If $S \in \Ss^i_\alpha$, no subset of $S$ appears in $\Bb^i$.
        \item If $S \in \Ss^i_{2\alpha}$, set $\mu^i_{S^i} =  \frac{\lambda_S}{\Lambda^i}$ for the corresponding copy of $S^i$ inserted in $\Bb^i$.
        \item if $S \in \Ss^i_{3\alpha}$, let $B_1, B_2, B_3 \subseteq S$ be sub-bundles of $S^i$ constructed as in \cref{partition}.
        For every $k = 1, 2, 3$, set $\mu^i_{B_k} = \frac{\lambda_S}{2\Lambda^i}$ for the corresponding copy of $B_k$ inserted in $\Bb^i$.
        \item if $S \in \Ss^i_{1}$, let $B_1, B_2 \subseteq S$ be sub-bundles of $S^i$ constructed as in \cref{partition}.
        For every $k = 1, 2$, set $\mu^i_{B_k} = \frac{\lambda_S}{\Lambda^i}$ for the corresponding copy of $B_k$ inserted in $\Bb^i$.
    \end{itemize}
    \begin{lemma}\label{distribution}
        It holds that $\sum_{B \in \Bb^i}\mu_B^i = 1$, so $\mu^i$ is a valid probability distribution.
         Furthermore, for every item $e \in \Mm$, the probability $\mu^i(e)$ of picking item $e$ if we sample a random $B \sim \mu^i$ is at most $\frac{1}{n\Lambda^i}$.
        That is,
        \[
            \mu^i(e) = \sum_{\substack{B \in \Bb^i \\ B \ni e}}\mu^i_B \leq \frac{1}{n\Lambda^i}.
        \]
\end{lemma}
\begin{proof}
    First, partition all bundles $B \in \Bb^i$ into classes $\Bb^i_\alpha, \Bb^i_{2\alpha}, \Bb^i_{3\alpha}, \Bb^i_{1}$ depending on what class $\Ss^i_\alpha, \Ss^i_{2\alpha}, \Ss^i_{3\alpha}, \Ss^i_1$ their ``parent'' bundle $S \in \Ss$ belongs to.
    If some bundle $B$ has several copies in $\Bb^i$, each copy comes from a separate ``parent'' bundle $S$, and we put each copy of $B$ into the corresponding $\Bb^i$-class based on each parent $S$.
    Then, observe that:
    \begin{enumerate}
        \item if $S \in \Ss^i_\alpha$, no sub-bundles of it contribute to $\mu^i$ (i.e $\Bb^i_\alpha = \varnothing$);
        \item if $S \in \Ss^i_{2\alpha}$, then $S^i \in \Bb^i$ and it contributes exactly $\frac{\lambda_S}{\Lambda^i}$ to $\mu^i$;
        \item if $S \in \Ss^i_{3\alpha}$, then sub-bundles of $S^i$ are in $\Bb^i_{3\alpha}$ and in total contribute exactly $\frac{3}{2}\cdot \frac{\lambda_S}{\Lambda^i}$ to $\mu^i$;
        \item if $S \in \Ss^i_{1}$, then sub-bundles of $S^i$ are in $\Bb^i_{1}$ and in total contribute exactly $2\cdot \frac{\lambda_S}{\Lambda^i}$ to $\mu^i$.
    \end{enumerate}
    As a result,
    \begin{multline*}
        \sum_{B \in \Bb^i}\mu^i_B = \sum_{B \in \Bb^i_\alpha}\mu_B^i + \sum_{B \in \Bb^i_{2\alpha}}\mu_B^i + \sum_{B \in \Bb^i_{3\alpha}}\mu_B^i + \sum_{B \in \Bb^i_1}\mu_B^i\\
        = \frac{1}{\Lambda^i}\left(0\cdot \Lambda^i_\alpha + 1\cdot \Lambda^i_{2\alpha} + \frac{3}{2}\cdot \Lambda^i_{3\alpha} + 2\cdot \Lambda^i_{1}\right) = 1.
    \end{multline*}

    Next, let $e \in \Mm$ be some item.
    It is clear that
    \[
        \sum_{\substack{B \in \Bb^i \\ B \ni e}}\mu^i_B = \sum_{\substack{B \in \Bb^i_{\alpha} \\ B \ni e}}\mu^i_B + \sum_{\substack{B \in \Bb^i_{2\alpha} \\ B \ni e}}\mu^i_B + \sum_{\substack{B \in \Bb^i_{3\alpha} \\ B \ni e}}\mu^i_B + \sum_{\substack{B \in \Bb^i_1 \\ B \ni e}}\mu^i_B.
    \]
    \begin{enumerate}
        \item if $B \in \Bb^i_\alpha$, its parent bundle $S$ contributes $0$ to the sum above.
        \item if $B \in \Bb^i_{2\alpha}$, its parent bundle $S$ contributes exactly $\frac{\lambda_S}{\Lambda^i}$ to the sum above, as there is only one sub-bundle of $S$ (set $S^i$ itself) containing $e$ in $\Bb^i$.
        \item if $B \in \Bb^i_{3\alpha}$, its parent bundle $S$ contributes exactly $2\cdot \frac{\lambda_S}{2\Lambda^i}$ to the sum above, as $e$ is contained in only two out of tree sub-bundles of $S^i$ that belong to $\Bb^i$.
        \item if $B \in \Bb^i_{1}$, its parent bundle $S$ contributes exactly $\frac{\lambda_S}{\Lambda^i}$ to the sum above, as there is only one out of two sub-bundles of $S^i$ containing $e$ in $\Bb^i$.
    \end{enumerate}
    As a result
    \[
        \sum_{\substack{B \in \Bb^i \\ B \ni e}}\mu^i_B = \sum_{\substack{S \in \Ss^i_\alpha \\ S \ni e}}0 + \sum_{\substack{S \in \Ss^i_{2\alpha} \\ S \ni e}}\frac{\lambda_S}{\Lambda^i} + \sum_{\substack{S \in \Ss^i_{3\alpha} \\ S \ni e}}2\cdot \frac{\lambda_S}{2\Lambda^i} + \sum_{\substack{S \in \Ss^i_1 \\ S \ni e}}\frac{\lambda_S}{\Lambda^i} \leq \frac{1}{\Lambda^i}\sum_{\substack{S \in \Ss\\ S\ni e}}\lambda_S \leq \frac{1}{n\Lambda^i}.
    \]
\end{proof}
The claim follows.
\end{proof}

\subsection{Bounding the probability}
\label{sec:pBound}

{
As shown in \cref{choice}, for every $i \in [n]$ it holds that $\beta^{i + 1} \geq (1 - \frac{1}{n\Lambda^i})\cdot \beta^i$ where
\[\Lambda^i = 0\cdot \Lambda^i_\alpha + 1\cdot \Lambda^i_{2\alpha} + \frac{3}{2}\cdot \Lambda^i_{3\alpha} + 2\cdot \Lambda^i_{1}.\]
This bound depends the values of $\Lambda^i_\alpha, \Lambda^i_{2\alpha}, \Lambda^i_{3\alpha}, \Lambda^1$, which are difficult to compute exactly without more information, as they depend on the exact execution of \cref{algo}.
To tackle this, for every step $i$ we obtain a universal upper-bound the expression $\frac{1}{n\Lambda^i}$ that depends only on $\beta^i$.
Note that to upper bound $\frac{1}{n\Lambda^i}$ is equivalent to lower bound $\Lambda^i$.}

\begin{claim}\label{lambdabound}
    Suppose that at step $i$ the total remaining value {$\beta^i$ satisfies} $\beta^i \geq \beta$ for some $\beta \geq \alpha$.
    Then $\Lambda^i$ is at least the optimum $W^i_\beta$ of the following optimization problem over the variables $w_\alpha, w_{2\alpha}, w_{3\alpha}, w_1$:
    \[
    \begin{aligned}
        &\minimize&& W := 0\cdot w_\alpha + 1\cdot w_{2\alpha} + \frac{3}{2}\cdot w_{3\alpha} + 2\cdot w_{1}\\
        &\subto&& w_\alpha + w_{2\alpha} + w_{3\alpha} + w_1 = 1;\\
        &&& \alpha \cdot w_\alpha + 2\alpha \cdot w_{2\alpha} + 3\alpha \cdot w_{3\alpha} + 1\cdot w_{1} \geq \beta; \\\
        &&& w_\alpha, w_{2\alpha}, w_{3\alpha}, w_1 \geq 0.
    \end{aligned}\tag{$\OPT^i_\beta$}\label{opt}
    \]
\end{claim}
\begin{proof}
    Consider the values $\Lambda^i_\alpha, \Lambda^i_{2\alpha}, \Lambda^i_{3\alpha}, \Lambda^1$ corresponding to $\Lambda^i$, and take the following solution: $w_\alpha = \Lambda^i_\alpha$, $w_{2\alpha} = \Lambda^i_{2\alpha}$, $w_{3\alpha} = \Lambda^i_{3\alpha}$, $w_1 = \Lambda^1$.
    Observe that for these values of $w_\alpha, w_{2\alpha}, w_{3\alpha}, w_1$, the sum $W = 0\cdot w_\alpha + 1\cdot w_{2\alpha} + \frac{3}{2}\cdot w_{3\alpha} + 2\cdot w_{1} = \Lambda^i$.
    At the same time, the constraint $\alpha \cdot w_\alpha + 2\alpha \cdot w_{2\alpha} + 3\alpha \cdot w_{3\alpha} + 1\cdot w_{1} \geq \beta$ does hold, as by definition of classes $\Ss^i$
    \begin{multline*}
        \beta = \sum_{S \in \Ss^i_\alpha}v^i(S)\lambda_S + \sum_{S \in \Ss^i_{2\alpha}}v^i(S)\lambda_S \sum_{S \in \Ss^i_{3\alpha}}v^i(S)\lambda_S \sum_{S \in \Ss^i_1}v^i(S)\lambda_S \\
        < \alpha \Lambda^i_\alpha + 2\alpha \Lambda^i_{2\alpha} + 3\alpha\Lambda^i_{3\alpha} + \Lambda^1 = \alpha \cdot w_\alpha + 2\alpha \cdot w_{2\alpha} + 3\alpha \cdot w_{3\alpha} + 1\cdot w_{1}.
    \end{multline*}
    {As a result, for a partition $\Ss^i_\alpha, \Ss^i_{2\alpha}, \Ss^i_{3\alpha}, \Ss^i_1$ at step $i$ of the algorithm, the corresponding weights $\Lambda^i_\alpha, \Lambda^i_{2\alpha}, \Lambda^i_{3\alpha}, \Lambda^1$ are a feasible solution to \ref{opt} with objective value $\Lambda^i$. The lemma follows.}
\end{proof}

Next, we will determine the exact structure of the optimal solutions of the LP in \cref{lambdabound}.
\begin{lemma}\label{wclass}
    Let $w_\alpha, w_{2\alpha}, w_{3\alpha}, w_{1}$ be an optimal solution to \ref{opt} for $\beta \geq \alpha$.
    Then:
    \begin{itemize}
        \item if $\alpha \leq 3/11$, we can assume that $w_{2\alpha} = w_{3\alpha} = 0$;
        \item if $\alpha \geq 3/11$ and $\beta < 3\alpha$, we can assume that $w_{2\alpha} = w_{1} = 0$;
        \item if $\alpha \geq 3/11$,  and $\beta \geq 3\alpha$, we can assume that $w_{\alpha} = w_{2\alpha} = 0$.
    \end{itemize}
    Note that when $\alpha = 3/11$ exactly, it is possible to make either assumption.
\end{lemma}
\begin{proof}
    {Since $w_\alpha, w_{2\alpha}, w_{3\alpha}, w_1$ is optimal, we can assume that $\alpha w_\alpha + 2\alpha w_{2\alpha} + 3\alpha w_{3\alpha} + w_1 = \beta$ exactly.}
    To begin with, we claim for any value of $\alpha \leq \alpha^* < 1/3$ that if $w_{2\alpha} > 0$, then we can split the entirety of its weight between $w_{\alpha}$ and $w_{1}$ in such a way that the total value {still equals} $\beta$, while the sum $W := 0\cdot w_\alpha + 1\cdot w_{2\alpha} + \frac{3}{2}\cdot w_{3\alpha} + 2\cdot w_{1}$ decreases.
    {
    Let $u_\alpha, u_{2\alpha}, u_{3\alpha}, u_1$ be the new solution, defined as follows:
    \[
        u_{2\alpha} = 0,\quad u_{3\alpha} = w_{3\alpha}, \quad u_\alpha = w_\alpha + \frac{1 - 2\alpha}{1-\alpha}w_{2\alpha},\quad u_1 = w_1 + \frac{\alpha}{1 - \alpha}w_{2\alpha}.
    \]
    It is easy to see that 
    \[u_\alpha + u_{2\alpha} + u_{3\alpha} + u_1 = w_\alpha + w_{2\alpha} + w_{3\alpha} + w_1 = 1,\]
    and
    \begin{multline*}
        \alpha u_\alpha + 2\alpha u_{2\alpha} + 3\alpha u_{3\alpha} + u_1 = \alpha w_\alpha + \frac{\alpha - 2\alpha^2}{1 - \alpha}w_{2\alpha} + 3\alpha w_{3\alpha} + w_1 + \frac{\alpha}{1 - \alpha}w_{2\alpha} \\
    = \alpha w_\alpha + 2\alpha w_{2\alpha} + 3\alpha w_{3\alpha} + w_1 = \beta,
    \end{multline*}
    so $u_\alpha, u_{2\alpha}, u_{3\alpha}, u_1$ is feasible.
    At the same time,
    \[0\cdot u_\alpha + u_{2\alpha} + \frac{3}{2}u_{3\alpha} + 2u_1 = \frac{2\alpha}{1-\alpha}w_{2\alpha} + \frac{3}{2}w_{3\alpha} + 2w_1 < w_{2\alpha} + \frac{3}{2}w_{3\alpha} + 2w_1\]
    where the last inequality holds as $\frac{2\alpha}{1 - \alpha} < 1$ when $\alpha < 1/3$.}
   {Thus, solution $u_\alpha, u_{2\alpha}, u_{3\alpha}, u_1$ achieves strictly smaller value than $w_\alpha, w_{2\alpha}, w_{3\alpha}, w_1$.
    Therefore, any solution with $w_{2\alpha} > 0$ is not optimal.
    For the rest of the proof, we will assume that $w_{2\alpha} = 0$ always.}
    
    Next, we claim that if $\alpha \leq 3/11$ and $w_{3\alpha} > 0$, we can split the entirety of its weight between $w_{\alpha}$ and $w_{1}$ in such a way that the total value still equals $\beta$, while the sum $W$ does not increase.
    {Let $u_\alpha, u_{2\alpha}, u_{3\alpha}, u_1$ be the new solution, defined as follows:
    \[
        u_{2\alpha} = 0,\quad u_{3\alpha} = 0, \quad u_\alpha = w_\alpha + \frac{1 - 3\alpha}{1-\alpha}w_{3\alpha},\quad u_1 = w_1 + \frac{2\alpha}{1 - \alpha}w_{3\alpha}.
    \]
    Note that $w_{2\alpha} = u_{2\alpha} = 0$.
    It is easy to see that
    \[
        u_\alpha + u_{3\alpha} + u_1 = w_\alpha + w_{3\alpha} + w_1 = 1,
    \]
    and 
    \[
        \alpha u_\alpha + 3\alpha u_{3\alpha} + u_1 = \alpha w_\alpha + \frac{\alpha - 3\alpha^2}{1 - \alpha}w_{3\alpha} + w_1 + \frac{2\alpha}{1-\alpha}w_{3\alpha} = \alpha w_\alpha + 3\alpha w_{3\alpha} + w_1 = \beta,
    \]
    so $u_\alpha, u_{2\alpha}, u_{3\alpha}, u_1$ is feasible.
    At the same time,
    \[
        0 \cdot u_\alpha + u_{2\alpha} + \frac{3}{2}u_{3\alpha} + 2u_1 = \frac{4\alpha}{1-\alpha}w_{3\alpha} + 2w_1 \leq \frac{3}{2}w_{3\alpha} + 2w_1, 
    \]
    where the last inequality holds if and only if
    \[
    \frac{4\alpha}{1-\alpha} \leq \frac{3}{2} \iff 8\alpha \leq 3 - 3\alpha \iff \alpha \leq 3/11.
    \]
    So, solution $u_\alpha, u_{2\alpha}, u_{3\alpha}, u_1$ is not worse than $w_\alpha, w_{2\alpha}, w_{3\alpha}, w_1$.
    Therefore, when $\alpha \leq 3/11$, any solution with $w_{3\alpha} > 0$ can be replaced with a solution that has $w_{3\alpha} = 0$ without value increase.}

    Suppose now that $\alpha \geq 3/11$.
    {We claim that if $w_{1} > 0$ and $\beta \leq 3\alpha$, then we can change the values of $w_\alpha, w_{3\alpha}$ and $w_1$ so that the total value still equals $\beta$, while $w_1 = 0$ and $W$ does not increase.
    Let $u_\alpha, u_{2\alpha}, u_{3\alpha}, u_1$ be a new solution, defined as follows:
    \[
        u_{2\alpha} = 0,\quad u_1 = 0,\quad u_{\alpha} = w_{\alpha} - \frac{1 - 3\alpha}{2\alpha}w_1,\quad u_{3\alpha} = w_{3\alpha} + \frac{1 - \alpha}{2\alpha}w_1.
    \]
    First, note that $w_{2\alpha} = u_{2\alpha} = 0$.
    Second, when $\beta \leq 3\alpha$ it must hold that $w_\alpha \geq \frac{1 - 3\alpha}{2\alpha}w_1$.
    Indeed, since $w_\alpha + w_{3\alpha} + w_1 = 1$, we have $w_{3\alpha} = 1 - w_\alpha - w_1$ and
    \[3\alpha \geq \beta = \alpha w_\alpha + 3\alpha w_{3\alpha} + w_1 = \alpha w_\alpha + 3\alpha - 3\alpha w_\alpha - 3\alpha w_1 + w_1 = 3\alpha - 2\alpha w_\alpha + (1 - 3\alpha)w_1,\]
    which holds if and only if $w_\alpha \geq \frac{1 - 3\alpha}{2\alpha}w_1$.
    As a result, $u_\alpha = w_\alpha - \frac{1 - 3\alpha}{2\alpha}w_1 \geq 0$.
    Therefore,
    \[
        u_\alpha + u_{3\alpha} + u_1 = w_\alpha - \frac{1 - 3\alpha}{2\alpha}w_1 + w_{3\alpha} + \frac{1 - \alpha}{2\alpha}w_1 = w_\alpha + w_{3\alpha} + w_1 = 1,
    \]
    and
    \[
        \alpha u_\alpha + 3\alpha u_{3\alpha} + u_1 = \alpha w_\alpha - \frac{1 - 3\alpha}{2}w_1 + 3\alpha w_{3\alpha} + \frac{3 - 3\alpha}{2}w_1 = \alpha w_\alpha + 3\alpha w_{3\alpha} + w_1 = \beta,
    \]
    so $u_\alpha, u_{2\alpha}, u_{3\alpha}, u_1$ is feasible.
    At the same time,
    \[
        0\cdot u_\alpha + u_{2\alpha} + \frac{3}{2}u_{3\alpha} + 2u_1 = \frac{3}{2}w_{3\alpha} + \frac{3- 3\alpha}{4\alpha}w_1 \leq \frac{3}{2}w_{3\alpha} + 2w_1,
    \]
    where the last inequality holds if and only if
    \[\frac{3- 3\alpha}{4\alpha} \leq 2 \iff 3 - 3\alpha \leq 8\alpha \iff \alpha \geq 3/11.\]
    Thus, the solution $u_\alpha, u_{2\alpha}, u_{3\alpha}, u_1$ achieves strictly smaller value than $w_\alpha, w_{2\alpha}, w_{3\alpha}, w_1$.
    Therefore, if $\alpha \geq 3/11$ and $\beta \leq 3\alpha$, any solution with $w_1 > 0$ can be replaced with a solution that has $w_1 = 0$ without increasing its value.}

    {Finally, suppose that $\alpha \geq 3/11$ but $\beta > 3\alpha$.
    We claim that if $w_\alpha > 0$, then we can change the values of $w_\alpha, w_{3\alpha}$ and $w_1$ so that the total value still equals $\beta$, while $w_\alpha = 0$ and $W$ does not increase.
    Let $u_\alpha, u_{2\alpha}, u_{3\alpha}, u_1$ be a new solution, defined as follows:
    \[
        u_{2\alpha} = 0,\quad u_\alpha = 0,\quad u_{3\alpha} =w_{3\alpha} + \frac{1 - \alpha}{1 - 3\alpha}w_\alpha ,\quad u_1 = w_1 - \frac{2\alpha}{1 - 3\alpha}w_\alpha.
    \]
    First, note that $w_{2\alpha} = u_{2\alpha} = 0$.
    Second, when $\beta > 3\alpha$, it must hold that $w_1 > \frac{2\alpha}{1 - 3\alpha}w_\alpha$.
    Indeed, since $w_\alpha + w_{3\alpha} + w_1 = $, we have $w_{3\alpha} = 1 - w_\alpha - w_1$ and
    \[
        3\alpha < \beta = \alpha w_\alpha + 3\alpha w_{3\alpha} + w_1 = \alpha w_\alpha + 3\alpha - 3\alpha w_\alpha - 3\alpha w_1 + w_1 = 3\alpha - 2\alpha w_\alpha + (1 - 3\alpha)w_1,
    \]
    which holds if and only if $w_1 > \frac{2\alpha}{1 - 3\alpha}w_\alpha$.
    As a result, $u_1 = w_1 - \frac{2\alpha}{1 - 3\alpha}w_\alpha > 0$.
    Therefore,
    \[
        u_\alpha + u_{3\alpha} + u_1 = w_{3\alpha} + \frac{1 - \alpha}{1 - 3\alpha}w_\alpha + w_1 -\frac{2\alpha}{1 - 3\alpha}w_\alpha = w_\alpha + w_{3\alpha} +w_1 = 1,
    \]
    and
    \[
        \alpha u_\alpha + 3\alpha u_{3\alpha} + u_1 =  3\alpha w_{3\alpha} + \frac{3\alpha -3\alpha^2 }{1 - 3\alpha}w_\alpha + w_1 - \frac{2\alpha}{1 - 3\alpha}w_\alpha = \alpha w_\alpha + 3\alpha w_{3\alpha} + w_1 = \beta,
    \]
    so $u_\alpha, u_{2\alpha}, u_{3\alpha}, u_1$ is feasible.
    At the same time,
    \[
        0 \cdot u_\alpha + u_{2\alpha} + \frac{3}{2}u_{3\alpha} + 2u_1 = \frac{3}{2}w_{3\alpha} + \frac{3 - 3\alpha}{2 - 6\alpha}w_\alpha + 2w_1 - \frac{4\alpha}{1 - 3\alpha}w_{\alpha} \leq \frac{3}{2}w_{3\alpha} + 2w_1,
    \]
    where the last inequality holds if and only if 
    \[
        \frac{3 - 3\alpha}{2 - 6\alpha} - \frac{4\alpha}{1 - 3\alpha} \leq 0 \iff 3 - 3\alpha - 8\alpha \leq 0 \iff \alpha \geq 3/11.
    \]
    Thus, the solution $u_\alpha, u_{2\alpha}, u_{3\alpha}, u_1$ achieves strictly smaller value that $w_\alpha, w_{2\alpha}, w_{3\alpha}, w_1$.
    Therefore, if $\alpha \geq 3/11$ and $\beta > 3\alpha$, any solution with $w_\alpha > 0$  can be replaced with a solution that has $w_\alpha = 0$ without increasing its value.}
\end{proof}

\begin{corollary}\label{valbound}
    Let $i \in [n]$, and suppose that at step $i$ the value $\beta^i$ is at least $\beta$, for some $\beta \geq \alpha$.
    \begin{enumerate}
        \item if $\alpha \leq 3/11$, then
        \[\beta^{i + 1} \geq \left(1 - \frac{1 - \alpha}{2(\beta - \alpha)n}\right)\cdot \beta^i;\]
        \item if $\alpha \geq 3/11$ and $\beta \geq 3\alpha$, then
        \[\beta^{i + 1} \geq \left(1 - \frac{2(1 - 3\alpha)}{(\beta - 12\alpha + 3)n}\right)\cdot \beta^i;\]
        \item if $\alpha \geq 3/11$ and $\beta < 3\alpha$, then
        \[\beta^{i + 1} \geq \left(1 - \frac{4\alpha}{3(\beta - \alpha)n}\right)\cdot \beta^i.\]
    \end{enumerate}
    {(When $\alpha = 3/11$, the three bounds coincide.)}
\end{corollary}
\begin{proof}
{As shown in \cref{choice}, for every $i \in [n]$ it holds that $\beta^{i + 1} \geq (1 - \frac{1}{n\Lambda^i})\cdot \beta^i$.
In addition, in \cref{lambdabound} we showed that $\Lambda^i \geq W^i_\beta$, where $W_i^\beta$ is the optimal value for \ref{opt}.
Hence, it holds that for every $i \in [n]$ we have $\beta^{i + 1} \geq (1 - \frac{1}{nW^i_\beta})\cdot \beta^i$}.

Consider the optimization problem \ref{opt} with optimal solution $w_\alpha, w_{2\alpha}, w_{3\alpha}, w_1$ and value $W^i_\beta$.
There are three cases to consider, defined by \cref{wclass}.
\begin{enumerate}
    \item if $\alpha \leq 3/11$, then $w_{2\alpha} = w_{3\alpha} = 0$ and $\beta \leq \alpha w_\alpha + w_1$.
    Since $w_\alpha + w_1 = 1$, we can express
    \[\beta \leq \alpha w_\alpha + (1 - w_\alpha) \implies w_\alpha \leq \frac{1 - \beta}{1 - \alpha}.\]
    Therefore,
    \[W^i_\beta = 2w_1 = 2(1 - w_\alpha) \geq \frac{2(\beta - \alpha)}{1 - \alpha}.\]

    \item if $\alpha \geq 3/11$ and $\beta \geq 3\alpha$, then $w_{\alpha} = w_{2\alpha} = 0$ and $\beta \leq 3\alpha w_{3\alpha} + w_1$.
    Since $w_{3\alpha} + w_1 = 1$, we can express
    \[\beta \leq 3\alpha w_{3\alpha} + (1 - w_{3\alpha}) \implies w_{3\alpha} \leq \frac{1 - \beta}{1 - 3\alpha}.\]
    Therefore,
    \[W^i_\beta = \frac{3}{2}w_{3\alpha} + 2w_1 = 2 - \frac{1}{2}w_{3\alpha} \geq 2 - \frac{1 - \beta}{2(1 - 3\alpha)} = \frac{\beta - 12\alpha + 3}{2(1 - 3\alpha)}.\]

    \item if $\alpha \geq 3/11$ and $\beta < 3\alpha$, then $w_{2\alpha} = w_1 = 0$ and $\beta \leq \alpha w_\alpha + 3\alpha w_{3\alpha}$.
    Since $w_\alpha + w_{3\alpha} = 1$, we can express
    \[\beta \leq \alpha w_\alpha + 3\alpha(1 - w_{\alpha}) \implies w_{\alpha} \leq \frac{3\alpha - \beta}{2\alpha}.\]
    Therefore,
    \[W^i_\beta = \frac{3}{2}w_{3\alpha} = \frac{3}{2} - \frac{3}{2}w_{\alpha} \geq \frac{3}{2} - \frac{9\alpha - 3\beta}{4\alpha} = \frac{3(\beta - \alpha)}{4\alpha}.\]
\end{enumerate}

{As mentioned earlier, from \cref{choice} and \cref{lambdabound}, $\beta^{i + 1} \geq (1 - \frac{1}{n\Lambda^i})\beta^i \geq (1 - \frac{1}{nW^i_\beta})\beta^i$.
Applying the lower bounds for $W^i_\beta$ obtained above for different cases finishes the proof.}
\end{proof}

\subsection{Algorithm performance when $\alpha \le 3/11$}

\label{sec:littlealpha}

In this section, we determine the largest value $\rho \geq \alpha$ possible, such that at the beginning of the $n$-th iteration of the algorithm the total remaining value of the APS partition is at least $\rho$.
To do so, we first introduce a special process $\gamma^i_n$, $i = 1, 2, \ldots$, inspired by \cref{valbound}.
This process serves as a lower bound on the total remaining value of the APS partition at every iteration, i.e $\beta^i \geq \gamma^i_n$ for all $i \geq 1$.
Then, we give sufficient conditions on $\rho$ and $n$ so that $\gamma^n_n$, the value of the process at iteration $n$, is at least $\rho$. 

\begin{definition}
    For $n$ agents, let $\gamma^i_n$ be the following process: $\gamma^1_n = 1$, and for $i \geq 1$,
    \[
        \gamma^{i + 1}_n = \left(1 - \frac{1 - \alpha}{2(\gamma^i_n - \alpha)n}\right)\gamma^i_n.
    \]
\end{definition}
\begin{lemma}\label{algogamma}
    Let $\alpha \leq 3/11$.
    For every iteration $i \geq 1$ of \cref{algo} for $n$ agents, $\beta^i \geq \gamma^i_n$.
\end{lemma}
\begin{proof}
    We prove the lemma via induction on $i$.
    For $i = 1$, it holds $\beta^1 = 1 = \gamma^1_n$.
    Suppose now that for some $i > 1$ it holds $\beta^i \geq \gamma^i_n$.
    As shown in \cref{valbound}, $
        \beta^{i + 1} \geq \left(1 - \frac{1 - \alpha}{2(\beta^i - \alpha)n}\right)\beta^i$.
    It remains to observe that $\beta^i \geq \gamma^i_n$ implies $\left(1 - \frac{1 - \alpha}{2(\beta^i - \alpha)n}\right)\beta^i \geq \left(1 - \frac{1 - \alpha}{2(\gamma^i_n - \alpha)n}\right)\gamma^i_n = \gamma^{i + 1}_n$.
\end{proof}

First, we would like to show that for a given value $\rho$, if number of agents $n$ is large enough and if $\rho$ is not too large, then it should hold $\gamma^i_n \geq \rho$.
To do so, we take $\rho$ at determine the number of iterations required for the process $\gamma^i_n$ to reach the value $\rho$ or smaller.
Specifically, we partition the steps of the process into epochs based on the current value of $\gamma^i_n$, i.e iteration $i$ is part of the epoch number $r$ if $\gamma^i_n \in [r/k, (r + 1)/k]$.
Then, we lower bound the number of steps $\gamma^i_n$ stays at each epoch, and determine for which values of $n$ the total sum of lower bounds exceeds $n$.
This sum exceeding $n$ implies that $\gamma^i_n$ stays above $\rho$ after $n$ iterations, thus $\gamma^n_n\geq \rho$.

In the following lemma we express the total number of iterations of $\gamma^i_n$ required to reach value $\rho$ as the sum of steps spent at each epoch.
Then, we lower bound the number of iterations spent at each epoch, and combine the lower bounds to obtain a sufficient condition on $n$ and $\rho$ for the inequality $\gamma^n_n\geq \rho$ to hold, in terms of lower bounds on epoch steps.

\begin{lemma}\label{iters}
    Consider the process $\gamma^i_n$ for some $n$.
    For $0 \leq \delta \leq \gamma < 1$, let $i_n(\gamma)$ be the largest integer $i$ such that $\gamma^i_n \geq \gamma$, and let $t_n(\gamma, \delta)$ be the smallest integer $t$ such that at the beginning of iteration $i_n(\gamma) + t$ it holds $\gamma^{i_n(\gamma) + t}_n < \gamma - \delta$.
    Then
    \begin{enumerate}
        \item For any $0 \leq \rho < 1$ and any integer $k$, the number of steps of $\gamma^i_n$ to reach value $\rho$ is
        \[i_n(\rho) = t_n\left(\frac{\lceil \rho k\rceil}{k},\frac{\lceil \rho k\rceil}{k} - \rho \right) + \sum_{r = \lceil\rho k\rceil}^{k - 1}t_n\left(\frac{r + 1}{k}, \frac{1}{k}\right) - \left(k - \lceil \rho k\rceil\right).\]

        \item For any $0 \leq \delta \leq \gamma < 1$ such that $\gamma > \alpha + \delta$, the number $t_n(\gamma, \delta)$ of steps of the process $\gamma^i_n$ after the iteration $i_n(\gamma)$ for the value to go below $\gamma - \delta$ satisfies
        \[
        t_n(\gamma, \delta) > \frac{2\delta(\gamma - \delta - \alpha)n}{\gamma(1 - \alpha)} =: l_n(\gamma, \delta)\]

        \item For any $\alpha \leq \rho < 1$ and any integer $k$, it holds that $\beta^n \geq \gamma_n^n\geq \rho$ if
        \[\sum_{r = \lceil\rho k\rceil}^{k - 1}l_n\left(\frac{r + 1}{k}, \frac{1}{k}\right) - \left(k - \lceil \rho k\rceil\right) \geq n.\]
    \end{enumerate}
\end{lemma}
We prove this lemma in \cref{sec:smallalpha} (\cref{lowergammaiters}).
We apply \cref{iters} and combine all epoch lower bounds together to obtain a sufficient condition on the values of $\rho, n$ guaranteeing that $\gamma_n^n \geq \rho$. 
Then, we explore the limiting behavior of $\rho$ and $n$ in the sufficient condition, in order to determine the largest possible  $\rho$ for which there still exists $n$ satisfying the condition.
\begin{theorem}\label{rhoinfinity}
    Let $\alpha \leq \rho < 1$.
    Consider the process $\gamma_n^i$, take integer $k \leq n$.
    If $n, k$ satisfy
    \begin{equation}
        \frac{2(1 - \rho + \alpha \ln \rho)}{1 - \alpha} \geq 1 + \frac{(k -  \rho k)}{n} + \frac{2(k - \rho k)}{\rho k ( \rho k + 1)} + \frac{2}{k},\label{mainineq}
    \end{equation}
    then the value of the process $\gamma_n^i$ at iteration $n$ satisfies $\gamma_n^n \geq \rho$.

    As a corollary, for any $\alpha \in [1/4, 3/11)$ and $\rho$ satisfying $2(1 - \rho + \alpha \ln \rho) = 1 - \alpha$, for every constant $0 < \delta < \rho - \alpha$ there exists $n_\delta$ such that for all $n \geq n_\delta$ it holds $\gamma_n^n \geq \alpha + \delta$.
\end{theorem}
We prove this theorem in \cref{sec:smallalpha} (\cref{lowergamma} and \cref{lowergammabetter}).
According to \cref{rhoinfinity}, if $\alpha \in (1/4, 3/11)$ and $\rho$ satisfies $2(1 - \rho + \alpha \ln \rho) = 1 - \alpha$, then for all $n \geq N$ we get the inequality $\gamma_n^n \geq \alpha$.
This is not good enough, as we would like the inequality $\gamma_n^n \geq \alpha$ to hold for small number of agents $n < N$ too.

One way to show $\gamma_n^n \geq \alpha$ for $n < N$ is to check it numerically --- process $\gamma_n^i$ is well-defined, and number $N$ can be explicitly obtained from the proof of \cref{rhoinfinity}.
However, the actual value of $N$ is of order $\approx 100000$.
While still is easy and fast to check on the computer all values of $\gamma_n^n$ for $n < N$, it may not be the most convenient approach.
Instead, we provide a method that allows to reduce the required number of checks for $n$ from $\approx 100000$ to just $10$.

To mitigate the requirement for large values of $n$ in \cref{rhoinfinity}, we are going to perform a ``doubling'' procedure on the process $\gamma_n^i$.
Specifically, we will show that if one considers a process $\gamma_n^i$ for $n$ agents, and a process $\gamma_{2n}^j$ for $2n$ agents, then the process $\gamma_{2n}^j$ \textbf{lower bounds} the process $\gamma_n^i$.
This property implies that, as long as the value of $\gamma_{2n}^{2n}$ is not too small, it holds $\gamma_{2n}^{2n} \leq \gamma_n^n$.
Given this relationship between $\gamma_n^i$ and $\gamma_{2n}^j$ we can extend \cref{rhoinfinity} to smaller values of $n$ as follows.
Given $\alpha$ and some small number of agents $n$, we ``double'' the $n$ by taking $N = 2^{q}\cdot n$ for $q = q(\alpha)$ large enough so that \cref{rhoinfinity} holds for $\alpha$ and $N$, i.e we have $\gamma_{N}^{N} \geq \alpha$.
But then the doubling property implies $\alpha \leq \gamma_N^N \leq \gamma_{N/2}^{N/2} \leq \gamma_{N/4}^{N/4}\leq \ldots \leq \gamma_n^n$.

In the following lemma, we first show that given processes $\gamma_n^i$ and $\gamma_{2n}^j$, the value of $\gamma_{2n}^j$ over two steps decreases by at least as much as the value of $\gamma_n^i$ does over one step.
This property guarantees that in the long run, $\gamma_{2n}^j$ will stay below $\gamma_n^i$.

\begin{lemma}\label{doublingineq}
    Recall the definition of the process $\gamma_n^i$: $\gamma_n^1 = 1$, and for all $i \geq 1$, $\gamma_n^{i + 1} = \left(1 - \frac{1 - \alpha}{2(\gamma_n^i - \alpha)n}\right)\gamma_n^i$.
    Suppose that for some $n$ and $i, j \geq 1$ it holds $\gamma_n^i\geq \beta = \gamma_{2n}^j$, for some $\beta$.
    If 
    \[
        \beta >\alpha + \frac{1 - \alpha + \sqrt{16n\alpha(1 - \alpha) + (1 - \alpha)^2}}{8n},  
    \]
    then $\gamma_{2n}^{j + 2} \leq \gamma_n^{i + 1}$.
\end{lemma}
We prove this lemma in \cref{sec:doubling} (\cref{doublingineqproof}).
This established relationship between $\gamma_n^i$ and $\gamma_{2n}^j$ gives the ``doubling'' property, i.e that final iterations of both processes satisfy $\gamma_{2n}^{2n} \leq \gamma_n^n$.
\begin{corollary}\label{doubling}
    Consider the process $\gamma_n^i$.
    If $\gamma_{2n}^{2n}> \alpha + \frac{1 - \alpha + \sqrt{16n\alpha(1 - \alpha) + (1 - \alpha)^2}}{8n}$,
    then $\gamma_{2n}^{2n} \leq \gamma_n^n$.
\end{corollary}
\begin{proof}
    Consider the first iterations of $\gamma_n^i$ and $\gamma_{2n}^j$.
    Since $\gamma_n^1= 1 = \gamma_{2n}^1$, we can apply the inequality of \cref{doublingineq} and obtain $\gamma_{2n}^3 \leq \gamma_n^2$.
    Suppose now that it holds
    $\gamma_{2n}^{2n} \geq \alpha + \frac{1 - \alpha + \sqrt{16n\alpha(1 - \alpha) + (1 - \alpha)^2}}{8n}$.
    Since the process is non-increasing, we have $\gamma_{2n}^3\geq \gamma_{2n}^{2n}$, hence $\gamma_{2n}^3 \leq \gamma_n^2$ implies that we can apply \cref{doublingineq} to $j = 3$ and $i = 2$ with value $\beta = \gamma_{2n}^3$, obtaining $\gamma_{2n}^5 \leq \gamma_n^3$.
    Since $\gamma_{2n}^j \geq \gamma_{2n}^{2n}$ for all $j \leq 2n$, we can repeat the same process, applying $\gamma_{2n}^{j + 2} \leq \gamma_n^{i + 1}$ subsequently for $j = 2i + 1$ and $i \leq n - 1$ with $\beta = \gamma_{2n}^j$.
    Doing so up to $n - 1$, we get $\gamma_{2n}^{2n - 1} \leq \gamma_n^n$, and the claim follows.
\end{proof}
Combining \cref{rhoinfinity} with \cref{doubling} allows us to prove the main result.
\begin{theorem}\label{311theorem}
   Assume that $\alpha \in [1/4, 3/11)$.
   Then $\beta^n$, the remaining value after $n$ steps of \cref{algo} for $n$ agents, satisfies $\beta^n \geq \alpha$ for all $n \geq 11$.
\end{theorem}
\begin{proof}
    By \cref{algogamma} we have $\beta^n \geq \gamma_n^n$, so to show $\beta^n \geq \alpha$ it suffices to prove that $\gamma_n^n \geq \alpha$.
    Let $\rho$ satisfy $2(1 - \rho + \alpha \ln \rho)=1 - \alpha$.
    Observe that for $\alpha^* = 3/11$, the value of $\rho^*$ satisfying $2(1 - \rho^* + \alpha^* \ln \rho^*)=1 - \alpha^*$ is approximately $0.3502$ (rounded down).
    Hence, the difference $\rho - \alpha$ is at least $\rho^* - \alpha^* >  0.07749$.
    We take $\delta = 0.07749 < \rho - \alpha$. 
    By \cref{rhoinfinity}, there exists $N_\delta$ such that for all $N \geq N_\delta$ we have $\gamma_N^N\geq\alpha$.
    Let $q = q(\alpha, \delta)$ be large enough integer such that $N_{q} := 2^{q} \cdot n \geq N_\delta$, then it holds $\gamma_{N_{q}}^{N_{q}} \geq \alpha + \delta$.
    We would like to show that $\gamma_n^n \geq \alpha$ too.

    For these values of $\alpha, \rho$ and $\delta$ it holds $
        \alpha + \delta > \alpha + \frac{1 - \alpha + \sqrt{16n\alpha(1 - \alpha) + (1 - \alpha)^2}}{8n}$
    when $n \geq 10.61$.
    Then, for the number $N_q$ we have
    \begin{equation*}
        \gamma_{N_q}^{N_q} \geq \alpha + \delta > \alpha + \frac{1 - \alpha + \sqrt{16n\alpha(1 - \alpha) + (1 - \alpha)^2}}{8n} \geq \alpha + \frac{1 - \alpha + \sqrt{16N_q\alpha(1 - \alpha) + (1 - \alpha)^2}}{8N_q},
    \end{equation*}
    thus we can apply doubling, and by \cref{doubling} it holds $\gamma_{N_q}^{N_q} \leq \gamma_{N_q/2}^{N_q/2}$.
    Applying the lower bound above repeatedly to $N_q/2$, $N_q/4$ and so on, we get
    \begin{multline*}
        \gamma_{N_q/2^r}^{N_q/2^r} \geq \alpha + \delta > \alpha + \frac{1 - \alpha + \sqrt{16n\alpha(1 - \alpha) + (1 - \alpha)^2}}{8n} \geq\\ \alpha + \frac{1 - \alpha + \sqrt{16N_q/2^r\alpha(1 - \alpha) + (1 - \alpha)^2}}{8N_q/2^r}\label{nqineq2}
    \end{multline*}
    for incrementing values of $r \geq 0$, and thus $\gamma_{N_q/2^r}^{N_q/2^r} \leq \gamma_{N_q/2^{r + 1}}^{ N_q/2^{r + 1}}$.
    Since $n \geq 11$, the chain of inequalities continues until we reach $r = q$, thus $\gamma_{N_q}^{N_q} \leq \gamma_n^n$.
    As a result, it holds $\gamma_n^n \geq \alpha + \delta$.
\end{proof}
\begin{figure}[h]
\includegraphics[width=17cm]{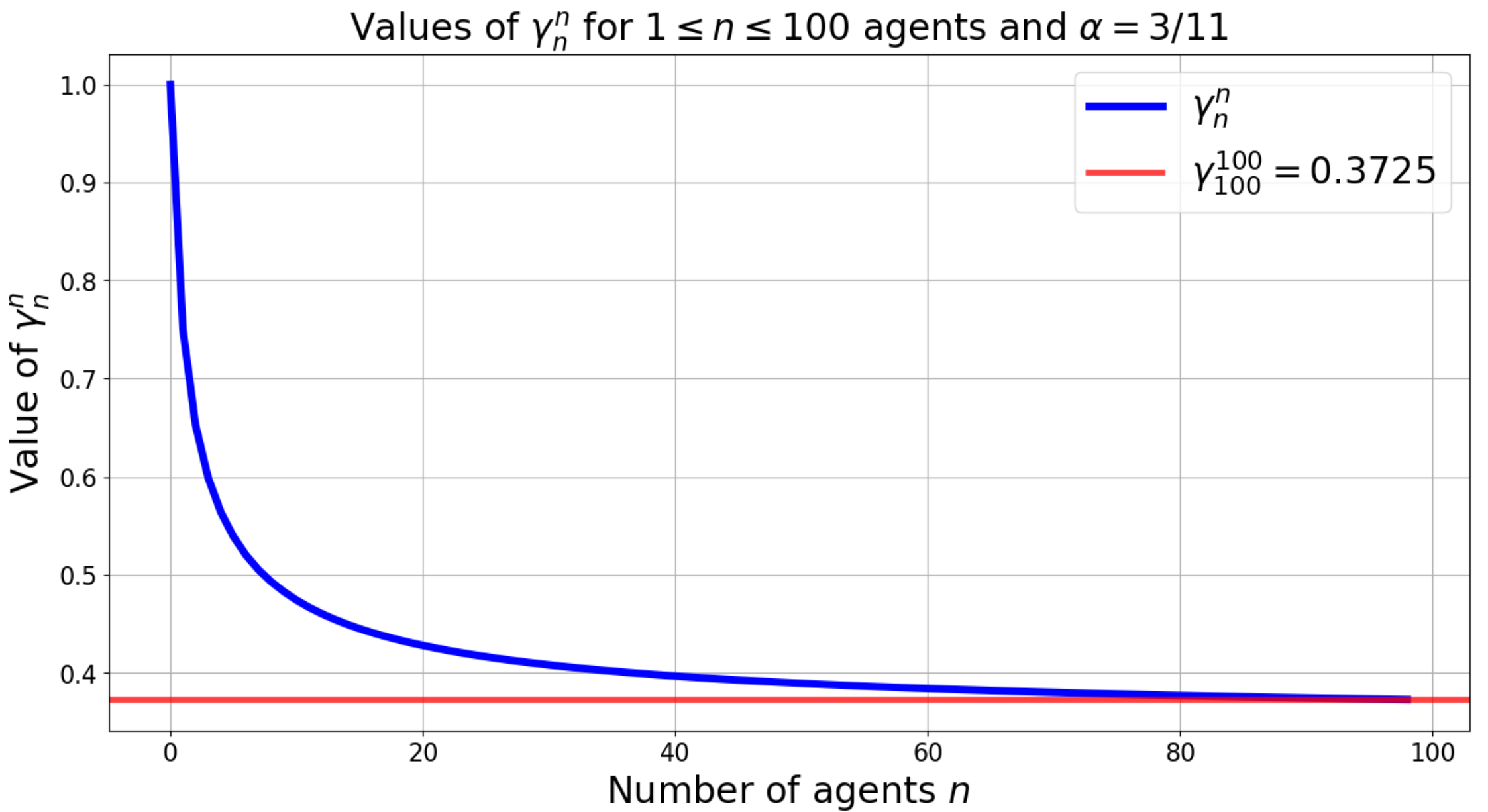}
\label{fig:smalla}
\caption{Numerical computation of values $\gamma_n^n$}
\label{fig:app1}
\end{figure}
Values of $n \leq 10$ were checked  umerically by running a simple calculation program (see Figure~\ref{fig:app1} for values of $\gamma_n^n$ for $n$ up to $100$). 
Since $\rho > 3/11 \geq \alpha$, there is still a lot of room left for improvement, so to maximize $\alpha$ we will turn to the more complicated case $\alpha \geq 3/11$.

\subsection{Algorithm performance when $\alpha \geq 3/11$}
\label{sec:bigAlpha}

The proofs for $\alpha \geq 3/11$ are completely is analogous, and just require more technical work.
We will need instead to use a different process $\gamma_n^i$.
The following lemma is proved by induction.
\begin{lemma}\label{algogamma2}
    For $n$ agents, let $\gamma_n^i$ be the following process: $\gamma_n^1 = 1$, and for $i \geq 1$,
    \begin{gather*}
        \text{if }\;\gamma^i_n \geq 3\alpha \;\text{ then }\;
        \gamma^{i + 1}_n = \left(1 - \frac{2(1 - 3\alpha)}{(\gamma^i_n -12\alpha +3)n}\right)\gamma^i_n;\\
        \text{if }\; \gamma^i_n < 3\alpha \;\text{ then }\; \gamma^{i + 1}_n = \left(1 - \frac{4\alpha}{3(\gamma^i_n - \alpha)n}\right)\gamma^i_n.
    \end{gather*}
    Assume $\alpha \in [3/11, 1/3)$.
    For every iteration $i \geq 1$ of \cref{algo} for $n$ agents, $\beta^i \geq \gamma_n^i$.
\end{lemma}

We state the analogue of \cref{rhoinfinity} for $\alpha > 3/11$, proved in \cref{sec:bigalpha} (\cref{lowergamma2better}).
\begin{theorem}\label{bigrhoinf}
    Let $\rho$ satisfy $2(12\rho - 3)\ln(3\rho) = (1 - 3\rho)(3\ln 3 - 4)$, and assume $\alpha \in [3/11, \rho)$.
    For any constant $0 < \delta < \rho - \alpha$, there exists $n_\delta$ such that for all $n \geq n_\delta$ it holds $\gamma_n^n \geq \alpha +\delta$.
\end{theorem}
The doubling and main result also extend to $\alpha > 3/11$, proved in \cref{sec:doubling} (\cref{doublingproof2}).
\begin{lemma}\label{doubling2}
    Suppose that for some $n$ and $i, j \geq 1$ it holds $\gamma_n^i \geq \beta=\gamma_{2n}^j$, for some value $\beta$.
    If $\beta >  \alpha + \frac{\alpha(1 + \sqrt{6n + 1})}{3n}$,
    then $\gamma_{2n}^{j + 2} \leq \gamma_n^{i + 1}$.
    Therefore, if $\gamma_{2n}^{2n}> \alpha + \frac{\alpha(1 + \sqrt{6n + 1})}{3n}$,
    then $\gamma_{2n}^{2n} \leq \gamma_n^n$.
\end{lemma}
\begin{theorem}
   Let $\alpha = \frac{11}{40} + \frac{1}{1000}$.
   Then $\beta^n$, the remaining value after $n$ steps of \cref{algo} for $n$ agents, satisfies $\beta^n \geq \alpha$ for all $n \geq 54$.  
\end{theorem}
\begin{proof}
    The proof is analogous to \cref{311theorem}.
    For $\alpha = \frac{11}{40} + \frac{1}{1000}$, the value $\rho$ for which  $
        \frac{3(3\alpha - \rho - \alpha\ln(3\alpha)  + \alpha\ln \rho)}{4\alpha} + \frac{1 - 3\alpha + (12\alpha - 3)\ln(3\alpha )}{2(1 - 3\alpha)} = 1$
    is $\rho \geq 0.30861$.
    Then, by \cref{lowergamma2better}, the gap between $\gamma_n^n \geq \rho$ and $\alpha$ will be at least $0.03261$, so we can take $\delta = 0.03261$ and have $\gamma_{N_\delta}^{N_\delta} \geq \alpha$ for the doubled process.
    Then, the doubling $n$ must satisfy $0.03261 \geq  \frac{\alpha(1 + \sqrt{6n + 1})}{3n}$,
    or $n \geq 53.4$.
\end{proof}
\begin{figure}[H]
\includegraphics[width=17cm]{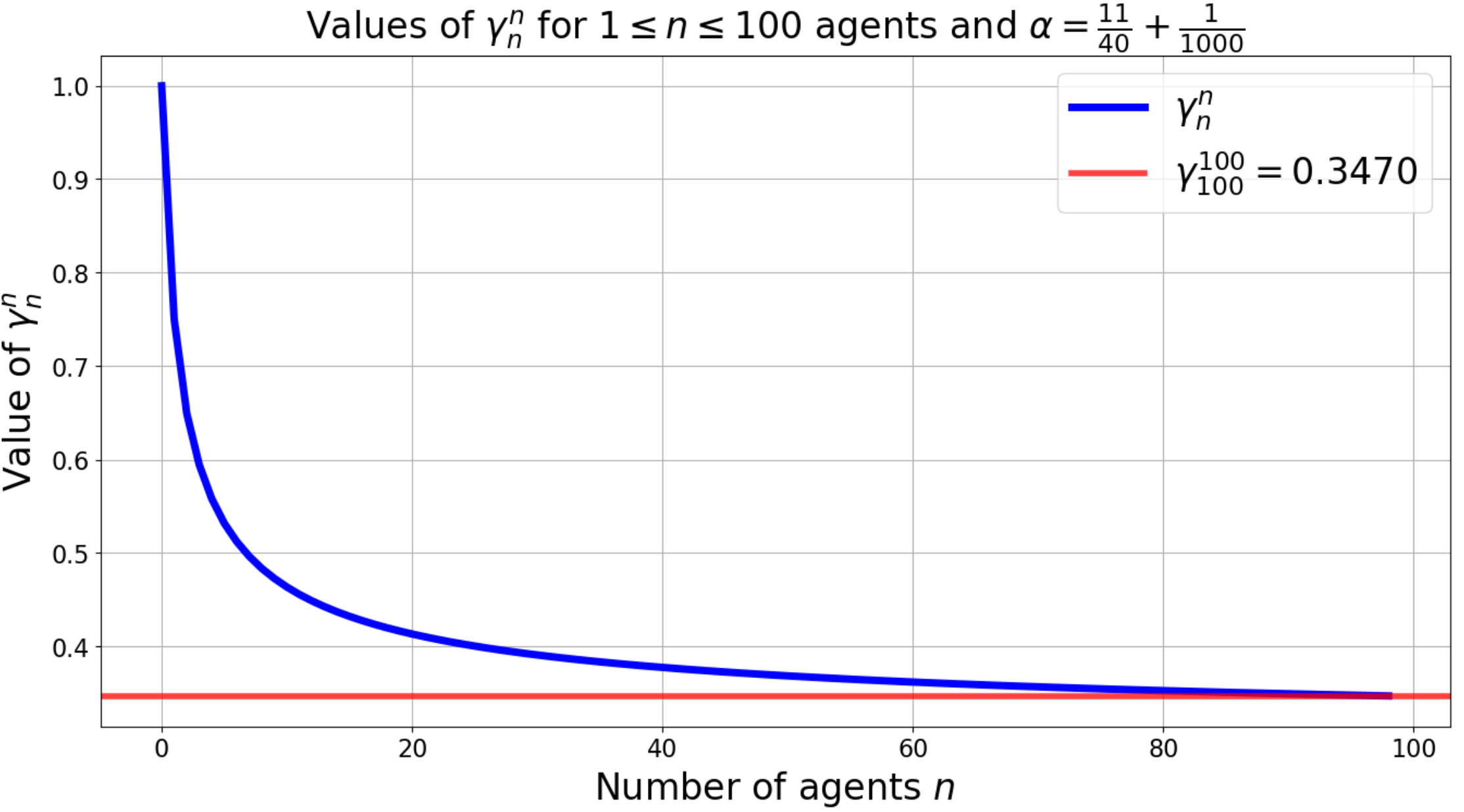}
\label{fig:smalla}
\caption{Numerical computation of values $\gamma_n^n$}
\label{fig:plotbig}
\end{figure}
The values $n\leq 100$ were checked directly via a computer program in \cref{fig:plotbig}.

\section{Different valuations}

\label{sec:diffvalsintro}

The algorithm for the general case, when valuations $v_i$ for agents $i \in [n]$ are not necessarily the same, will have the same high-level structure.
There will be $n$ iterations, and at iteration $k \in [n]$ it picks some agent $i_k$ and allocates a bundle to this agent.
There major difficulty that the algorithm needs to address in the multiple-agent case is the fact that different valuations result in different APS-partitions, which can be have problematic overlaps between each other.
So, if say bundle $B$ is acceptable to agent $i$, removing items of $B$ from APS-partitions of other agents may result in a significant loss of value.
Therefore, the algorithm needs to ensure that not only every agents receives an acceptable bundle in the end, but also that giving a bundle to one agent keeps enough value in the APS-partitions of the agents that are still in the game.

For agent $i$, let $\{(S_i, \lambda_i(S_i)\}_{S_i \in \Ss_i}$ be her APS partition satisfying \cref{wlog}.
In the following sections, we will also assume that these APS partitions satisfy \cref{smallitems}.
This assumption is later removed in \cref{sec:bigitems}, where we solve different instances with items of value at least $\alpha$ using results obtained for the case when \cref{smallitems} does hold.

In our algorithm, each agent will keep track of their personal APS-partition, which is modified over the course of the algorithm.
At every iteration, we remove some items from every bundle of the partition, while keeping the weights the same.
We will need a similar terminology of filtered partition, used for the same-valuation case, extended to multiple agents.

\begin{definition}
    For $i \in [n]$, consider agent $i$ with APS-partition $\{(S_i, \lambda_i(S_i)\}_{S_i \in \Ss_i}$, and suppose that this agent participates in the following high-level algorithm: there are $n$ iterations in total, and at every iteration $k \in [n]$, it removes some (possibly none) items from every bundle $S_i \in \Ss_i$.
    We denote by $S_i^k \subseteq S_i$ the remaining items of $S_i$ at iteration $k$, and let $\Ss_i^k := \{S_i^k : S_i \in \Ss_i\}$.
    We call $\Ss_i^k$ a \textbf{filtration of $\Ss_i$ at iteration $k$}, and $\{(S_i^k, \lambda_i(S_i)\}_{S_i \in \Ss_i}$ a \textbf{filtered APS-partition}.
    We will refer to $\beta_i^k := \sum_{S_i \in \Ss_i}\lambda_i(S_i) \cdot v_i(S_i^k)$ as the \textbf{remaining value} of the filtered partition $\Ss_i^k$.
\end{definition}

Throughout the algorithm, every agent $i \in [n]$ maintains a filtration of $\Ss_i$.
At some iteration $k$, we may choose agent $i_k$ and allocate her an acceptable bundle $B_{i_k}$.
Unlike in \cref{algo} before, the different-valuations algorithm will pick a bundle for agent $i_k$ probabilistically, using a distribution over subsets of the filtered partition $\Ss_i^k$.
In order to construct this distribution, we use essentially the same ideas that were earlier described in \cref{partition}, \cref{choice} and \cref{distribution}.
For completeness, we summarize these ideas in the following lemma.

\begin{lemma}\label{distpartition}
    For $i \in [n]$, consider agent $i$ with valuation $v_i$ and APS-partition $\{(S_i, \lambda_i(S_i)\}_{S_i \in \Ss_i}$ satisfying \cref{wlog}, and let $\Ss_i^k = \{S_i^k : S_i \in \Ss_i\}$ be some filtration of $\Ss_i$ at iteration $k \in [n]$.
    For any $\beta \geq \alpha$ such that $\beta_i^k:=\sum_{S_i \in \Ss_i}\lambda_i(S_i) \cdot v_i(S_i^k)$ satisfies $\beta_i^k\geq \beta$, the following holds.
    
    There exists a probability distribution $\mu_i^k$ over subsets of $S_i^k$ for $S_i^k \in \Ss_i^k$, such that for every $B\subseteq \Mm$ with $\mu_i^k(B) > 0$ it holds $v_i(B) \geq \alpha$.
    In addition, for every item $e \in \Mm$ the probability $\mu_i^k(e)$ to pick $e$ if we sample a random set $B\sim \mu_i^k$ satisfies $\mu_i^k(e) \leq p = p(\beta)$, where
    \begin{enumerate}
        \item if $\alpha \leq 3/11$, then $p := \frac{1 - \alpha}{2(\beta - \alpha)n}$;
        
        \item if $\alpha \geq 3/11$ and $\beta \geq 3\alpha$, then $p := \frac{2(1 - 3\alpha)}{(\beta - 12\alpha + 3)n}$;

        \item if $\alpha \geq 3/11$ and $\beta < 3\alpha$, then $p := \frac{4\alpha}{3(\beta - \alpha)n}$.
    \end{enumerate}
\end{lemma}
\begin{proof}
    We use \cref{partition} on valuation $v_i$ and partition $\Ss_i^k$ to distribute all bundles $S_i^k \in \Ss_i^k$ into classes $\Ss_i^k(t\alpha)$ for $t \geq 1$ as follows: bundle $S_i^k$ belongs to $\Ss_i^k(t\alpha)$ if $(t - 1)\alpha \leq v_i(S_i^k) < t\alpha$.
    As $\alpha \leq 1/3$, and by \cref{wlog} there are no sets of value greater than $1$, we will only have bundles of classes $\Ss_i^k(\alpha), \Ss_i^k(2\alpha), \Ss_i^k(3\alpha), \Ss_i^k(4\alpha)$, where for every $S_i^k \in \Ss_i^k(4\alpha)$ it holds $3\alpha \leq v(S_i^k) \leq 1$.
    We denote $\Ss_i^k(1) := \Ss_i^k(4\alpha)$ for clarity, and proceed to construct $\mu_i^k$ using \cref{partition}.

    Let $\Lambda_i^k(\alpha) := \sum_{S_i^k \in \Ss_i^k(\alpha)}\lambda_i(S_i)$, and let $\Lambda_i^k(2\alpha), \Lambda_i^k(3\alpha), \Lambda_i^k(1)$ be defined accordingly as sums of weights of sets in corresponding classes.
    Let
    \[
        \Lambda_i^k := 0 \cdot \Lambda_i^k(\alpha) + 1\cdot \Lambda_i^k(2\alpha) + \frac{3}{2}\cdot \Lambda_i^k(3\alpha) + 2\cdot \Lambda_i^k(1).
    \]
    Let $\Bb_i^k$ be a (multi)-collection of sets, initially $\Bb_i^k = \varnothing$.
    For every $S_i^k \in \Ss_i^k$ with $\lambda_i(S_i) > 0$ we check which of the $\Ss_i^k$-classes it belongs to.
    \begin{itemize}
        \item If $S_i^k \in \Ss_i^k(\alpha)$, we do nothing, effectively setting $\mu_i^k(S_i^k) = 0$.
        \item If $S_i^k \in \Ss_i^k(2\alpha)$, add a copy of $S_i^k$ into $\Bb_i^k$, and assign it probability $\mu_i^k(S_i^k) = \frac{\lambda_i(S_i)}{\Lambda_i^k}$.
        \item If $S_i^k \in \Ss_i^k(3\alpha)$, let $B_1, B_2, B_3 \subseteq S_i^k$ be sub-bundles of $S_i^k$ constructed as in \cref{partition}.
        For every $r = 1, 2, 3$, add a copy of $B_r$ to $\Bb_i^k$, and assign it probability $\mu_i^k(B_r) = \frac{\lambda_i(S_i)}{2\Lambda_i^k}$.
        \item If $S_i^k \in \Ss_i^k(1)$, let $B_1, B_2 \subseteq S_i^k$ be sub-bundles of $S_i^k$ constructed as in \cref{partition}.
        For every $r = 1, 2$, add a copy of $B_r$ to $\Bb_i^k$, and assign it probability $\mu_i^k(B_r) = \frac{\lambda_i(S_i)}{\Lambda_i^k}$.
    \end{itemize}
    By construction, any $B\subseteq \Mm$ with $\mu_i^k(B) > 0$ satisfies $v_i(B) \geq \alpha$.
    By \cref{distribution}, $\mu_i^k$ is a valid probability distribution over $\Bb_i^k$, and for any item $e \in \Mm$ it holds
    \[
        \mu_i^k(e) = \sum_{\substack{B \in \Bb_i^k\\ B\ni e}}\mu_i^k(B) \leq \frac{1}{n\Lambda_i^k}.
    \]
    Finally, by \cref{lambdabound} and \cref{valbound}, we can upper bound $1/(n\Lambda_i^k)$ by $p(\beta)$, with corresponding values of $p(\beta)$ for different of $\alpha$ and $\beta \leq \beta_i^k$ (as stated in lemma).
\end{proof}
\subsection{Algorithm description}
\label{sec:algorithm}

The high-level description of the algorithm for different valuations is quite similar to the same-valuation case (\cref{algo}).
The algorithm has $n$ iterations and maintains a set $\Aa^k$ of \textbf{active} agents at iteration $k$, i.e the agents that still have not been allocated a bundle.
At each iteration $k$, we pick some agent $i_k \in \Aa^k$, construct a distribution $\mu_{i_k}^k$ as in \cref{distpartition}, sample a bundle $B_{i_k} \sim \mu_{i_k}^k$ at random and give the bundle $B_{i_k}$ to agent $i_k$.
Then, agent $i_k$ becomes inactive, i.e $\Aa^{k + 1} = \Aa^{k}\setminus \{i_k\}$, and we move to iteration $k + 1$.

In order to succeed, the algorithm needs to ensure that for the next agent $j \in \Aa^{k + 1}$ picked at iteration $k + 1$, the remaining value of the filtered partition $\Ss_{j}^{k + 1}$ is at least $\alpha$, i.e $\beta_{j}^{k + 1} \geq \alpha$ (with high probability), otherwise \cref{distpartition} cannot be applied and we would not be able to choose an acceptable bundle for $j$.
One could try mimicking previous approach of \cref{algo}, which, having allocated $B_{i_k}$ to $i_k$, creates $\Ss_j^{k + 1}$ out of $\Ss_j^{k}$ by simply removing items of $B_{i_k}$ from all bundles $S_j^k \in \Ss_j^k$, i.e $S_j^{k + 1} = S_j^k\setminus B_{i_k}$.
Unfortunately, doing so in the different valuation case may cause the remaining agents in $\Aa^{k + 1}$ lose a lot of value.
For example, we could sample an acceptable bundle $B_{i_k}$ for $i_k$, such that for some agent $j \in \Aa^{k + 1}$ that remains active, removing items of $B_{i_k}$ from bundles $S_j^k \in \Ss_{j}^k$ will force $\sum_{S_j \in \Ss_j}\lambda_j(S_j)\cdot v_j(S_j^k \setminus B_{i_k}) < \alpha$.
And, if APS-partitions of different agents have significant overlaps, i.e multiple bundles of one agent have huge intersections with many different bundles of the other agents, the probability to pick a ``wrong'' bundle and remove a lot of value from the other agents can be quite substantial.

To mitigate this issue, our algorithm will include the following additional steps:
\begin{itemize}
    \item Before any allocations, for every agent $i$ we scan the initial partition $\Ss_i$ and identify bundles $S_i \in \Ss_i$ that other agents $j\neq i$ on average find very valuable --- we will refer to such bundles as ``dangerous''.
    Then, we create a filtered partition $\Ss_i^1$ where every dangerous bundle $S_i$ is replaced with $S_i^1 = \varnothing$, while all the other bundles are untouched. 
    {This effectively ``bans'' the algorithm from allocating any subset of a dangerous $S_i$-s to $i$ in the future.
    Indeed, when we construct a distribution $\mu_i^k$ from \cref{distpartition}, every subset $S\subseteq S_i$ will have $\mu_i^k(S) = 0$, while the probabilities of other bundles are  increased.}
    We show that the total weight of the dangerous bundles w.r.t APS-partition of $i$ is small, thus the value $\beta_i^1$ of $\Ss_i^1$ is still close to $1$, i.e $i$ loses almost no value by ignoring the dangerous bundles.

    \item After picking agent $i_k \in \Aa^k$ and allocating her an acceptable bundle $B_{i_k} \sim \mu_{i_k}^k$, for every remaining agent $j \in \Aa^{k + 1}$ the algorithm keeps removing items of $B_{i_k}$ from bundles $S_j^k \in \Ss_j^k$ only as long as the total loss inflicted to $\beta_j^k$ does not exceed some value $c > 0$, decided in advance.
    This procedure, referred to as ``filtering'', creates $\Ss_j^{k + 1}$ out of $\Ss_j^k$.
    Note that from the perspective of agent $j \in \Aa^{k + 1}$, agent $i_k$ only ``partially owns'' the allocated bundle $B_{i_k}$, as the filtering procedure intentionally keeps some of the intersections of $B_{i_k}$ with $S_j^k \in \Ss_j^{k}$.

    \item While the filtering procedure guarantees that for every $j \in \Aa^{k + 1}$ the total loss inflicted to $\beta_j^k$ is at most $C$, at the same time it creates conflicts.
    It is now possible that at the future iteration, when agent $j \in \Aa^{k + 1}$ is picked, the algorithm allocates her an acceptable bundle $B_j$ such that $B_j \cap B_{i_k} \neq \varnothing$, i.e we allocate several items twice.
    To fix this, the algorithm ``steals'' the intersection $B_j \cap B_{i_k}$ from already allocated agent $i_k$, and instead gives it to the currently picked agent $j$, so that $j$ receives the bundle $B_j$ fully (at least for now).
    We will show that 1) if the filtering procedure removes items of $B_{i_k}$ from bundles $S_j^k \in \Ss_j^k$ in a specific way, then $i_k$ will not lose much value from $j$ stealing items of $B_j\cap B_{i_k}$, and 2) since the filtered partition $\Ss_{i_k}^k$ of agent $i_k$ did not have any dangerous bundles, then the number of agents $j \in \Aa^{k + 1}$ that can potentially steal items from $i_k$ in the future is small.
\end{itemize}

With the steps above implemented, we guarantee that for every iteration $k \in [n]$ and for every still active agent $j \in \Aa^{k + 1}$, the value loss $\beta_j^{k} - \beta_j^{k + 1}$ is at most $c$.
Suppose that agents do not lose much value by removing dangerous bundles and having some of their items stolen in the future.
Then, as long as for every iteration $k$ and agent $i_k \in \Aa^k$ picked at this iteration, the bundle picked for $i_k$ is acceptable, the algorithm succeeds.
Note that by construction, if at iteration $k$ one can make the distribution $\mu_{i_k}^k$ from \cref{distpartition}, then a bundle $B_{i_k}\sim \mu_{i_k}^k$ sampled from this distribution is guaranteed to be acceptable for $i_k$, i.e have $v_{i_k}(B_{i_k})\geq \alpha$.
The only prerequisite is actually being able to construct $\mu_{i_k}^k$, which by \cref{distpartition} requires the remaining value of agent $i_k$, $\beta_{i_k}^k$, satisfy $\beta_{i_k}^k \geq \alpha$.
As stated, by removing dangerous bundles agents do not lose a lot of value, so one can hope that for every $i \in [n]$ we have $\beta_i^1 \geq 1 - o(1) \geq \alpha$, and the agent picked first, $i_1$, gets an acceptable bundle.
However, it is not clear whether the same will be true for agents picked at iterations $k > 1$, especially after removing items and updating the filtered partitions.
For the algorithm to succeed it needs to assure that for every iteration $k \in [n]$ and agent $j \in \Aa^{k + 1}$ picked at some point in the future, the remaining value $\beta_j^{k + 1}$ is at least $\alpha$.

We claim that this can be achieved if the agent $i_k \in \Aa^k$ chosen at iteration $k$ is exactly the agent with the \textbf{largest} remaining value $\beta_{i_k}^k$.
We show that if agents are is chosen this way at every iteration, over the course of the algorithm with high probability every agent $j$ keeps her remaining value $\beta_j$ at least $\alpha$, as long as the value loss $\beta_j^{k} - \beta_j^{k + 1} \leq c$ is sufficiently small. 

We are now ready to give a complete description of the algorithm for the different valuations case.
First, we formally define what does it mean for a bundle to be ``dangerous''.
\begin{definition}
    Suppose that for every $j \in [n]$, APS-partition $\{(\lambda_j(S_j), S_j\}_{S_j \in \Ss_j}$ of agent $j$ satisfies \cref{wlog}.
    For every set $B\subseteq \Mm$ and agent $j \in [n]$, let
    \[
        \Delta(B, j) := \sum_{S_j \in \Ss_j}\lambda_j(S_j)\cdot v_j(B\cap S_j)
    \]
    be the value lost by $j$ if we remove the items of the set $B$ from all bundles of the partition $\Ss_j$.
    Additionally, for every set $B\subseteq \Mm$ and agent $i \in [n]$, let
    \[
        \Delta_i(B) := \sum_{j\neq i}\Delta(B, j)
    \]
    be the total value lost by all agents, except for $i$, if we remove items of $B$ from the their partitions.
\end{definition}

We will pick some large value $D > 0$ (say, of order $\sqrt{n}$) in advance, and say that a bundle $S_i\in \Ss_i$ of agent $i \in [n]$ is \textit{dangerous} if $\Delta_i(S_i) > D$.
The algorithm for different valuations is presented in \cref{diffalgo} and \cref{filteralgo}.

\begin{algorithm}[h]
\caption{Different valuations allocation algorithm}\label{diffalgo}
	\begin{algorithmic}[1]
		\Require{Instance $(n, m, \Mm, v_1, \ldots, v_n)$, $n$ APS partitions $\{(S_i, \lambda_i(S_i)\}_{S_i \in \Ss_i}$, parameters $c, D > 0$.}
		\Ensure{Disjoint sets $Q_1, \ldots, Q_n$ such that $\bigsqcup_{i = 1}^nQ_i \subseteq \Mm$.}
        \State Let $\Aa^k$ be the set of active agents at iteration $k$, initially $\Aa^1 := [n]$
        \State For $i \in [n]$ and every $S_i \in \Ss_i^1$, let $S_i^k\subseteq S_i$ be the items remaining in $S_i$ at iteration $k$
        \For{$i = 1, \ldots, n$}
            \State For every $S_i \in \Ss_i$, if $\Delta_i(S_i) > D$, let $S_i^1 := \varnothing$, otherwise $S_i^1 := S_i$
            \State Let $\Ss_i^1 := \{S_i^1 : S_i \in \Ss_i\}$ --- initial filtered partition
        \EndFor
		\For{$k = 1,\ldots, n$}
            \State For every $i \in \Aa^k$, compute $\beta^k_i := \sum_{S_i \in \Ss_i}\lambda_i(S_i) \cdot v_i(S_i^k)$
            \State Sort $\Aa^{k} = \{i_k, i_{k + 1}, \ldots, i_n\}$ so that $\beta^k_{i_k} \geq \beta^k_{i_{k + 1}} \geq \ldots \geq \beta^k_{i_n}$, and pick the agent $i_k$
            \State Let $\mu_{i_k}^k$ be a probability distribution over sub-bundles of $\Ss_{i_k}^k$ from \cref{distpartition}
            \State Sample $B_{i_k} \sim \mu_{i_k}^k$ at random, and let $Q_{i_k}^{k} := B_{i_k}$
            \For{$h \in [n] \setminus\Aa^{k}$}
                \State Update $Q_h^{k} := Q_h^{k - 1} \setminus B_{i_k}$
            \EndFor
            \State Update $\Aa^{k + 1} := \Aa^k \setminus \{i_k\}$
            \For{$j \in \Aa^{k+1}$}
                \State Update $\Ss_j^{k + 1} := \mathbf{Filtering}(v_{i_k}, v_j, B_{i_k}, \Ss_j^k, c)$
            \EndFor
		\EndFor
		\Return sets $Q_1^n, \ldots, Q_n^n$.
	\end{algorithmic}
\end{algorithm}

\begin{algorithm}[h]
\caption{\textbf{Filtering}}\label{filteralgo}
	\begin{algorithmic}[1]
		\Require{valuations $v_i, v_j$, set $B_i$, filtered partition $\Ss_j^k$, parameter $c > 0$.}
		\Ensure{New filtered partition $\Ss_j^{k + 1}$ for agent $j$}
        \State Sort $\Ss_j^k = \{S_{j, 1}^k, \ldots, S_{j, |\Ss_j^k|}^k\}$, so that for all $t \geq 1$ it holds $v_i(B_i \cap S_{j, t}^k) \geq v_i(B_i \cap S_{j, t + 1}^k)$
        \State Let $T_j^k \in [ |\Ss_j^k|]$ be the largest number such that $\sum_{t = 1}^{T_j^k}\lambda_j(S_{j, t}) \cdot v_j(B_i \cap S_{j, t}^k)\leq c$
        \For{$t = 1, \ldots, |\Ss_j^k|$}
            \If{$t \leq T_j^k$} update $S_{j, t}^{k + 1} := S_{j, t}^k\setminus B_{i}$
            \EndIf
            \If{$t > T_j^k$} update $S_{j, t}^{k + 1} := S_{j, t}^k$
            \EndIf
        \EndFor
        \Return $\Ss_j^{k + 1} := \{S_j^{k + 1} : S_j \in \Ss_j\}$
	\end{algorithmic}
\end{algorithm}

We will prove the following theorem about the performance of \cref{diffalgo}.
\begin{theorem}\label{difftheorem}
    Let $\alpha$ be some value in the segment $[1/4, \alpha^*)$, where $\alpha^*$ is the solution to 
    \[
        \frac{(12\alpha - 3)\ln(3\alpha)}{1 - 3\alpha} = \frac{3\ln 3}{2} - 2 \implies \alpha^* \approx 0.2767738...
    \]
    There exists $n_\alpha$ such that for all $n \geq n_\alpha$, the following holds: there exist values of $c, D > 0$ such that after running \cref{diffalgo} on $n$ agents with APS-partitions satisfying \cref{wlog}, with probability at least $1- 1/n$ for every agent $i \in [n]$ it holds $v_i(Q_i^n) \geq \alpha$.
\end{theorem}

\subsection{Analysis of the algorithm}
\label{sec:stealinganalysis}

To prove \cref{difftheorem}, we will analyze the performance of \cref{diffalgo} step by step.
First, we will show the deterministic guarantees, specifically regarding the value loss of dangerous bundles and the filtering procedure of \cref{filteralgo}.
We start with a simple observation that at every iteration, the amount of value lost by agents is at most $c$.
\begin{lemma}\label{bounddiff}
    For every $k \in [n - 1]$ and every agent $j \in \Aa^{k + 1}$ that \textbf{remains active} after iteration $k$ of \cref{diffalgo}, it holds that $0 \leq \beta_j^k - \beta_j^{k + 1} \leq c$.
\end{lemma}
\begin{proof}
    Fix $k \in [n - 1]$ and $j \in \Aa^{k + 1}$.
    Since for every $S_j \in \Ss_j$ it holds $S_j^{k + 1} \subseteq S_j^k$, $\beta_j^k \geq \beta_j^{k + 1}$.
    Next, let $i_k$ be the agent picked at iteration $k$ of \cref{diffalgo}, and let $B_{i_k}$ be the initially assigned bundle.
    In order to obtain $\Ss_j^{k + 1}$ from $\Ss_j^k$, the algorithm applies filtering procedure (\cref{filteralgo}) to $\Ss_j$ and set $B_{i_k}$: set $\Ss_j^k = \{S_{j, 1}^k, \ldots, S_{j, |\Ss_j^k|}^k\}$ is sorted in decreasing order of $v_{i_k}(B_{i_k} \cap S_{j, t}^k)$, the value of the intersection of $B_{i_k} \cap S_{j, t}^k$ for agent $i_k$.
    Then, \cref{filteralgo} removes items of $B_{i_k}$ only from sets $S_{j, t}^k$ for $t \in [T_j^k]$, i.e $S_{j, t}^{k + 1} = S_{j, t}^k \setminus B_{i_k}$, where $T_j^k$ satisfies $\sum_{t = 1}^{T_j^k}\lambda_j(S_{j, t})\cdot v_j(B_{i_k} \cap S_{j, t}^k)\leq c$, and for all other $t > T_j^k$, $S_{j, t}^{k + 1} = S_{j, t}^k$.
    By \cref{wlog}, valuation $v_j$ acts like an additive function on all $S_j \in \Ss_j$, so the lemma follows 
    from $v_j(S_{j, t}^{k + 1}) = v_j(S_{j, t}^{k }\setminus B_{i_k}) = v_j(S_{j, t}^{k })  - v_j(S_{j, t}^{k }\cap B_{i_k})$ for all $t \in [T_j^k]$ and the definition $\beta_j^{k + 1}= \sum_{t = 1}^{T_j^k}\lambda_j(S_{j_t})\cdot v_j(S_{j, t}^{k + 1}) + \sum_{t = T_j^k + 1}^{|S_j|}\lambda_j(S_{j_t})\cdot v_j(S_{j, t}^{k + 1})$.
\end{proof}

Next, before \cref{diffalgo} starts allocating bundles, for every agent $i \in [n]$ it creates a filtered partition $\Ss_i^1$ by effectively removing all dangerous bundles $S_i \in \Ss_i$.
That is, if $S_i \in \Ss_i$ has $\Delta_i(S_i) > D$ for the value $D > 0$, then $S_i^1 = \varnothing$, and otherwise $S_i^1 = S_i$.
We will show that if $D$ is large enough, then the total weight of dangerous bundles of $i$ is small.
As a corollary, removing all dangerous bundles from $\Ss_i$ inflicts a very small loss of value, so $\beta_i^1$, the value of $\Ss_i^1$, will still be almost $1$.
\begin{claim}\label{dangerous}
    For every $D > 0$ and every agent $i \in [n]$, the amount of value lost by removing all $S_i \in \Ss_i$ such that $\Delta_i(S_i) > D$ is small.
    More specifically,
    \[
        \beta_i^1 = \sum_{S_i^1 \in \Ss_i^1}\lambda_i(S_i)\cdot v_i(S_i^1) = \sum_{\substack{S_i \in \Ss_i\\\Delta_i(S_i) \leq D}}\lambda_i(S_i)\cdot v_i(S_i) \geq 1 - \frac{n - 1}{Dn}.
    \]
\end{claim}
\begin{proof}
    By \cref{wlog}, for all $S_j \in \Ss_j$ valuation $v_j$ acts like an additive function on subsets of $S_j$, and also $v_j(S_j) = 1$.
    Then, expanding the definition of $\Delta(S_i, j)$ and re-arranging summands,
    \begin{multline*}
        \sum_{S_i \in \Ss_i}\lambda_i(S_i)\cdot \Delta_i(S_i) = \sum_{S_i \in \Ss_i}\sum_{j \neq i}\lambda_i(S_i)\cdot \Delta(S_i, j) = \sum_{S_i \in \Ss_i}\sum_{j \neq i}\sum_{S_j \in \Ss_j}\lambda_i(S_i)\cdot \lambda_j(S_j)\cdot v_j(S_j \cap S_i)
        \\= \sum_{j \neq i}\sum_{S_j \in \Ss_j}\lambda_j(S_j)\sum_{S_i\in \Ss_i}\lambda_i(S_i)\cdot v_j(S_j \cap S_i) 
        = \sum_{j \neq i}\sum_{S_j \in \Ss_j}\lambda_j(S_j)\sum_{S_i\in \Ss_i}\lambda_i(S_i)\sum_{e \in S_j}v_j(\{e\}\cap S_i) \\= \sum_{j \neq i}\sum_{S_j \in \Ss_j}\lambda_j(S_j)\sum_{e \in S_j}\sum_{\substack{S_i \in \Ss_i\\ S_i \ni e}}\lambda_i(S_i)\cdot v_j(e)
        = \sum_{j \neq i}\sum_{S_j \in \Ss_j}\lambda_j(S_j)\sum_{e \in S_j}\frac{v_j(e)}{n} \\=\frac{1}{n}\sum_{j \neq i}\sum_{S_j \in \Ss_j}\lambda_j(S_j)\cdot v_j(S_j) = \frac{1}{n}\sum_{j \neq i}1 = \frac{n - 1}{n}.
    \end{multline*}
    Now, it is easy to see that for any number $D > 0$,
    \[
        \sum_{S_i \in \Ss_i}\lambda_i(S_i)\cdot \Delta_i(S_i) < \frac{n - 1}{n} \implies \sum_{\substack{S_i \in \Ss_i\\\Delta_i(S_i) > D}}\lambda_i(S_i) < \frac{n - 1}{Dn}.
    \]
    The lemma follows by $v_i(S_i) = 1$ for all $S_i \in \Ss_i$.
\end{proof}

Finally, we show that for every agent $i \in [n]$, the amount of value $i$ may lose from stealing items is small.
We do so by analyzing the filtering procedure of \cref{filteralgo}.
First, we observe that the number of agents picked after $i$ that may potentially steal items from $i$ is small, as long as there are no dangerous items.
Second, we prove that if some other agent $j$ ends up stealing items from $i$, then from the perspective of $i$ the stolen value is not too large.
We do so by noting that with respect to $i$, the average value of intersection of the bundle picked for $i$ and the remaining bundles of $j$ must be small, and that by construction of \cref{filteralgo} agent $j$ can steal only lower valued intersections of $i$.
As a result, the value of the stolen intersection cannot be too large, and $i$ does not lose much value to $j$.
Combining these two observations together gives us the desired.
\begin{claim}\label{stolenvalue}
    For $i \in [n]$, let $B_{i}$ be the bundle initially assigned to agent $i$ by \cref{diffalgo} at the iteration when $i$ was picked.
    Then $Q_i^n$, the final bundle of agent $i$, satisfies
    \[
        v_{i}(Q_{i}^n) \geq v_i(B_{i}) - \frac{D}{c(cn - 1)}.
    \]
\end{claim}
\begin{proof}
    Suppose that agent $i$ was picked at iteration $k$.
    For every $j \in \Aa^{k + 1}$, denote by $B_{j}$ the bundle \cref{diffalgo} \textbf{initially} gave to agent $j$ at a future iteration when $j$ was picked.
    By construction, the final bundle $Q_i^n$ that agent $i$ gets is exactly
    \[
        Q_i^n = B_i \setminus\bigcup_{j \in \Aa^{k + 1}}B_{j} \;\;= B_i \setminus \bigcup_{j \in \Aa^{k + 1}}(B_{i} \cap B_{j}).
    \]
    Then, by subadditivity of $v_i$,
    \[
         v_i(B_i) \leq v_i(Q_i^n) +\sum_{j \in \Aa^{k + 1}}v_i(B_i \cap B_{j}) \iff  v_i(Q_i^n) \geq v_i(B_i) - \sum_{j \in \Aa^{k + 1}}v_i(B_i \cap B_{j}).
    \]

    Fix some $j \in \Aa^{k + 1}$, and consider the filtering procedure, \cref{filteralgo}, applied to agent $j$ at step $k$.
    Let $\Ss_j^k = \{S_{j, 1}^k, \ldots, S_{j, |\Ss_j^k|}^k\}$ be sorted in decreasing order of $v_i(B_i \cap S_{j, t}^k)$, the value of the intersection of $B_i \cap S_{j, t}^k$ for agent $i$.
    In addition, let $T_j^k$ be the largest number such that $\sum_{t = 1}^{T_j^k}\lambda_j(S_{j, t})\cdot v_j(B_i \cap S_{j, t}^k)\leq c$, i.e the total value of intersections of $B_i$ with bundles from $1$ up to $T_j^k$ is at most $c$.
    There are two cases to consider.

    First: $T_j^k = |\Ss_j^k|$.
    In this case, $c \geq \sum_{t = 1}^{T_j^k}\lambda_j(S_{j, t}) \cdot v_j(B_i \cap S_{j, t}^k) = \sum_{S_j \in \Ss_j}\lambda_j(S_{j}) \cdot v_j(B_i \cap S_{j}^k)$.
    That is, the total value the items in $B_i$ that are still remaining in the partition $\Ss_j^k$ is at most $c$.
    Then, for every $S_j^k \in \Ss_j^k$ it would hold that $S_j^{k + 1} = S_j^k \setminus B_i$, i.e the filtering procedure removes all items of $B_i$ from the bundles of agent $j$'s partition.
    But then, if $B_j$ is the bundle that \cref{diffalgo} initially gives to $j$ at one of the future iterations, it must hold that $B_i \cap B_j = \varnothing$.
    Thus, agent $i$ does not lose any value from such agent $j$.
    
    Second: $T_j^k < |\Ss_j^k|$.
    In this case, by maximality of $T_j^k$,
    \[
        \Delta(B_i, j) = \sum_{S_j \in \Ss_j}\lambda_j(S_j)\cdot v_j(B_i \cap S_j) \geq \sum_{t = 1}^{|S_j^k|}\lambda_j(S_{j, t})\cdot v_j(B_i \cap S_{j, t}^k) > c.
    \]
    Suppose that this particular copy of $B_i$ came from a bundle $S_i^k \in \Ss_i^k$, then $\Delta(S_i, j) \geq \Delta(B_i, j) > c$.
    On the other hand, it must be that $S_i^1 \neq\varnothing$, hence by construction $\Delta_i(S_i) = \sum_{j \neq i}\Delta(S_i, j)\leq D$.
    Then, the number of agents $j$ such that $T_j^k < |\Ss_j^k|$ is at most $D/c$.

    Let $B_j$ be the bundle that \cref{diffalgo} initially gives to $j$ at one of the future iterations (after $k$), and suppose that this particular copy of $B_j$ came from a bundle $S_{j, r}^k$ (as in \cref{distpartition}), for some $1 \leq r \leq |S_j^k|$.
    There are two possible options.
    \begin{enumerate}
        \item if $r \leq T_j^k$,  by construction $S_{j, r}^{k + 1} = S_{j, r}^k \setminus B_i$.
        Since $j$ is picked after iteration $k$, $B_j \subseteq S_{j, r}^{k + 1}$ and $B_i \cap B_j \subseteq B_i \cap S_{j, r}^{k + 1} = \varnothing$.
        In this case, agent $i$ does not lose any value from agent $j$.
        \item if $r > T_j^k$, then by construction $S_{j, r}^{k + 1} = S_{j, r}^k$, so agent $i$ may lose some value from agent $j$, as filtering procedure does not remove items of $B_i$ from $S_{j, r}^k$.
    \end{enumerate}
    Suppose that $r > T_j^k$.
    Then, for every $1 \leq t\leq T_j^k$ it holds that $v_i(B_i \cap S_{j, r}^k) \leq v_i(B_i \cap S_{j, t}^k)$.
    Therefore,
    \[
        v_i(B_i \cap S_{j, r}^k)\cdot \sum_{t = 1}^{T_j^k}\lambda_j(S_{j, t}) \leq \sum_{t = 1}^{T_j^k}\lambda_j(S_{j, t})\cdot v_i(B_i \cap S_{j, t}^k) \leq \sum_{S_j \in \Ss_j}\lambda_j(S_j)\cdot v_i(S_j \cap B_i).
    \]
    
    Observe that by maximality of $T_j^k$ it must holds $\sum_{t = 1}^{T_j^k}\lambda_j(S_{j, t}) \geq c - 1/n$.
    Indeed, if the opposite is true, as by \cref{wlog} for every $S_j \in \Ss_j$ we have $v_j(S_j) = 1$,
    \[
        \sum_{t = 1}^{T_j^k}\lambda_j(S_{j, t})\cdot v_j(B_i \cap S_{j, t}^k) \leq \sum_{t = 1}^{T_j^k}\lambda_j(S_{j, t}) < c - 1/n.
    \]
    By \cref{wlog}, for every $S_j \in \Ss_j$ it holds that $\lambda_j(S_j) \leq 1/n$, so adding $\lambda_j(S_{j, T_j^k + 1})\cdot v_j(S_{j, T_j^k + 1}^k)$ to the sum above would increase it by at most $1/n$, implying that $T_j^k$ is not maximal.

    Now, since $B_i \subseteq S_i^k$, by \cref{wlog} $v_i(B_i) \leq v_i(S_i)= 1$, furthermore, valuation $v_i$ acts like an additive function on $B_i$.
    It follows that
    \begin{multline*}
        \sum_{S_j \in \Ss_j}\lambda_j(S_j)\cdot v_i(S_j \cap B_i)  =  \sum_{S_j \in \Ss_j}\lambda_j(S_j)\cdot \sum_{e \in S_j}v_i(\{e\} \cap B_i)\\
        = \sum_{e \in B_i}v_i(e)\sum_{\substack{S_j \in \Ss_j\\S_j\ni e}}\lambda_j(S_j) = \sum_{e \in B_i}\frac{v_i(e)}{n} = \frac{v_i(B_i)}{n} \leq \frac{1}{n}.
    \end{multline*}
    Combining all the inequalities together, we get that if $r > T_j^k$, the corresponding bundle $S^k_{j, r}$ of agent $j$ satisfies
    \begin{multline*}
        v_i(B_i \cap S_{j, r}^k)\cdot \left(c - \frac{1}{n}\right) \leq v_i(B_i \cap S_{j, r}^k)\cdot \sum_{t = 1}^{T_j^k}\lambda_j(S_{j, t}) \leq \sum_{S_j \in \Ss_j}\lambda_j(S_j)\cdot v_i(S_j \cap B_i)  \leq \frac{1}{n}\\
        \implies  v_i(B_i \cap S_{j, r}^k) \leq \frac{1}{cn - 1}.
    \end{multline*}
    \label{proof:stealing}
    As a result, if agent $j$ picks a bundle $B_j\subseteq S^k_{j, r}$ such that $r \leq T_j^k$, then $B_j\subseteq S^k_{j, r}\setminus B_i$ and hence $j$ does not steal anything from agent $i$.
    On the other hand, if $r \geq T_j^k$, agent $i$ loses at most $v_i(B_i \cap B_j) \leq v_i(B_i \cap S_{j, r}^k) \leq1/(cn - 1)$ value from agent $j$, as shown above.
    Note that every agent $j$ who steals something from $i$ can steal only once, as they pick a single bundle $B_j\subseteq S_{j, r}^k$.
    And, as we proved earlier, the number of such agents $j$ with $T_j^k < |\Ss_j^k|$ is at most $D/c$.
    Thus,
    \[
        v_i(Q_i^n) \geq v_i(B_i) - \sum_{l = k + 1}^nv_i(B_i \cap B_{j_l}) \geq v_i(B_i) - \frac{D}{c(cn - 1)}.
    \]
\end{proof}

\cref{stolenvalue} tells us the following: if for every agent $i \in [n]$ the initially assigned bundle $B_i$ had value $v_i(B_i) \geq \alpha$, then after \cref{diffalgo} finishes, agent $i$ would still have a bundle of value at least $\alpha - \frac{D}{c(cn - 1)}$.
If, for example, we pick $D = \Theta(\sqrt{n})$ and $c = \polylog n$, then each agent would receive value at least $(1 - o(1))\alpha$.
So, as mentioned earlier, in order to prove \cref{difftheorem} we need to guarantee that for every agent $i \in [n]$ and every $k \in [n]$, if $i$ was picked by \cref{diffalgo} at iteration $k$, it is possible to construct the distribution $\mu_{i}^k$ from \cref{distpartition}.
That is, to show that with high probability, for every agent $i \in [n]$ and iteration $k \in [n]$, if $i \in \Aa^k$ then the remaining value $\beta_{i}^k \geq \alpha$.
This will be the main focus of the next section, where we prove the following theorem.
\begin{theorem}\label{diffvalue}
    Let $\alpha$ be some value in the segment $(1/4, \alpha^*)$, where $\alpha^*$ is defined in \cref{difftheorem}.
    There exists $n_\alpha$ such that for all $n \geq n_\alpha$, the following holds: there exist values of $c, D > 0$ such that if we run \cref{diffalgo} on $n$ agents with APS-partitions satisfying \cref{wlog}, with probability at least $1- 1/n$ for every agent $i \in [n]$ and every $k \in [n]$, if $i \in \Aa^k$ it holds $\beta_{i}^k \geq \alpha + (\alpha^* - \alpha)/2$.
\end{theorem}

\section{Proof of \cref{diffvalue}}

\label{sec:martingale}

Our proof will rely on the observation made earlier in \cref{probval}.
We will restate it in an equivalent way, using the terminology of \cref{distpartition}.
\begin{lemma}[\cref{probval} restated]\label{diffprob}
    For $i \in [n]$, consider agent $i$ with valuation $v_i$ and APS-partition $\{(S_i, \lambda_i(S_i)\}_{S_i \in \Ss_i}$, satisfying \cref{wlog}, and let $\Ss_i^k$ be some filtration of $S_i$ at iteration $k \in [n]$, such that $\beta_i^k \geq \alpha$.
    Let $\mu_i^k$ be the distribution over subsets of $S_i^k$ for $S_i^k \in \Ss_i^k$, and let $p > 0$ be such that $\mu_i^k(e) \leq p$ for all $e \in \Mm$.
    Then
    \[
        \E[\mu_{i}^k]{\beta_j^{k + 1} \mid \Ss_i^k} \geq (1 - p)\beta_j^k.
    \]
\end{lemma}

The main idea of the proof  of \cref{diffvalue} is as follows.
Consider iteration $k \in [n]$ and agent $i_k \in \Aa^k$ picked by \cref{diffalgo} at this iteration.
By \cref{distpartition}, the distribution $\mu_{i_k}^k$ possesses the following property: for every item $e \in \Mm$, the probability $\mu_{i_k}^k(e)$ for item $e$ to belong in the bundle $B_{i_k}\sim \mu_{i_k}^k$ is at most $p(\beta_{i_k}^k)$.
Suppose for a moment that for every $e \in \Mm$ the probability $\mu_{i_k}^k(e)$ is exactly $p(\beta_{i_k}^k)$.
Then, as shown earlier in the proof of \cref{diffprob} (\cref{probval}), for every agent $j \in \Aa^{k + 1}$ that stays active after iteration $k$, the expected value of $\beta_{j}^{k + 1}$ would be exactly $(1 - p(\beta_{i_k}^k))\beta_j^k$.
That is, on average every remaining agent $j \in \Aa^{k + 1}$ would lose exactly $p(\beta_{i_k}^k)$-fraction of their current value $\beta_j^k$.
Of course, in reality agents are not guaranteed to lose exactly $p(\beta_{i_k}^k)$-fraction, they may lose much more or much less value.
However, we will show that with high probability, for every agent $j \in [n]$ the remaining value $\beta_j^k$ at all iterations $k$ (until the agent is picked) behaves not much worse than as if the agent loses exactly $p(\beta_{i_k}^k)$-fraction.
And, the sequence where each agent loses exactly $p(\beta_{i_k}^k)$-fraction or better is essentially the process $\gamma$ that was considered for the same-valuation case in the earlier section.

\subsection{Expected process}

Fix some agent $j \in [n]$, we will track the value $\beta_j^k$ at different iterations $k$ until $j$ is picked by \cref{diffalgo}.
By \cref{dangerous}, with the appropriately chosen $D > 0$ the initial value $\beta_j^1$ is at least $1 - o(1)$, and by \cref{bounddiff}, at every iteration $k$ the value $\beta_j^k$ does not decrease by more than $C > 0$.
We will show that with high probability, the sequence of values $\beta_j^k$ is not much worse than the sequence of ``expected'' values of $\beta_j^k$, where the agent loses exactly $p(\beta_{i_k}^k)$-fraction of her value.
First and foremost, we define a so-called ``expected process'' $\wgamma_j$ of agent $j$, which is essentially a sequence of values $\beta_j^k$ that is always decreased exactly by the expected fraction.

\begin{definition}
    For $k \in [n]$, let $i_k$ denote the agent that was picked at iteration $k$ of \cref{diffalgo}, and let $p_{i_k}$ be \textbf{the smallest possible} number such that for every item $e \in \Mm$, $\mu_{i_k}^k(e) \leq p_{i_k}$.
    For agent $j \in [n]$, let $\wgamma_j^k$ be the following process: $\wgamma_j^1 = \beta_j^1$ and for $k \geq 1$,
    \[
        \wgamma_j^{k + 1} = \left(1 - p_{i_k}\right)\wgamma_j^k.
    \]
\end{definition}

We will show that if the maximum value and the variance of $\wgamma_j^k - \beta_j^k$ for various $k$ is bounded, then with high probability the expected process $\wgamma_j$ essentially serves as a lower bound to the sequence of values $\beta_j$.
\begin{theorem}\label{martingale}
    Fix some $1 \leq t < n$, and let $j \in \Aa^{t + 1}$ be some agent that \textbf{remains active} after iteration $t$ of \cref{diffalgo}.
    For $k \in [t]$, let $X_k := \wgamma_j^k - \beta_j^k$, and let $\Ff_k$ be the smallest $\sigma$-algebra for random variables $\{X_1, \ldots, X_k\}$, i.e minimal $\sigma$-algebra w.r.t values $\beta_i^k$ for agents $i \in \Aa^k$.
    Suppose that
    \[
       \forall k \in [t]: |X_{k + 1} - X_k| \leq \theta,\qquad\text{and}\qquad \sum_{k = 1}^t\E{(X_{k + 1} - X_k)^2 \mid \Ff_k} \leq \eta^2;
    \]
    for some $\theta, \eta^2 > 0$.
    Then, for any $\eps > 2c$,
    \[
        \P{\wgamma_j^{t + 1}- \beta_j^{t + 1}> \eps} \leq \exp\left[\frac{\eta^2}{\theta^2}\psi\left(\frac{\eps\theta}{2\eta^2}\right)\right],
    \]
    where $\psi(x) := x - (1 + x)\ln(1 + x)$.
\end{theorem}
The idea of the proof is the following.
Suppose that \cref{diffalgo} picks agent $i_k$ at iteration $k$.
Then, by definition $\wgamma_j^{k + 1} = (1 - p_{i_k})\wgamma_j^k$ exactly.
On the other hand, since $p_{i_k}$ is an upper bound on the probability $\mu_{i_k}^k(e)$, by \cref{diffprob} in expectation the value of $\beta_j^{k + 1}$ is at least $(1 - p_{i_k})\beta_j^k$.
Then, the difference $\wgamma_j^{k + 1} - \beta_j^{k + 1}$ in expectation is at most $(1 - p_{i_k})(\wgamma_j^k - \beta_j^k)$.
Note that if $\wgamma_j^k - \beta_j^k > 0$, then $\wgamma_j^{k + 1} - \beta_j^{k + 1}$ is at most $\wgamma_j^k - \beta_j^k$.
That is, when the value of the process $\wgamma_j$ is above the value of $\beta_j$, the difference $\wgamma_j^k - \beta_j^k$ is shrinking in expectation.
Therefore, at iterations $k$ when $\wgamma_j^k > \beta_j^k$, the difference $\wgamma_j^k - \beta_j^k$ behaves like a \textit{super-martingale}.
So, to prove that $\wgamma_j$ does not exceed $\beta_j$ by more than $\eps$, we can consider the iterations where $\wgamma_j$ is above $\beta_j$ and apply concentration inequalities for super-martingales.

\begin{proof}[Proof of \cref{martingale}]
    We will consider the sequence $X \equiv X_1, \ldots, X_t, X_{t + 1}$ as a probabilistic process $\{X_k\}_{k = 1}^\infty$, where for all $k \geq t + 1$ we set $X_{k + 1} := X_{k}$.
    For integer $s\geq 1$, define the stopping time 
    \[
        \tau(s) := s + \inf\{r \geq 0 : X_{s + r} < 0\},
    \]
    i.e the first iteration of $X$ after (and including) $s$ at which the process $X$ goes below $0$.
    \begin{lemma}\label{subprocess}
        For integer $s \geq 1$, let $X^{s, \tau(s)}$ be a stopped sub-process of $X$ which starts at time $s$ and stops at time $\tau(s)$.
        That is,
        \[
            X^{s, \tau(s)} : = \{X_{\min(s + r), \tau(s)}\}_{r = 0}^\infty.
        \]
        Then, for every $s \geq 1$, the process $X^{s, \tau(s)}$ is a super-martingale.
    \end{lemma}
    \begin{proof}
        Recall that a sequence of random variables $\{\xi_{r = 1}^\infty\}$ is called a (discrete time) \textit{super-martingale} if for every time $r$ it holds
        \[
            \E{|\xi_r|} < \infty \qquad \text{and}\qquad \E{\xi_{r + 1}\mid \xi_1,\ldots, \xi_r} \leq \xi_r.
        \]
        Observe that the process $X = \{X_k\}_{k = 1}^\infty$ has a strong Markov property, for every $k\in[t]$, the values of $\wgamma_{n}^{k + 1}$ and $\beta_j^{k + 1}$ depend only on $\wgamma_j^k, \beta_j^k$ and on agents $i \in \Aa^k$ that are active in the beginning of iteration $k$ of the algorithm.
        Then, the process $\{X_{s + r}\}_{r = 0}^\infty$ is a valid random process that depends only on $X_s$ and corresponds to a sub-process of $X$ starting from time $s$ onward.
        Consequently, the process $\{X_{\min(s + r), \tau(s)}\}_{r = 0}^\infty$ is a valid random process that corresponds to $\{X_{s + r}\}_{r = 0}^\infty$ stopped at time $\tau(s)$, i.e after time $\tau(s)$, this process remains constant equal to $X_{\tau(s)}$.

        Now, for $k \in [n]$, let $i_k$ denote the agent that was picked at iteration $k$ of \cref{diffalgo}, and let $p_{i_k}$ be \textbf{the smallest possible} number such that for every item $e \in \Mm$, $\mu_{i_k}^k(e) \leq p_{i_k}$.
        By definition, for all $k \in [t]$ we have $\wgamma_j^{k + 1} = (1 - p_{i_k})\wgamma_j^k$, and by \cref{diffprob}, $\E{\beta_j^{k + 1}} \geq (1 - p_{i_k})\beta_j^k$.
        Hence, for every $k \in [t]$
        \[
            \E{X_{k + 1}\mid \Ff_k}  = \wgamma_j^{k + 1}  - \E{\beta_j^{k + 1}} \leq (1 - p_{i_k})(\wgamma_j^k - \beta_j^k) = (1 - p_{i_k})X_k.
        \]
        Then, it must hold for every $r \geq 0$ that $\E{X_{r + 1}^{s, \tau(s)}\mid X_r^{s, \tau(s)}} \leq X_r^{s, \tau(s)}$.
        Indeed,
        \begin{itemize}
            \item if $X_r^{s, \tau(s)} < 0$, then $s + r \geq \tau(s)$ and the process has already stopped, so $X_{r + 1}^{s, \tau(s)} = X_r^{s, \tau(s)}$;
            \item if $X_r^{s, \tau(s)} \geq 0$, then $s + r < \tau(s)$, and hence $X_r^{s, \tau(s)} = X_{s + r}$ and $X_{r + 1}^{s, \tau(s)} = X_{s + r + 1}^s$, implying
            \[
                \E{X_{r + 1}^{s, \tau(s)} \mid X_r^{s, \tau(s)}} = \E{X_{s + r + 1}^s\mid X_{s + r}^s} \leq (1 - p_{i_{s + r}})X_{s + r} = (1 - p_{i_{s + r}})X_r^{s, \tau(s)} \leq X_r^{\tau(s)}.
            \]
        \end{itemize}
        Finally, as $|X_k| \leq 1$ for any $k$, we get that the process $X^{s, \tau(s)}$ is a super-martingale.
    \end{proof}

    Suppose that for the given $1 \leq t < n$ it holds $\wgamma_j^{t + 1} - \beta_j^{t + 1} > \eps$.
    Then, it must also hold that $\wgamma_j^t - \beta_j^t > 0$.
    Indeed, as $\wgamma_j^{t + 1} \leq \wgamma_j^t$ and by \cref{bounddiff}, $\beta_j^{t + 1} \geq \beta_j^t - c$, then 
    \[
        \eps < \wgamma_j^{t + 1} - \beta_j^{t + 1} \leq \wgamma_j^t - \beta_j^t + c \implies \wgamma_j^t - \beta_j^t > \eps - c \geq 0.
    \]
    Let $1 \leq s \leq t$ be the smallest integer such that for every $s \leq k \leq t$ it holds $\wgamma_j^s - \beta_j^s \geq 0$, not that for such $s$ either $\wgamma_j^{s - 1} - \beta_j^{s - 1} < 0$ or $s = 1$.
    In other words, $s$ is the start of the segment $[s, t] = \{s, s + 1, \ldots, t\}$ where for every $k \in [s, t]$ the value $X_k= \wgamma_j^k - \beta_j^k$ stays non-negative.
    Since $\wgamma_j^t - \beta_j^t > 0$, such $1 \leq s \leq t$ must exist.

    Note that for this $s$, it must hold that $\wgamma_j^s - \beta_j^s \leq c$.
    Indeed, by the choice of $s$, either $s = 1$ and $\wgamma_j^s -\beta_j^s = 0$, or $s > 1$ and $\wgamma_j^{s - 1} - \beta_j^{s - 1} < 0$.
    Since $\wgamma_j^{s} \leq \wgamma_j^{s - 1}$ and by \cref{bounddiff}, $\beta_j^{s} \geq \beta_j^{s - 1} - c$, if one had $\wgamma_j^s - \beta_j^s > c$ then
    \[
        c < \wgamma_j^s - \beta_j^s \leq \wgamma_j^{s - 1} - \beta_j^{s- 1} + c < c,
    \]
    leading to a contradiction.
    Therefore, $\wgamma_j^{t + 1} - \beta_j^{t + 1} > \eps \geq \wgamma_j^s - \beta_j^s + \eps - c \geq \wgamma_j^s - \beta_j^s + \eps/2$, and
    \[
        \P{\wgamma_j^{t + 1} - \beta_j^{t + 1} > \eps} \leq \P{\wgamma_j^{t + 1} - \beta_j^{t + 1} > \wgamma_j^s - \beta_j^s + \eps/2} = \P{X_{t + 1} - X_s > \eps/2}.
    \]
    
    Consider the stopped sub-process $X^{s, \tau(s)}$ which starts at time $s$ and stops at time $\tau(s)$.
    By the choice of $s$ it holds that $s < \tau(s)$ and $t + 1 < \tau(s)$ and $t - s + 1 > 0$, so for every $0 \leq r \leq t - s + 1$ we have $X^{s, \tau(s)}_r = X_{s + r}$, and
    \[
        \P{X_{t + 1} - X_s > \eps/2} = \P{X^{s, \tau(s)}_{t - s + 1} - X^{s, \tau(s)}_0 > \eps/2}.
    \]
    We will use the following bound by \cite{Pin94} on the deviation probabilities of super-martingales.
    \begin{theorem}[\cite{Pin94}, theorem 8.2]\label{pinelis}
        For a given $K$, let $\{Y_k\}_{k = 0}^K$ be a real-valued super-martingale with respect to filtration $\{\Ff_k\}_{k = 0}^K$ of $\sigma$-algebras.
        Suppose that
        \[
            \forall k \leq K - 1: |Y_{k + 1} - Y_k|\leq \theta; \qquad \text{and}\qquad \sum_{k = 0}^{K - 1}\E{(Y_{k + 1} - Y_k)^2\mid\Ff_k} \leq \eta^2.
        \]
        for some $\theta, \eta^2 > 0$.
        Then, for any $\eps > 0$,
        \[
            \P{Y_K - Y_0 \geq \eps} \leq \exp\left[\frac{\eta^2}{\theta^2}\psi\left(\frac{\eps\theta}{\eta^2}\right)\right],
        \]
        where $\psi(x) := x - (1 + x)\ln(1 + x)$.
    \end{theorem}
    As shown in \cref{subprocess}, $\{X^{s, \tau(s)}_r\}_{r = 0}^{t - s + 1}$ is a super-martingale, with respect to filtration $\{\Ff_r^{s, \tau(s)}\}_{r = 0}^{t - s + 1}$ where for $0 \leq r \leq t - s + 1$ the $\sigma$-algebra $\Ff_r^{s, \tau(s)}$ is equal to $\Ff_{s + r}$ of the process $X$.
    Then, under the conditions on $X$ given by the statement of the theorem, we have
    \[ 
        \forall r \leq t -s: |X^{s, \tau(s)}_{r + 1} - X^{s, \tau(s)}_r| = |X_{s + r + 1} - X_{s + r}| \leq\theta,
    \]
    and
    \begin{multline*}
        \sum_{r = 0}^{t - s}\E{\left(X^{s, \tau(s)}_{r + 1} - X^{s, \tau(s)}_r\right)^2\mid\Ff_r^{s, \tau(s)}} = \sum_{r = 0}^{t - s}\E{\left(X_{s + r + 1} - X_{s + r}\right)^2\mid\Ff_{s + r}} \\\leq \sum_{k = 1}^t\E{(X_{k + 1} - X_k)^2\mid \Ff_k} \leq\eta^2,
    \end{multline*}
    so by applying \cref{pinelis} to $\{X^{s, \tau(s)}_r\}_{r = 0}^{t - s + 1}$ and $\eps / 2$ we obtain the bound
    \[
         \P{\wgamma_j^{t + 1} - \beta_j^{t + 1} > \eps} \leq  \P{X^{s, \tau(s)}_{t - s + 1} - X^{s, \tau(s)}_0 > \eps/2} \leq \exp\left[\frac{\eta^2}{\theta^2}\psi\left(\frac{\eps\theta}{2\eta^2}\right)\right].
    \]
\end{proof}

In order to apply \cref{martingale}, we need to upper bound the difference $|X_{k + 1} - X_k|$, as well as the sum of variances $\sum_{k = 1}^t\E{(X_{k + 1} - X_k)^2\mid \Ff_k}$.
This is the focus of the next two lemmas.

\begin{lemma}\label{boundexpec}
    Fix some $1 \leq t < n$ and $j \in \Aa^{t + 1}$.
    For $k \in [t]$, let $X_k := \wgamma_j^k - \beta_j^k$, and let $\Ff_k$ be the smallest $\sigma$-algebra for random variables $\{X_1, \ldots, X_k\}$.
    For any $k \in [t]$:
    \[
         |X_{k + 1} - X_k| \leq |\E{X_{k + 1} - X_k \mid \Ff_k}| + c \leq  p_{i_k} + c.
    \]
\end{lemma}
\begin{proof}
    It is easy to see that 
    \begin{multline*}
        X_{k + 1} - X_k - \E{X_{k + 1} - X_k \mid \Ff_k}\\
        = (\wgamma^{k + 1}_j - \beta_j^{k + 1}) - (\wgamma^k_j - \beta_j^k) - \E{(\wgamma^{k + 1}_j - \beta_j^{k + 1}) - (\wgamma^k_j - \beta_j^k) \mid \Ff_k}\\
        =(\wgamma^{k + 1}_j - \beta_j^{k + 1}) - (\wgamma^k_j - \beta_j^k) - (\wgamma^{k + 1}_j - \wgamma^k_j) + \E{\beta_j^{k + 1} - \beta_j^k \mid \Ff_k}\\
        = \E{\beta_j^{k + 1} - \beta_j^k \mid \Ff_k} - (\beta_j^{k + 1} - \beta_j^k).
    \end{multline*}
    By \cref{bounddiff} it holds $0 \geq \beta_j^{k + 1} - \beta_j^k \geq -c$, and $0 \geq \E{\beta_j^{k + 1} - \beta_j^k \mid \Ff_k} \geq -c$.
    Hence, $-c \leq \E{\beta_j^{k + 1} - \beta_j^k \mid \Ff_k} - (\beta_j^{k + 1} - \beta_j^k) \leq c$, which implies
    \begin{multline*}
        \E{X_{k + 1} - X_k \mid \Ff_k} - c \leq X_{k + 1} - X_k \leq \E{X_{k + 1} - X_k \mid \Ff_k} + c \\
        \iff |X_{k + 1} - X_k| \leq |\E{X_{k + 1} - X_k \mid \Ff_k}| + c.
    \end{multline*}
    Next, by definition $\wgamma^{k + 1}_j = (1 - p_{i_k})\wgamma^k_j$, and by \cref{diffprob}, $0 \geq \E{\beta_j^{k + 1} - \beta_j^k\mid \Ff_k} \geq -p_{i_k}\beta_j^k$.
    Then, since
    \[
        \E{X_{k + 1} - X_k\mid \Ff_k}  =(\wgamma^{k + 1}_j - \wgamma^k_j) - \E{\beta_j^{k + 1} - \beta_j^k\mid \Ff_k}.
    \]
    the value $\E{X_{k + 1} - X_k\mid \Ff_k}$ is at most $-p_{i_k}(\wgamma_j^k - \beta_j^k)$ and at least $-p_{i_k}\wgamma_j^k$.
    As both $\wgamma_j^k$ and $\beta_j^k$ are between $0$ and $1$, we get $|\E{X_{k + 1} - X_k\mid \Ff_k}| \leq p_{i_k}$.
\end{proof}

\begin{lemma}\label{boundvar}
    Fix some $1 \leq t < n$ and $j \in \Aa^{t + 1}$.
    For $k \in [t]$, let $X_k := \wgamma_j^k - \beta_j^k$, and let $\Ff_k$ be the smallest $\sigma$-algebra for random variables $\{X_1, \ldots, X_k\}$.
    For any $k \in [t]$:
    \[
        \E{(X_{k + 1} - X_k)^2 \mid \Ff_k} \leq p_{i_k}(5p_{i_k} + c).
    \]
\end{lemma}
\begin{proof}
    Observe that
    \[
        \E{(X_{k + 1} - X_k)^2 \mid \Ff_k} = \E{(X_{k + 1} - X_k - \E{X_{k + 1} - X_k \mid \Ff_k})^2 \mid \Ff_k} + \E{X_{k + 1} - X_k\mid \Ff_k}^2
    \]
    By \cref{boundexpec}, $|\E{X_{k + 1} - X_k\mid \Ff_k}| \leq p_{i_k}$.
    Furthermore, the proof of \cref{boundexpec} implies that
    \begin{multline*}
        \left|X_{k + 1} - X_k - \E{X_{k + 1} - X_k\mid \Ff_k}\right| = \left|\E{\beta_j^{k + 1} - \beta_j^k\mid \Ff_k} - (\beta_j^{k + 1} - \beta_j^k)\right|\\
        \leq \left|\E{\beta_j^{k + 1} - \beta_j^k\mid \Ff_k}\right| + |\beta_j^{k + 1} - \beta_j^k| = p_{i_k}\beta_j^k + (\beta_j^k - \beta_j^{k + 1}),
    \end{multline*}
    hence
    \begin{multline*}
        \left(X_{k + 1} - X_k - \E{X_{k + 1} - X_k\mid \Ff_k}\right)^2 \leq (p_{i_k}\beta_j^k + (\beta_j^k - \beta_j^{k + 1}))^2 \\
        = (p_{i_k}\beta_j^k)^2 + 2p_{i_k}\beta_j^k(\beta_j^k - \beta_j^{k + 1}) + (\beta_j^k - \beta_j^{k + 1})^2.
    \end{multline*}
    Therefore, using $\E{\beta_j^k - \beta_j^{k + 1}\mid \Ff_k}\leq p_{i_k}\beta_j^k$,
    \begin{multline*}
        \E{\left(X_{k + 1} - X_k - \E{X_{k + 1} - X_k\mid \Ff_k}\right)^2\mid \Ff_k} \\
        \leq (p_{i_k}\beta_j^k)^2 + 2p_{i_k}\beta_j^k\E{\beta_j^k - \beta_j^{k + 1}\mid \Ff_k} + \E{(\beta_j^k - \beta_j^{k + 1})^2\mid \Ff_k}\\
        \leq 3(p_{i_k}\beta_j^k)^2 + \E{(\beta_j^k - \beta_j^{k + 1})^2\mid \Ff_k}.
    \end{multline*}
    To bound $\E{(\beta_j^k - \beta_j^{k + 1})^2\mid \Ff_k}$, consider the function $f(x) := (x - \E{x})^2$ for an arbitrary random variable $x$ distributed on the segment $[0, 1]$.
    Since $f$ is convex, it holds that
    \[
        \forall x : f(x) = f\big((1 - x)\cdot 0 + x\cdot 1\big) \leq (1 - x)f(0) + xf(1).
    \]
    But then, the variance $\V{x} = \E{f(x)}$ is at most
    \begin{multline*}
        \E{f(x)} \leq \E{1 - x}\cdot f(0) + \E{x}\cdot f(1) 
        = (1 - \E{x})\cdot \E{x}^2 + \E{x}\cdot (1 - \E{x})^2\\
        =\E{x}^2 - \E{x}^3 + \E{x} - 2\E{x}^2 + \E{x}^3 = \E{x}\cdot (1 - \E{x}).
    \end{multline*}
    Note that $(\beta_j^k - \beta_j^{k + 1} \mid \Ff_k)$ is a non-negative random variable with expectation at most $p_{i_k}\beta_j^k$, and maximum value at most $c$ (by \cref{bounddiff}).
    Then, $\V{\beta_j^k - \beta_j^{k + 1} \mid \Ff_k}$ is at most
    \begin{multline*}
        \V{\beta_j^k - \beta_j^{k + 1} \mid \Ff_k} = c^2\V{\frac{\beta_j^k - \beta_j^{k + 1}}{c} \mid \Ff_k} \\
        \leq c^2\E{\frac{\beta_j^k - \beta_j^{k + 1}}{c} \mid \Ff_k}\left(1 - \E{\frac{\beta_j^k - \beta_j^{k + 1}}{c} \mid \Ff_k}\right) \leq c\E{\beta_j^k - \beta_j^{k + 1} \mid \Ff_k} \leq cp_{i_k}\beta_j^k.
    \end{multline*}
    It follows that
    \[
        \E{(\beta_j^k - \beta_j^{k + 1})^2\mid \Ff_k} = \V{(\beta_j^k - \beta_j^{k + 1})^2\mid \Ff_k} + \E{\beta_j^k - \beta_j^{k + 1}\mid \Ff_k}^2 \leq cp_{i_k}\beta_j^k + (p_{i_k}\beta_j^k)^2,
    \]
    and thus
    \[
        \E{\left(X_{k + 1} - X_k - \E{X_{k + 1} - X_k\mid \Ff_k}\right)^2\mid \Ff_k}  \leq 4(p_{i_k}\beta_j^k)^2 + cp_{i_k}\beta_j^k.
    \]
    Using the fact that $\beta_j^k \leq 1$ always, we obtain the final bound
    \begin{multline*}
        \E{(X_{k + 1} - X_k)^2 \mid \Ff_k} 
        \\= \E{(X_{k + 1} - X_k - \E{X_{k + 1} - X_k \mid \Ff_k})^2 \mid \Ff_k} + \E{X_{k + 1} - X_k\mid \Ff_k}^2\\
        \leq 4(p_{i_k}\beta_j^k)^2 + cp_{i_k}\beta_j^k + (p_{i_k})^2 \leq 5(p_{i_k})^2 + cp_{i_k}.
    \end{multline*}
\end{proof}

One could combine \cref{martingale}, \cref{boundexpec} and \cref{boundvar}, we get the following bound: if $j \in \Aa^{t + 1}$ is some agent that remains active after iteration $t$ of \cref{diffalgo}, then for any $\eps > 2c$
\begin{equation}
    \P{\wgamma_j^{t + 1}- \beta_j^{t + 1}> \eps \mid \Ff_t} \leq \exp\left[\frac{\eta^2}{\theta^2}\psi\left(\frac{\eps\theta}{2\eta^2}\right)\right],\label{badbound}
\end{equation}
where $\Ff_t$ is the smallest $\sigma$-algebra for random variables $\{\beta_j^1, \ldots, \beta_j^t\}$ and $\theta = \max_{k \in [t]}(p_{i_k} + c)$, $\eta^2 = \sum_{k = 1}^tp_{i_k}(5p_{i_k} + c)$.
The issue with such a bound is that the values of $p_{i_k}$ for $k \in [t]$ are random variables themselves, and depend on $\Ss_{i_k}^k$, the filtered partition of the agent picked at corresponding iteration.
As a result, the process $\wgamma_j$ is also a random process, so the bound obtained in \cref{badbound} just by itself does not allow us to definitely conclude that the value of $\beta_j^{t + 1}$ is at least $\alpha$, one would also need to bound the process $\wgamma_j^{t + 1}$ first, as well as all the values of $p_{i_k}$ for $k \in [t]$.
Instead, we would like to replace $\wgamma_j$ with some \textit{deterministic} process, similar to what we did in the same-valuation case (\cref{algogamma}).

Consider \cref{diffalgo} and iteration $k$, and for simplicity we assume that $\alpha\leq 3/11$.
Then, according to \cref{distpartition}, we can upper bound the value $p_{i_k}$ by
\[
    p_{i_k} \leq \frac{1 - \alpha}{2(\beta_{i_k}^k - \alpha)n}.
\]
By construction, agent $i_k$ picked at iteration $k$ had the largest value $\beta_{i_k}^k$.
But then, $\beta_{i_k}^k \geq \beta_j^k$ and
\[
    p_{i_k} \leq \frac{1 - \alpha}{2(\beta_{j}^k - \alpha)n}.
\]
Suppose now that $\beta_j^k \geq \wgamma_j^k - \eps$ for some $\eps > 0$.
Then,
\begin{equation}
    p_{i_k} \leq \frac{1 - \alpha}{2(\wgamma_{j}^k - \eps - \alpha)n} \qquad\implies\qquad \wgamma_j^{k + 1} = (1 - p_{i_k})\wgamma_j^k \geq \left(1 - \frac{1 - \alpha}{2(\wgamma_{j}^k - \eps - \alpha)n}\right)\wgamma_j^k \label{detpro}.
\end{equation}
Note that the latter expression is very similar to the original process $\gamma_n^k$ that we used in the same-valuation case, and it provides a lower bound on $\wgamma_j^{k + 1}$ in terms of $\wgamma_j^k$.
Then, intuitively, if for all $k \in [t]$ the lower bound $\beta_j^k \geq \wgamma_j^k - \eps$ holds, then the random process $\wgamma_j^k$ should be lower bounded by the deterministic process which matches the expression in \cref{detpro} exactly.
And, if $\beta_j^1$ is close to $1$, and the deviation $\eps$ is small, then this deterministic process should be close to the original $\gamma_n^k$, which in turn was proven to be at least $\alpha - c$ for any constant $c$.

We will first cover a simpler case of $\alpha \leq 3/11$ and present a complete proof for this range of $\alpha$.
The case of $\alpha > 3/11$ uses exactly the same ideas, but is a bit more involving technically and more formula-heavy, and we will cover it in the next subsection.

\subsection{Proof of \cref{diffvalue} for $\alpha \leq 3/11$}

We will generalize the process $\gamma_n^k$ used in the same valuation case from Appendix~\ref{identical}, allowing the starting point to be some $\beta \leq 1$, as well as introducing an error $\eps$ inside the step-size.
We will call the new process $\gamma_n^k(\beta, \eps)$.

\begin{definition}
    For $n$ agents, $1/4 \leq \alpha \leq 3/11$, $\eps > 0$ and $1\geq \beta > \alpha + \eps$, let $\gamma_n^{k}(\beta, \eps)$ be the following process: $\gamma_n^1(\beta, \eps) = \beta$ and for $k \geq 1$,
    \[
            \gamma_n^{k + 1}(\beta, \eps) = \left(1 - \frac{1 - \alpha}{2(\gamma_n^k(\beta, \eps) - \eps - \alpha)n}\right)\gamma_n^k(\beta, \eps);
    \]
\end{definition}

We will show that if for all $k \in [t]$ the expected process $\wgamma_j^k$ lower bounds $\beta_j^k$ (up to an $\eps$-error), then the deterministic process $\gamma_n^k(\beta_j^1, \eps)$ also lower bounds $\beta_j^k$ up to an $\eps$-error.

\begin{claim}\label{epsprocess}
    For $1 \leq t < n$, let $j \in \Aa^{t + 1}$ be some agent that \textbf{remains active} after iteration $t$ of \cref{diffalgo}.
    For $\eps >0$, let $\Ee_j^t(\eps)$ be the following event:
    \[
        \Ee_j^t(\eps) = \left\{\forall k \in [t]: \beta_j^k \geq \wgamma_j^k - \eps\right\}.
    \]
    Under the event $\Ee_j^t(\eps)$, it holds for all $k \in [t]$ that $\beta_j^k \geq \gamma_n^k(\beta_j^1, \eps) - \eps$.
\end{claim}
\begin{proof}
    It suffices to show that for all $k \in [t]$ it holds $\wgamma_j^k \geq \gamma_n^k(\beta_j^1, \eps)$.
    We prove this via induction on $k \in [t]$.
    By definition, $\wgamma_j^1 = \beta_j^1 = \gamma_n^1(\beta_j^1, \eps)$.
    Next, suppose that for all $1 \leq k'\leq k$ it holds $\wgamma_j^{k'} \geq \gamma_n^{k'}(\beta_j^1, \eps)$.
    We show that $\wgamma_j^{k + 1} \geq \gamma_n^{k + 1}(\beta_j^1, \eps)$.
    
    For $k \in [n]$, let $i_k$ denote the agent that was picked at iteration $k$.
    Recall that $p_{i_k}$ is \textbf{the smallest possible} number such that for every item $e \in \Mm$, $\mu_{i_k}^k(e) \leq p_{i_k}$.
    By \cref{distpartition},
    \[
        p_{i_k} \leq \frac{1 - \alpha}{2(\beta_{i_k}^k - \alpha)n}.
    \]
    Then, by definition of $\wgamma_j^k$, 
    \[
        \wgamma_j^{k + 1} = (1 - p_{i_k})\wgamma_j^k \geq  \left(1 - \frac{1 - \alpha}{2(\beta_{i_k}^k - \alpha)n}\right)\wgamma_j^k.
    \]
    Note that since $j \in \Aa^{t + 1}$, it must hold for every $k \in [t]$ that $\beta_j^k \leq \beta_{i_k}^k$, as otherwise agent $j$ would have been picked at step $k$.
    Then, it holds for every $k \in [t]$ that
    \[
        \wgamma_j^{k + 1} \geq \left(1 - \frac{1 - \alpha}{2(\beta_{i_k}^k - \alpha)n}\right)\wgamma_j^k \geq \left(1 - \frac{1 - \alpha}{2(\beta_{j}^k - \alpha)n}\right)\wgamma_j^k.
    \]
    Suppose that event $\Ee_j^t(\eps)$ holds, then for all $k \in [t]$ it holds $\beta_j^k \geq \wgamma_j^k - \eps$, and
    \[
        \wgamma_j^{k + 1}\geq \left(1 - \frac{1 - \alpha}{2(\beta_{j}^k - \alpha)n}\right)\wgamma_j^k \geq \left(1 - \frac{1 - \alpha}{2(\wgamma_j^k - \eps - \alpha)n}\right)\wgamma_j^k.
    \]
    Finally, it holds by induction that $\wgamma_j^k \geq \gamma_n^k(\beta_j^1, \eps)$.
    Then, at step $k + 1$:
    \[
        \wgamma_j^{k + 1}\geq\left(1 - \frac{1 - \alpha}{2(\wgamma_j^k - \eps - \alpha)n}\right)\wgamma_j^k \\
        \geq \left(1 - \frac{1 - \alpha}{2(\gamma_n^k(\beta_j^1, \eps) - \eps - \alpha)n}\right)\gamma_n^k(\beta_j^1, \eps) = \gamma_n^{k + 1}(\beta_j^1, \eps).
    \]
    It follows that under the event $\Ee_j^t(\eps)$, for all $k \in [t]$ it holds $\beta_j^k \geq \gamma_n^k(\beta_j^1, \eps) - \eps$.
\end{proof}

The lower bound of \cref{epsprocess} allows to obtain a concrete upper bound on the values of $p_{i_k}$.
\begin{corollary}\label{pikbound}
    For $k \in [n]$, let $i_k$ denote the agent that was picked at iteration $k$ of \cref{diffalgo}.
    For any $1 \leq t < n$, under the event $\Ee_j^t(\eps)$ it holds  for every $k \in [t]$ that $p_{i_k} \leq \frac{1 - \alpha}{2(\gamma_n^k(\beta_j^1, \eps) - \eps - \alpha)n}$.
\end{corollary}
\begin{proof}
    Follows immediately by definition of $p_{i_k}$, \cref{distpartition} and (the proof of) \cref{epsprocess}.
    %
\end{proof}

Before making use of this upper bound, we will describe the main idea of the rest of the proof.
Let $i \in \Aa^t$ be the agent \textbf{picked} at iteration $t \in [n]$.
Let $\rho$ satisfy $2(1 - \rho + \alpha \ln \rho) = 1-\alpha$.
Suppose that, for some value of $n$ and functions $D = D(n)$ and $\eps = \eps(n)$ we have $\gamma_n^n(\beta_{i}^1, \eps) - \eps\geq \alpha + (\rho - \alpha)/2$.
Then, by \cref{epsprocess} under the event $\Ee_{i}^t(\eps)$ it holds $\beta_{i}^t \geq \gamma_n^t(\beta_{i}^1, \eps) - \eps \geq  \gamma_n^n(\beta_{i}^1, \eps) - \eps \geq \alpha + (\rho - \alpha)/2$ --- the desired lower bound in \cref{diffvalue}.
So, to prove \cref{diffvalue} it suffices to show that for $n$ large enough, with the appropriate choice of $c = c(n), D$ and $\eps$, the event $\Ee_i^t(\eps)$ occurs with probability at least $1 - 1/n^3$.
Then, taking union bound over all iterations $t \in [n]$ for a fixed agent $i$, and after that over all agents $i \in [n]$, gives the desired probability of $1 - 1/n$.

The lower bound $\gamma_n^n(\beta_i^1, \eps) - \eps \geq \alpha +(\rho - \alpha)/2
$ for large $n$ will also be crucial in proving that the event $\Ee_i^t(\eps)$ happens with high probability, since this lower bound by \cref{pikbound} implies that under the event $\Ee_i^k(\eps)$ for any $k< t$ it holds $p_{i_k} = O(1/n)$.
As shown in \cref{dangerous}, for every agent $i \in [n]$ we have $\beta_i^1 \geq 1 - \frac{1}{D}$, hence (by induction) for every $k \in [n]$, $\gamma_n^k(\beta_i^1, \eps) \geq \gamma_n^k(1 -  \frac{1}{D}, \eps)$.
So, it suffices to prove that for any $\eps, \delta \xrightarrow{n\to\infty}0$ there exists some large $n_\alpha$ such that for all $n \geq n_\alpha$ it holds $\gamma_n^n(1 - \delta, \eps) - \eps \geq \alpha + (\rho - \alpha)/2$.
We prove this lower bound in the following theorem.
\begin{theorem}\label{diffgamma}
    Let $\alpha \in [1/4, 3/11)$, and let $\rho$ satisfy $2(1 - \rho + \alpha \ln \rho) = 1-\alpha$.
    For any $\eps, \delta \xrightarrow{n\to\infty}0$, there exists $n_\alpha$ such that for all $n \geq n_\alpha$ it holds $\gamma_n^n(1 - \delta, \eps) -\eps \geq \alpha  + (\rho - \alpha)/2$.
\end{theorem}

The proof is rather technical, and we refer the reader to \cref{sec:smallalpha}.
With \cref{diffgamma} equipped, we are ready to show that the event $\Ee_i^t(\eps)$ for fixed $i$ and $t$ occurs with high probability.
This is done via the law of total probability and repeatedly applying the bound of \cref{martingale}.

\begin{corollary}\label{totalprob}
    Let $\alpha \in [1/4, 3/11)$ and $\rho$ satisfy $2(1 - \rho + \alpha \ln \rho) = 1-\alpha$.
    Let $\eps, c, D > 0$ satisfy  $\eps, c\xrightarrow{n\to\infty}0$, $D\xrightarrow{n\to\infty}\infty$ and $\eps > 2c$.
    Let $n_\alpha$ be such that for all $n \geq n_\alpha$ it holds $\gamma_n^n(1 - \frac{1}{D}, \eps)- \eps \geq \alpha + (\rho-\alpha)/2$.
    For a fixed $n \geq n_\alpha$ and $1 \leq t \leq n$, let $i \in \Aa^t$ be the agent \textbf{picked} at iteration $t$ of \cref{diffalgo} for $n$ agents.
    Then,
    \[
        \P{\Ee_i^t(\eps)} \geq 1 - t\cdot \exp\left[-\frac{\eps^2}{24\eta^2}\right].
    \]
    where
    \[
        \eta^2 = \left(\frac{5(1 - \alpha)}{(\rho - \alpha)n} + c\right) \frac{1 - \alpha}{\rho - \alpha}.
    \]
\end{corollary}
\begin{proof}
    By definition of the event $\Ee_i^t(\eps)$, for any $1 \leq s < t$
    \[
        \Ee_i^{s + 1}(\eps) = \left\{\forall k \in [s + 1]: \beta_j^{k} \geq \wgamma_j^{k} - \eps\right\} = \left\{\beta_i^{s + 1} \geq \wgamma_i^{s + 1} - \eps\right\} \cap \Ee_i^{s}(\eps).
    \]
    Then, by the law of total expectation,
    \begin{multline*}
        \P{\Ee_i^t(\eps)} = \P{\beta_i^t \geq \wgamma_i^t - \eps \mid \Ee_i^{t - 1}(\eps)}\cdot \P{\Ee_i^{t - 1}(\eps)} \\
        =\ldots = \P{\Ee_i^1(\eps)}\cdot \prod_{s = 1}^{t - 1}\P{\beta_i^{s + 1} \geq \wgamma_i^{s + 1} - \eps \mid \Ee_i^s(\eps)}.
    \end{multline*}
    By definition, $\wgamma_i^1 = \beta_i^1$, so $\P{\Ee_i^1(\eps)} = 1$.
    Next, in order to bound $\P{\beta_i^{s + 1} \geq \wgamma_i^{s + 1} - \eps \mid \Ee_i^s(\eps)}$ for $s \in [t]$, we use \cref{martingale}:
    \[
        \P{\beta_i^{s + 1} \geq \wgamma_i^{s + 1} - \eps \mid \Ee_i^s(\eps)} = 1 - \P{\wgamma_i^{s + 1} - \beta_i^{s + 1}\geq \eps \mid \Ee_i^s(\eps)} \geq 1 - \exp\left[\frac{\eta^2}{\theta^2}\psi\left(\frac{\eps \theta}{2\eta^2}\right)\right],
    \]
    where $\psi(x) := x - (1 + x)\ln(1 + x)$ and $\theta, \eta^2 > 0$ are such that for $X_k := \wgamma_i^{k} - \beta_i^{k}$, under $\Ee_i^s(\eps)$
    \[
       \forall k \in [s]: |X_{k + 1} - X_k| \leq \theta,\qquad\text{and}\qquad \sum_{k = 1}^s\E{(X_{k + 1} - X_k)^2 \mid \Ff_k} \leq \eta^2.
    \]
    By \cref{boundexpec} and \cref{boundvar} we have the following upper bounds:
    \[
        \forall k \in [s]: |X_{k + 1} - X_k| \leq p_{i_k} + c
    \]
    and
    \[
        \sum_{k = 1}^s\E{(X_{k + 1} - X_k)^2 \mid \Ff_k} \leq  \sum_{k = 1}^sp_{i_k}(5p_{i_k} +c).
    \]
    Now, by \cref{pikbound} under the event $\Ee_i^s(\eps)$ it holds
    \[
        \forall k \in [s]: p_{i_k} \leq \frac{1 - \alpha}{2(\gamma_n^k(\beta_i^1, \eps) - \eps - \alpha)n} \leq \frac{1 - \alpha}{2(\gamma_n^n(\beta_i^1, \eps) - \eps - \alpha)n},
    \]
    hence under the event $\Ee_i^s(\eps)$ we have upper bounds
    \[
        \forall k \in [s]: |X_{k + 1} - X_k| \leq \frac{1 - \alpha}{2(\gamma_n^n(\beta_i^1, \eps) - \eps - \alpha)n} + c  
    \]
    and 
    \[
        \sum_{k = 1}^s\E{(X_{k + 1} - X_k)^2 \mid \Ff_k} \leq \left(\frac{5(1 - \alpha)}{2(\gamma_n^n(\beta_i^1, \eps) - \eps - \alpha)n} + c\right) \frac{1 - \alpha}{2(\gamma_n^n(\beta_i^1, \eps) - \eps - \alpha)}.
    \]
    Recall that by \cref{dangerous}, for every agent $i \in [n]$ we have $\beta_i^1 \geq 1 - \frac{1}{D}$.
    It is easy to see (via induction on $k$ as in \cref{epsprocess}) that for every $k \in [n]$ we have a lower bound  $\gamma_n^k(\beta_i^1, \eps) \geq \gamma_n^k(1 - \frac{1}{D}, \eps)$.
    Then, since $n \geq n_\alpha$, for these $n$ and agent $i$ it holds $\gamma_{n}^n(\beta_i^1, \eps)-\eps \geq \alpha + (\rho - \alpha)/2$.
    So under the event $\Ee_i^s(\eps)$ we have
    \[
        \forall k \in [s]: |X_{k + 1} - X_k| \leq \frac{1 - \alpha}{(\rho - \alpha)n} + c,\qquad \sum_{k = 1}^s\E{(X_{k + 1} - X_k)^2 \mid \Ff_k} \leq\left(\frac{5(1 - \alpha)}{(\rho - \alpha)n} + c\right) \frac{1 - \alpha}{\rho - \alpha}.
    \]
    We pick these upper bounds as $\theta$ and $\eta^2$ respectively.
    Observe that in this case
    \[
        \frac{\eps\theta}{2\eta^2} < \frac{(\rho -\alpha)\eps}{1 - \alpha} < 1.
    \]
    Now, since $\ln(1 + x) \geq x- x^2/2$ for any $x\geq 0$, it holds for $\psi(x)$ with $x \in (0, 1)$ that
    \[
        \psi(x) = x - (1 + x)\ln(1 + x) \leq x - (1 + x)(x  - x^2/2) = -x^2/2 + x^3/3 \leq -x^2/6,
    \]
    hence
    \begin{multline*}
        \exp\left[\frac{\eta^2}{\theta^2}\psi\left(\frac{\eps \theta}{2\eta^2}\right)\right] \leq \exp\left[-\frac{\eta^2}{\theta^2}\cdot \frac{\eps^2 \theta^2}{24(\eta^2)^2}\right] = \exp\left[-\frac{\eps^2}{24\eta^2}\right]\\
        \implies \P{\beta_i^{s + 1} \geq \wgamma_i^{s + 1} - \eps \mid \Ee_i^s(\eps)} \geq 1 - \exp\left[-\frac{\eps^2}{24\eta^2}\right].
    \end{multline*}
    Since this bound holds for all $s \in [t - 1]$, we have
    \[
        \P{\Ee_i^t(\eps)} = \prod_{s = 1}^{t - 1}\P{\beta_i^{s + 1} \geq \wgamma_i^{s + 1} - \eps \mid \Ee_i^s(\eps)} \geq \left(1 - \exp\left[-\frac{\eps^2}{24\eta^2}\right]\right)^{t - 1} \geq 1 - t\cdot \exp\left[-\frac{\eps^2}{24\eta^2}\right].
    \]
    This finishes the proof.
\end{proof}

\cref{totalprob} allows us to show that for the choice $C = O(\ln^{-3}n)$ and $\eps = O(\sqrt{C\ln n})$, if agent $i_k$ is picked at iteration $k$ of the algorithm, then $\beta_{i_k}^k \geq \alpha + (\rho - \alpha)/2$.
This will complete the proof of \cref{diffvalue}.
As a result, with high probability all iterations of \cref{diffalgo} will succeed, i.e it will able to construct a distribution $\mu_{i_k}^k$ and allocate an acceptable bundle for $i_k$.

\begin{corollary}\label{finalprob}
    Let $\alpha \in [1/4, 3/11)$ and $\rho$ satisfy $2(1 - \rho + \alpha \ln \rho) = 1-\alpha$.
    For a fixed $n$, let $c = \frac{5(1 - \alpha)}{\rho - \alpha}\ln^{-3}n$, $\eps = 7\sqrt{\frac{4c(1 - \alpha)}{\rho - \alpha}\ln n}$ and $D = D(n)\xrightarrow{n\to\infty}0$ as functions of $n$.
    There exists $n_\alpha$ such that for all $n \geq n_\alpha$, running \cref{diffalgo} on $n$ agents, with probability at least $1 - 1/n$ for every agent $i \in [n]$ and every iteration $t \in [n]$, if $i \in \Aa^t$ then $\beta_i^t \geq \alpha + (\rho - \alpha)/2$.
\end{corollary}
\begin{proof}
    By \cref{diffgamma}, for these values of $\eps, D$ there exists $n_\alpha$ such that for all $n \geq n_\alpha$ it holds $\gamma_n^n(1 - \frac{1}{D}, \eps) - \eps \geq \alpha  + (\rho - \alpha)/2$.
    Fix some $n \geq n_\alpha$.
    For every $t \in [n]$ and agent $i \in \Aa^t$, by \cref{epsprocess} under the event $\Ee_i^t(\eps)$ it holds that $\beta_i^t \geq \gamma_n^n(\beta_i^1, \eps) - \eps$.
    As mentioned in the proof of \cref{totalprob}, \cref{dangerous} implies that $\gamma_n^n(\beta_i^1, \eps) \geq \gamma_n^n(1 - \frac{1}{D}, \eps)$, hence for this $n$ it holds under the event $\Ee_i^t(\eps)$ that $\beta_i^t \geq \alpha + (\rho-\alpha)/2$.
    So, it suffices to show that the probability that there exists any agent $i \in [n]$ and any iteration $t \in [n]$, such that $i \in \Aa^t$ but event $\Ee_i^t(\eps)$ does not occur, is small.
    By union bound,
    \[
        \P{\exists i, t \in [n]: i \in \Aa^t\text{ and }\overline{\Ee_i^t(\eps)}} \leq \sum_{i = 1}^n\P{\exists t \in [n]: i \in \Aa^t\text{ and }\overline{\Ee_i^t(\eps)}} \leq \sum_{i = 1}^n\sum_{t = 1}^n\P{\overline{\Ee_i^t(\eps)} \mid i \in \Aa^t}.
    \]
    By the choice of $\eps, c$ we have $\eps > 2C$.
    So, by \cref{totalprob}, for every $i \in \Aa^t$ it holds
    \[
        \P{\overline{\Ee_i^t(\eps)} \mid i \in \Aa^t} \leq t\cdot \exp\left[-\frac{\eps^2}{24\eta^2}\right],\qquad\text{where}\qquad \eta^2 = \left(\frac{5(1 - \alpha)}{(\rho - \alpha)n} + c\right)\frac{1 - \alpha}{\rho - \alpha}.
    \]
    By the choice of $c$, $\eta^2 \leq 2c\frac{1 - \alpha}{\rho - \alpha}$, hence by the choice of $\eps$ we get
    \[
        \P{\overline{\Ee_i^t(\eps)} \mid i \in \Aa^t} \leq t\cdot \exp\left[-\frac{\eps^2\cdot \frac{\rho - \alpha}{2}}{24c(1 - \alpha)}\right]\\
        = t\cdot \exp\left[-\frac{98c(1 - \alpha)\ln n}{24c(1 - \alpha)}\right] < \frac{t}{n^4} \leq \frac{1}{n^3}.
    \]
    We conclude that
    \[
        \P{\exists i, t \in [n]: i \in \Aa^t\text{ and }\overline{\Ee_i^t(\eps)}} \leq\sum_{i = 1}^n\sum_{t = 1}^n\P{\overline{\Ee_i^t(\eps)} \mid i \in \Aa^t} \leq \frac{1}{n}.
    \]
\end{proof}
It is easy to see that \cref{finalprob} directly implies \cref{diffvalue} for $\alpha \leq 3/11$.
We combine all the obtained bounds to prove the following theorem.
\begin{theorem}\label{finalallocation}
    For every $\alpha \in [1/4, 3/11)$, there exists $n_\alpha$, such that for all $n \geq n_\alpha$ the following holds.
    For every fair allocation instance of $m$ indivisible goods to $n$ agents with XOS valuations $v_1,\ldots, v_n$ and APS-partitions $\{(S_i, \lambda_i(S_i))\}_{S_i \in \Ss_i}$ satisfying \cref{wlog} and \cref{smallitems}, there exists an allocation algorithm that with probability $1 - 1/n$ produces an allocation where each agent receives at least $\alpha$-fraction of her APS.
\end{theorem}
\begin{proof}
    Let $\rho$ satisfy $2(1 - \rho + \alpha \ln \rho) = 1-\alpha$, and let $\alpha' := \alpha + (\rho - \alpha)/2$.
    For some value $n$, consider \cref{diffalgo} on $(n, m, \Mm, v_1, \ldots, v_n)$ with this value $\alpha'$, note that as $\alpha' > \alpha$, there are no items of value over $\alpha'$ in any of APS-partitions, thus \cref{partition} is applicable.
    We pick values $c = \frac{5(1 - \alpha')}{\rho - \alpha'}\ln^{-3}n$ and $D = (\frac{5(1 - \alpha')}{\rho - \alpha'})^2\sqrt{n}$.
    By \cref{finalprob}, for $\eps = 7\sqrt{\frac{4c(1 - \alpha')}{\rho - \alpha'}\ln n}$ there exists $n_{\alpha'}$ such that if $n \geq n_{\alpha'}$, it holds with probability at least $1 - 1/n$ that for every agent $i \in [n]$ and every iteration $t \in [n]$, if $i \in \Aa^t$ then $\beta_i^t \geq \alpha' + (\rho - \alpha')/2$.
    
    As a result, if $n \geq n_{\alpha'}$, with probability at least $1 - 1/n$ every iteration of \cref{diffalgo} on this instance will succeed.
    It follows that if the agent $i$ was picked at iteration $k \in [n]$ and initially assigned a bundle $B_{i} \sim \mu_{i}^k$, then $v_{i}(B_{i}) \geq \alpha'$.
    Hence, by \cref{stolenvalue}
    \[
        v_{i}(Q_{i}^n) \geq v_i(B_i) - \frac{D}{c(cn - 1)} \geq \alpha' - \frac{\ln^6 n}{(\sqrt{n} - \frac{2c}{5(1 - \alpha)}\cdot\frac{\ln^3 n}{\sqrt{n}})} \geq \alpha' -\frac{2\ln^6n}{\sqrt{n}}.
    \]
    Since $\alpha' = \alpha + (\rho - \alpha)/2$, there exists $n_\alpha \geq n_{\alpha'}$ such that for all $n \geq n_\alpha$,
    \[
        \alpha' -\frac{2\ln^6n}{\sqrt{n}} = \alpha + \frac{\rho - \alpha}{2} - \frac{2\ln^6n}{\sqrt{n}} \geq \alpha.
    \]
    But then for every agent $i \in [n]$ we have $v_i(Q_i^n) \geq \alpha$.
\end{proof}

\subsection{$\alpha \geq 3/11$}

The proofs for $\alpha \geq 3/11$ are essentially the same as for the case $\alpha \leq 3/11$, and differ only in minor technicalities.
When $\alpha \geq 3/11$, upper bounds on probabilities $p_{i_k}$ arising from \cref{distpartition} vary depending on whether the agent $i_k$'s current value, $\beta_{i_k}^k$, is above or below $3\alpha$.
Because of this, the definition of the deterministic process $\gamma_n^k$ will also be adaptive.

\begin{definition}
    For $n$ agents, $3/11 \leq \alpha < 1/3$, $\eps > 0$ and $1 \geq \beta > \alpha + \eps$, let $\gamma_n^{k}(\beta, \eps)$ be the following process: $\gamma_n^1(\beta, \eps) = \beta$ and for $k \geq 1$,
    \begin{itemize}
        \item if $\gamma_n^k(\beta, \eps) - \eps \geq 3\alpha$, then
        \[
            \gamma_n^{k + 1}(\beta, \eps) = \left(1 - \frac{2(1 - 3\alpha)}{(\gamma_n^k(\beta, \eps) - \eps - 12\alpha + 3)n}\right)\gamma_n^k(\beta, \eps);
        \]
        \item if $\gamma_n^k(\beta, \eps) - \eps < 3\alpha$, then
        \[
            \gamma_n^{k + 1}(\beta, \eps) = \left(1 - \frac{4\alpha}{3(\gamma_n^k(\beta, \eps) - \eps - \alpha)n}\right)\gamma_n^k(\beta, \eps).
        \]
    \end{itemize}
\end{definition}

With this definition of $\gamma_n^k$, we show that, as in \cref{epsprocess} for $\alpha \leq 3/11$, when $\alpha \geq 3/11$ we also have a lower bound on $\beta_j^k$ under the event $\Ee_j^t$.
\begin{claim}\label{epsprocess2}
    For $1 \leq t < n$, let $j \in \Aa^{t + 1}$ be some agent that \textbf{remains active} after iteration $t$ of \cref{diffalgo}.
    For $\eps >0$, let $\Ee_j^t(\eps)$ be the following event:
    \[
        \Ee_j^t(\eps) = \left\{\forall k \in [t]: \beta_j^k \geq \wgamma_j^k - \eps\right\}.
    \]
    Under the event $\Ee_j^t(\eps)$, it holds for all $k \in [t]$ that $\beta_j^k \geq \gamma_n^k(\beta_j^1, \eps) - \eps$.
\end{claim}
\begin{proof}    
    Once again, it suffices to show that for all $k \in [t]$, $\wgamma_j^k \geq \gamma_n^k(\beta_j^1, \eps)$.
    We prove this via induction on $k \in [t]$, the equality $\wgamma_j^1 = \beta_j^1 = \gamma_n^1(\beta_j^1, \eps)$ holds by definition.
    Suppose now that for all $1 \leq k'\leq k$ it holds $\wgamma_j^{k'} \geq \gamma_n^{k'}(\beta_j^1, \eps)$, we show that $\wgamma_j^{k+ 1} \geq \gamma_n^{k+ 1}(\beta_j^1, \eps)$.
    
    Observe that when $\alpha \geq 3/11$, for any $x \geq 3\alpha$ it holds
    \begin{equation}
        \frac{2(1 - 3\alpha)}{(x  - 12\alpha + 3)} \leq \frac{4\alpha}{3(x - \alpha)} \implies \left(1 - \frac{2(1 - 3\alpha)}{(x  - 12\alpha + 3)}\right) \geq \left(1 - \frac{4\alpha}{3(x - \alpha)}\right).\label{3alineq}
    \end{equation}
    Therefore, using the upper bound on $p_{i_k}$ from \cref{distpartition}, 
    \begin{gather*}
        \text{if }\;\beta_{i_k}^k\geq 3\alpha\; \text{ then } \; \wgamma_j^{k + 1} \geq \left(1 - \frac{2(1 - 3\alpha)}{(\beta_{i_k}^k - 12\alpha + 3)n}\right)\wgamma_j^k \geq  \left(1 - \frac{4\alpha}{3(\beta_{i_k}^k - \alpha)n}\right)\wgamma_j^k;\\
        \text{if }\;\beta_{i_k}^k< 3\alpha\; \text{ then } \; \wgamma_j^{k + 1} \geq  \left(1 - \frac{4\alpha}{3(\beta_{i_k}^k - \alpha)n}\right)\wgamma_j^k.
    \end{gather*}
    Similar to $\alpha \leq 3/11$, when $\alpha \geq 3/11$ under the event $\Ee_j^t(\eps)$ it also holds that $\beta_{i_k}^k \geq\beta_j^k \geq \wgamma_j^k - \eps$ for all $k \in [t]$.
    Now each of the values $\beta_{i_k}^k, \beta_j^k, \wgamma_j^k - \eps$ can be above or below $3\alpha$, and we will have to consider different cases.
    Note that if one value is above $3\alpha$, by \cref{3alineq} we lower bound for ``both cases'', as demonstrated above for $\beta_{i_k}^k$.
    Then, depending on what is the value of the next variable in $\beta_{i_k}^k \geq\beta_j^k \geq \wgamma_j^k - \eps$, we can pick the appropriate lower bound.
    
    Repeating the same chain of inequalities: since $\beta_{i_k}^k \geq \beta_j^k$,
    \begin{gather*}
        \text{if }\;\beta_{j}^k\geq 3\alpha\; \text{ then } \; \wgamma_j^{k + 1} \geq \left(1 - \frac{2(1 - 3\alpha)}{(\beta_{j}^k - 12\alpha + 3)n}\right)\wgamma_j^k \geq  \left(1 - \frac{4\alpha}{3(\beta_{j}^k - \alpha)n}\right)\wgamma_j^k;\\
        \text{if }\;\beta_{j}^k< 3\alpha\; \text{ then } \; \wgamma_j^{k + 1} \geq  \left(1 - \frac{4\alpha}{3(\beta_{j}^k - \alpha)n}\right)\wgamma_j^k.
    \end{gather*}
    Since $\beta_j^k \geq \wgamma_j^k - \eps$,
    \begin{gather*}
        \text{if }\;\wgamma_j^k - \eps\geq 3\alpha\; \text{ then } \; \wgamma_j^{k + 1} \geq \left(1 - \frac{2(1 - 3\alpha)}{(\wgamma_j^k - \eps - 12\alpha + 3)n}\right)\wgamma_j^k \geq  \left(1 - \frac{4\alpha}{3(\wgamma_j^k - \eps - \alpha)n}\right)\wgamma_j^k;\\
        \text{if }\;\wgamma_j^k - \eps< 3\alpha\; \text{ then } \; \wgamma_j^{k + 1} \geq  \left(1 - \frac{4\alpha}{3(\wgamma_j^k - \eps - \alpha)n}\right)\wgamma_j^k.
    \end{gather*}
    Finally, by induction $\wgamma_j^k \geq \gamma_n^k(\beta_j^1, \eps)$.
    Hence,    
    \begin{gather*}
        \text{if }\;\gamma_n^k(\beta_j^1, \eps) - \eps\geq 3\alpha\; \text{ then } \;  \wgamma_j^{k + 1} \geq \left(1 - \frac{2(1 - 3\alpha)}{(\gamma_n^k(\beta_j^1, \eps) - \eps - 12\alpha + 3)n}\right)\gamma_n^k(\beta_j^1, \eps) = \gamma_n^{k + 1}(\beta_j^1, \eps);\\
        \text{if }\;\gamma_n^k(\beta_j^1, \eps) - \eps < 3\alpha \; \text{ then } \;  \wgamma_j^{k + 1}  \geq \left(1 - \frac{4\alpha}{3(\gamma_n^k(\beta_j^1, \eps) - \eps - \alpha)n}\right)\gamma_n^k(\beta_j^1, \eps) = \gamma_n^{k + 1}(\beta_j^1, \eps).
    \end{gather*}
    It follows that under $\Ee_j^t(\eps)$, for all $k \in [t]$ it holds $\beta_j^k \geq \gamma_n^{k }(\beta_j^1, \eps) - \eps$.
\end{proof}

\begin{corollary}\label{pikbound2}
    For $k \in [n]$, let $i_k$ denote the agent that was picked at iteration $k$ of \cref{diffalgo}.
    For any $1 \leq t < n$, under the event $\Ee_j^t(\eps)$ it holds  for every $k \in [t]$ that $p_{i_k} \leq \leq \frac{4\alpha}{3(\gamma_n^k(\beta_j^1, \eps) - \eps - \alpha)n}$.
\end{corollary}
\begin{proof}
    Follows immediately by definition of $p_{i_k}$, \cref{distpartition} and (the proof of) \cref{epsprocess2} and the fact that when $\alpha \geq 3/11$, for any $x \geq 3\alpha$ it holds $\frac{2(1 - 3\alpha)}{(x  - 12\alpha + 3)} \leq \frac{4\alpha}{3(x - \alpha)}$.
    %
\end{proof}

As with $\alpha \leq 3/11$, when $\alpha \geq 3/11$ we also need a lower bound of the form $\gamma_n^t(\beta_i^1, \eps) - \eps \geq \alpha + c$ for every agent $i \in \Aa^t$.
The following theorem is the extension of \cref{diffgamma} to values $\alpha \geq 3/11$.
\begin{theorem}\label{diffgamma2}
    Let $\rho$ satisfy $2(12\rho - 3)\ln(3\rho) = (1 - 3\rho)(3\ln 3 - 4)$, and assume $\alpha \in [3/11, \rho)$.
    For any $\eps, \delta \xrightarrow{n\to\infty}0$, there exists $n_\alpha$ such that for all $n \geq n_\alpha$ it holds $\gamma_n^n(1 -\delta, \eps) - \eps \geq \alpha + (\rho - \alpha)/2$.
\end{theorem}

The proof is essentially the same as the proof of \cref{diffgamma}, just a bit more technical.
We refer the reader to \cref{sec:bigalpha}.
Now, we can show that for fixed $i, t$, the event $\Ee_i^t(\eps)$ occurs with high probability also for $\alpha > 3/11$.
The following statements are extensions of \cref{totalprob} and \cref{finalprob} to $\alpha > 3/11$, and are proven analogously.

\begin{corollary}\label{totalprob2}
    Let $\rho$ satisfy $2(12\rho - 3)\ln(3\rho) = (1 -3\rho)(3\ln 3 - 4)$, and assume $\alpha \in [3/11, \rho)$.
    Let $\eps, c, D > 0$ satisfy  $\eps, c\xrightarrow{n\to\infty}0$, $D\xrightarrow{n\to\infty}\infty$ and $\eps > 2c$.
    Let $n_\alpha$ be such that for all $n \geq n_\alpha$ it holds $\gamma_n^n(1 - \frac{1}{D}, \eps)- \eps \geq \alpha + (\rho - \alpha)/2$.
    For a fixed $n \geq n_\alpha$ and $1 \leq t \leq n$, let $i \in \Aa^t$ be the agent \textbf{picked} at iteration $t$ of \cref{diffalgo}.
    Then,
    \[
        \P{\Ee_i^t(\eps)} \geq 1 - t\cdot \exp\left[-\frac{\eps^2}{24\eta^2}\right].
    \]
    where
    \[
        \eta^2 = \left(\frac{40\alpha}{3(\rho - \alpha)n} + c\right) \frac{8\alpha}{3(\rho -\alpha)}.
    \]
\end{corollary}

\begin{corollary}\label{finalprob2}
    Let $\rho$ satisfy $2(12\rho - 3)\ln(3\rho) = (1 -3\rho)(3\ln 3 - 4)$, and assume $\alpha \in [3/11, \rho)$.
    For a fixed $n$, let $c = \frac{40\alpha}{3(\rho -\alpha)}\ln^{-3}n$, $\eps = 17\sqrt{\frac{2c\alpha}{\rho - \alpha}\ln n}$ and $D = D(n) \xrightarrow{n\to\infty}0$ as functions of $n$.
    There exists $n_\alpha$ such that for all $n \geq n_\alpha$, running \cref{diffalgo} on $n$ agents, with probability at least $1 - 1/n$ for every agent $i \in [n]$ and every iteration $t \in [n]$, if $i \in \Aa^t$ then $\beta_i^t \geq \alpha + (\rho - \alpha)/2$.
\end{corollary}
Finally, we combine \cref{totalprob2} and \cref{finalprob2} to prove the analogue of \cref{finalallocation} for the case $\alpha > 3/11$ (picking $c = \frac{40\alpha'}{3(\rho -\alpha')}\ln^{-3}n$, $D = (\frac{40\alpha'}{3(\rho - \alpha')})^2\sqrt{n}$ and $\eps = 17\sqrt{\frac{2c\alpha'}{\rho - \alpha'}\ln n}$ for $\alpha' := \alpha + (\rho -\alpha)/2$).
\begin{theorem}\label{finalallocation2}
    Let $\rho$ satisfy $2(12\rho - 3)\ln(3\rho) = (1 -3\rho)(3\ln3 - 4)$.
    For every $\alpha \in [3/11, \rho)$, there exists $n_\alpha$, such that for all $n \geq n_\alpha$ the following holds.
    For every fair allocation instance of $m$ indivisible goods to $n$ agents with XOS valuations $v_1,\ldots, v_n$ and APS-partitions $\{(S_i, \lambda_i(S_i))\}_{S_i \in \Ss_i}$ satisfying \cref{wlog} and \cref{smallitems}, there exists an allocation algorithm that with probability $1 - 1/n$ produces an allocation where each agent receives at least $\alpha$-fraction of her APS.
\end{theorem}

\section{Big Items}

\label{sec:bigitems}

For every agent $i \in [n]$ and item $e \in \Mm$, we say that item $e$ is \textbf{big} for $i$ if $v_i(e) \geq \alpha$, and \textbf{small} for $i$ if $v_i(e) < \alpha$.
We have showed in previous sections that if all items are small for all agents, i.e \cref{smallitems} is satisfied, then there exists an allocation where each agent gets at least $\alpha$-APS for $\alpha > 11/40$.
In this section, we show how to get rid of the requirement of \cref{smallitems}.

\subsection{Properties of bipartite graph}

Below we show that any bipartite graph $G = (U, W, E)$ with $|U| = n$ must satisfy one of the four structural conditions.
Then, we consider the graph $G = G(U, W, E)$ with $U = [n]$, $W = \Mm$ and $E = \{(i, e) : v_i(e) \geq \alpha\}$, and prove that if $G$ satisfies one of these four conditions, there exists an $\alpha$-APS allocation for $\alpha > 11/40$.

\begin{lemma}\label{bipartite}
    Let $0 < \eps \leq 1/240$ and $n$ be arbitrary.
    For every bipartite graph $G = G(U, W, E)$ with $|U| = n$ and $|W| \geq n$, one of the following must hold.
    \begin{enumerate}
        \item $G$ has a matching in which all of $U$ is matched.

        \item $G$ has a maximal matching with at most $(1 - \eps)n$ edges.

        \item There are partitions $(U_1, U_2)$ of $U$ and $(W_1, W_2)$ of $W$ with the following properties:
        \begin{enumerate}
            \item $|U_1| = |W_1|$, and $(U_1, W_1)$ form a perfect matching which is maximal in $G$;
            
            \item for every vertex $u \in U_2$ we have $|N(u)| \leq n - 3|U_2|$.
        \end{enumerate}

        \item There are partitions $(U_1, U_2, U_3)$ of $U$ and $(W_1, W_2, W_3)$ of $W$ with the following properties:
        \begin{enumerate}
            \item $|U_1|= |W_1|$, and $(U_1, W_1)$ form a perfect matching;

            \item there are no edges between $U_2$ and $W_3$ and between $U_3$ and $W_2 \cup W_3$;

            \item $(1 - 60\eps)n \leq |W_2| < |U_2|$;
            
            \item for every subset $U_2'\subseteq U_2$ of size $|W_2|$ there is a perfect matching between $U_2'$ and $W_2$.
        \end{enumerate}
    \end{enumerate}
\end{lemma}
\begin{proof}
    We will show that if items 1), 2), 3) do not hold, then 4) must hold.
    
    We first prove that if 1), 2), 3) do not hold, there exists a set $U_0 \subseteq U$ of size at least $(1 - 3\eps)n$, such that every vertex $u \in U_0$  has at least $(1 - 3\eps)n$ neighbors in $W$.
    To see this, observe that 
    \begin{itemize}
        \item by 1), every maximal matching $M$ in $G$ has size $|M| < n$, so $U\setminus M$ is nonempty;
        \item by 2), $|M| \geq (1 - \eps)n$, thus $|U\setminus M| \leq \eps n$;
        \item hence, by 3), for every maximal matching $M$ in $G$ there exists a vertex $u\in U\setminus M$ such that $|N(u)| > n - 3|U \setminus M| \geq (1 - 3\eps)n$.
    \end{itemize}
    We will pick vertices into $U_0$ iteratively, starting with $U_0 = \varnothing$.
    Let $M_1$ be arbitrary maximal matching in $G$, then there exists some vertex $u_1 \in U\setminus M_1$ such that $|N(u)| \geq (1 -3\eps)n$.
    We add $u_1$ to $U_0$, and, if possible, extend $U_0 = \{u_1\}$ to a maximal matching $M_2$ in $G$, containing $U_0$.
    Then, there exists some vertex $u_2 \in U\setminus M_2$ such that $|N(u)|\geq (1 - 3\eps)n$.
    We add $u_2$ to $U_0$ and then, if possible, extend $U_0 = \{u_1, u_2\}$ to a maximal matching $M_3$ containing $U_0$, then take vertex $u_3$ and so on.
    We continue doing so until we no longer can find a maximal matching containing all $U_0$.
    
    Since at any point for every vertex $u \in U_0$ it holds $|N(u)|\geq (1 -3\eps)n$, it is always possible to perfectly match all current vertices of $U_0$ as long as $|U_0| \leq (1 - 3\eps)n$, and therefore possible to obtain a maximal matching containing all $U_0$.
    Thus, if at some point it is no longer possible to extend $U_0$ to a maximal matching, we must have $|U_0| > (1 - 3\eps)n$.
    We then add to this $U_0$ all other remaining vertices with at least $(1 - 3\eps)n$ neighbors (if there are any).

    Consider this $U_0$, we have $|U_0|\geq (1 - 3\eps)n$, and for all $u \in U_0$ it holds $|N(u)| \geq (1 - 3\eps)n$.
    Denote $W_0 := N(U_0)$, $\oU_0 := U\setminus U_0$ and $\oW_0 := W\setminus W_0$, we have $|\oU_0|\leq 3\eps n$.
    We start forming sets $U_1, U_2, U_3$ and $W_1, W_2, W_3$, initially all of them are empty.
    
    Go over all vertices $u \in \oU_0$ in arbitrary order, if $N(u) \setminus W_1$ is nonempty, take $w \in N(u)\setminus W_1$ and match $(u, w)$ together, adding $u$ to $U_1$ and $w$ to $W_1$.
    It is clear that in the end of this procedure, if some $u \in \oU_0$ was unmatched, then $N(u)\setminus W_1 = \varnothing$, i.e there are no edges between $u$ and $W\setminus W_1$.
    We then add all unmatched vertices of $\oU_0$ to $U_3$.
    Similarly, for every unmatched vertex $w \in \oW_0$ there are no edges between $w$ and $U\setminus U_1$ --- there can't be edges with $U_0$ by definition, and all neighbors of $w$ in $\oU_0$ have already been matched.
    We then add all unmatched vertices of $\oW_0$ to $W_3$.
    In addition, as $|U_0| \geq (1 - 3\eps)n$, we could not have matched more than $3\eps n$ vertices of $\oU_0$, therefore we could not have matched more than $3\eps n$ vertices of $W_0$.
    As a result, for every vertex $u \in U_0$ there will be at least $(1 - 6\eps)n$ neighbors left in $W_0$.

    Consider the subgraph of $G$ on $(U_0, W_0)$.
    We claim that there is no perfect matching between $U_0$ and $W_0$.
    Suppose the contrary, then $U_0 \cup U_1$ and $W_0\cup W_1$ together form a perfect matching which is maximal in $G$.
    Indeed, by construction the only unmatched vertices would be $U_3$ and $W_3$, and there are no edges between them.
    Now, take some vertex $u \in U_3$.
    By construction, $u \in \oU_0$, hence $|N(u)| < (1 - 3\eps)n$.
    At the same time, since $(U_0\cup U_1, W_0 \cup W_1)$ is maximal, case 2) not holding implies that $|U_0\cup U_1| \geq (1 - \eps)n$, so $|U_3| \leq \eps n$.
    Thus, all $u \in U_3$ have $|N(u)|< (1 - 3\eps)n \leq n - 3|U_3|$.
    But then we find ourselves in case 3), which does not hold by assumption.

    {As a result,} there is no perfect matching between $U_0$ and $W_0$.
    {Then,} by Hall's theorem there exists a subset $U_0'\subseteq U_0$ such that $W_0':= N(U_0')$ satisfies $|U_0'| > |W_0'|$.
    Since every vertex in $U_0'$ has at least $(1 - 6\eps)n$ neighbors in $W_0$, we have $|W_0'| \geq (1 - 6\eps)n$ and $|U_0'| \geq (1 -6\eps)n$.

    Let $\overline{U_0'}:= U_0 \setminus U_0'$, and $\overline{W_0'}:= W_0 \setminus W_0'$.
    We go over all vertices $u \in \overline{U_0'}$ in arbitrary order, and if $N(u)\setminus W_1$ is nonempty, take $w \in N(u)\setminus W_1$ and match $(u, w)$ together, adding $u$ to $U_1$ and $w$ to $W_1$.
    Similar to what we did before, in the end of this procedure all unmatched vertices of $\overline{U_0'}$ have no edges into $W\setminus W_1$, and we add them to $U_3$.
    At the same time, all unmatched vertices of $\overline{W_0'}$ have no edges into $U\setminus U_1$, and we add them to $W_3$.
    Since $|U_0'| \geq (1 - 6\eps)n$, the total number of matched vertices of both $\overline{U_0}$ and $\overline{U_0'}$ cannot exceed $6\eps n$.
    Note that for every vertex $u \in U_0'$, $N(u)\cap W_0'$ contains exactly the vertices $w \in W_0$ that have not been matched. 
    Every vertex $u \in U_0'$ initially had at least $(1 - 3\eps)n$ neighbors in $W_0$, and the total number of matched vertices of $W_0$ cannot exceed $6\eps n$.
    Then, for every vertex $u \in U_0'$ there will be at least $(1 - 9\eps)n$ neighbors left in $W_0'$, hence $|U_0'| > |W_0'| \geq (1 - 9\eps)n$.

    Next, let $W_0'' \subseteq W_0'$ contain all vertices $w \in W_0'$ such that $|N(w) \cap U_0'| \geq (3/4) n$.
    We lemma that $|W_0''| \geq (1 - 48\eps)n$.
    To estimate $|W_0''|$, consider the subgraph induced on $(U_0', W_0')$.
    Every vertex in $W_0''$ has degree at most $n$, and every vertex \textbf{not} in $W_0''$ has degree at most $(3/4) n$.
    Since $|U_0'| \geq (1 - 6\eps)n$ and for every $u \in U_0'$, $|N(u)|\geq (1 - 9\eps)n$, the total number of edges in this subgraph is at least $(1 - 9\eps)n|U_0'| \geq (1 - 6\eps)(1 - 9\eps)n^2 \geq (1 - 15\eps)n^2$, and $|W_0'| < n$, we have the following lower bound:
    \begin{multline*}
        (1 - 15\eps)n^2 \leq |W_0''|\cdot n + (n - |W_0''|) \cdot  (3/4) n\\
        \implies (1 - 15\eps)n \leq |W_0''| + (3/4)n - (3/4)|W_0''|\\
        \implies |W_0''| \geq 4\cdot (1/4 -15\eps)n =  (1 - 60\eps)n.
    \end{multline*}
    Let $\overline{W_0''} = W_0'\setminus W_0''$.
    Now, we go over all vertices $w \in \overline{W_0''}$ in arbitrary order, and if $N(w)\cap U_1$ is nonempty, take $u \in N(w)\setminus U_1$, and match $(u, w)$ together, adding $w$ to $W_1$ and $u$ to $U_1$.
    In the end of the procedure, all unmatched vertices of $\overline{W_0''}$ have no edges into $U\setminus U_1$, and we add them to $W_3$.
    Let $U_0''\subseteq U_0'$ denote the remaining vertices of $U_0'$, and $\overline{U_0''} := U_0'\setminus U_0''$.
    Clearly, $|U_0''| = |U_0'| - |\overline{U_0''}|$.
    Since $|W_0''| \geq (1 - 60\eps)n$, then $|\overline{W_0''}| \leq 60\eps n$.
    Therefore, the number $|\overline{U_0''}|$ of matched vertices of $U_0'$ satisfies $|\overline{U_0''}| \leq 60\eps n$.
    Furthermore, every vertex $u \in U_0''$ has lost at most $|\overline{W_0''}| \leq 60\eps n$ neighbors, so $|N(u)|\geq (1 - 69\eps)n$.
    Finally, consider the remaining vertices $W_0''$.
    Observe that 
    \[
        |W_0''| = |W_0'| - |\overline{W_0''}| < |U_0'| - |\overline{W_0''}| = |U_0''| + |\overline{U_0''}| - |\overline{W_0''}|.
    \]
    Note that the vertices $\overline{W_0''}$ include all the vertices of $W_0'$ that have been matched with $\overline{U_0''}$, plus some additional vertices that remained unmatched.
    Then, $|\overline{U_0''}| \leq |\overline{W_0''}|$ and $|W_0''| < |U_0''|$.
    And, since $|\overline{U_0''}| \leq 60\eps n$, each vertex  $w \in W_0''$ has at least 
    \[
        |N(w) \cap U_0'| - 60\eps n \geq (3/4)n - 60\eps n = n/2 + (1/4 - 60\eps)n \geq n/2
    \]
    neighbors left in $U_0''$.
    The last inequality uses $\eps \leq 1/240$.

    We take $U_2 := U_0''$ and $W_2:= W_0''$.
    To finish the proof, it remains to show that if $U_2'\subseteq U_2$ has size $|W_2|$, there exists a perfect matching between $U_2'$ and $|W_2|$.
    By construction for every $u \in U_2$ it holds $|N(u)|\geq (1 - 69\eps)n \geq n/2 > |W_2|/2$.
    At the same time, we showed above that every vertex $w \in W_2$ has at least $n/2 > |W_2|/2$ neighbors left in $U_2$.
    Thus, in the subgraph induced on $(U_2', W_2)$, every vertex has degree at least $|W_2|/2$.
    The lemma follows from the fact that any $k \times k$ bipartite graph where every vertex has degree at least $k/2$ has a perfect matching.
\end{proof}
We apply \cref{bipartite} to the following graph $G = G(U, W, E)$: $U = [n]$, $W = \Mm$, and we draw an edge $(i, e)\in E$ between agent $i \in [n]$ and item $e \in \Mm$ if $v_i(e) \geq \alpha$, i.e item $e$ is big for agent $i$.
In the following subsections, we prove that in each case of \cref{bipartite}, there must exist an $\alpha$-APS allocation for $\alpha > 11/40$, as long as the number of agents $n$ is large enough.

\subsection{Allocations for cases 1), 2), 3)}

We first deal with cases 1) and 2), they are pretty straightforward.

\begin{claim}
    Let $\rho$ satisfy $2(12\rho - 3)\ln(3\rho) = (1 - 3\rho)(3\ln 3 - 4)$.
    For every $\alpha \in [3/11, \rho)$, there exists $n_\alpha$, such that for all $n \geq n_\alpha$ the following holds.
    
    Consider the graph $G = G(U, W, E)$ with $U = [n]$, $W = \Mm$ and $E = \{(i, e) : v_i(e) \geq \alpha\}$.
    If $G$ satisfies cases 1) or 2) of \cref{bipartite}, then there exists an allocation where each agent receives at least $\alpha$-fraction of her APS.
\end{claim}
\begin{proof}
Suppose that \textbf{case 1)} of \cref{bipartite} holds for $G$.
Then, there exists a matching in $G$ such that every agent $i \in [n]$ is matched with a unique item $e \in \Mm$ satisfying $v_i(e)\geq \alpha$.
So, we trivially get an $\alpha$-APS allocation.

Next, suppose that \textbf{case 2)} of \cref{bipartite} holds for $G$ and $\eps_2$.
Let $M$ be the maximal matching in $G$ with at most $(1 - \eps_2)n$ edges.
Recall \cref{giveitem}, which we restate it here for convenience:
\begin{lemma}\label{giveitem2}
    For any item $e \in \Mm$ and any valuation $v$, $\aps(\Mm\setminus\{e\}, v, \frac{1}{n - 1}) \geq \aps(\Mm, v, \frac{1}{n})$.
\end{lemma}
For every edge $(i, e)$ of the matching $M$, we give item $e$ to $i$, then $i$ receives $\alpha$-fraction of her APS.
By \cref{giveitem2}, for all the remaining agents in $U\setminus M$ the APS will not decrease after we remove all items and agents from $M$.
Since the matching $M$ was maximal, for every agent $i \in U\setminus M$ there are no big items remaining in $\Mm\setminus M$.
Furthermore, the number $|U\setminus M|$ of remaining agents is at least $|U| - |M| \geq \eps_2 n$.

As shown in \cref{finalallocation2}, there exists $n_\alpha$ such that for all $n'\geq n_\alpha$, every fair allocation instance with $n'$ XOS agents and APS partitions satisfying \cref{wlog} and \cref{smallitems} has an $\alpha$-APS allocation.
Then, if $\eps_2 n \geq n_\alpha \iff n\geq n_\alpha/\eps_2$, we can apply \cref{finalallocation2} to the agents $U\setminus M$, items $\Mm\setminus M$ and the corresponding APS partitions, and therefore get an $\alpha$-APS allocation for all agents of $U$.
\end{proof}

Case 3) is quite more convoluted, we analyze it in the following lemma.
\begin{lemma}
    Let $0 < \alpha \leq 1/3$, and consider the graph $G = G(U, W, E)$ with $U = [n]$, $W = \Mm$ and $E = \{(i, e) : v_i(e) \geq \alpha\}$.
    If $G$ satisfies case 3) of \cref{bipartite}, then there exists an allocation where each agent receives at least $\alpha$-fraction of her APS.
\end{lemma}
\begin{proof}
Suppose that \textbf{case 3)} of \cref{bipartite} holds for $G$.
Let $U = (U_1, U_2)$ and $W = (W_1, W_2)$ be the corresponding partitions, where $(U_1, W_1)$ form a perfect matching.
Similarly to case 2), for every edge $(i, e) \in M$, we give item $e$ to $i$, then $i$ receives $\alpha$-fraction of her APS.
By \cref{giveitem2}, for all the remaining agents in $U_2$, the APS will not decrease after we removed all items and agents from $M$.
So, it remains to allocate the remaining items $W_2$ among agents $U_2$.

For every agent $i \in U_2$, consider her APS partition $\{(S_i, \lambda_i(S_i)\}_{S_i \in \Ss_i}$, and let $\Ss_i^+$ be the set of all bundles $S_i \in \Ss_i$ such that $S_i$ contains some big item $e$ (i.e $v_i(e)\geq \alpha$) that \textbf{was allocated} to someone in $U_1$.
In other words, if $e \in S_i$ and $S_i \in \Ss_i^+$, then $e \in N(i) \cap W_1$.
Denote $k := |U_2|$.
By construction, for every agent $i \in U_2$ we have $|N(u)| \leq n - 3k$, so the total weight $\lambda_i(\Ss_i^+) = \sum_{S_i \in \Ss_i^+}\lambda_i(S_i)$ is at most $1 - 3k/n$.
Indeed, for every item $e \in \Mm$ we have $\sum_{S_i \in \Ss_i : S_i\ni e}\lambda_i(S_i) \leq 1/n$, and there are at most $n - 3k$ big items, so in the ``worst'' case each of the big item would have its own disjoint set of bundles containing it.

Note that since $|W_1| = n -k$, there are at least $2k$ allocated items that were small for agent $i$.
For convenience, we ``pretend'' that some of these small allocated items were actually ``big''.
Specifically, we go over small allocated items one by one, and for item $e$ we include the bundles of $\Ss_i$ containing $e$ into $\Ss_i^+$.
We do so until there are exactly $n - 3k$ ``big'' allocated items.
Denote these ``big'' items by $W_1^+$, and the remaining small items by $W_1^-$.
Note that $|W_1^-| = 2k$.

Consider the bundles $\Ss_i^- := \Ss_i \setminus \Ss_i^+$, by definition any $S_i \in \Ss_i^-$ contains no items from $W_1^+$, and either lost no items from the allocation $(U_1, W_1)$, or could have potentially lost some small items from $W_1^-$.
For every item $e \in W_1^-$, denote by $\Ss_i^-(e)$ the set of bundles $S_i \in \Ss_i^-$ that contain $e$.
In addition, denote by $\Ss_i^-(0), \Ss_i^-(1), \Ss_i^-(2)$ the subsets of $\Ss_i^-$ containing bundles that lost respectively exactly $0$, exactly $1$ and at least $2$ small items from $W_1^-$.
We say that two sets $\Ss_i^-(e)$ and $\Ss_i^-(e')$ are disjoint, if there are no $S_i \in \Ss_i^-(e)$ such that $\{e, e'\} \in S_i$ or $S_i'\in \Ss_i^-(e')$ such that $\{e, e'\} \in S_i'$.
Note that if all $2k$ sets $\Ss_i^-(e)$ for $e \in W_1^-$ are disjoint, then $\Ss_i^-(2)=\varnothing$ and $\lambda_i(\Ss_i^-(0)) + \lambda_i(\Ss_i^-(1)) = \lambda_i(\Ss_i^-)\geq 3k/n$.

Let $t \geq 0$ be such that there are exactly $2k - t$ disjoint sets $\Ss_i^-(e)$ for $e \in W_1^-$.
Without loss of generality, let $W_1^- = \{e_1, \ldots, e_{2k}\}$, and sets $\Ss_i^-(e_j)$ for $1 \leq j \leq 2k - t$ are disjoint.
Note that by definition, for every $2k - t + 1 \leq j'\leq 2k$, non-disjoint set $\Ss_i^-(e_{j'})$ can contain only bundles $S_i \in \Ss_i^-(2)$.
Indeed, if $S_i \in \Ss_i^-(e_{j'})$, there exists $e_j$ for $1 \leq j \leq 2k - t$ such that $S_i \in \Ss_i^{-}(e_j)$ and hence $\{e_j, e_{j'}\} \in S_i$.
It will be important that this set $\Ss_i^{-}(e_j)$ must contain at least one bundle from $\Ss_i^-(2)$.
We will use this observation to show that $\lambda_i(S_i^-(2)) \leq k/n$.

Suppose that there are $k_2 > k$ disjoint sets $\Ss_i^-(e)$ such that there exists $S_i \in \Ss_i^-(e)$ which lost at least $2$ items, i.e $S_i \in \Ss_i^-(2)$.
The fact that these sets are disjoint implies that there must exist at least $k_2$ disjoint bundles $S_{i, 1}, \ldots, S_{i, k_2} \in \Ss_i^-(2)$.
But then, since each of these bundles contains at least $2$ items from $W_1^-$, we have $|W_1^-| \geq 2k_2 > 2k$, which is impossible.

As a result, there must be $k_2 \leq k$ disjoint sets $\Ss_i^-(e)$ such that contain some $S_i \in \Ss_i^-(2)$.
Without loss of generality, let these sets be $\Ss_i^-(e_1), \ldots, \Ss_i^-(e_{k_2})$.
We lemma that any bundle $S_i \in \Ss_i^-(2)$ belongs to exactly one of disjoint sets $\Ss_i^-(e_{j})$ for $1 \leq j \leq k_2$.
Indeed, as mentioned earlier, for every $S_i \in \Ss_i^-(e_{j'})$ where $2k - t + 1 \leq j'\leq 2k$, i.e $\Ss_i^-(e_{j'})$ is non-disjoint,  we must have $S_i \in \Ss_i^-(2)$, and there must exist some $1 \leq j \leq k_2$ such that $S_i \in \Ss_i^-(e_j)$.
But then
\[
    \lambda_i(\Ss_i^-(2)) \leq \sum_{j = 1}^{k_2}\lambda_i(\Ss_i^-(e_j)) \\
    \leq \sum_{j = 1}^{k_2}\sum_{\substack{S_i \in \Ss_i\\S_i \ni e_j}}\lambda_i(S_i) \leq \frac{k_2}{n} \leq \frac{k}{n}.
\]

As a corollary, for every agent $i \in U_2$ that was not allocated, the total weight of the bundles in $\Ss_i^-(0) \cup \Ss_i^-(1)$ is at least $2k/n$.
By definition for every $S_i \in \Ss_i^-(0) \cup \Ss_i^-(1)$ we have $v_i(S_i \cap W_2) \geq 1- \alpha$, as each bundle could lose at most one small item.
In order to allocate the remaining $|U_2| = k$ agents, we will use the approach developed in \cite{GhodsiHSSY22, FG25}.

Without loss of generality, let $U_2 := [k]$.
For every agent $i \in [k]$, let $\hv_i$ be a capped valuation function defined as $\hv_i(S) := \min(1 - \alpha, v_i(S))$, note that when $v_i$ is XOS, so is valuation $\hv_i$.
Let $B_1, \ldots, B_k$ be the allocation of items $W_2$ among agents $U_2$, such that agent $i \in [k]$ receives bundle $B_i$, and the welfare $\sum_{i = 1}^k\hv_i(B_i)$ w.r.t $\hv_i$-s is maximal.
We lemma that for every agent $i \in [k]$, $v_i(B_i) \geq \hv_i(B_i) \geq (1-\alpha)/2$.

To see this, for a fixed agent $i \in U_2$ consider its APS-partition $\{(S_i, \lambda_i(S_i)\}_{S_i \in \Ss_i}$ for the original sets $(U, W)$.
We construct the following distribution $\mu_i$ over subsets of $W_2$: for every bundle $S_i \in \Ss_i^-(0)\cup \Ss_i^-(1)$, let $\mu_i(S_i) := \lambda_i(S_i)/\lambda_i(\Ss_i^-(0)\cup \Ss_i^-(1))$.
Suppose that for $i$ we have $\hv_i(B_i) < (1 - \alpha)/2$.
Then, sample a random bundle $B_i^*\sim \mu_i$, and give $i$ the bundle $B_i^* \cap W_2$ instead, and to every other agent $j \in U_2\setminus\{i\}$ give bundle $B_j^* := B_j\setminus B_i^*$.
Since $\mu_i(S_i)$ is non-zero only if $S_i \in \Ss_i^-(0)\cup \Ss_i^-(1)$, we must have $\hv_i(B_i^*) = \min(1 - \alpha, v_i(B_i^*\cap W_2)) \geq 1 - \alpha$, so agent $i$ gained at least $(1 - \alpha)/2$ welfare.
At the same time, for every item $e \in W_2$ the probability $\mu_i(e)$ that $S_i \sim \mu_i$ contains $e$ satisfies
\[
    \mu_i(e) = \sum_{\substack{S_i \in \Ss_i\\S_i\ni e}}\mu_i(S_i) = \frac{1}{\lambda_i(\Ss_i^-(0)\cup \Ss_i^-(1))}\sum_{\substack{S_i \in \Ss_i\\S_i\ni e}}\lambda_i(S_i) \leq \frac{1}{\lambda_i(\Ss_i^-(0)\cup \Ss_i^-(1))}\cdot \frac{1}{n} \leq \frac{1}{2k},
\]
because $\lambda_i(\Ss_i^-(0)\cup \Ss_i^-(1)) \geq 2k/n$.
Now, take some agent $j\in U_2\setminus\{i\}$.
Since her valuation $\hv_j$ acts like an additive function on $B_j$, in expectation agent $j$ lost at most
\[
    \sum_{e \in B_j}\mu_i(e)\cdot \hv_j(e) \leq \frac{1}{2k}\sum_{e \in B_j}\hv_j(e) = \frac{\hv_j(B_j)}{2k} \leq \frac{1- \alpha}{2k}.
\]
But then in expectation the total value lost by all other $|U_2\setminus \{i\}| = k - 1$ agents is at most 
\[
    \frac{1 - \alpha}{2k}\cdot (k - 1)< \frac{1 - \alpha}{2},
\]
implying that there exists a choice of $B_i^*$ such that the allocation $B_1^*, \ldots, B_k^*$ will achieve strictly larger welfare than $B_1, \ldots, B_k$, leading to a contradiction.
Hence, for a maximal welfare $B_1, \ldots, B_k$, for every $i \in [k]$ we must have $v_i(B_i) \geq \hv_i(B_i) \geq (1 - \alpha)/2$.
Finally, $\alpha \leq 1/3$, so $(1 - \alpha)/2 \geq \alpha$, which finishes the proof.
\end{proof}

\subsection{Allocation for case 4): agents $U_1$ and $U_3$}

\label{sec:bigcase4}

\textbf{Case 4} is the hardest one to deal with, it corresponds to the instance where the vast majority of agents agree which items are big and which items are small, while a small fraction of agents do not have any big items at all.
In order to get an $\alpha$-APS allocation for this case, we design a new algorithm, inspired by both \cref{algo} and \cref{diffalgo}.

Suppose that case 4) of \cref{bipartite} holds for $G$ and some $\eps_4'\leq 1/240$, denote $\eps_4 := 60\eps_4'\leq 1/4$.
Let $U = (U_1, U_2, U_3)$ and $W = (W_1, W_2, W_3)$  be the corresponding partitions.
Let $M_1$ be the perfect matching between $U_1$ and $W_1$.
Similarly to cases 2) and 3), for every edge $(i, e) \in M_1$, we give item $e$ to $i$, then $i$ receives $\alpha$-fraction of her APS.
By \cref{giveitem2}, for all the remaining agents in $U_2\cup U_3$ the APS will not decrease after we remove all items and agents from $M_1$.
So, it remains to allocate the remaining items $W_2\cup W_3$ among agents $U_2\cup U_3$.

By construction, items $W_3$ are small for both agents in $U_2$ and in $U_3$, while items $W_2\cup W_3$ are small only for agents in $U_3$.
Since $|U_2| > |W_2| \geq (1 - \eps_4)n$, the set $U_2$ contains at least $(1 - \eps_4)n$ agents, and there are at most $\eps_4 n$ agents in $U_3$.
Since $U_2$ is large, we cannot simply give big items to as many agents in $U_2$ as possible and then try to allocate the leftovers of $U_2$ and $U_3$ using only small items of $W_3$, using for example the algorithm of \cref{finalallocation2}.
First, \cref{finalallocation2} holds only when applied to at least $n_\alpha$ agents, while we will have less than $2\eps_4 n$ agents left.
Second, if $U_2$ is allocated first, there is no guarantee that for the leftovers of $U_2$ and $U_3$ the remaining items will provide enough APS-value, if anything, the vast majority of the APS-value will be lost.

Instead, we will allocate agents of $U_3$ first.
Since for $U_3$ all items of $W_2 \cup W_3$ are small, one may try and apply \cref{finalallocation2} and \cref{diffalgo} to agents $U_3$.
However, $|U_3|\leq \eps_4 n$, so if we assume that $U_3$ are the only agents participating in \cref{diffalgo} and consider their APS-partitions of items $W_2\cup W_3$ among only agents $U_3$, the algorithm will not succeed.
The issue is that when there are substantially less than $n$ total agents, for every step $k$ of \cref{diffalgo} the probabilities $p_{i_k}$ for distributions $\mu_{i_k}^k$ will be too large, thus agents $j \in \Aa^{k + 1}$ that remain active after step $k$ will lose a lot of their APS-value $\beta_j^k$, forcing it to drop below $\alpha$ very fast.

Another shortcoming of the approach suggested above is that if we allow agents $i\in U_3$ in \cref{diffalgo} to pick items from $W_2$, agents $U_3$ can take almost all items of $W_2$ just for themselves, and then we will not be able to match $U_2$ with the remainder of $W_2$.
This is due to the fact that all items of $W_2$ are small for all agents of $U_3$, and \cref{diffalgo} gives no upper bound on the number of items each allocated agent may get.

To mitigate these issues, we modify \cref{diffalgo} and run it on all $n_4 := |U_2 \cup U_3|$ agents together, and consider their APS-partitions of items $W_2 \cup W_3$ among all agents $U_2\cup U_3$.
However, we will force the algorithm to pick and allocate \textbf{only agents from $U_3$} at every step, and use \textbf{only items from $W_3$}, and run it until all agents of $U_3$ are allocated.
One may ask whether there is enough value left for agents $U_3$ if we restrict them only to items from $W_3$, and also whether after we allocate $U_3$, there will be enough value left for agents $U_2$ from the remaining items of $W_2\cup W_3$.
To answer these questions, we first observe that agents $U_2$, the vast majority of APS-value must come from the items of $W_3$, while agents $U_3$ possesses an APS partition of items $W_3$ among $|U_2|- |W_2|$ agents of value at least $1$.
\begin{claim}\label{w3value}
    \begin{enumerate}
        \item For agent $i \in U_3$, let $\{(S_i, \lambda_i(S_i)\}_{S_i \in \Ss_i}$ be her APS partition.
        Then 
        \[
            \sum_{S_i \in \Ss_i}\lambda_i(S_i)\cdot v_i(S_i \cap W_3) \geq 1 - \alpha.
        \]
        \item For every agent $j \in U_2$, there exists an APS-partition $\{(\hS_j, \hlam_j(\hS_j)\}_{\hS_j \in \Hss_j}$ of \textbf{only} items $W_3$ among $|U_2| - |W_2|$ agents such that its value $\hbeta_j := \hbeta_j(\Hss_j)$ satisfies
        \[
            \hbeta_j = \sum_{\hS_j \in \Hss_j}\hlam_j(\hS_i)\cdot v_j(\hS_j) \geq 1.
        \]
    \end{enumerate}
\end{claim}
\begin{proof}
    By construction, every item $e \in W_2$ is small for every agent $i \in U_3$.
    Therefore, by deleting $e$ from every bundle $S_i \in \Ss_i$ the agent loses at most $\sum_{S_i\in \Ss_i : S_i \ni e}\lambda_i(S_i)\cdot \alpha \leq \alpha / n$ value.
    And, as $|W_2| < n$, removing all $W_2$ from bundles of $\Ss_i$ will decrease the APS-value of $i$ by at most $\alpha$.
    
    Next, for $j \in U_2$, consider the initial APS-partition $\{(S_j, \lambda_j(S_j)\}_{S_j \in \Ss_j}$ of agent $j$ over all agents $U_2\cup U_3$ and all items $W_2\cup W_3$, its value is at least $1$.
    By construction of $U_2$, there exists a matching $M_j$ such that $j \notin M_j$ and $|M_j| = |W_2|$.
    Then, we can apply \cref{giveitem2} to all agents of $M_j$ and their matched items of $W_2$, repeatedly removing pairs $(j', e') \in M_j$.
    By \cref{giveitem2}, for all the remaining $|U_2| - |W_2|$ agents, including $j$, their APS will not decrease.
    So, for the obtained APS-partition $\{(\hS_j, \hlam_j(\hS_j)\}_{\hS_j \in \Hss_j}$ of the remaining items, i.e $W_3$, among the remaining $|U_2| - |W_2|$ agents, including $j$, we will have $\hbeta_j(\hSs_j) \geq 1$.
\end{proof}
For every agent $j \in U_2$, we the APS-partition $\{(\hS_j, \hlam_j(\hS_j)\}_{\hS_j \in \Hss_j}$ from \cref{w3value}, without loss of generality we can assume that it satisfies \cref{wlog}.
Then, for every $j \in U_2$ we would have $\hbeta_j := \hbeta_j(\Hss_j) = 1$.

Since the algorithm only picks agents from $U_3$, values of $\beta_j$ for agents $j \in U_2$ will not affect its performance at all.
Because of this, we will filter-out dangerous bundles, as well as ensure that the APS-value decrease at every step is bounded, only for agents from $U_3$.
In essence, while agents from $U_2$ do participate in the algorithm, they play an ``observant'' role and just lose APS-value at every step from $U_3$ allocations.
That is, every time agent $i \in U_3$ is given some bundle, say $B_i$, by the algorithm, for $j \in U_2$ the algorithm completely removes the items of $B_i$ from $j$'s APS partition.
The full description of the algorithm is given below.

\begin{algorithm}[h]
\caption{Algorithm for agents $U_3$ in case 4) of \cref{bipartite}}\label{bigalgo}
	\begin{algorithmic}[1]
		\Require{Agents $U_3$ with APS partitions $\{(S_i, \lambda_i(S_i)\}_{S_i \in \Ss_i}$, agents $U_2$ with APS partitions $\{(\hS_j, \hlam_j(\hS_j)\}_{\hS_j \in \Hss_j}$ from \cref{w3value} satisfying \cref{wlog}, parameters $C, D > 0$.}
        \State Let $\Aa^k$ be the set of active agents at iteration $k$, initially $\Aa^1 := U_3$
        \State For $j \in U_2$ every $\hS_j \in \Hss_j$, let $\hS_j^k\subseteq \hS_j$ be the items remaining in $\hS_j$ at iteration $k$
        \State For $i \in U_3$ every $S_i \in \Ss_i$, let $S_i^k\subseteq S_i$ be the items remaining in $S_i$ at iteration $k$
        \For{$j \in U_2$}
            \State For every $\hS_j \in \Hss_j$, let $\hS_j^1 := \hS_j$
            \State Let $\hSs_j^1 := \{\hS_j^1 : \hS_j \in \hSs_j\}$ --- initial filtered partition for $U_2$-agent
        \EndFor
        \For{$i \in U_3$}
            \State For every $S_i \in \Ss_i$, let $S_i^0 := S_i \cap W_3$
            \State For every $S_i \in \Ss_i$, if $\Delta_i(S_i) > D$, let $S_i^1 := \varnothing$, otherwise $S_i^1 := S_i^0$
            \State Let $\Ss_i^1 := \{S_i^1 : S_i \in \Ss_i\}$ --- initial filtered partition for $U_3$-agent
        \EndFor
		\For{$k = 1,\ldots, |U_3|$}
            \State For every $i \in \Aa^k$, compute $\beta^k_i := \sum_{S_i \in \Ss_i}\lambda_i(S_i) \cdot v_i(S_i^k)$
            \State Sort $\Aa^{k} = \{i_k, i_{k + 1}, \ldots, i_{|U_3|}\}$ so that $\beta^k_{i_k} \geq \beta^k_{i_{k + 1}} \geq \ldots \geq \beta^k_{i_{|U_3|}}$, and pick the agent $i_k$
            \State Let $\mu_{i_k}^k$ be a probability distribution over sub-bundles of $\Ss_{i_k}^k$ from \cref{distpartition}
            \State Sample $B_{i_k} \sim \mu_{i_k}^k$ at random, and let $Q_{i_k}^{k} := B_{i_k}$
            \For{$h \in U_3 \setminus\Aa^{k}$}
                \State Update $Q_h^{k} := Q_h^{k - 1} \setminus B_{i_k}$
            \EndFor
            \State Update $\Aa^{k + 1} := \Aa^k \setminus \{i_k\}$
            \For{$i \in \Aa^{k+1}$}
                \State Update $\Ss_i^{k + 1} := \mathbf{Filtering}(v_{i_k}, v_i, B_{i_k}, \Ss_i^k, C)$ (\cref{filteralgo})
            \EndFor
            \For{$j \in U_2$}
                \State For every $\hS_j \in \hSs_j$, update $\hS_j^{k + 1} := \hS_j^k \setminus B_{i_k}$
                \State Update $\hSs_j^{k + 1} := \{\hS_j^{k + 1} : \hS_j \in \hSs_j\}$
            \EndFor
		\EndFor
		\Return sets $Q_1^{|U_3|}, \ldots, Q_n^{|U_3|}$.
	\end{algorithmic}
\end{algorithm}

We will prove the following theorem about the performance of \cref{bigalgo}.
\begin{theorem}\label{bigcorrectness}
    Let $\rho$ satisfy $2(12\rho - 3)\ln(3\rho) = (1 - 3\rho)(3\ln3 - 4)$.
    For every $\alpha \in [3/11, \rho)$ there exists $n_\alpha$, such that for all $n_4 \geq n_\alpha$ the following holds: after running \cref{bigalgo}, with probability at least $1/2 - 1/(2n_4)$
    \begin{enumerate}
        \item every agent $i \in U_3$ receives a bundle $Q_i^{|U_3|}$ satisfying $v_i(Q_i^{|U_3|}) \geq \alpha - \frac{2\ln^6n_4}{\sqrt{n_4}}$;
        \item remaining APS-values $\hbeta_j^{|U_3|} := \sum_{\hS_j \in \hSs_j}\hlam_j(\hS_j)\cdot v_j(\hS_j^{|U_3|})$ for agents $j \in U_2$ satisfy
        \[
            \frac{1}{|U_2|}\sum_{j \in U_2}\hbeta_j^{|U_3|} \geq 1 - \frac{16\alpha \eps_4}{3(1 - \eps_4)(\rho - \alpha)}.
        \]
    \end{enumerate}
\end{theorem}
Given this theorem, we then can do similar trick as in \cref{finalallocation} and \cref{finalallocation2}: run \cref{bigalgo} w.r.t value $\alpha' = \alpha + (\rho -\alpha)/2$.
Since every agent $i \in U_3$ has only items of value less than $\alpha < \alpha'$ in their APS partition, these APS partitions satisfy \cref{wlog} and \textbf{Assumption for item value} for $\alpha'$.
Then, \cref{distpartition} applies for $\alpha'$, and the \cref{bigalgo} will try to obtain an $(\alpha'- \frac{2\ln^6n_4}{\sqrt{n_4}})$-allocation for agents in $U_3$.
Taking $n$ large enough so that $n_4 \geq (1 - \eps_4)n \geq n_\alpha$ will guarantee that $\alpha'- \frac{2\ln^6n_4}{\sqrt{n_4}} \geq \alpha$.
At the same time, agents $j \in U_2$ will have 
\[
    \frac{1}{|U_2|}\sum_{j \in U_2}\hbeta_j^{|U_3|} \geq 1 - \frac{16\alpha' \eps_4}{3(1 - \eps_4)(\rho - \alpha')} = 1 - \frac{16(\rho + \alpha)\eps_4}{3(1 - \eps_4)(\rho - \alpha)}.
\]
Thus we obtain the following corollary.

\begin{corollary}\label{bigallocation}
    Let $\rho$ satisfy $2(12\rho - 3)\ln(3\rho) = (1 - 3\rho)(3\ln3 - 4)$.
    For every $\alpha \in [3/11, \rho)$ there exists $n_\alpha$, such that for all $n_4 \geq n_\alpha$ the following holds: there exists an allocation of some of the items $W_3$ among agents $U_3$ such that
    \begin{enumerate}
        \item every agent $i \in U_3$ receives a bundle $Q_i$ satisfying $v_i(Q_i) \geq \alpha$;
        \item values $\hbeta_j^{|U_3|} := \sum_{\hS_j \in \hSs_j}\hlam_j(\hS_j)\cdot v_j(\hS_j\setminus \bigcup_{i\in U_3}Q_i)$ for agents $j \in U_2$ satisfy
        \[
            \frac{1}{|U_2|}\sum_{j \in U_2}\hbeta_j^{|U_3|} \geq 1 - \frac{16(\rho + \alpha)\eps_4}{3(1 - \eps_4)(\rho - \alpha)}.
        \]
    \end{enumerate}
\end{corollary}

\begin{proof}[Proof of \cref{bigcorrectness}]
For $k \in [|U_3|]$, let $i_k$ denote the agent that was picked at iteration $k$ of \cref{bigalgo}, and let $p_{i_k}$ be the smallest possible number such that for every item $e \in W_2\cup W_3$, $\mu_{i_k}^k(e) \leq p_{i_k}$.
Once again, for every $j \in U_3$ we consider the process $\wgamma_j^k$ defined as $\gamma_j^1 = \beta_j^1$ and for $k \geq 1$, $\wgamma_j^{k + 1} = (1 - p_{i_k})\wgamma_j^k$.
Consider the deterministic process $\gamma_{n_4}^k(\beta_j^1, \eps)$ for some value $\eps$.
\begin{claim}\label{bigepsprocess}
    For $1 \leq t < |U_3|$, let $j \in \Aa^{t + 1}$ be some agent from $U_3$ that remains active after iteration $t$ of \cref{bigalgo}.
    For $\eps > 0$, let $\Ee_j^t(\eps)$ be the following event:
    \[
        \Ee_j^t(\eps) = \{\forall k \in [t]: \beta_j^k \geq \wgamma_j^k - \eps\}.
    \]
    Under the event $\Ee_j^t(\eps)$, it holds for all $k \in [t]$ that $\beta_j^k \geq \gamma_{n_4}^k(\beta_j^1, \eps)$.
    Furthermore, if agent $i_k$ was picked at iteration $k$ of \cref{bigalgo}, under the event $\Ee_j^t(\eps)$ for every $k \in [t]$ it holds 
    \[
        p_{i_k} \leq \frac{4\alpha}{3(\gamma_{n_4}^k(\beta_j^1, \eps) - \eps - \alpha)n_4}.
    \]
\end{claim}
\begin{proof}
    Follows directly from \cref{epsprocess2} and \cref{pikbound2} by observing that all $\beta_{i_k}^k$ was the largest among $\beta_j^k$ for agents $j \in U_3$ that were still active, i.e $j \in \Aa^k$.
\end{proof}

In order to prove \cref{bigcorrectness}, we need to show that \cref{bigalgo} at every iteration manages to successfully construct distribution $\mu_{i_k}^k$.
As before, this requires showing that 1) for every $i \in U_3$, process $\gamma_{n_4}^k(\beta_i^1, \eps) - \eps$ is at least $\alpha + (\rho - \alpha)/2$ (provided that $n_4$ is large enough), and 2) for every agent $j \in U_3$ and every iteration $t$ of \cref{bigalgo}, the probability that the event $\{\Ee_i^t(\eps) \mid i \in \Aa^t\}$ does not occur is small. 

Note that for every $i \in U_3$ values $\beta_i^k$ at all iterations $k$ are affected only by other $\beta_j^k$ for $j \in U_3$.
In addition, by \cref{w3value} we have $\beta_i^1 \geq 1 - \alpha - 1/D$ for all $i \in U_3$.
Now, consider the simplified statement of \cref{lowergamma2better}.
\begin{theorem}[\cref{lowergamma2better} restated]\label{gammasimple}
    For $\eps > 0$, let $\alpha \leq \rho^* < 1 - \eps$.
    Consider the process $\gamma_n^k(1 - \delta, \eps)$ for this $\eps$, $0 < \delta \leq 1 - (\alpha + \eps)$ and some ${n^*}$, and take integer $d$ such that $d\leq n$.
    Then, $\gamma_{n^*}^{n_0}(1 - \delta, \eps) \geq \rho^* + \eps$ if ${n^*}, n_0$ satisfy
    \begin{multline}
        \frac{3(3\alpha - \rho^* - \alpha\ln(3\alpha)  + \alpha\ln \rho^*)}{4\alpha} + \frac{1 - 3\alpha + (12\alpha - 3)\ln(3\alpha )}{2(1 - 3\alpha)} \\
        \geq \frac{n_0}{{n^*}} + \eps\left(\frac{5}{4\alpha} +  \frac{(1 - \ln(3\alpha))}{2(1 - 3\alpha)}\right) + 2\delta + O\left(\frac{1}{\sqrt{{n^*}}}\right).\label{gammasimpleineq}
    \end{multline}
\end{theorem}
In our case, $\rho^* = \alpha + (\rho - \alpha)/2 < \rho$, where $\rho$ satisfies $2(12\rho - 3)\ln(3\rho) = (1 - 3\rho)(3\ln3 - 4)$.
For these values of parameters $\rho^*$ and $\alpha$, the lefthand side expression of \cref{gammasimpleineq} is at least $1 + \const$.
Next, as $\beta_j^1 \geq 1 - \alpha - 1/D$, we pick $\delta = \alpha + 1/D$ in \cref{gammasimple}.
Finally, for the process $\gamma_{n_4}^k(\beta_j^1, \eps)$ we pick $n^* = n_4 \geq (1 - \eps_4)n$ and $n_0 = |U_3| \leq \eps_4 n$.
Then, the righthand side of \cref{gammasimpleineq} is at most
\[
    \frac{\eps_4}{1 - \eps_4} + 2\alpha + \eps\left(\frac{5}{4\alpha} +  \frac{(1 - \ln(3\alpha))}{2(1 - 3\alpha)}\right) + \frac{2}{D} + O\left(\frac{1}{\sqrt{{n}}}\right).
\]
Choosing small enough constant $\eps_4$ so that $\frac{\eps_4}{1 - \eps_4} + 2\alpha < 1$, and picking $\eps, D$ such that $\eps, \frac{1}{D}\xrightarrow{n\to\infty}0$ guarantees that for large enough value of $n_4$ (i.e, large enough $n$), we have $\gamma_{n_4}^k(\beta_i^1, \eps) - \eps \geq \alpha + (\rho - \alpha)/2$ for all $k \leq |U_3|$.

It follows that, similar to \cref{totalprob2} and \cref{finalprob2}, one can pick values of $c, D, \eps$ in \cref{bigalgo} and choose $\eps\xrightarrow{n_4\to\infty}0$ such that $\eps, \frac{1}{D} \xrightarrow{n_4\to\infty}0$ and with probability at least $1 - 1/n_4$ for every agent $i \in U_3$ and iteration $1 \leq t \leq |U_3|$, if $i \in \Aa^t$ then $\beta_i^t \geq \alpha + (\rho - \alpha)/2$.
But then it  follows directly from the proof of \cref{finalprob2} and \cref{dangerous} that after running \cref{bigalgo}, every agent $i \in U_3$ receives a bundle $Q_i^{|U_3|}$ such that $v_i(Q_i^{|U_3|}) \geq \alpha - \frac{2\ln^6n_4}{\sqrt{n_4}}$.

Now we show that with probability at least $1/2 - 1/(2n_4)$, agents $j \in U_2$ must satisfy $\frac{1}{|U_2|}\sum_{j \in U_2}\hbeta_j^{|U_3|} \geq 1 - \frac{16\alpha \eps_4}{3(1 - \eps_4)(\rho - \alpha)}$.
To see this, fix some $j \in U_2$ and consider the expectation $\E{\hbeta_j^{|U_3|}}$ over distributions $\mu_{i_1}^1, \mu_{i_2}^2,\ldots, \mu_{i_{|U_3|}}^{|U_3|}$.
By \cref{probval}, for every $1 \leq k \leq |U_3|$ we have
\[
    \E[\mu_{i_1}^1,\ldots, \mu_{i_k}^k, \mu_{i_{k + 1}}^{k + 1}]{\hbeta_j^{k + 1}} \geq (1 - p_{i_k}) \E[\mu_{i_1}^1,\ldots, \mu_{i_k}^k]{\hbeta_j^{k}}.
\]
Since $\hbeta_j^1 = 1$, by \cref{bigepsprocess} with probability at least $1 - 1/n_4$ for every $j \in U_2$ it holds
\[
    \E{\hbeta_j^{|U_3|}} \geq \prod_{k = 1}^{|U_3|}(1 - p_{i_k})\hbeta_j^1 \geq \prod_{k = 1}^{|U_3|}\left(1 - \frac{4\alpha}{3(\gamma_{n_4}^k(\beta_{i_k}^1, \eps) - \eps - \alpha)n_4}\right)
\]
for our choice of $C, D, \eps$.
Since $n_4$ is large enough, we have $\gamma_{n_4}^k(\beta_{i_k}^1, \eps) - \eps \geq \alpha + (\rho - \alpha)/2$ for every picked $i_k \in U_3$, thus
\[
    \E{\hbeta_j^{|U_3|}} \geq \prod_{k = 1}^{|U_3|}\left(1 - \frac{8\alpha}{3(\rho -\alpha)n_4}\right) \geq 1 - |U_3|\cdot \frac{8\alpha}{3(\rho -\alpha)n_4} \geq 1 - \frac{8\alpha \eps_4}{3(1 - \eps_4)(\rho - \alpha)},
\]
where the last inequality follows from $n_4 \geq (1 - \eps_4)n$ and $|U_3| \leq \eps_4 n$.
Then, by Markov inequality, with probability at least $1 - 1/n_4$ it holds
\begin{multline*}
    \P{|U_2| - \sum_{j \in U_2}\hbeta_j^{|U_3|} >  \frac{16\alpha \eps_4}{3(1 - \eps_4)(\rho - \alpha)}\cdot |U_2|} \\
    \leq \frac{1}{|U_2|}\cdot \frac{3(1 - \eps_4)(\rho - \alpha)}{16\alpha \eps_4}\cdot \E{|U_2| - \sum_{j \in U_2}\hbeta_j^{|U_3|}}\\
    = \frac{1}{|U_2|}\cdot \frac{3(1 - \eps_4)(\rho - \alpha)}{16\alpha \eps_4}\cdot \sum_{j \in U_2}\left(1 - \E{\hbeta_j^{|U_3|}}\right)
    \\
    \leq \frac{1}{|U_2|} \cdot \frac{3(1 - \eps_4)(\rho - \alpha)}{16\alpha \eps_4}\cdot \frac{8\alpha \eps_4}{3(1 - \eps_4)(\rho - \alpha)} \cdot |U_2| = \frac{1}{2}.
\end{multline*}
It follows by the law of total probability that with probability at least $1/2 - 1/(2n_4)$ it holds
\[
    \frac{1}{|U_2|}\sum_{j \in U_2}\hbeta_j^{|U_3|} \geq 1 - \frac{16\alpha \eps_4}{3(1 - \eps_4)(\rho - \alpha)}.
\]
\end{proof}
\subsection{Allocation for case 4): agents $U_2$}

For every agent $j \in U_2$, let $\{(\hS_j, \hlam_j(\hS_j))\}_{\hS_j \in \Hss_j}$ be the APS partition of only items $W_3$ among $|U_2|-|W_2|$ agents from \cref{w3value}, satisfying \cref{wlog}.
We showed in \cref{bigcorrectness} and \cref{bigallocation} that there exists an allocation of items $W_3$ among agents $U_3$ with the following properties:
\begin{itemize}
    \item every agent $i \in U_3$ receives a bundle $Q_i$ of value at least $\alpha$-APS;
    \item values $\hbeta_j := \sum_{\hS_j \in \hSs_j}\hlam_j(\hS_j)\cdot v_j(\hS_j\setminus \bigcup_{i\in U_3}Q_i)$ for agents $j \in U_2$ satisfy
        \[
            \frac{1}{|U_2|}\sum_{j \in U_2}\hbeta_j\geq 1 - \frac{16(\rho + \alpha)\eps_4}{3(1 - \eps_4)(\rho - \alpha)}.
        \]
\end{itemize}
We will now focus on allocating agents $U_2$.
By construction, every subset $U_2'\subseteq U_2$ of size $|W_2| \geq (1 - \eps_4)n$ forms a perfect matching with $W_2$.
Hence, it is sufficient to allocate arbitrary $|U_2| - |W_2| \leq \eps_4 n$ agents of $U_2$ using the items of $W_3$, and then use the matching to give a single personal big item to each $|W_2|$ of the remaining agents of $U_2$.

We will do this using a generalization of our initially presented approach for the case of same-valuation agents, specifically the greedy algorithm \cref{algo}, to the case of agents with different valuations.
We will abuse notation and let $\hSs_j^0$ denote the filtered APS-partition $\{(\hS_j^0, \hlam_j(\hS_j))\}_{\hS_j \in \Hss_0}$ of agent $j \in U_2$, obtained from the original $\{(\hS_j, \hlam_j(\hS_j))\}_{\hS_j \in \Hss_j}$ after \cref{bigalgo} (with the property on average above).

Denote $n_0 := |U_2| - |W_2|$, so $|W_2| =|U_2| - n_0$, and let $\eps_0$ be the smallest constant satisfying $n_0 \leq \eps_0|U_2|$.
Our algorithm will run for $n_0$ iterations, and at every iteration $k$ we maintain a set of agents $\Aa^k\subseteq U_2$ that are still active, as well as total values $\hbeta_j^k$ of their filtered APS partitions $\{(\hS_j^k, \hlam_j(\hS_j))\}_{\hS_j \in \Hss_j}$ over items $W_3$.
For iteration $k$ and any set $B$, denote
\[
    \hbeta_j^k(B) := \sum_{\hS_j \in \hSs_j}\hlam_j(\hS_j)\cdot v_j(\hS_j^k \setminus B).
\]
At every iteration $k$, we pick an active agent $j_k \in \Aa^k$ with \textbf{largest} value $\hbeta_{j_k}^k$, and then greedily choose for her bundle $B_{j_k}$ that maximizes the average value $\frac{1}{|\Aa^k| - 1}\sum_{j \in \Aa^k\setminus\{j_k\}}\hbeta_j^{k}(B_{j_k})$.
We allocate $B_{j_k}$ to $j_k$ and set $\Aa^{k + 1} = \Aa^k\setminus\{j_k\}$.
We run this algorithm until exactly $n_0 = |U_2| - |W_2|$ agents are allocated, and then give out single big items to the remaining $|W_2|$ agents, similar to what we did in previous cases 1), 2) and 3).
The full description of this greedy-average algorithm, \cref{greedavgalgo}, is given below.

\begin{algorithm}[h]
\caption{Algorithm for agents $U_2$ in case 4) of \cref{bipartite}}\label{greedavgalgo}
	\begin{algorithmic}[1]
		\Require{Agents $U_2$ with filtered APS partitions $\{(\hS_j^0, \hlam_j(\hS_j)\}_{\hS_j \in \Hss_j}$ after \cref{bigalgo}.}
        \State Let $\Aa^k$ be the set of active agents at iteration $k$, initially $\Aa^1 := U_2$
        \State For $j \in U_2$ and every $\hS_j \in \Hss_j$, let $\hS_j^k\subseteq \hS_j$ be the items remaining in $\hS_j$ at iteration $k$
		\For{$k = 1,\ldots, n_0$}
            \State For every $j \in \Aa^k$, compute $\hbeta^k_j := \sum_{\hS_j \in \hSs_j}\lambda_j(\hS_j) \cdot v_j(\hS_j^k)$
            \State Sort $\Aa^{k} = \{j_k, j_{k + 1}, \ldots, j_{|U_3|}\}$ so that $\hbeta^k_{j_k} \geq \hbeta^k_{j_{k + 1}} \geq \ldots \geq \hbeta^k_{j_{n_0}}$, and pick the agent $j_k$
            \State Use \cref{partition} to construct the multi-set $\Bb_{j_k}^k$ from the APS-partition $\hSs_{j_k}^k$
            \State Let $B_{j_k} \in \Bb_{j_k}^k$ be an acceptable bundle for $j_k$ maximizing $\frac{1}{|\Aa^k| - 1}\sum_{j \in \Aa^{k}\setminus\{j_k\}}\hbeta_j^k(B_{j_k})$
            \State Give agent $j_k$ bundle $B_{j_k}$, and set $\Aa^{k + 1} := \Aa^{k}\setminus \{j_k\}$
            \For{$j \in \Aa^{k + 1}$}
                \State For every $\hS_j \in \hSs_j$, update $\hS_j^{k + 1} := \hS_j^k \setminus B_{j_k}$
                \State Update $\hSs_j^{k + 1} := \{\hS_j^{k + 1} : \hS_j \in \hSs_j\}$
            \EndFor
		\EndFor
		\Return sets $B_{j_1}, \ldots, B_{j_{n_0}}$.
	\end{algorithmic}
\end{algorithm}

We will now analyze the performance of \cref{greedavgalgo}.
It is going to be quite similar to the analysis of \cref{algo}, with additional caveats.
First, we are going to define a new special process $\hgam_{n_0}^k$.
\begin{definition}
    Let $\hgam_{n_0}^k$ be the following process: $\hgam_n^1 = 1 - \frac{16(\rho + \alpha)\eps_4}{3(1 - \eps_4)(\rho - \alpha)}$, and for $k \geq 1$,
    \begin{gather*}
        \text{if }\; \hgam_{n_0}^k - \tau \geq 3\alpha,\;\text{ then }\; \hgam_{n_0}^{k + 1} = \left(1 - \frac{2(1 - 3\alpha)}{(\hgam_{n_0}^k -\tau - 12\alpha + 3)n_0}\right)\hgam_{n_0}^k;\\
        \text{if }\; \hgam_{n_0}^k - \tau < 3\alpha,\; \text{ then }\; \hgam_{n_0}^{k + 1} = \left(1 - \frac{4\alpha}{3(\hgam_{n_0}^k -\tau- \alpha){n_0}}\right)\hgam_{n_0}^k,
    \end{gather*}
    where $\tau = \frac{\eps_0}{(1-2\eps_0)}$.
\end{definition}
Similarly to the case of same valuations covered in \cref{algogamma} and \cref{algogamma2}, we use $\hgam_{n_0}^k$ to lower bound the average remaining value of active agents $\Aa^k$.
Then, if the value $\hgam_{n_0}^k$ stays large enough after $n_0$ iterations of \cref{greedavgalgo}, it guarantees that all $n_0$ iterations will succeed.
\begin{theorem}\label{avggamma}
    For every iteration $k \in [n_0]$ it holds
    \[
        \frac{1}{|\Aa^k|}\sum_{j \in \Aa^k}\hbeta_j^k \geq \hgam_{n_0}^k - \tau.
    \]
\end{theorem}
\begin{proof}
    We will use induction to prove the following claim.
    Define the process $\tau_{n_0}^k$ as follows $\tau_{n_0}^1 := 0$, and for all $k \geq 1$,
    \[
        \tau_{n_0}^{k + 1} = \tau_{n_0}^k\left(1 + \frac{1}{|\Aa^k| - 1}\right) + \frac{1}{|\Aa^k| - 1}.
    \]
    The base case $k = 1$ follows trivially by definition of $\hgam_{n_0}^k$ and $\tau_{n_0}^k$.
    Suppose that the lower bound holds for all $1 \leq k' \leq k$, we want to show it for $k + 1$.

    Let $j_k$ denote the agent that was picked at iteration $k$ of the algorithm.
    It follows from \cref{probval}, \cref{choice} and \cref{valbound} that the chosen bundle $B_{j_k}$ from $\Bb_{j_k}^k$ satisfies 
    \[
        \hbeta_j^{k + 1} = \hbeta_j^k(B_{j_k}) \geq (1 - \hp_{j_k})\hbeta_j^k,
    \]
    where if $\hbeta_{j_k}^k \geq 3\alpha$, then $\hp_{j_k} \leq \frac{2(1 - 3\alpha)}{(\hbeta_{j_k}^k - 12\alpha + 3)n}$, and if $\hbeta_{j_k}^k < 3\alpha$, then $\hp_{j_k} \leq \frac{4\alpha}{3(\hbeta_{j_k}^k - \alpha)n}$.
    Since $j_k$ has the largest $\hbeta_{j_k}^k$ among $\Aa^k$, applying averaging and using the inductive assumption yields us $\hbeta_{j_k}^k \geq \hgam_{n_0}^k - \tau_{n_0}^k$.
    Hence, using inductive step,
    \begin{multline*}
        \frac{1}{|\Aa^k| - 1}\sum_{j \in \Aa^{k + 1}}\hbeta_j^{k + 1} \geq \frac{1 - \hp_{j_k}}{|\Aa^k| - 1} \sum_{j \in \Aa^{k + 1}}\hbeta_j^{k}\\
        = \frac{1 - \hp_{j_k}}{|\Aa^k| - 1} \sum_{j \in \Aa^{k}}\hbeta_j^{k} - \frac{(1 - \hp_{j_k})\hbeta_{j_k}^k}{|\Aa^k| - 1} \geq \frac{|\Aa^k|}{|\Aa^k| - 1}\cdot (1 - \hp_{j_k})(\hgam_{n_0}^k - \tau_{n_0}^k)- \frac{(1 - \hp_{j_k})\hbeta_{j_k}^k}{|\Aa^k| - 1}\\
        = (1 - \hp_{j_k})\hgam_{n_0}^k - (1 - \hp_{j_k})\tau_{n_0}^k - \frac{(1 - \hp_{j_k})(\hbeta_{j_k}^k + \tau_{n_0}^k -\hgam_{n_0}^k)}{|\Aa^k| - 1}.
    \end{multline*}
    Since $1 - \hp_{j_k} \leq 1$, $\hbeta_{j_k}^k \leq 1$ and $\hgam_{n_0}^k \geq 0$, the inequality above implies
    \[
        \frac{1}{|\Aa^k| - 1}\sum_{j \in \Aa^{k + 1}}\hbeta_j^{k + 1} \geq (1 - \hp_{j_k})\hgam_{n_0}^k - \tau_{n_0}^k - \frac{1 + \tau_{n_0}^k}{|\Aa^k| - 1} = (1 - \hp_{j_k})\hgam_{n_0}^k - \tau_{n_0}^{k + 1}.
    \]
    Now, it is easy to see (e.g via induction) that the recursive process $\tau_{n_0}^k$ satisfies $\tau_{n_0}^{k} = \frac{k}{|\Aa^k| - 1}$.
    Furthermore, for every $k$ we have $|\Aa^k| - 1 = |U_2| - k$.
    Then, since $k \leq n_0\leq \eps_0 |U_2|$ and $|U_2| - k \geq (1 - \eps_0)|U_2|$, we obtain that for every $k \in [n_0]$,
    \[
        \tau_{n_0}^k \leq \frac{\eps_0}{1 - \eps_0} = \tau.
    \]
    Next, $\hbeta_{j_k}^k \geq \hgam_{n_0}^k - \tau_{n_0}^k$ and $\tau_{n_0}^k \leq\tau$ imply that $(1 - \hp_{j_k})\hgam_{n_0}^k \geq \hgam_{n_0}^{k + 1}$.
    To see this, we use a proof similar to \cref{epsprocess2} and utilize the upper bound on $\hp_{j_k}$ above, replacing $\hbeta_{j_k}^k$ first with $\hgam_j^k - \tau_{n_0}^k$ and then with $\hgam_j^k - \tau$.
    Thus, induction step is also proved, and the lemma follows.
\end{proof}

Recall that \cref{partition} and \cref{choice} succeed to pick an acceptable $B_{j_k}$ from $\Bb_{j_k}^k$ for agent $j_k$ only if $\hbeta_{j_k}^k \geq \alpha$.
So, to prove that \cref{greedavgalgo} indeed works, by \cref{avggamma} it suffices to show that for every $k \in [n_0]$ we have $\hgam_{n_0}^k \geq \alpha + \tau$.

Let $d_0$ be some constant (dependent on $\eps_0$) that satisfies $\alpha + d_0\tau < \rho$, and will be defined later.
First, we use \cref{gammasimple} and show that there exists $\eps_0$ such that when $n_0$ is large enough, it holds $\hgam_{n_0}^{n_0} \geq \alpha + d_0\tau$.
Recall that by definition, $\hgam_{n_0}^1 = 1 - \frac{16(\rho + \alpha)\eps_0}{3(1 - \eps_0)(\rho - \alpha)}$.
We substitute values $\rho^* = \alpha + (d_0 -1)\tau$, $\eps = \tau = \frac{\eps_0}{1 - 2\eps_0}$, $\delta = \frac{16(\rho + \alpha)\eps_0}{3(1 - \eps_0)(\rho - \alpha)}$, same $n_0$ and $n^* = n_0$ into \cref{gammasimple}, and obtain that $\hgam_{n_0}^{n_0} \geq \alpha + d_0\tau$ if
\begin{multline}
        \frac{3(3\alpha - \rho^* - \alpha\ln(3\alpha)  + \alpha\ln \rho^*)}{4\alpha} + \frac{1 - 3\alpha + (12\alpha - 3)\ln(3\alpha )}{2(1 - 3\alpha)} \\
        \geq 1 + \frac{\eps_0}{1 - 2\eps_0}\left(\frac{5}{4\alpha} +  \frac{(1 - \ln(3\alpha))}{2(1 - 3\alpha)}\right) + \frac{16(\rho + \alpha)\eps_0}{3(1 - \eps_0)(\rho - \alpha)} + O\left(\frac{1}{\sqrt{{n_0}}}\right).\label{finalineq}
\end{multline}
As shown in the proof of \cref{diffgamma2}, for values of $\alpha + d_0\tau< \rho$ the lefthand side of \cref{finalineq} is at least $1 + \const$.
Then, we can pick sufficiently small $\eps_0$ and the largest value $d_0$ for this $\eps_0$ so that the righthand side of \cref{finalineq} without $O(\frac{1}{\sqrt{n_0}})$ part is strictly less than the lefthand side, we will get $\hgam_{n_0}^{n_0} \geq \alpha + 2\tau$ for sufficiently large $n_0$.

However, $n_0 \leq \eps_0 n$, so the bound on $\hgam_{n_0}^{n_0}$ that works only when $n_0$ is large is not good enough!
Fortunately, the bound $\hgam_{n_0}^{n_0}$ actually holds for all $n_0 \geq 1$.
This is due to the fact that even though process $\hgam_{n_0}^k$ loses an additional term of $\tau$, when compared with the original process $\gamma_n^i$ for the same valuation case, this $\tau$ being a fixed constant allows us to apply the \textbf{doubling} procedure also to $\hgam_{n_0}^k$!
Specifically, we prove the following lemma in \cref{sec:doubling} (\cref{doublingproof2}):
\begin{lemma}\label{bigdoubling}
    Suppose that for some $n_0$ and $k, l \geq 1$ it holds $\hgam_{n_0}^k \geq \beta=\gamma_{2n_0}^l$, for some value $\beta$.
    If
    \[
        \beta >  \alpha + \tau +  \frac{\alpha + \sqrt{6(\alpha + \tau)\alpha n_0 + \alpha^2}}{3n_0},
    \]
    where $\tau$ is constant, then $\hgam_{2n_0}^{l + 2} \leq \hgam_{n_0}^{k + 1}$.
\end{lemma}
If $n_0$ is large enough for \cref{finalineq} to hold by itself, it already holds $\hgam_{n_0}^{n_0} \geq \alpha + d_0\tau \geq \alpha$.
Otherwise, we can find $N_0$ large enough such that $\hgam_{N_0}^{N_0} \geq \alpha + d_0\tau$, and for large enough choice of $d_0$ we would have $(d_0 - 1)\tau \geq \frac{\alpha + \sqrt{6(\alpha + \tau)\alpha N_0 + \alpha^2}}{3N_0}$.
Then by \cref{doubling} we could reduce $N_0$ to $N_0/2$ and so on.
Note that the smaller $n_0$ is, the smaller we could take the constant $\eps_0$ and hence the larger we could take the constant $d_0$.
Then, we continue doing the doubling for all values of $n_0$ until some constant, and from there the values of $\hgam_{n_0}^{n_0}$ are checked numerically.
Hence, for every given $\alpha$ one can pick small enough values $\eps_4$, $\eps_0$ and $d_0$ to guarantee doubling for any any target number $n'$ of values to check.


\section{Lower bound on $\gamma_n^n(\beta, \eps)$ for $\alpha \leq 3/11$}
\label{sec:smallalpha}

Recall the definition of process $\gamma_n^k(\beta, \eps)$ for $\alpha \leq 3/11$.

\begin{definition}
    For $n$ agents, $\alpha \leq 3/11$, $\eps > 0$ and $\beta > \alpha + \eps$, let $\gamma^k_n(\beta, \eps)$ be the following process: $\gamma^1_n(\beta, \eps) = \beta$, and for $i \geq 1$,
    \[
        \gamma^{i + 1}_n(\beta, \eps) = \left(1 - \frac{1 - \alpha}{2(\gamma^k_n(\beta, \eps) - \eps -\alpha)n}\right)\gamma^k_n(\beta, \eps).
    \]
\end{definition}

In this section, we prove \cref{diffgamma}, restated below.
Note that the claim of (\cref{rhoinfinity}) is a special case of this theorem with $\eps = \delta = 0$.
\begin{theorem}\label{diffgammaproof}
    For $\alpha \in [1/4, 3/11)$ let $\rho$ satisfy $2(1 - \rho + \alpha \ln \rho) = 1- \alpha$.
    For any constant $0 < x < \rho - \alpha$ and any $\eps, \delta \xrightarrow{n\to\infty}0$, there exists $n_x$ such that for all $n \geq n_x$ and all $n_0 \leq n$ it holds $\gamma_n^{n_0}(1 - \delta, \eps) - \eps \geq \alpha + x$.
\end{theorem}

First, we consider a process $\gamma_n^k(\beta, \eps)$ for arbitrary $\beta \geq \alpha + \eps$ and $\rho$, and obtain a lower bound on the number of iterations required for such process, having started at value $\beta$, to drop below value $\rho$.
To do so, we break the segment $[\rho, \beta]$ into many epochs of the form $[r/d, (r + 1)/d]$ w.r.t current value of $\gamma_n^k(\beta, \eps)$.
Then, we estimate the minimum number of iterations the process spends in each epoch, and determine for which values of $n$ the total number of estimated iterations over all epochs exceeds some $n_0$.
\begin{lemma}\label{lowergammaiters}
    Consider the process $\gamma^k_n := \gamma^k_n(\beta, \eps)$ for some $\eps > 0$, $\beta \geq \alpha + \eps$ and $n$.
    For $0 \leq \tau \leq \gamma < 1$, let $k_n(\gamma)$ be the largest integer $k$ such that $\gamma^k_n \geq \gamma$, and let $t_n(\gamma, \tau)$ be the smallest integer $t$ such that at the beginning of iteration $k_n(\gamma) + t$ it holds $\gamma^{k_n(\gamma) + t}_n < \gamma - \tau$.
    Then
    \begin{enumerate}
        \item For any $0 \leq \rho < 1$ and any integer $d$, the minimum number of steps of $\gamma^k_n$ to reach $\rho$ is
        \[k_n(\rho) \geq  \sum_{r = \lceil\rho d/\beta\rceil}^{d - 1}t_n\left(\frac{\beta(r + 1)}{d}, \frac{\beta}{d}\right) - \left(d - \lceil \rho d/\beta\rceil\right).\]

        \item For any $0 \leq \tau \leq \gamma < 1$ such that $\gamma > \alpha + \tau + \eps$, the number $t_n(\gamma, \tau)$ of steps of the process $\gamma^k_n$ after the iteration $k_n(\gamma)$ for the value to go below $\gamma - \tau$ satisfies
        \[
        t_n(\gamma, \tau) > \frac{2\tau(\gamma - \tau - \eps - \alpha)n}{\gamma(1 - \alpha)} =: l_n(\gamma, \tau).\]

        \item For any $\alpha + \eps \leq \rho < 1$ and any integer $d$, it holds that $\gamma_n^{n_0}\geq \rho$ if
        \[\sum_{r = \lceil\rho d/\beta\rceil}^{d - 1}l_n\left(\frac{\beta(r + 1)}{d}, \frac{\beta}{d}\right) \geq n_0 + \left(d - \lceil \rho d/\beta\rceil\right).\]
    \end{enumerate}
\end{lemma}
\begin{proof}
    To prove 1), first assume that $\rho d/\beta$ is an integer.
    Divide the segment $[\rho, \beta]$ into $(1 - \rho/\beta)d$ smaller sub-segments of the form $[\beta r/d, \beta(r + 1)/d]$ where $\rho d  \leq r \leq d - 1$.
    Note that by definition, for any $0 \leq \tau \leq \gamma < 1$ it holds $k_n(\gamma) = k_n(\gamma + \tau) + t_n(\gamma + \tau, \tau) - 1$.
    In particular,
    \[
        k_n\left(\frac{\beta r}{d}\right) = k_n\left(\frac{\beta(r + 1)}{d}\right) + t_n\left(\frac{\beta(r + 1)}{d}, \frac{\beta}{d}\right) - 1.
    \]
    Therefore,
    \begin{multline*}
        k_n(\rho) = k_n\left(\rho + \frac{\beta}{d}\right) + t_n\left(\rho + \frac{\beta}{d}, \frac{\beta}{d}\right) - 1
        = k_n\left(\rho + \frac{2\beta}{d}\right) + t_n\left(\rho + \frac{2\beta}{d}, \frac{\beta}{d}\right) + t_n\left(\rho + \frac{\beta}{d}, \frac{\beta}{d}\right) - 2 \\
        = \ldots
    = k_n(1) + \sum_{r = \rho d/\beta}^{d - 1}t_n\left(\frac{\beta(r + 1)}{d}, \frac{\beta}{d}\right) - (d - \rho d/\beta).
    \end{multline*}
    Since $k_n(1) = 1$, we get 
    \[
        k_n(\rho) = 1 + \sum_{r = \rho d/\beta}^{d - 1}t_n\left(\frac{\beta(r + 1)}{d}, \frac{\beta}{d}\right) - (d - \rho d/\beta).
    \]

    If $\rho d/\beta$ is not an integer, we consider a segment $[\lceil \rho d/\beta\rceil\cdot \beta / d, \beta]$, divide it into $d - \lceil \rho d/\beta\rceil$ sub-segments and repeat the process above.
    This gives us 
    \[
        k_n\left(\frac{\lceil \rho d/\beta \rceil\beta}{d}\right) = 1 + \sum_{r = \lceil\rho d/\beta\rceil}^{d - 1}t_n\left(\frac{\beta(r + 1)}{d}, \frac{\beta}{d}\right) - (d - \lceil\rho d/\beta\rceil).
    \]
    We finish the proof by observing 
    \begin{multline*}
        k_n(\rho) = k_n\left(\rho + \frac{\lceil \rho d/\beta\rceil\beta}{d} - \rho\right) + t_n\left(\frac{\lceil \rho d/\beta\rceil\beta}{d},\frac{\lceil \rho d/\beta\rceil\beta}{d} - \rho \right) - 1\\
        = t_n\left(\frac{\lceil \rho d/\beta\rceil\beta}{d},\frac{\lceil \rho d/\beta\rceil\beta}{d} - \rho \right)  +  \sum_{r = \lceil\rho d/\beta\rceil}^{d - 1}t_n\left(\frac{\beta(r + 1)}{d}, \frac{\beta}{d}\right) - (d - \lceil\rho d/\beta\rceil).
    \end{multline*}

    Next, for 2), observe that by definition of $t_n(\gamma, \tau)$, for every integer $0 \leq t \leq t_n(\gamma, \tau) - 1$ it holds $\gamma^{k_n(\gamma) + t}_n \geq \gamma - \tau$.
    Then, by definition of the process $\gamma_n^k$, for every such $t$ we have
    \[
        \gamma_n^{k_n(\gamma) + t + 1} = \left(1 - \frac{1 - \alpha}{2(\gamma_n^{k_n(\gamma) + t} - \eps -  \alpha)n}\right)\cdot \gamma_n^{k_n(\gamma) + t}\\
        \geq \left(1 - \frac{1 - \alpha}{2(\gamma - \tau - \eps - \alpha)n}\right)\cdot \gamma_n^{k_n(\gamma) + t}.
    \]
    Applying the inequality above repeatedly, we get
    \begin{multline*}
        \gamma - \tau > \gamma_n^{k_n(\gamma) + t_n(\gamma, \tau)} \geq \left(1-\frac{1 - \alpha}{2(\gamma - \tau - \eps - \alpha)n}\right)^{t_n(\gamma, \tau)}\cdot \gamma_n^{k_n(\gamma)}\\
        \geq  \left(1-\frac{1 - \alpha}{2(\gamma - \tau - \eps - \alpha)n}\right)^{t_n(\gamma, \tau)}\cdot \gamma,
    \end{multline*}
    implying
    \[
        t_n(\gamma, \tau) > \frac{\ln(\frac{\gamma - \tau}{\gamma})}{\ln(1-\frac{1 - \alpha}{2(\gamma - \tau - \eps -  \alpha)n})}. 
    \]
    Now, since for all $x > - 1$ it holds $\frac{x}{1 + x} \leq \ln(1 + x) \leq x$, we have
    \begin{multline*}
        \frac{\ln(1 -\frac{\tau}{\gamma})}{\ln(1-\frac{1 - \alpha}{2(\gamma - \tau - \eps - \alpha)n})} \geq \frac{-\frac{\tau/\gamma}{1 - \tau/\gamma}}{-\frac{1 - \alpha}{2(\gamma - \tau - \eps - \alpha)n}} \\
        = \frac{2\tau(\gamma - \tau - \eps - \alpha)n}{\gamma(1 - \alpha)}\cdot \frac{1}{1 - \tau/\gamma} \geq \frac{2\tau(\gamma - \tau - \eps - \alpha)n}{\gamma(1 - \alpha)} =: l_n(\gamma, \tau).
    \end{multline*}

    To see 3), note that $k_n(\rho) \geq n_0$ implies that at iteration $n_0$ of the process $\gamma_n^k$ we have $\gamma_n^{n_0} \geq \rho$.
    Combining the expression for $k_n(\rho)$ from 1) and lower bound on $t_n(\gamma, \tau) > l_n(\gamma, \tau)$ from 2) completes the proof.
\end{proof}

Using the lower bound on the number of iterations required for $\gamma_n^k(\beta, \eps)$ to stay above $\rho$ after $n_0$ iterations, we can get a necessary condition on $n$ and $n_0$ for the inequality $\gamma_n^{n_0}(\beta, \eps)$ to hold.
To do so, we take the inequality of 3) from \cref{lowergammaiters} and perform multiple equivalent transformations.

\begin{theorem}\label{lowergamma}
    For $\eps > 0$, let $\alpha + \eps \leq \rho < 1$.
    Consider the process $\gamma_n^k:= \gamma_n^k(\beta, \eps)$ for this $\eps$, $\beta \geq \alpha + \eps$ and some $n$, and take integer $d$ such that $d \leq n$.
    If numbers $n, n_0, d$ satisfy
    \begin{equation}
        \frac{2(\beta - \rho + (\alpha + \eps)\ln (\rho/\beta))}{1 - \alpha} \geq \frac{n_0}{n} + \frac{d(1 - \rho /\beta)}{n} + \frac{2(1 - \rho /\beta)}{(\rho /\beta)(\rho d/\beta  + 1)} + \frac{2(\beta - \eps - \alpha)}{d(1 - \alpha)},\label{mainineq}
    \end{equation}
    then the value of the process $\gamma_n^k$ at iteration $n_0$ satisfies $\gamma_n^{n_0} \geq \rho$.
\end{theorem}
\begin{proof}
From \cref{lowergammaiters} we know that $\gamma_n^{n_0} \geq \rho$ if integers $n, n_0, d$ with $d \leq n$ are such that
\begin{equation}
    \sum_{r = \lceil\rho d/\beta\rceil}^{d - 1}l_n\left(\frac{\beta(r + 1)}{d}, \frac{\beta}{d}\right) \geq n_0 + \left(d - \lceil \rho d/\beta\rceil\right).\label{gammalb}
\end{equation}
Consider the expression for $l_n(\gamma, \tau)$,  derived in \cref{lowergammaiters}, with $\gamma = \frac{\beta(r + 1)}{d}$ and $\tau = \frac{\beta}{d}$:
\begin{multline*}
    l_n\left(\frac{\beta(r + 1)}{d}, \frac{\beta}{d}\right) = \frac{2\cdot \frac{\beta}{d}\cdot (\beta r/d -\eps - \alpha)n}{\frac{\beta(r + 1)}{d}\cdot (1 - \alpha)} = \frac{2(\beta r/d -\eps - \alpha)n}{(r + 1)(1 - \alpha)} \\= \frac{2(\beta r/d - \eps - \alpha)n}{r(1 - \alpha)}\cdot \left(1 - \frac{1}{r + 1}\right). 
\end{multline*}
We have $l_n(\frac{\beta(r + 1)}{d}, \frac{\beta}{d}) = \frac{2(\beta r/d - \eps - \alpha)n}{r(1 - \alpha)} - \frac{2(\beta r/d - \eps - \alpha)n}{r(r + 1)(1 - \alpha)}$.
Observe that
\[\sum_{r = \lceil\rho d /\beta \rceil}^{d - 1}\frac{2(\beta r/d -\eps - \alpha)n}{r(r + 1)(1 - \alpha)} \leq \frac{2n(\lceil\rho d /\beta \rceil\beta /d - \eps - \alpha)(d - \lceil\rho d/\beta\rceil)}{\lceil \rho d/\beta \rceil (\lceil \rho d/\beta \rceil + 1)(1 - \alpha)} \leq \frac{2n(d - \lceil \rho d /\beta\rceil)}{\lceil \rho d/\beta\rceil (\lceil \rho d/\beta \rceil + 1)}.\]
Hence, if
\[\sum_{r = \lceil \rho d/\beta\rceil}^{d - 1}\frac{2(\beta r / d - \eps - \alpha)n}{r(1 - \alpha)} \geq n_0 + (d - \lceil \rho d/\beta \rceil) + \frac{2n(d - \lceil \rho d/\beta \rceil)}{\lceil \rho d/\beta\rceil (\lceil \rho d/\beta\rceil + 1)},\]
then \cref{gammalb} holds.
Since $\rho d / \beta \leq \lceil \rho d/\beta\rceil$, we have 
\[
    (d - \lceil \rho d/\beta \rceil) + \frac{2n(d - \lceil\rho d /\beta \rceil)}{\lceil \rho d/\beta \rceil (\lceil \rho d/\beta \rceil + 1)} \leq (d - \rho d/\beta) + \frac{2n(d - \rho d /\beta)}{(\rho d/\beta)(\rho d/\beta  + 1)}
\]
Additionally, since
\[\left.\frac{2(\beta r / d -\eps - \alpha)n}{r(1 - \alpha)}\right|_{r = d} = \frac{2n(\beta -\eps - \alpha)}{d(1 - \alpha)},\]
it is easy to see that \cref{gammalb} also follows from inequality
\begin{multline*}
    \sum_{r = \lceil \rho d/\beta\rceil}^{d}\frac{2(\beta r - (\alpha + \eps) d)n}{rd(1 - \alpha)} = \sum_{r = \lceil \rho d/\beta\rceil}^{d}\frac{2(\beta r / d - \eps -\alpha)n}{r(1 - \alpha)} \\
    \geq n_0 + d(1 - \rho /\beta) + \frac{2n(1 - \rho  /\beta)}{(\rho /\beta)(\rho d/\beta  + 1)} + \frac{2n(\beta -\eps - \alpha)}{d(1 - \alpha)}.
\end{multline*}
Note that for values of $\lceil \rho d /\beta \rceil \leq r \leq d$, the function $f(r) := \frac{2(\beta r - (\alpha + \eps) d)n}{rd(1 - \alpha)}$ is non-decreasing:
\[\frac{f(r + 1)}{f(r)} = \frac{\frac{2(\beta r + \beta - (\alpha + \eps) d)n}{(r + 1)d(1 - \alpha)}}{\frac{2(\beta r - (\alpha + \eps) d)n}{rd(1 - \alpha)}} = \frac{(\beta r - (\alpha + \eps) d + \beta )r}{(\beta r - (\alpha + \eps) d)(r + 1)} \geq 1 \iff r \geq (\alpha + \eps) d/\beta.\]
Therefore, it holds
\[\sum_{r = \lceil \rho d/\beta \rceil}^{d}\frac{2(\beta r - (\alpha + \eps) d)n}{rd(1 - \alpha)} \geq \int_{\lceil \rho d/\beta  \rceil - 1}^d\frac{2(\beta r - (\alpha + \eps) d)n}{rd(1 - \alpha)} \geq \int_{\rho d /\beta }^d\frac{2(\beta r - (\alpha + \eps) d)n}{rd(1 - \alpha)},\]
where the last inequality follows from $\lceil \rho d/\beta \rceil - 1 \leq \rho d/\beta $ and shrinking the integration domain.

Putting all of the above together and dividing by $n$, we get that if 
\[\int_{\rho d/\beta}^d\frac{2(\beta r - (\alpha + \eps) d)}{rd(1 - \alpha)} \geq \frac{n_0}{n} + \frac{d(1 - \rho /\beta)}{n} + \frac{2(1 - \rho  /\beta)}{(\rho /\beta)(\rho d/\beta  + 1)} + \frac{2(\beta - \eps - \alpha)}{d(1 - \alpha)},\]
then \cref{gammalb} holds, which by \cref{lowergammaiters} implies $\gamma_n^{n_0} \geq \rho$.
Taking the integral, we get
\begin{multline*}
    \int_{\rho d/\beta}^d\frac{2(\beta r - (\alpha + \eps) d)}{rd(1 - \alpha)} = \left.\frac{2(\beta r - (\alpha + \eps) d\ln r)}{d(1 - \alpha)}\right|_{r = \rho d/\beta }^d \\
    = \frac{2(\beta d - (\alpha + \eps) d\ln d)}{d(1 - \alpha)} - \frac{2(\rho d - (\alpha + \eps) d \ln d - (\alpha + \eps) d\ln (\rho/\beta))}{d(1 - \alpha)} \\
    = \frac{2(\beta - \rho + (\alpha + \eps)\ln (\rho/\beta))}{1 - \alpha}.
\end{multline*}
This finishes the proof.
\end{proof}

The proof of \cref{diffgammaproof} follows directly from the following refinement of \cref{lowergamma}.
\begin{corollary}\label{lowergammabetter}
    For $\eps > 0$, let $\alpha \leq \rho < 1 - \eps$.
    Consider the process $\gamma_n^k(1 - \delta, \eps)$ for this $\eps$, $0 < \delta \leq 1 - (\alpha + \eps)$ and some $n$, and take integer $d$ such that $d \leq n$.
    Then, $\gamma_n^{n_0}(1 - \delta, \eps) \geq \rho + \eps$ if $n, n_0, d$ satisfy
    \begin{equation}
        \frac{2(1 - \rho + \alpha\ln \rho)}{1 - \alpha} 
        \geq \frac{n_0}{n} +\frac{6\eps}{1 - \alpha} + 2\delta\left(1 - \frac{\eps}{1 - \alpha}\right) +  \frac{d(1 - \rho)}{n} + \frac{2(1 - \rho)}{\rho (\rho d  + 1)} + \frac{2(1 - \delta - \eps - \alpha)}{d(1 - \alpha)}.\label{rhoeq}
    \end{equation}
\end{corollary}
\begin{proof}
    Consider the expression obtained in \cref{lowergamma} for some $\beta$ and $\rho + \eps$ instead of $\rho$.
    Observe that
    \begin{multline*}
        \frac{2(\beta - \rho - \eps + (\alpha + \eps)\ln ((\rho + \eps)/\beta))}{1 - \alpha}
        \\= \frac{2(1 - \rho + \alpha \ln (\rho + \eps))}{1 - \alpha} - \frac{2\eps}{1 - \alpha}+ \frac{2\eps\ln(\rho + \eps)}{1 - \alpha}- \frac{2(1 - \beta + (\alpha + \eps)\ln \beta)}{1 - \alpha} .
    \end{multline*}
    Since $\rho > \alpha \geq e^{-2}$, then $\ln (\rho + \eps) \geq \ln\rho > -2$ and $\frac{2\eps\ln(\rho + \eps)}{1 - \alpha} \geq -\frac{4\eps}{1 - \alpha}$.
    In addition, when $\beta= 1- \delta$,
    \[
        \frac{2(1 - \beta + (\alpha + \eps)\ln \beta)}{1 - \alpha} = \frac{2(\delta + (\alpha + \eps)\ln (1 - \delta))}{1 - \alpha} \leq \frac{2(\delta - (\alpha + \eps)\delta)}{1 - \alpha} = 2\delta\left(1 - \frac{\eps}{1 - \alpha}\right).
    \]
    Finally, as $(\rho + \eps)/(1 - \delta) \geq \rho$, we have
    \[
        \frac{d(1 - (\rho + \eps)/(1 - \delta))}{n} + \frac{2(1 - (\rho + \eps)/(1 - \delta))}{(\rho + \eps)/(1 - \delta) ((\rho + \eps) d/(1 - \delta)  + 1)} \leq \frac{d(1 - \rho)}{n} + \frac{2(1 - \rho)}{\rho (\rho d  + 1)} .
    \]
    Combining all of the above and noting $\alpha \ln (\rho + \eps) \geq \alpha \ln \rho$, \cref{lowergamma} holds if $n, d$ satisfy
    \[
        \frac{2(1 - \rho + \alpha\ln \rho)}{1 - \alpha} 
        \geq \frac{n_0}{n} +\frac{6\eps}{1 - \alpha} + 2\delta\left(1 - \frac{\eps}{1 - \alpha}\right) +  \frac{d - \rho d}{n} + \frac{2(d - \rho d)}{\rho d(\rho d  + 1)} + \frac{2(1 - \delta - \eps - \alpha)}{d(1 - \alpha)}.
    \]
\end{proof}

\begin{proof}[Proof of \cref{diffgammaproof}]
    First, observe that for all $\alpha$, as long as $\rho  \geq \alpha$, the lefthand side of \cref{rhoeq} in \cref{lowergammabetter} is a decreasing function of $\rho$.
    Indeed, the derivative of $\frac{2(1 - \rho + \alpha \ln \rho)}{1 - \alpha}$ is $ \frac{2(\alpha - \rho)}{\rho(1 - \alpha)}$, which is non-positive for all $\rho \geq \alpha$.
    Then, if $\rho^*$ is a solution to $\frac{2(1 - \rho^* + \alpha\ln \rho^*)}{1 - \alpha} = 1$, for any constant $0 < x < \rho^* - \alpha$ it holds for $\rho := \alpha + x$ that $\frac{2(1 - \rho + \alpha \ln \rho)}{1 - \alpha} > 1 + \const$.
    
    At the same time, when $n_0 = n$ the righthand side of \cref{rhoeq} in \cref{lowergammabetter} is at least $1$.
    Since $\rho < 1$, both $\rho d$ and $d - \rho d$ are of order $\Theta(d)$, and taking any integer $d = \Theta(\sqrt{n})$ in
    \[ \frac{d(1 - \rho)}{n} + \frac{2(1 - \rho)}{\rho (\rho d  + 1)} + \frac{2(1 - \delta - \eps - \alpha)}{d(1 - \alpha)}\]
    makes the expression into a decreasing function of $n$, specifically of order $\Theta(n^{-1/2})$.
    But then, with this value $d$, if $\eps, \delta \xrightarrow{n\to\infty}0$, when $n_0 = n$ the righthand side of \cref{rhoeq} converges to $1$ as $n \to \infty$.
    Thus, for inequality to be satisfied for all $n_0 \leq n$ it suffices to take $n_x$ large enough such that the righthand side of \cref{rhoeq} with $n \geq n_x$ is at most $\frac{2(1 - \rho + \alpha \ln \rho)}{1 - \alpha}$.
    Since $\frac{2(1 - \rho + \alpha \ln \rho)}{1 - \alpha} > 1$, such $n_x$ must exist.
\end{proof}

\section{Lower bound on $\gamma_n^n(\beta, \eps)$ for $\alpha \geq 3/11$}
\label{sec:bigalpha}

Recall the definition of the process $\gamma_n^k(\beta, \eps)$ for $\alpha \geq 3/11$.
\begin{definition}
    For $n$ agents, $3/11 \leq \alpha < 1/3$, $\eps > 0$ and $\beta > \alpha + \eps$, let $\gamma^k_n(\beta, \eps)$ be the following process: $\gamma^1_n(\beta, \eps) = \beta$, and for $k \geq 1$,
    \begin{gather*}
        \text{if }\;\gamma^k_n(\beta, \eps) - \eps \geq 3\alpha \;\text{ then }\;
        \gamma^{k + 1}_n(\beta, \eps) = \left(1 - \frac{2(1 - 3\alpha)}{(\gamma^k_n(\beta, \eps) - \eps -12\alpha +3)n}\right)\gamma^k_n(\beta, \eps);\\
        \text{if }\; \gamma^k_n(\beta, \eps) - \eps < 3\alpha \;\text{ then }\; \gamma^{k + 1}_n(\beta, \eps) = \left(1 - \frac{4\alpha}{3(\gamma^k_n(\beta, \eps) - \eps - \alpha)n}\right)\gamma^k_n(\beta, \eps).
    \end{gather*}
\end{definition}

In this section, we prove \cref{diffgamma2}, restated below.
Similarly to how it was with $\alpha \leq 3/11$, the claim for the same-valuation case (\cref{bigrhoinf}) follows by setting $\eps = \delta = 0$.
\begin{theorem}
    Let $\rho$ satisfy $2(12\rho - 3)\ln(3\rho) = (1 - 3\rho)(3\ln 3 - 4)$, and assume $\alpha \in [3/11, \rho)$.
    For any constant $0 < x < \rho - \alpha$ and any $\eps, \delta \xrightarrow{n\to\infty}0$, there exists $n_x$ such that for all $n \geq n_x$ and any $n_0 \leq n$ it holds $\gamma_n^{n_0}(1 -\delta, \eps) - \eps \geq \alpha + x$.
\end{theorem}

The structure of the proof is essentially the same as in \cref{lowergammaiters}, \cref{lowergamma} and \cref{lowergammabetter} for the case $\alpha \leq 3/11$, with only difference that now one has to consider regimes above and below $3\alpha$.
As before, we first consider arbitrary $\beta \geq \alpha +\eps$ and $\rho$, and obtain a necessary condition on $n$ so that $\gamma_n^{n_0}(\beta, \eps) \geq \rho$.
\begin{theorem}\label{lowergamma2}
    For $\eps > 0$, let $\alpha +\eps \leq \rho < 1$.
    Consider the process $\gamma_n^k := \gamma_n^k(\beta, \eps)$ for this $\eps$, $0 < \delta \leq 1 - (\alpha + \eps)$ and some $n$, and take integer $d$ such that $d \leq n$.
    Then, $\gamma_n^{n_0} \geq \rho$ if $n, n_0, d$ satisfy
    \begin{multline}
        \frac{3(3\alpha - \rho - \alpha\ln(3\alpha)  + \alpha\ln \rho)}{4\alpha} + \frac{1 - 3\alpha + (12\alpha - 3)\ln(3\alpha )}{2(1 - 3\alpha)} \\
        \geq \frac{n_0}{n} + \eps\left(\frac{1}{\alpha} +  \frac{1 - \ln(3\alpha)}{2(1 - 3\alpha)}\right) + \frac{1 - \beta + (12\alpha - 3)\ln \beta}{2(1 - 3\alpha)} + \frac{d(1 - \rho /\beta)}{n}\\
         + \frac{3((3\alpha + \eps)/\beta - \rho /\beta + 1/d)}{2(\rho /\beta) ( \rho d/\beta + 1)} +\frac{3(1 -  (3\alpha + \eps) /\beta )}{2( (3\alpha + \eps) /\beta)((3\alpha + \eps) d/\beta + 1)} \\
        +\frac{3\beta (\beta / d + 4\alpha)}{4(3\alpha + \eps)\alpha d} +  \frac{\beta - \eps - 12\alpha + 3}{2d(1 - 3\alpha)}.\label{bigalphaineq}
    \end{multline}
\end{theorem}
\begin{proof}
    By the first item of \cref{lowergammaiters} (which holds for all values of $\alpha$), the minimum number of steps of $\gamma_n^k$ to reach $\rho$ is
    \[
        k_n(\rho) \geq \sum_{r = \lceil \rho d/\beta \rceil}^{d - 1}t_n\left(\frac{\beta(r + 1)}{d}, \frac{\beta}{d}\right) - (d - \lceil \rho d / \beta \rceil). 
    \]
    Observe that
    \[
        \sum_{r = \lceil \rho d/\beta \rceil}^{d - 1}t_n\left(\frac{\beta(r + 1)}{d}, \frac{\beta}{d}\right) = \sum_{r = \lceil \rho d/\beta \rceil}^{\lceil (3\alpha + \eps)d/\beta\rceil - 1}t_n\left(\frac{\beta(r + 1)}{d}, \frac{\beta}{d}\right) + \sum_{r = \lceil (3\alpha + \eps) d/\beta \rceil}^{d - 1}t_n\left(\frac{\beta(r + 1)}{d}, \frac{\beta}{d}\right).
    \]
    It is easy to see that
    \begin{gather*}
        \text{if }\; \lceil (3\alpha + \eps) d/\beta \rceil \leq r \leq d - 1,\;\text{ then }\;  \frac{\beta(r + 1)}{d} - \frac{\beta}{d} \geq \frac{\beta\lceil (3\alpha + \eps) d/\beta \rceil)}{d} \geq 3\alpha + \eps;\\
        \text{if }\;\lceil \rho d/\beta \rceil \leq r \leq \lceil (3\alpha + \eps) d/\beta \rceil - 1,\;\text{ then }\;\frac{\beta(r + 1)}{d} - \frac{\beta}{d} \leq \frac{\beta\lceil (3\alpha + \eps) d/\beta \rceil - 1)}{d} \leq 3\alpha + \eps.
    \end{gather*}
    Consider $t_n(\gamma, \tau)$ for $\gamma - \tau \geq 3\alpha + \eps$.
    For this range of values, as was done in \cref{lowergammaiters},
    \begin{multline*}
        \gamma - \tau \geq \gamma_n^{k_n(\gamma) + t_n(\gamma, \tau)} \geq \left(1 - \frac{2(1 - 3\alpha)}{(\gamma - \tau - \eps - 12\alpha + 3)n}\right)^{t_n(\gamma, \tau)}\cdot \gamma_n^{k_n(\gamma)}\\
        \geq \left(1 - \frac{2(1 - 3\alpha)}{(\gamma - \tau - \eps - 12\alpha + 3)n}\right)^{t_n(\gamma, \tau)}\cdot \gamma,
    \end{multline*}
    implying
    \begin{multline*}
         t_n(\gamma, \tau) > \frac{\ln(\frac{\gamma - \tau}{\gamma})}{\ln(1 - \frac{2(1 - 3\alpha)}{(\gamma - \tau - \eps - 12\alpha + 3)n})}\geq \frac{-\frac{\tau/\gamma}{1 - \tau/\gamma}}{-\frac{2(1 - 3\alpha)}{(\gamma - \tau - \eps - 12\alpha + 3)n}}\\
         = \frac{\tau(\gamma - \tau - \eps - 12\alpha + 3)n}{2\gamma (1 - 3\alpha)}\cdot \frac{1}{1 - \tau/\gamma} \geq \frac{\tau(\gamma - \tau - \eps - 12\alpha + 3)n}{2\gamma (1 - 3\alpha)}=:l_n^>(\gamma, \tau).
    \end{multline*}
    Similarly, consider $t_n(\gamma, \tau)$ for $\gamma - \tau\leq 3\alpha + \eps$.
    For this range of values it holds
    \begin{multline*}
        \gamma - \tau \geq \gamma_n^{k_n(\gamma) + t_n(\gamma, \tau)}\geq \left(1 - \frac{4\alpha}{3(\gamma - \tau - \eps - \alpha)n}\right)^{t_n(\gamma, \tau)}\cdot \gamma_n^{k_n(\gamma)}\\
        \geq \left(1 - \frac{4\alpha}{3(\gamma - \tau - \eps - \alpha)n}\right)^{t_n(\gamma, \tau)}\cdot \gamma,
    \end{multline*}
    implying
    \begin{multline*}
        t_n(\gamma, \tau) > \frac{\ln(\frac{\gamma - \tau}{\gamma})}{\ln(1 - \frac{4\alpha}{3(\gamma - \tau - \eps - \alpha)n})} \geq \frac{-\frac{\tau/\gamma}{1 - \tau/\gamma}}{-\frac{4\alpha}{3(\gamma - \tau - \eps - \alpha)n}}\\
        = \frac{3\tau(\gamma - \tau - \eps - \alpha)n}{4\gamma\alpha}\cdot \frac{1}{1 - \tau/\gamma} \geq \frac{3\tau(\gamma - \tau - \eps - \alpha)n}{4\gamma\alpha} = : l_n^<(\gamma, \tau).
    \end{multline*}
    As a result,
    \[
        k_n(\rho) \geq \sum_{r = \lceil \rho d/\beta \rceil}^{\lceil (3\alpha + \eps)d/\beta\rceil - 1}l_n^<\left(\frac{\beta(r + 1)}{d}, \frac{\beta}{d}\right) + \sum_{r = \lceil (3\alpha + \eps) d/\beta \rceil}^{d - 1}l_n^>\left(\frac{\beta(r + 1)}{d}, \frac{\beta}{d}\right) - (d - \lceil \rho d/\beta\lceil).
    \]
    Now, we substitute $\gamma = \frac{\beta(r + 1)}{d}$ and $\tau = \frac{\beta}{d}$ into $l_n^<(\gamma, \tau)$ and $l_n^>(\gamma, \tau)$.
    \begin{gather*}
        l_n^<\left(\frac{\beta(r + 1)}{d}, \frac{\beta}{d}\right) = \frac{3\cdot \frac{\beta}{d}\cdot (\beta r/d - \eps - \alpha)n}{4\cdot \frac{\beta(r + 1)}{d}\cdot \alpha} = \frac{3(\beta r / d - \eps - \alpha)n}{4r\alpha}\cdot \left(1 - \frac{1}{r + 1}\right),\\
        l_n^>\left(\frac{\beta(r + 1)}{d}, \frac{\beta}{d}\right)  = \frac{\frac{\beta}{d}\cdot (\beta r / d - \eps - 12\alpha + 3)n}{2\cdot \frac{\beta(r + 1)}{d}\cdot  (1 - 3\alpha)} = \frac{(\beta r / d - \eps - 12\alpha + 3)n}{2r(1 - 3\alpha)}\cdot \left(1 - \frac{1}{r + 1}\right).
    \end{gather*}
    As was done in \cref{lowergamma}, we observe that
    \begin{multline*}
        \sum_{r = \lceil \rho d/\beta \rceil}^{\lceil (3\alpha + \eps)d/\beta\rceil - 1}\frac{3(\beta r / d - \eps - \alpha)n}{4r(r + 1)\alpha}\leq \frac{3n(\lceil (3\alpha + \eps)d/\beta\rceil - \lceil \rho d/\beta \rceil)}{2\lceil \rho d/\beta \rceil(\lceil \rho d/\beta \rceil + 1)} \\\leq \frac{3n((3\alpha + \eps)/\beta - \rho /\beta + 1/d)}{2(\rho /\beta) ( \rho d/\beta + 1)},
    \end{multline*}
    and
    \begin{multline*}
        \sum_{r = \lceil (3\alpha + \eps) d/\beta \rceil}^{d - 1}\frac{(\beta r / d - \eps - 12\alpha + 3)n}{2r(r + 1)(1 - 3\alpha)} \leq \frac{3n(d - \lceil (3\alpha + \eps) d/\beta \rceil)}{2\lceil (3\alpha + \eps) d/\beta \rceil(\lceil (3\alpha + \eps) d/\beta \rceil + 1)}\\
        \leq \frac{3n(1 -  (3\alpha + \eps) /\beta )}{2( (3\alpha + \eps) /\beta)((3\alpha + \eps) d/\beta + 1)}.
    \end{multline*}
    Combining all of the above with $\lceil \rho d/\beta \rceil \geq \rho d / \beta$, we get that $\gamma_n^{n_0} \geq \rho$ if $n, n_0, d$ satisfy
    \begin{multline*}
        \sum_{r = \lceil \rho d/\beta \rceil}^{\lceil (3\alpha + \eps)d/\beta\rceil - 1}\frac{3(\beta r / d - \eps - \alpha)n}{4r\alpha} + \sum_{r = \lceil (3\alpha + \eps) d/\beta \rceil}^{d - 1}\frac{(\beta r / d - \eps - 12\alpha + 3)n}{2r(1 - 3\alpha)}\\ \geq n_0 + d(1 - \rho /\beta) + \frac{3n((3\alpha + \eps)/\beta - \rho /\beta + 1/d)}{2(\rho /\beta) ( \rho d/\beta + 1)} +\frac{3n(1 -  (3\alpha + \eps)/\beta )}{2( (3\alpha + \eps) /\beta)((3\alpha + \eps) d/\beta + 1)} .
    \end{multline*}
    Next,
    \begin{gather*}
        \left.\frac{3(\beta r / d - \eps - \alpha)n}{4r\alpha}\right|_{r = \lceil (3\alpha + \eps)d/\beta\rceil} \leq \frac{3\beta (\beta / d + 4\alpha)n}{4(3\alpha + \eps)\alpha d};\\\left. \frac{(\beta r / d - \eps - 12\alpha + 3)n}{2r(1 - 3\alpha)}\right|_{r = d} = \frac{(\beta - \eps - 12\alpha + 3)n}{2d(1 - 3\alpha)},
    \end{gather*}
    so we get an equivalent condition (after dividing by $n$)
    \begin{multline*}
        \sum_{r = \lceil \rho d/\beta \rceil}^{\lceil (3\alpha + \eps)d/\beta\rceil}\frac{3(\beta r / d - \eps - \alpha)}{4r\alpha} + \sum_{r = \lceil (3\alpha + \eps) d/\beta \rceil}^{d}\frac{\beta r / d - \eps - 12\alpha + 3}{2r(1 - 3\alpha)}\\ \geq \frac{n_0}{n}+ \frac{d(1 - \rho /\beta)}{n} + \frac{3((3\alpha + \eps)/\beta - \rho /\beta + 1/d)}{2(\rho /\beta) ( \rho d/\beta + 1)} +\frac{3(1 -  (3\alpha + \eps) /\beta )}{2( (3\alpha + \eps) /\beta)((3\alpha + \eps) d/\beta + 1)} \\
        +\frac{3\beta (\beta / d + 4\alpha)}{4(3\alpha + \eps)\alpha d} +  \frac{\beta - \eps - 12\alpha + 3}{2d(1 - 3\alpha)}.
    \end{multline*}
    Next, since both $f(r) := \frac{3(\beta r - (\alpha + \eps)d)}{4dr\alpha}$ and $g(r) := \frac{\beta r - (12\alpha  -3 + \eps)d}{2dr(1 - 3\alpha)}$ are decreasing functions of $r$ on respective ranges, and by shrinking the integration domain, we have
    \begin{gather*}
        \sum_{r = \lceil \rho d/\beta \rceil}^{\lceil (3\alpha + \eps)d/\beta\rceil}\frac{3(\beta r / d - \eps - \alpha)}{4r\alpha} \geq \int_{ \rho d/\beta }^{(3\alpha + \eps)d/\beta}\frac{3(\beta r - (\alpha + \eps)d)}{4dr\alpha};\\
        \sum_{r = \lceil (3\alpha + \eps) d/\beta \rceil}^{d}\frac{\beta r / d - \eps - 12\alpha + 3}{2r(1 - 3\alpha)} \geq \int_{(3\alpha + \eps) d/\beta}^d\frac{\beta r - (12\alpha - 3 + \eps)d}{2dr(1 - 3\alpha)}.
    \end{gather*}
    Taking the integrals:
    \begin{multline*}
         \int_{ \rho d/\beta }^{(3\alpha + \eps)d/\beta}\frac{3(\beta r - (\alpha + \eps)d)}{4dr\alpha} = \left. \frac{3(\beta r - (\alpha + \eps)d\ln r)}{4\alpha d}\right|_{r = \rho d / \beta}^{(3\alpha + \eps)d/\beta}\\
         = \frac{3((3\alpha + \eps) - (\alpha + \eps)\ln(3\alpha + \eps) - (\alpha + \eps)\ln d + (\alpha + \eps)\ln \beta)}{4\alpha}\\
         - \frac{3(\rho - (\alpha + \eps)\ln \rho - (\alpha + \eps)\ln d + (\alpha + \eps)\ln \beta)}{4\alpha}\\
         = \frac{3((3\alpha + \eps) - \rho - (\alpha + \eps)\ln(3\alpha + \eps) + (\alpha + \eps)\ln \rho)}{4\alpha},
    \end{multline*}
    and
    \begin{multline*}
        \int_{(3\alpha + \eps) d/\beta}^d\frac{\beta r - (12\alpha - 3 + \eps)d}{2dr(1 - 3\alpha)} = \left. \frac{\beta r - (12\alpha - 3 + \eps)d\ln r}{2d(1 - 3\alpha)}\right|_{r = (3\alpha + \eps)d/\beta}^d\\
        = \frac{\beta  - (12\alpha - 3 + \eps)\ln d}{2(1 - 3\alpha)}\\ - \frac{(3\alpha + \eps) - (12\alpha - 3 + \eps)\ln (3\alpha + \eps) -(12\alpha - 3 + \eps)\ln d + (12\alpha - 3 + \eps)\ln \beta }{2(1 - 3\alpha)}\\
        = \frac{\beta - (3\alpha + \eps) + (12\alpha - 3 + \eps)\ln(3\alpha + \eps) - (12\alpha - 3 + \eps)\ln \beta}{2(1 - 3\alpha)}.
    \end{multline*}
    Now, $\ln(3\alpha) \leq \ln(3\alpha + \eps) \leq \ln(3\alpha) + \frac{\eps}{3\alpha}$ and $\ln(\rho) \geq -2$ and $\ln(3\alpha + \eps) \leq 0$, so
    \begin{multline*}
        \frac{3((3\alpha + \eps) - \rho - (\alpha + \eps)\ln(3\alpha + \eps) + (\alpha + \eps)\ln \rho)}{4\alpha} \\
        \geq  \frac{3(3\alpha - \rho - \alpha\ln(3\alpha)  + \alpha\ln \rho)}{4\alpha} + \frac{3(2\eps/3 - \eps\ln(3\alpha + \eps) + \eps\ln \rho )}{4\alpha}\\
        \geq  \frac{3(3\alpha - \rho - \alpha\ln(3\alpha)  + \alpha\ln \rho)}{4\alpha}  - \frac{\eps}{\alpha},
    \end{multline*}
    and
    \begin{multline*}
        \frac{\beta - (3\alpha + \eps) + (12\alpha - 3 + \eps)\ln(3\alpha + \eps) - (12\alpha - 3 + \eps)\ln \beta}{2(1 - 3\alpha)}\\
        \geq \frac{\beta - 3\alpha + (12\alpha - 3)\ln(3\alpha ) - (12\alpha - 3)\ln \beta}{2(1 - 3\alpha)} + \frac{\eps\ln(3\alpha + \eps)- \eps \ln \beta - \eps }{2(1 - 3\alpha)}\\
        \geq \frac{1 - 3\alpha + (12\alpha - 3)\ln(3\alpha )}{2(1 - 3\alpha)} - \frac{1 - \beta + (12\alpha - 3)\ln \beta}{2(1 - 3\alpha)}-\frac{\eps(1 - \ln(3\alpha))}{2(1 - 3\alpha)}.
    \end{multline*}
    So, we have an equivalent condition on $n, d$:
    \begin{multline*}
        \frac{3(3\alpha - \rho - \alpha\ln(3\alpha)  + \alpha\ln \rho)}{4\alpha} + \frac{1 - 3\alpha + (12\alpha - 3)\ln(3\alpha )}{2(1 - 3\alpha)} \\
        \geq \frac{n_0}{n} + \frac{\eps}{\alpha} +  \frac{\eps(1 - \ln(3\alpha))}{2(1 - 3\alpha)} + \frac{1 - \beta + (12\alpha - 3)\ln \beta}{2(1 - 3\alpha)} + \frac{d(1 - \rho /\beta)}{n}\\
         + \frac{3((3\alpha + \eps)/\beta - \rho /\beta + 1/d)}{2(\rho /\beta) ( \rho d/\beta + 1)} +\frac{3(1 -  (3\alpha + \eps) /\beta )}{2( (3\alpha + \eps) /\beta)((3\alpha + \eps) d/\beta + 1)} \\
        +\frac{3\beta (\beta / d + 4\alpha)}{4(3\alpha + \eps)\alpha d} +  \frac{\beta - \eps - 12\alpha + 3}{2d(1 - 3\alpha)}.
    \end{multline*}
\end{proof}

Similarly to \cref{lowergammabetter}, we obtain a refinement of \cref{lowergamma2} for $\beta = 1- \delta$ and $\rho + \eps$.
\begin{corollary}\label{lowergamma2better}
    For $\eps > 0$, let $\alpha \leq \rho < 1 - \eps$.
    Consider the process $\gamma_n^k(1 - \delta, \eps)$ for this $\eps$, $0 < \delta \leq 1 - (\alpha + \eps)$ and some $n$, and take integer $d$ such that $d\leq n$.
    Then, $\gamma_n^{n_0}(1 - \delta, \eps) \geq \rho + \eps$ if $n, n_0, d$ satisfy
    \begin{multline}
        \frac{3(3\alpha - \rho - \alpha\ln(3\alpha)  + \alpha\ln \rho)}{4\alpha} + \frac{1 - 3\alpha + (12\alpha - 3)\ln(3\alpha )}{2(1 - 3\alpha)} \\
        \geq \frac{n_0}{n} + \eps\left(\frac{5}{4\alpha} +  \frac{(1 - \ln(3\alpha))}{2(1 - 3\alpha)}\right) + 2\delta + \frac{d(1 - \rho)}{n}\\
         + \frac{3(3 - \rho  + 1/d)}{2\rho  ( \rho d + 1)} +\frac{3(1 -  (3\alpha + \eps))}{2 (3\alpha + \eps)((3\alpha + \eps) d + 1)} \\
        +\frac{3(1 / d + 4\alpha)}{4(3\alpha + \eps)\alpha d} +  \frac{4 -\delta - \eps - 12\alpha}{2d(1 - 3\alpha)}.\label{bigalphaineqbetter}
    \end{multline}
\end{corollary}
\begin{proof}
    First, replacing $\beta = 1- \delta$ and using $\ln(1 - \delta) \leq - \delta$ and $1 \leq 1/(1 - \delta)\leq 1/(\alpha + \eps)$, from \cref{bigalphaineq} we get a condition for any $\alpha + \eps < \rho' < 1$:
    \begin{multline*}
        \frac{3(3\alpha - \rho' - \alpha\ln(3\alpha)  + \alpha\ln \rho')}{4\alpha} + \frac{1 - 3\alpha + (12\alpha - 3)\ln(3\alpha )}{2(1 - 3\alpha)} \\
        \geq \frac{n_0}{n} + \eps\left(\frac{1}{\alpha} +  \frac{(1 - \ln(3\alpha))}{2(1 - 3\alpha)}\right) + 2\delta + \frac{d(1 - \rho')}{n}\\
         + \frac{3(3 - \rho'  + 1/d)}{2\rho'  ( \rho' d + 1)} +\frac{3(1 -  (3\alpha + \eps))}{2 (3\alpha + \eps)((3\alpha + \eps) d + 1)} \\
        +\frac{3(1 / d + 4\alpha)}{4(3\alpha + \eps)\alpha d} +  \frac{4 -\delta - \eps - 12\alpha}{2d(1 - 3\alpha)}.
    \end{multline*}
    Next, we substitute $\rho + \eps$ instead of $\rho'$ and make the following observations:
    \[
        \frac{d(1 - \rho - \eps)}{n}+ \frac{3(3 - \rho - \eps  + 1/d)}{2(\rho + \eps)  ( (\rho + \eps) d + 1)} \leq \frac{d(1 - \rho)}{n}+ \frac{3(3 - \rho  + 1/d)}{2\rho  ( \rho d + 1)},
    \]
    and
    \[
        \frac{3(3\alpha - \rho - \eps - \alpha\ln(3\alpha)  + \alpha\ln (\rho + \eps))}{4\alpha} \geq \frac{3(3\alpha - \rho - \alpha\ln(3\alpha)  + \alpha\ln \rho)}{4\alpha} - \frac{\eps}{4\alpha}.
    \]
    We conclude that $\gamma_n^n(1 - \delta, \eps) \geq \rho + \eps$ if $n, d$ satisfy
    \begin{multline*}
        \frac{3(3\alpha - \rho - \alpha\ln(3\alpha)  + \alpha\ln \rho)}{4\alpha} + \frac{1 - 3\alpha + (12\alpha - 3)\ln(3\alpha )}{2(1 - 3\alpha)} \\
        \geq \frac{n_0}{n} + \eps\left(\frac{5}{4\alpha} +  \frac{(1 - \ln(3\alpha))}{2(1 - 3\alpha)}\right) + 2\delta + \frac{d(1 - \rho)}{n}\\
         + \frac{3(3 - \rho  + 1/d)}{2\rho  ( \rho d + 1)} +\frac{3(1 -  (3\alpha + \eps))}{2 (3\alpha + \eps)((3\alpha + \eps) d + 1)} \\
        +\frac{3(1 / d + 4\alpha)}{4(3\alpha + \eps)\alpha d} +  \frac{4 -\delta - \eps - 12\alpha}{2d(1 - 3\alpha)}.
    \end{multline*}
\end{proof}

\begin{proof}[Proof of \cref{diffgamma2}]
    First, observe that for all $\alpha \geq 3/11$, as long as $\rho \geq \alpha$, the lefthand side of \cref{bigalphaineqbetter} in \cref{lowergamma2better} is a decreasing function of $\rho$.
    Indeed, the derivative of $\frac{3(3\alpha - \rho - \alpha\ln(3\alpha)  + \alpha\ln \rho)}{4\alpha} + \frac{1 - 3\alpha + (12\alpha - 3)\ln(3\alpha )}{2(1 - 3\alpha)}$ is $\frac{3(\alpha/\rho - 1)}{4\alpha}$, which is non-positive for all $\rho \geq \alpha$.
    Then, if $\rho^* \geq \alpha$ is a solution to 
    \begin{equation}
        \frac{3(3\alpha - \rho^* - \alpha\ln(3\alpha)  + \alpha\ln \rho^*)}{4\alpha} + \frac{1 - 3\alpha + (12\alpha - 3)\ln(3\alpha )}{2(1 - 3\alpha)} = 1, \label{rhoalphaeq}
    \end{equation}
    for any constant $0 < c < \rho^* - \alpha$ it holds for $\rho := \alpha + c$ that $\frac{3(3\alpha - \rho - \alpha\ln(3\alpha)  + \alpha\ln \rho)}{4\alpha} + \frac{1 - 3\alpha + (12\alpha - 3)\ln(3\alpha )}{2(1 - 3\alpha)}$ is strictly greater than $1 + \const$.

    Note that, unlike in \cref{diffgammaproof} and $\alpha \leq 3/11$, when $\alpha \geq 3/11$ there is an upper bound on the values of $\alpha$ for which there would exist such a solution $\rho^* \geq \alpha$ to \cref{rhoalphaeq}.
    Specifically, one could increase the value of $\alpha$ until the only viable solution $\rho^* \geq \alpha$ is exactly $\rho^* = \alpha$.
    This largest possible $\alpha^*$ then must satisfy
    \begin{multline*}
        \frac{3(3\alpha^* - \alpha^* - \alpha^*\ln(3\alpha^*)  + \alpha^*\ln \alpha^*)}{4\alpha^*} + \frac{1 - 3\alpha^* + (12\alpha^* - 3)\ln(3\alpha^* )}{2(1 - 3\alpha^*)} = 1 \\
        \iff \frac{3(2\alpha^*  - \alpha^*\ln3)}{4\alpha^*} + \frac{1}{2} + \frac{(12\alpha^* - 3)\ln(3\alpha^* )}{2(1 - 3\alpha^*)} = 1\\
        \iff2(12\alpha^* - 3)\ln(3\alpha^* ) = (1 - 3\alpha^*)(3\ln 3 - 4),
    \end{multline*}
    which gives us maximum possible $\rho^* = \alpha^* \approx  0.2767738$.
    Thus, if $3/11 \leq \alpha < \rho^*$ and $x > 0$ is such that $x < \rho^* - \alpha$, it holds for $\rho = \alpha + x$ that $\frac{3(3\alpha - \rho - \alpha\ln(3\alpha)  + \alpha\ln \rho)}{4\alpha} + \frac{1 - 3\alpha + (12\alpha - 3)\ln(3\alpha )}{2(1 - 3\alpha)}$ is a constant greater than $1$.

    At the same time, when $n_0 = n$ the righthand side of \cref{bigalphaineqbetter} in \cref{lowergamma2better} is at least $1$.
    Since $\rho < 1$, both $\rho d$ and $d(1 - \rho)$ are of order $\Theta(d)$, and taking any integer $d = \Theta(\sqrt{n})$ in 
    \[
        \frac{d(1 - \rho)}{n}
         + \frac{3(3 - \rho  + 1/d)}{2\rho  ( \rho d + 1)} +\frac{3(1 -  (3\alpha + \eps))}{2 (3\alpha + \eps)((3\alpha + \eps) d + 1)} 
        +\frac{3(1 / d + 4\alpha)}{4(3\alpha + \eps)\alpha d} +  \frac{4 -\delta - \eps - 12\alpha}{2d(1 - 3\alpha)}
    \]
    makes the expression into a decreasing function of $n$, specifically of order $\Theta(n^{- 1/2})$.
    Taking this value of $d$ with $\eps, \delta \xrightarrow{n\to\infty}0$ when $n_0 = n$ makes the righthand side of \cref{bigalphaineqbetter} converge to $1$ as $n \to \infty$.
    Thus, for inequality to be satisfied for all $n_0 \leq n$ it suffices to take $n_x$ large enough such that the righthand side of \cref{bigalphaineqbetter} with $n \geq n_x$ is at most $\frac{3(3\alpha - \rho - \alpha\ln(3\alpha)  + \alpha\ln \rho)}{4\alpha} + \frac{1 - 3\alpha + (12\alpha - 3)\ln(3\alpha )}{2(1 - 3\alpha)}$ with $\rho = \alpha + x < \rho^*$.
    Since the latter expression is strictly greater than $1$, such $n_x$ must exist.
\end{proof}

\section{Doubling}

\label{sec:doubling}

First, we prove the doubling lemma for the case $\alpha \leq 3/11$.
\begin{lemma}\label{doublingineqproof}
    Assume that $\alpha \in [1/4, 3/11]$ and $\tau \in [0, \alpha)$ is constant, and consider the following process $\gamma_n^k$: $\gamma_n^1 = 1$, and for all $i \geq 1$,
    \[
        \gamma_n^{i + 1} = \left(1 - \frac{1 - \alpha}{2(\gamma_n^k - \tau - \alpha)n}\right)\gamma_n^k.
    \]
    Suppose that for some $n$ and $k, l \geq 1$ it holds $\gamma_n^k \geq \beta = \gamma_{2n}^l$ for some value $\beta$.
    If
    \[
        \beta > \alpha + \tau + \frac{1 - \alpha + \sqrt{16n(\alpha + \tau)(1 - \alpha) + (1 - \alpha)^2}}{8n},
    \]
    then $\gamma_{2n}^{l + 2} \leq \gamma_n^{k + 1}$.
\end{lemma}
\begin{proof}
    By definition of the process $\gamma_n^k$, if $\gamma_{2n}^{j} = \beta$ then
    \[
        \gamma_{2n}^{l + 1} = \left(1 - \frac{1 - \alpha}{4(\gamma_{2n}^{l} - \tau - \alpha)n}\right)\gamma_{2n}^{l} = \left(1 - \frac{1 - \alpha}{4(\beta  - \tau- \alpha)n}\right)\beta,
    \]
    and
    \begin{multline*}
        \gamma_{2n}^{l + 2} = \left(1 - \frac{1 - \alpha}{4(\gamma_{2n}^{l + 1} - \tau - \alpha)n}\right)\gamma_{2n}^{l + 1} \\
        =  \left(1 - \frac{1 - \alpha}{4\left(\left(1 - \frac{1 - \alpha}{4(\beta - \tau - \alpha)n}\right)\beta  - \tau- \alpha\right)n}\right) \left(1 - \frac{1 - \alpha}{4(\beta - \tau - \alpha)n}\right)\beta.
    \end{multline*}
    To prove the lemma, it remains to determine for which values of $\alpha, \beta$ the expression above is at most $\gamma_n^{k + 1}$.
    Since $\gamma_n^k \geq \beta$, we have by definition
    \[
        \gamma_n^{k + 1} = \left(1 - \frac{1 - \alpha}{2(\gamma_n^k - \tau- \alpha)n}\right)\gamma_n^k \geq \left(1 - \frac{1 - \alpha}{2(\beta- \tau - \alpha)n}\right)\beta.
    \]
    Thus, to prove $\gamma_{2n}^{l +2 } \leq \gamma_{n}^{k + 1}$, it suffices to show that
    \[
        \left(1 - \frac{1 - \alpha}{4\left(\left(1 - \frac{1 - \alpha}{4(\beta - \tau- \alpha)n}\right)\beta- \tau - \alpha\right)n}\right)\ \left(1 - \frac{1 - \alpha}{4(\beta- \tau - \alpha)n}\right) \leq \left(1 - \frac{1 - \alpha}{2(\beta - \tau- \alpha)n}\right).
    \]
    Opening up the parenthesis, we get
    \begin{multline*}
        1 - \frac{1 - \alpha}{4(\beta- \tau - \alpha)n} -\frac{1 - \alpha}{4\left(\left(1 - \frac{1 - \alpha}{4(\beta- \tau - \alpha)n}\right)\beta - \alpha\right)n} \\+ \frac{(1 - \alpha)^2}{16\left(\left(1 - \frac{1 - \alpha}{4(\beta- \tau - \alpha)n}\right)\beta- \tau - \alpha\right)(\beta- \tau - \alpha)n^2} \leq 1 - \frac{1 - \alpha}{2(\beta- \tau - \alpha)n}.
    \end{multline*}
    Removing $1$ from each side, diving by $1-\alpha$ and multiplying by $4(\beta - \tau - \alpha)n$ gives
    \[
         - 1 -\frac{\beta - \tau - \alpha}{\left(1 - \frac{1 - \alpha}{4(\beta- \tau  - \alpha)n}\right)\beta- \tau  - \alpha} + \frac{1 - \alpha}{4\left(\left(1 - \frac{1 - \alpha}{4(\beta - \tau - \alpha)n}\right)\beta - \tau - \alpha\right)n} \leq  - 2.
    \]
    Adding $1$ to each side and multiplying by $\left(1 - \frac{1 - \alpha}{4(\beta - \tau - \alpha)n}\right)\beta - \tau - \alpha$ (assuming it is positive) gives
    \begin{multline*}
        -(\beta - \tau - \alpha) + \frac{1 - \alpha}{4n} \leq - \left(\left(1 - \frac{1 - \alpha}{4(\beta- \tau  - \alpha)n}\right)\beta - \tau - \alpha\right) \\
        \iff  -\beta + \tau + \alpha + \frac{1 - \alpha}{4n} \leq - \beta + \frac{\beta(1 - \alpha)}{4(\beta- \tau  - \alpha)n} + \tau + \alpha\\
        \iff \frac{1 - \alpha}{4n} \leq \frac{\beta(1 - \alpha)}{4(\beta- \tau  - \alpha)n} \iff \beta - \tau - \alpha \leq \beta.
    \end{multline*}
    The last inequality always holds.
    As mentioned, this is an equivalent transformation only if
    \begin{multline*}
        \left(1 - \frac{1 - \alpha}{4(\beta- \tau  - \alpha)n}\right)\beta- \tau  - \alpha > 0 \iff \beta- \tau  - \alpha > \frac{\beta(1- \alpha)}{4(\beta- \tau  - \alpha)n} \\\iff 4(\beta - \tau - \alpha)^2n > \beta(1 - \alpha).
    \end{multline*}
    Solving the equation for $\beta$, we have
    \begin{multline*}
        4n\beta^2 - 8n(\alpha + \tau)\beta + 4n(\alpha + \tau)^2 > \beta(1 - \alpha) \\\iff 4n\beta^2 - (8n(\alpha + \tau) + 1 - \alpha)\beta + 4n(\alpha + \tau)^2 > 0\\
        \iff \beta > \frac{8n(\alpha + \tau) + 1 - \alpha + \sqrt{(8n(\alpha + \tau) + 1 - \alpha)^2 - 64n^2(\alpha + \tau)^2}}{8n} \\
        = \alpha + \tau + \frac{1 - \alpha + \sqrt{16n(\alpha + \tau)(1 - \alpha) + (1 - \alpha)^2}}{8n}.
    \end{multline*}
\end{proof}

The proof for the case $\alpha \geq 3/11$ is completely analogous, the only major difference is that we have to consider cases of whether $\gamma_n^k, \gamma_{2n}^l, \gamma_{2n}^{l + 1}$ are above or below $3\alpha$, which makes the proof quite longer (though technically as simple).
\begin{lemma}\label{doublingproof2}
    Let $\rho$ satisfy $2(12\rho -3)\ln(3\rho) = (1 - 3\rho)(3\ln 3 - 4)$, and assume $\alpha \in [3/11, \rho)$ and $\tau \in [0, \alpha)$ is constant.
    Consider the following process $\gamma_n^k$: $\gamma_n^1 = 1$, and for all $k \geq 1$,
    \begin{gather*}
        \text{if }\; \gamma_n^k - \tau \geq 3\alpha,\;\text{ then }\; \gamma_n^{k + 1} = \left(1 - \frac{2(1 - 3\alpha)}{(\gamma_n^k -\tau - 12\alpha + 3)n}\right)\gamma_n^k;\\
        \text{if }\; \gamma_n^k - \tau < 3\alpha,\; \text{ then }\; \gamma_n^{k + 1} = \left(1 - \frac{4\alpha}{3(\gamma_n^k -\tau- \alpha)n}\right)\gamma_n^k.
    \end{gather*}
    Suppose that for some $n$ and $k, l \geq 1$ it holds $\gamma_n^k \geq \beta=\gamma_{2n}^l$, for some value $\beta$.
    If
    \[
        \beta >  \alpha + \tau +  \frac{\alpha + \sqrt{6(\alpha + \tau)\alpha n + \alpha^2}}{3n},
    \]
    then $\gamma_{2n}^{l + 2} \leq \gamma_n^{k + 1}$.
\end{lemma}
\begin{proof}
    There are four cases to consider, depending on whether $\gamma_n^k, \gamma_{2n}^l, \gamma_{2n}^{l + 1}$ are above or below $3\alpha + \tau$.
    First case is $\gamma_n^k, \gamma_{2n}^l, \gamma_{2n}^{l + 1} \geq 3\alpha+ \tau$, note that we also have $\beta - \tau \geq 3\alpha$.
    Then, by definition
    \[
        \gamma_{2n}^{l + 1} = \left(1 - \frac{2(1 - 3\alpha)}{2(\gamma_{2n}^l -\tau- 12\alpha + 3)n}\right)\gamma_{2n}^l =\left(1 - \frac{1 - 3\alpha}{(\beta-\tau - 12\alpha + 3)n}\right)\beta, 
    \]
    and
    \begin{multline*}
        \gamma_{2n}^{l + 2} = \left(1 - \frac{2(1 - 3\alpha)}{2(\gamma_{2n}^{l + 1} -\tau- 12\alpha + 3)n}\right)\gamma_{2n}^{l + 1} \\= \left(1 - \frac{1 - 3\alpha}{\left(\left(1 - \frac{1 - 3\alpha}{(\beta -\tau- 12\alpha + 3)n}\right)\beta -\tau- 12\alpha + 3\right)n}\right)\left(1 - \frac{1 - 3\alpha}{(\beta -\tau- 12\alpha + 3)n}\right)\beta.
    \end{multline*}
    Once again, we would like to determine for which values of $\alpha, \beta$ the expression above is at most $\gamma_n^{k + 1}$.
    By definition,
    \[
        \gamma_n^{k + 1} =  \left(1 - \frac{2(1 - 3\alpha)}{(\gamma_{n}^{k}-\tau - 12\alpha + 3)n}\right)\gamma_{n}^{k} \geq \left(1 - \frac{2(1 - 3\alpha)}{(\beta -\tau- 12\alpha + 3)n}\right)\beta.
    \]
    Thus, to prove $\gamma_{2n}^{l + 2} \leq \gamma_n^{k + 1}$, it suffices to show that
    \begin{multline*}
        \left(1 - \frac{1 - 3\alpha}{\left(\left(1 - \frac{1 - 3\alpha}{(\beta -\tau- 12\alpha + 3)n}\right)\beta -\tau- 12\alpha + 3\right)n}\right)\left(1 - \frac{1 - 3\alpha}{(\beta -\tau- 12\alpha + 3)n}\right) \\\leq \left(1 - \frac{2(1 - 3\alpha)}{(\beta-\tau - 12\alpha + 3)n}\right).
    \end{multline*}
    Opening up the parenthesis,
    \begin{multline*}
        1 - \frac{1 - 3\alpha}{(\beta -\tau- 12\alpha + 3)n} - \frac{1 - 3\alpha}{\left(\left(1 - \frac{1 - 3\alpha}{(\beta-\tau - 12\alpha + 3)n}\right)\beta -\tau- 12\alpha + 3\right)n} \\ + \frac{(1 - 3\alpha)^2}{\left(\left(1 - \frac{1 - 3\alpha}{(\beta-\tau - 12\alpha + 3)n}\right)\beta -\tau- 12\alpha + 3\right)(\beta -\tau- 12\alpha + 3)n^2} \leq 1 - \frac{2(1 - 3\alpha)}{(\beta -\tau- 12\alpha + 3)n}.
    \end{multline*}
    Removing $1$ from each side, dividing by $1 - 3\alpha$ and multiplying by $(\beta-\tau - 12\alpha + 3)n$ gives
    \[
         -1- \frac{\beta-\tau - 12\alpha + 3}{\left(1 - \frac{1 - 3\alpha}{(\beta-\tau - 12\alpha + 3)n}\right)\beta-\tau - 12\alpha + 3} + \frac{1 - 3\alpha}{\left(\left(1 - \frac{1 - 3\alpha}{(\beta-\tau - 12\alpha + 3)n}\right)\beta-\tau - 12\alpha + 3\right)n} \leq  - 2.
    \]
    Adding $1$ to each side and multiplying by $\left(1 - \frac{1 - 3\alpha}{(\beta-\tau - 12\alpha + 3)n}\right)\beta-\tau - 12\alpha + 3$ (assuming it is positive):
    \begin{multline*}
        -(\beta -\tau- 12\alpha + 3) + \frac{1 - 3\alpha}{n} \leq  - \left(\left(1 - \frac{1 - 3\alpha}{(\beta -\tau- 12\alpha + 3)n}\right)\beta -\tau- 12\alpha + 3\right)\\
        \iff - \beta + \tau + 12\alpha - 3 + \frac{1 - 3\alpha}{n} \leq -\beta + \frac{\beta(1 - 3\alpha)}{(\beta-\tau - 12\alpha + 3)n} +\tau+ 12\alpha - 3\\
        \iff \frac{1 - 3\alpha}{n} \leq \frac{\beta(1 - 3\alpha)}{(\beta -\tau- 12\alpha + 3)n} \iff \beta -\tau- 12\alpha + 3 \leq \beta \iff \alpha \geq 1/4.
    \end{multline*}
    The last inequality always holds.
    As mentioned, this is an equivalent transformation only if 
    \begin{multline*}
        \left(1 - \frac{1 - 3\alpha}{(\beta -\tau- 12\alpha + 3)n}\right)\beta -\tau- 12\alpha + 3 > 0\\ \iff \beta-\tau - 12\alpha + 3 > \frac{\beta(1 - 3\alpha)}{(\beta -\tau- 12\alpha + 3)n} \iff (\beta-\tau - 12\alpha + 3)^2n > \beta(1 - 3\alpha).
    \end{multline*}
    Note that by assumption, $\beta - \tau \geq 3\alpha$.
    Then, $(\beta -\tau- 12\alpha + 3)^2n \geq 9(1 - 3\alpha)^2n$, and the inequality above holds if $9(1 - 3\alpha)n > \beta$.
    Since $\alpha < \rho < 3/10$, $9(1 - 3\alpha)n > 9n/10 > 1 \geq \beta$ for all $n \geq 2$.
    So, the expression above is always positive for our ranges of $\alpha$ and $\beta$.

    Second case is $\gamma_n^k, \gamma_{2n}^l \geq 3\alpha + \tau$ while $\gamma_{2n}^{l + 1} < 3\alpha + \tau$, note that still $\beta \geq 3\alpha + \tau$.
    By definition,
    \[
        \gamma_{2n}^{l + 1} = \left(1 - \frac{2(1 - 3\alpha)}{2(\gamma_{2n}^l -\tau- 12\alpha + 3)n}\right)\gamma_{2n}^l =\left(1 - \frac{1 - 3\alpha}{(\beta-\tau - 12\alpha + 3)n}\right)\beta, 
    \]
    and,
    \begin{multline*}
        \gamma_{2n}^{l + 2} = \left(1 - \frac{4\alpha}{6(\gamma_{2n}^{l + 1}-\tau -\alpha)n}\right)\gamma_{2n}^{l + 1} \\= \left(1 - \frac{2\alpha}{3\left(\left(1 - \frac{1 - 3\alpha}{(\beta -\tau- 12\alpha + 3)n}\right)\beta -\tau- \alpha\right)n}\right)\left(1 - \frac{1 - 3\alpha}{(\beta -\tau- 12\alpha + 3)n}\right)\beta,
    \end{multline*}
    as well as
    \[
        \gamma_n^{k + 1} =  \left(1 - \frac{2(1 - 3\alpha)}{(\gamma_{n}^{k}-\tau - 12\alpha + 3)n}\right)\gamma_{n}^{k} \geq \left(1 - \frac{2(1 - 3\alpha)}{(\beta -\tau- 12\alpha + 3)n}\right)\beta.
    \]
    Thus, to prove $\gamma_{2n}^{l + 2} \leq \gamma_n^{k + 1}$, it suffices to show that
    \begin{multline*}
        \left(1 - \frac{2\alpha}{3\left(\left(1 - \frac{1 - 3\alpha}{(\beta-\tau - 12\alpha + 3)n}\right)\beta-\tau - \alpha\right)n}\right)\left(1 - \frac{1 - 3\alpha}{(\beta -\tau- 12\alpha + 3)n}\right) \\\leq \left(1 - \frac{2(1 - 3\alpha)}{(\beta -\tau- 12\alpha + 3)n}\right).
    \end{multline*}
    Opening up the parenthesis and combining like terms,
    \begin{multline*}
        - \frac{2\alpha}{3\left(\left(1 - \frac{1 - 3\alpha}{(\beta -\tau- 12\alpha + 3)n}\right)\beta-\tau - \alpha\right)n} \\+ \frac{2\alpha(1 - 3\alpha)}{3\left(\left(1 - \frac{1 - 3\alpha}{(\beta -\tau- 12\alpha + 3)n}\right)\beta-\tau - \alpha\right)(\beta-\tau- 12\alpha + 3)n^2} \leq - \frac{1 - 3\alpha}{(\beta-\tau - 12\alpha + 3)n}.
    \end{multline*}
    Multiplying everything by $3(\beta-\tau - 12\alpha + 3)n$, we have
    \[
        - \frac{2\alpha(\beta -\tau- 12\alpha + 3)}{\left(1 - \frac{1 - 3\alpha}{(\beta -\tau- 12\alpha + 3)n}\right)\beta -\tau- \alpha} + \frac{2\alpha(1 - 3\alpha)}{\left(\left(1 - \frac{1 - 3\alpha}{(\beta -\tau- 12\alpha + 3)n}\right)\beta -\tau- \alpha\right)n} \leq - 3(1 - 3\alpha).
    \]
    Multiplying by $\left(1 - \frac{1 - 3\alpha}{(\beta-\tau - 12\alpha + 3)n}\right)\beta-\tau - \alpha$ (assuming it is positive):
    \[
        - 2\alpha(\beta -\tau- 12\alpha + 3) + \frac{2\alpha(1 - 3\alpha)}{n} \leq - 3(1 - 3\alpha)\left(\left(1 - \frac{1 - 3\alpha}{(\beta -\tau- 12\alpha + 3)n}\right)\beta-\tau - \alpha\right),
    \]
    which is equivalent to
    \[
        -2\alpha(\beta - \tau) -6\alpha(1 - 4\alpha) + \frac{2\alpha(1 - 3\alpha)}{n} \leq -3(1-3\alpha)(\beta - \tau) + \frac{3\beta(1 - 3\alpha)^2}{(\beta -\tau- 12\alpha + 3)n} + 3\alpha(1 - 3\alpha),
    \]
    which holds if and only if 
    \[
        - 3\alpha(3 - 11\alpha) + \frac{2\alpha(1 - 3\alpha)}{n} \leq (\beta- \tau)(11\alpha - 3) +   \frac{3\beta(1 - 3\alpha)^2}{(\beta - \tau- 12\alpha + 3)n}.
    \]
    Since $1 \geq \beta - \tau \geq 3\alpha$ and $\alpha \geq 3/11$, it holds 
    \begin{multline*}
        (\beta- \tau)(11\alpha - 3) +   \frac{3\beta(1 - 3\alpha)^2}{(\beta - \tau- 12\alpha + 3)n} \\
        \geq 3\alpha(11\alpha - 3) + \frac{9\alpha(1 - 3\alpha)^2}{(1 - 12\alpha + 3)n} = 3\alpha(11\alpha - 3) + \frac{9\alpha(1 - 3\alpha)}{4n},
    \end{multline*}
    so it suffices to show that
    \begin{multline*}
        - 3\alpha(3 - 11\alpha) + \frac{2\alpha(1 - 3\alpha)}{n} \leq 3\alpha(11\alpha - 3) + \frac{9\alpha(1 - 3\alpha)}{4n}\\
        \iff  \frac{2\alpha(1 - 3\alpha)}{n}\leq  \frac{9\alpha(1 - 3\alpha)}{4n}\iff \left(\frac{9}{4} - 2\right)(1- 3\alpha) \geq 0,
    \end{multline*}
    which is always true for $\alpha < 1/3$.
    So, the inequality holds as long as
    \begin{multline*}
        \left(1 - \frac{1 - 3\alpha}{(\beta- \tau - 12\alpha + 3)n}\right)\beta - \tau- \alpha > 0\\
        \iff \beta - \tau- \alpha > \frac{\beta(1 - 3\alpha)}{(\beta - \tau- 12\alpha + 3)n} \\\iff (\beta- \tau - \alpha)(\beta - \tau- 12\alpha + 3)n > \beta(1 - 3\alpha).
    \end{multline*}
    Once again, since $\beta-\tau \geq 3\alpha$, we have $(\beta - \tau- \alpha)(\beta - \tau- 12\alpha + 3)n \geq 6\alpha(1 - 3\alpha)n$, and the inequality above holds if $6\alpha n > \beta$.
    As $\alpha \geq 3/11$, inequality holds if $\frac{18}{11}n > \beta$, which is true for $n \geq 1$.

    The third case is $\gamma_n^k \geq 3\alpha + \tau$, while $\gamma_{2n}^l, \gamma_{2n}^{l + 1} < 3\alpha + \tau$, and in this case $\beta < 3\alpha + \tau$.
    Observe that when $\alpha \geq 3/11$, for any $x \geq 3\alpha$ we have
    \[
        \frac{2(1 - 3\alpha)}{(x - 12\alpha + 3)n} \leq \frac{4\alpha}{3(x - \alpha)n} \iff \left(1 - \frac{2(1 - 3\alpha)}{(x - 12\alpha + 3)n} \right) \geq \left(1 - \frac{4\alpha}{3(x - \alpha)n}\right).
    \]
    Then, by definition
    \[
        \gamma_n^{k + 1} =  \left(1 - \frac{2(1 - 3\alpha)}{(\gamma_{n}^{k} - \tau- 12\alpha + 3)n}\right)\gamma_{n}^{k} \geq\left(1 - \frac{4\alpha}{3(\gamma_{n}^{k} - \tau- \alpha)n}\right)\gamma_{n}^{k},
    \]
    which would be a formula for $\gamma_{n}^{k + 1}$ if $\gamma_n^k < 3\alpha + \tau$.
    So, we can reduce the case of  $\gamma_n^k \geq 3\alpha + \tau$ and $\gamma_{2n}^l, \gamma_{2n}^{l + 1} < 3\alpha + \tau$ to the final case, where all $\gamma_n^k, \gamma_{2n}^l, \gamma_{2n}^{l + 1} < 3\alpha + \tau$.
    Once again, by definition
    \[
        \gamma_{2n}^{l + 1} = \left(1 - \frac{4\alpha}{6(\gamma_{2n}^l- \tau - \alpha)n}\right)\gamma_{2n}^l = \left(1 - \frac{2\alpha}{3(\beta - \tau- \alpha)n}\right)\beta,
    \]
    and
    \begin{multline*}
        \gamma_{2n}^{l + 2} = \left(1 - \frac{4\alpha}{6(\gamma_{2n}^{l + 1}- \tau - \alpha)n}\right)\gamma_{2n}^{l + 1}\\
        = \left(1 - \frac{2\alpha}{3\left(\left(1 - \frac{2\alpha}{3(\beta- \tau - \alpha)n}\right)\beta- \tau - \alpha\right)n}\right)\left(1 - \frac{2\alpha}{3(\beta - \tau- \alpha)n}\right)\beta,
    \end{multline*}
    as well as
    \[
        \gamma_n^{k + 1} =  \left(1 - \frac{4\alpha}{3(\gamma_{n}^{k} - \tau- \alpha)n}\right)\gamma_{n}^{k} \geq \left(1 - \frac{4\alpha}{3(\beta- \tau - \alpha)n}\right)\beta.
    \]
    Thus, to prove $\gamma_{2n}^{l + 2} \leq \gamma_n^{k + 1}$, it suffices to show that
    \[
        \left(1 - \frac{2\alpha}{3\left(\left(1 - \frac{2\alpha}{3(\beta- \tau - \alpha)n}\right)\beta- \tau - \alpha\right)n}\right)\left(1 - \frac{2\alpha}{3(\beta - \tau- \alpha)n}\right) \leq \left(1 - \frac{4\alpha}{3(\beta- \tau - \alpha)n}\right).
    \]
    Opening up the parenthesis and combining like terms,
    \begin{multline*}
        - \frac{2\alpha}{3\left(\left(1 - \frac{2\alpha}{3(\beta - \tau- \alpha)n}\right)\beta- \tau - \alpha\right)n} \\+ \frac{4\alpha^2}{9\left(\left(1 - \frac{2\alpha}{3(\beta- \tau - \alpha)n}\right)\beta - \tau- \alpha\right)(\beta - \tau- \alpha)n^2} \leq - \frac{2\alpha}{3(\beta - \tau- \alpha)n}.
    \end{multline*}
    Diving everything by $2\alpha$ and multiplying $3(\beta - \tau- \alpha)n$, we get
    \[
        - \frac{\beta - \tau- \alpha}{\left(1 - \frac{2\alpha}{3(\beta- \tau - \alpha)n}\right)\beta - \tau- \alpha} + \frac{2\alpha}{3\left(\left(1 - \frac{2\alpha}{3(\beta- \tau - \alpha)n}\right)\beta - \tau- \alpha\right)n} \leq -1.
    \]
    Multiplying by $\left(1 - \frac{2\alpha}{3(\beta - \tau- \alpha)n}\right)\beta - \tau- \alpha$ (assuming it is positive):
    \begin{multline*}
        - (\beta - \tau- \alpha) + \frac{2\alpha}{3n} \leq -\left(\left(1 - \frac{2\alpha}{3(\beta- \tau - \alpha)n}\right)\beta- \tau - \alpha\right)
        \\\iff  \frac{2\alpha}{3n} \leq  \frac{2\alpha \beta }{3(\beta - \tau- \alpha)n} \iff \beta- \tau - \alpha \leq \beta.
    \end{multline*}
    The last inequality always holds.
    As mentioned, this is an equivalent transformation only if
    \[
        \left(1 - \frac{2\alpha}{3(\beta - \tau- \alpha)n}\right)\beta - \tau- \alpha > 0 \iff \beta - \tau- \alpha > \frac{2\alpha\beta}{3(\beta - \tau- \alpha)n} \iff 3(\beta - \tau- \alpha)^2n > 2\alpha \beta.
    \]
    Solving the equation for $\beta$, we have
    \begin{multline*}
        3n\beta^2 - 6n(\alpha + \tau)\beta + 3n(\alpha + \tau)^2 > 2\alpha \beta \iff 3n\beta^2 - 2((\alpha+\tau)3n + \alpha)\beta + 3n(\alpha + \tau)^2 > 0\\
        \iff \beta > \frac{2((\alpha + \tau)3n + \alpha) + \sqrt{4((\alpha + \tau)3n + \alpha)^2 - 36n^2(\alpha + \tau)^2}}{6n}\\
        = \alpha + \tau  + \frac{2\alpha + \sqrt{4(6(\alpha + \tau)\alpha n + \alpha^2)}}{6n} = \alpha + \tau +  \frac{\alpha + \sqrt{6(\alpha + \tau)\alpha n + \alpha^2}}{3n}.
    \end{multline*}
\end{proof}

\end{document}